\documentclass[leqno,a4paper,12pt]{article}

\usepackage{latexsym}
\usepackage{amsmath}
\usepackage{amsthm}
\usepackage{amssymb}
\usepackage{enumerate}
\usepackage{multicol}
\usepackage{graphicx}

\usepackage{pict2e}

\allowdisplaybreaks[4]

\theoremstyle{plain}
\newtheorem{theorem}{Theorem}[section]
\newtheorem{proposition}[theorem]{Proposition}
\newtheorem{lemma}[theorem]{Lemma}

\theoremstyle{definition}
\newtheorem{definition}[theorem]{Definition}

\newtheorem{remark}[theorem]{Remark}
\numberwithin{equation}{section}

\def\cB{{\mathcal{B}}}

\def\C{{\mathbb{C}}}
\def\cC{{\mathcal{C}}}

\def\be{{\mathbf{e}}}
\def\cE{{\mathcal{E}}}

\def\sF{{\mathsf{F}}}

\def\cH{{\mathcal{H}}}
\def\sH{{\mathsf{H}}}

\def\bk{{\mathbf{k}}}

\def\cM{{\mathcal{M}}}

\def\N{{\mathbb{N}}}

\def\cO{{\mathcal{O}}}

\def\cQ{{\mathcal{Q}}}

\def\R{{\mathbb{R}}}
\def\cR{{\mathcal{R}}}

\def\S{{\mathbb{S}}}
\def\cS{{\mathcal{S}}}

\def\sS{{\mathsf{S}}}

\def\bu{{\mathbf{u}}}

\def\cV{{\mathcal{V}}}
\def\sV{{\mathsf{V}}}
\def\bv{{\mathbf{v}}}
\def\hbv{{\hat{\mathbf{v}}}}

\def\bw{{\mathbf{w}}}

\def\bW{{\mathbf{W}}}
\def\cW{{\mathcal{W}}}

\def\bx{{\mathbf{x}}}
\def\bX{{\mathbf{X}}}

\def\hbx{{\hat{\mathbf{x}}}}

\def\by{{\mathbf{y}}}
\def\bY{{\mathbf{Y}}}

\def\hby{{\hat{\mathbf{y}}}}

\def\Z{{\mathbb{Z}}}

\def\bZ{{\mathbf{Z}}}

\def\opsi{{\overline{\psi}}}

\def\oeta{{\overline{\eta}}}

\def\otau{{\overline{\tau}}}

\def\spin{{\{\uparrow,\downarrow\}}}
\def\ua{{\uparrow}}
\def\da{{\downarrow}}

\def\la{{\lambda}}
\def\bla{{\mbox{\boldmath$\lambda$}}}

\def\o{{\omega}}

\def\eps{{\varepsilon}}

\def\g{{\gamma}}

\def\G{{\Gamma}}
\def\s{{\sigma}}

\def\hrho{\hat{\rho}}
\def\orho{\overline{\rho}}
\def\heta{\hat{\eta}}
\def\oeta{\overline{\eta}}
\def\D{{\Delta}}

\def\<{{\langle}}
\def\>{{\rangle}}

\def\liminf{\mathop{\mathrm{liminf}}}
\def\Tr{\mathop{\mathrm{Tr}}\nolimits}

\def\Mat{\mathop{\mathrm{Mat}}}
\def\Map{\mathop{\mathrm{Map}}}

\def\sgn{\mathop{\mathrm{sgn}}\nolimits}

\def\Im{\mathop{\mathrm{Im}}}
\def\Re{\mathop{\mathrm{Re}}}

\def\b0{{\mathbf{0}}}

\def\frah{{\left(\frac{1}{h}\right)}}

\def\0betah{{[0,\beta)_h}}

\begin{document}

\title{Superconducting phase in the BCS model with imaginary magnetic
field. III. \\
Non-vanishing free dispersion relations}

\author{Yohei Kashima
\medskip\\
Division of Mathematical Science,\\
Graduate School of Engineering Science, Osaka University,\\
Toyonaka, Osaka 560-8531, Japan. \\
E-mail: kashima@sigmath.es.osaka-u.ac.jp}

\date{}

\maketitle

\begin{abstract} We analyze a class of the BCS model, whose free dispersion
 relation is non-vanishing, under the influence of imaginary magnetic
 field at positive temperature. The magnitude of the negative coupling
 constant must be small but is allowed to be independent of the
 temperature and the imaginary magnetic field. The infinite-volume limit
 of the free energy density is characterized. A spontaneous symmetry
 breaking and an off-diagonal long range order are proved to occur only
 in high temperatures. This is because the gap equation in this model
 has a positive solution only if the temperature is higher than a
 critical value. The proof is based on a double-scale integration 
 of the Grassmann integral formulation. In this scheme we
integrate with the infrared covariance first and with the
 ultra-violet covariance afterwards, which is opposite to 
the previous schemes in [Y. Kashima, accepted for publication in
 J. Math. Sci. Univ. Tokyo, arXiv:1609.06121], 
[Y. Kashima, accepted for publication in
 J. Math. Sci. Univ. Tokyo,  arXiv:1709.06714] or \cite{K_BCS_I},
 \cite{K_BCS_II} 
in short.  As the other focus, we  study geometric properties of the 
phase boundaries, which are periodic copies of a closed curve in the
 two-dimensional space of the temperature and the real time variable. 
Here we adopt the real time variable in place of the temperature times 
the imaginary magnetic field by considering its relevance within 
contemporary physics of  
 dynamical phase transition at positive temperature. As the main result,
 we show that for any choice of a non-vanishing free dispersion relation
 the representative curve of the phase boundaries has only one local
 minimum point, or in other words the phase boundaries do not oscillate
 with temperature, if and only if the minimum of the magnitude of the
 free dispersion relation over the maximum is larger than the critical
 value $\sqrt{17-12\sqrt{2}}$. Overall we use the same notational
 conventions as in \cite{K_BCS_I}, \cite{K_BCS_II}. So this work is a
 continuation of these preceding papers.
\footnote{
\textit{keywords and phrases.} The BCS model, gap equation, phase boundaries,
spontaneous symmetry
 breaking, off-diagonal long range order, Grassmann integral formulation}
\end{abstract}

\tableofcontents

\section{Introduction}\label{sec_introduction}
 
\subsection{Introductory remarks}\label{subsec_introduction}

Since the proposal in 1957 (\cite{BCS}), the Bardeen-Cooper-Schrieffer
(BCS) model of interacting electrons has been considered as a primal
model to explain superconductivity from a microscopic principle. Apart
from the conventional reduction of its quartic Fermionic
interaction to a solvable quadratic one, we are still unable to make
explicit the thermodynamic limit of the BCS model for full set of
physical parameters. It is our longstanding desire to complete the
rigorous derivation of the thermodynamics and acquire fully
coherent applications of the BCS model.

It was shown in our previous works \cite{K_BCS_I}, \cite{K_BCS_II} that
the infinite-volume limit of the BCS model interacting with imaginary
magnetic field can be rigorously derived. The main difference between
these two constructions lies in properties of the free dispersion
relation. In \cite{K_BCS_I} we assumed the nearest-neighbor hopping and
tuned the chemical potential in a way that the free Fermi surface does
not degenerate. On the contrary, in \cite{K_BCS_II} we considered a
class of free dispersion relations which widely cover the ones with
degenerate but not empty free Fermi surface. Our mission here is to
achieve the same goal for non-vanishing free dispersion relations. We
characterize the infinite-volume limit of the free energy density and
the thermal expectation values of Cooper pair operators. The proof is
based on a multi-scale analysis of the Grassmann integral formulation. As an
illustration, let us summarize the applicability of the main theorems of
this series
to the typical free dispersion relation of nearest-neighbor hopping electron,
$e(\bk)=2\sum_{j=1}^d\cos k_j-\mu$ $:\R^d\to\R$, where $d$ $(\in \N)$ is
the spatial dimension and $\mu$ $(\in \R)$ is the chemical potential.
\begin{itemize}
\item \cite[\mbox{Theorem 1.3}]{K_BCS_I} applies to the case that $d$ is
      arbitrary and $|\mu|<2d$.
\item \cite[\mbox{Theorem 1.3}]{K_BCS_II} applies to the case that
      $d\in\{3,4\}$ and $|\mu|=2d$.
\item Theorem \ref{thm_infinite_volume_limit} of the present paper
      applies to the case that $d$ is arbitrary and $|\mu|>2d$. 
\end{itemize}
Qualitative properties of the free dispersion relation around its zero
points deeply affect the possible magnitude of interaction in this
approach. Therefore, characteristics of each paper of this series can be
explained in terms of dependency of the allowed magnitude of the
coupling constant on the temperature and the imaginary magnetic
field. In \cite{K_BCS_I} the magnitude of the coupling constant must be
smaller than some power of these parameters. Though the claimed
dependency is most complicated in this series, we can actually choose
the parameters so that they obey the necessary constraint and the gap
equation has a positive solution at the same time. In \cite{K_BCS_II} 
the magnitude of the coupling
constant can be largely independent of the temperature and the imaginary
magnetic field if the temperature is lower than a certain constant. As
the result, we were able to  prove phase transitions in arbitrarily small
temperatures for a fixed coupling constant. In this paper the magnitude
of negative coupling constant must be small but is independent of the 
temperature and the imaginary magnetic field. 
It turns out that the gap equation has a positive solution
only if the temperature is higher than a critical value. Accordingly,
the phase transitions characterized by spontaneous symmetry breaking (SSB)
and off-diagonal long range order (ODLRO) are proved to occur in the
high-temperature regions. 

The gapped property of the free dispersion relation is one essential
factor to make it possible to analyze the system independently of the
temperature and the imaginary magnetic field. However, a direct
combination of the non-vanishing free dispersion relation and the same
strategy as the core part of the multi-scale integrations of \cite{K_BCS_I},
\cite{K_BCS_II} does not lead to the desired result. We can see from the
constraints on the coupling constant \cite[\mbox{(1.2)}]{K_BCS_I}, 
\cite[\mbox{(1.18)}]{K_BCS_II} that the magnitude of the coupling
constant must be arbitrarily small in high temperatures for some
choices of the imaginary magnetic field in our previous constructions. 
The extra constraint in high
temperatures stems from a determinant bound on the covariance of the
last integration scale, which is tactically manipulated to be
independent of (imaginary) time variables. This constraint remains
regardless of the gapped property of the free dispersion relation as
long as we follow the same strategy as in \cite{K_BCS_I},
\cite{K_BCS_II}.

Let us explain this issue more by using formulas in a simple
way, as it also shows a novel aspect of the present construction. 
As usual, let $\beta$ $(\in\R_{>0})$ denote the inverse
temperature. Take an artificial parameter $h\in \frac{2}{\beta}\N$ and
set  
$$
[0,\beta)_h:=\left\{0,\frac{1}{h},\frac{2}{h},\cdots,\beta-\frac{1}{h}\right\},
$$
which is a discrete analogue of the interval $[0,\beta)$. For a finite
set $S$, which should be considered as a generalization of the product
set of the spatial lattice points and the orbital index, let $C:(S\times
[0,\beta)_h)^2\to\C$ denote the full covariance of the Grassmann
Gaussian integral formulation of our system. The main object to analyze
is the Grassmann Gaussian integral 
$$
\int e^{V(\psi)}d\mu_C(\psi)
$$ 
with a quartic Grassmann polynomial $V(\psi)$, which is as before a
correction term left after extracting the reference Grassmann polynomial. The full
covariance $C$ can be decomposed as follows. 
\begin{align}
C(Xs,Yt)=e^{i\frac{\pi}{\beta}(s-t)}(C_0(Xs,Yt)+C_1(Xs,Yt)),\ (\forall
 X,Y\in S,\ s,t\in [0,\beta)_h),\label{eq_formal_decomposition}
\end{align}
where the covariance $C_0:(S\times[0,\beta)_h)^2\to\C$ is in particular
independent of the time variables.
$$
C_0(Xs,Yt)=C_0(X0,Y0),\ (\forall
 X,Y\in S,\ s,t\in [0,\beta)_h).
$$
In essence the Matsubara frequency is fixed to be $\pi/\beta$
inside $C_0$ and $C_1$ sums over all the Matsubara frequencies but
$\pi/\beta$. 
Due to the gapped property of the free dispersion relation and
the partition of the Matsubara frequencies, the covariances $C_0$, $C_1$
satisfy the following bound properties. 
\begin{align*}
&|\det(C_0(X_is_i,Y_jt_j))_{1\le i,j\le n}|\le {c_{onst}}^n\beta^{-n},\\
&|\det(C_1(X_is_i,Y_jt_j))_{1\le i,j\le n}|\le {c_{onst}}^n,\\
&(\forall n\in \N,\ X_j,Y_j\in S,\ s_j,t_j\in [0,\beta)_h\
 (j=1,\cdots,n)),\\
&\sup_{(Y,t)\in S\times [0,\beta)_h}\Bigg(\frac{1}{h}\sum_{(X,s)\in
 S\times [0,\beta)_h}(|C_a(Xs,Yt)|+|C_a(Yt,Xs)|)\Bigg)\le {c_{onst}},\
 (\forall a\in \{0,1\}),
\end{align*}
where $c_{onst}$ $(\in \R_{>0})$ is independent of $\beta$ and the
imaginary magnetic field, though it may depend on other parameters
such as the spatial dimension or the minimum value of the magnitude of 
the free
dispersion relation. By \eqref{eq_formal_decomposition} and a gauge
invariance we can transform as follows. 
\begin{align*}
\int e^{V(\psi)}d\mu_C(\psi)&=\int\int
 e^{V(\psi^0+\psi^1)}d\mu_{C_0}(\psi^0)d\mu_{C_1}(\psi^1)\\
&=\int\int e^{V(\psi^0+\psi^1)}d\mu_{C_1}(\psi^1)d\mu_{C_0}(\psi^0).
\end{align*}
At this point we have two ways to proceed, either integrating with $C_0$
first or with $C_1$ first. Integrating with $C_1$ first is essentially
the same strategy as in the previous papers and the determinant bound on
$C_0$ remains to affect the possible magnitude of the coupling constant
at the end. This is the reason why the coupling constant needed to be
small even in high temperatures in \cite{K_BCS_I}, \cite{K_BCS_II}. We
can see from the $\beta$-dependent determinant bound on $C_0$ claimed
above that this is not the way to achieve our goal. Interestingly we
find that the determinant bound on $C_0$ does not affect the magnitude
of the coupling constant at all if we integrate with $C_0$ first and
make use of a vanishing property of the kernel function of
$V(\psi)$. Since the other bounds on $C_0$, $C_1$ listed above are
independent of $\beta$ and the imaginary magnetic field, this way leads
to the goal.

We can apply many of the general estimates established in
\cite{K_BCS_I}, \cite{K_BCS_II} and the Grassmann Gaussian integral
formulation stated in \cite{K_BCS_II} without any modification. At the
same time we need some modified versions of the previous general estimates in
order to implement the present double-scale integration scheme. However,
the modification can be done in a systematic way so that it does not
require a widespread reconstruction. Therefore, as far as it concerns
the general estimation of the Grassmann integration, the present
construction should not be longer than the previous ones. Moreover, the
conclusive part of the derivation of the infinite-volume limit after
building the general integration regime is essentially parallel to that
of the previous papers. 
Not to disappoint the readers later, we should clearly mention at this stage
that we will only explain which lemmas are
necessary to complete the proof of each claim of the theorem
in the final part of our construction (Subsection
\ref{subsec_infinite_volume}). 
On the other hand, estimation of the real covariance needs to be carefully
performed so that it does not yield any extra dependency on the
temperature and the imaginary magnetic field in the resulting theory. In
particular the determinant bound on the ultra-violet covariance $C_1$
requires a complicated application of the useful general determinant
bound by de Siqueira Pedra and Salmhofer \cite[\mbox{Theorem
1.3}]{PS}. The parts making up the proof of the derivation of the
infinite-volume limit are presented in the second half of the paper,
namely Section \ref{sec_derivation}.

As yet we cannot prove a superconducting order characterized by SSB and
ODLRO by this method in the BCS model without imaginary magnetic
field. In this approach we fail to take the coupling constant large
enough to ensure the solvability of the gap equation without the
imaginary magnetic field. The present class of free dispersion relations
includes the non-zero constant ones, with which the Hamiltonian is
called the strong coupling limit of the BCS model. We should remark that
a totally different method based on characterization of equilibrium
state on $C^*$-algebra applies to the strong coupling limit of the BCS
model and proves SSB and ODLRO (\cite{BP}). However, the method is not known to
be applicable to the BCS model with imaginary magnetic field, which is
not hermitian, at present. 

In the first half of the paper we analyze the free energy density, which
is made explicit by the theorem proved in the second half of the paper,
as a real-valued function of the temperature and the real time
variable. Here let us introduce the free energy density at a formal
level for illustrative purposes. The official definition will be given in
the next subsection. Let $\sH$, $\sS_z$ denote the BCS model Hamiltonian
and the $z$-component of the spin operator respectively. For $\theta\in
\R$ we consider the operator $\sH+i\theta\sS_z$ as the BCS model
interacting with the imaginary magnetic field. The infinite-volume limit
of the free energy density is the following. 
\begin{align*}
\lim_{L\to \infty\atop L\in\N}\left(-\frac{1}{\beta L^d}\log(\Tr
 e^{-\beta(\sH+i\theta \sS_z)})\right),
\end{align*}
where the parameter $L$ $(\in \N)$ controls the size of a $d$-dimensional
spatial lattice. By admitting the explicit form of the limit we study
regularity of the function
\begin{align}
(\beta,t)\mapsto \lim_{L\to \infty\atop L\in\N}\left(-\frac{1}{\beta L^d}\log(\Tr
 e^{-\beta\sH+i t \sS_z})\right):\R_{>0}\times
 \R\to\R\label{eq_formal_free_energy}
\end{align}
and geometric properties of the subset of $\R_{>0}\times\R$ where this
function loses analyticity. The reason why we study the free energy
density as a function of $(\beta,t)$ rather than $(\beta,\theta)$ is
that functions of the form 
\begin{align}
(\beta,t)\mapsto \lim_{L\to\infty}\frac{1}{L^d}\log\left(
\frac{\Tr e^{-\beta \sH+it\sS_z}}{\Tr e^{-\beta
 \sH}}\right):\R_{>0}\times\R\to \R
\label{eq_formal_free_energy_original}
\end{align}
 are becoming
relevant in contemporary physics of dynamical phase
transition (DPT) at positive temperature (\cite{BBD}, \cite{HB},
\cite{AK}, \cite{SFS}, \cite{Metal} and so on). In this context 
the function \eqref{eq_formal_free_energy_original} is seen as a finite-temperature
version of the infinite-volume limit of the overlap amplitude 
$$
\lim_{L\to \infty}\frac{1}{L^d}\log\<\psi_0,e^{it\sS_z}\psi_0\>,
$$
where $\psi_0$ is a ground state of $\sH$. Since the function 
$$
\beta\mapsto \lim_{L\to\infty}\frac{1}{L^d}\log(\Tr e^{-\beta \sH}):\R_{>0}\to\R
$$
is real analytic in the weak coupling regime of this paper (see
Proposition \ref{prop_phase_transition} \eqref{item_general_regularity}), the
regularity of the function \eqref{eq_formal_free_energy_original} is
equivalent to that of \eqref{eq_formal_free_energy}. 
In fact the concept of dynamical quantum phase transition at zero 
temperature has become a notable
topic of physics (\cite{HPK}, \cite{H}, \cite{Z}) and it recently reached a
state of experimental confirmation (see e.g. \cite{Jetal},
\cite{ZPHKBKGGM}, 
\cite{Fetal}). As the term indicates, non-analyticity with the
real time variable $t$ defines an occurrence of DPT both at zero
temperature and at positive temperature. DPTs at positive temperature
have been shown in quantum many-body systems which can be mapped to Fermionic
systems governed by quadratic Hamiltonians (see e.g. \cite{BBD},
\cite{HB}). 
To the author's
knowledge, no rigorous result of DPT in the BCS model at positive
temperature has been reported. In this situation we believe that
we should push forward mathematical analysis of the function
\eqref{eq_formal_free_energy} for possible future physical applications.

It is advantageous that with the present class of free dispersion
relations the characterization of the function
\eqref{eq_formal_free_energy} is justified for any $(\beta,t)\in
\R_{>0}\times\R$ as long as the coupling constant is fixed to be
small. The contents of the first half of this paper, which is Section
\ref{sec_free_energy} plus Appendices \ref{app_special_function},
\ref{app_integral}, are essentially independent of the second half. The
readers who want to complete the proof of the characterization of the
function \eqref{eq_formal_free_energy} can read the second half
first. We prove that the function is $C^1$-class in $\R_{>0}\times \R$
and the second order derivatives have jump discontinuity across a
one-dimensional submanifold of $\R_{>0}\times \R$
which we call phase boundaries. Then we
focus on describing geometric properties of the phase boundaries. We
find that the phase boundaries consist of periodic copies of one closed curve
(or more precisely periodic copies of the restriction of one closed curve
in $\R^2$ to $\R_{>0}\times \R$) and the representative curve is axially symmetric with
respect to the horizontal line $\{(\beta,2\pi)\ |\ \beta\in \R_{>0}\}$. Therefore,
letting $\beta_c$ denote the critical inverse temperature, 
the problem is reduced to an analysis of graph of a function on
$(0,\beta_c)$, which is
the lower half of the representative curve. In particular we focus on
determining when the function has only one local minimum point in
$(0,\beta_c)$, or in other words, when the representative curve of the
phase boundaries does not oscillate with temperature. It will turn out
that answers to this question can be expressed in terms of the ratio of
the maximum and the minimum of the magnitude of the free dispersion relation. The results
are summarized in Theorem \ref{thm_boundary_characterization} as the
second main result of this paper. 

Overall we keep using the same notational conventions as in
\cite{K_BCS_I}, \cite{K_BCS_II}. We will often refer the readers to
related parts of these papers for the meaning of notations rather
than restating them. We provide a supplementary short list of notations
which only contains new notations at the end of the paper. The readers
should refer to the comprehensive lists presented at the end of \cite{K_BCS_I},
\cite{K_BCS_II} for the other notations. 

This paper is organized as follows. In the next subsection we state the
theorem concerning the infinite-volume limit of the BCS model with
imaginary magnetic field at positive temperature and outline the main
results concerning the analysis of the free energy density and the phase
boundaries. In Section \ref{sec_free_energy} by admitting the explicit
form of the free energy density we study its regularity and geometric
properties of the phase boundaries. Moreover we analyze the phase
boundaries for a couple of specific examples of the free Hamiltonian. In
Section \ref{sec_derivation} we prove the theorem concerning the
infinite-volume limit in the constructive manner. In Appendix
\ref{app_special_function} we prepare a lemma which is used to study
the phase boundaries in Section \ref{sec_free_energy}. In
Appendix \ref{app_integral} we give a formula of a definite integral 
which we need to analyze a specific model in Sub-subsection 
\ref{subsubsec_one_band}.

\subsection{The main results}\label{subsec_main_results}

First let us state our main results on the derivation of the
infinite-volume limit of the free energy density and the thermal
expectation values. Let $d$ $(\in \N)$ denote the spatial dimension. Let 
$\{\bv_j\}_{j=1}^d$, $\{\hbv_j\}_{j=1}^d$ denote a basis of $\R^d$, its
dual basis respectively. They satisfy that
$$
\<\bv_i,\hbv_j\>=\delta_{i,j},\quad (\forall i,j\in \{1,\cdots,d\}),
$$
where $\<\cdot,\cdot\>$ is the canonical inner product of $\R^d$. With
$L\in \N$ the spatial lattice $\G$ and the momentum lattice $\G^*$
are defined by 
\begin{align*}
&\G:=\left\{\sum_{j=1}^dm_j\bv_j\ \Big|\ m_j\in \{0,1,\cdots,L-1\}\
 (j=1,\cdots,d)\right\},\\
&\G^*:=\left\{\sum_{j=1}^d\hat{m}_j\hbv_j\ \Big|\ \hat{m}_j\in
 \left\{0,\frac{2\pi}{L},\frac{4\pi}{L},\cdots,2\pi-\frac{2\pi}{L}\right\}\ (j=1,\cdots,d)\right\}.
\end{align*}
To formulate the infinite-volume limit of our interest, we use the infinite sets
$\G_{\infty}$, $\G_{\infty}^*$ defined by
\begin{align*}
&\G_{\infty}:=\left\{\sum_{j=1}^dm_j\bv_j\ \Big|\ m_j\in \Z\
 (j=1,\cdots,d)\right\},\\
&\G_{\infty}^*:=\left\{\sum_{j=1}^d\hat{k}_j\hbv_j\ \Big|\ \hat{k}_j\in
 [0,2\pi]\ (j=1,\cdots,d)\right\}.
\end{align*}

Let us define a set of matrix-valued functions, which are one-particle
Hamiltonians in momentum space. Using a function belonging to the set,
we will define the free part of our Hamiltonian. For $b\in\N$ and $e_{min}$,
$e_{max}\in\R_{>0}$ with $e_{min}\le e_{max}$ we define the subset
$\cE(e_{min},e_{max})$ of $\Map(\R^d,\Mat(b,\C))$ as follows. $E$
belongs to $\cE(e_{min},e_{max})$ if and only if 
\begin{align}
&E\in C^{\infty}(\R^d,\Mat(b,\C)),\notag\\
&E(\bk)=E(\bk)^*,\quad (\forall \bk\in
 \R^d),\label{eq_one_particle_hermitian}\\
&E(\bk+2\pi \hbv_j)=E(\bk),\quad (\forall \bk\in \R^d,\ j\in
 \{1,\cdots,d\}),\notag\\
&E(\bk)=\overline{E(-\bk)},\ (\forall
 \bk\in\R^d),\label{eq_time_reversal_symmetry}\\
&\inf_{\bk\in \R^d}\inf_{\bu\in\C^b\atop\text{with
 }\|\bu\|_{\C^b}=1}\|E(\bk)\bu\|_{\C^b}=e_{min}(>0),\label{eq_one_particle_lowest_band}\\
&\sup_{\bk\in\R^d}\|E(\bk)\|_{b\times b}=e_{max}.\notag
\end{align}
We remark that for $n\in\N$ $\Mat(n,\C)$ denotes the set of $n\times n$ complex
matrices, $\|\cdot\|_{n\times n}$ denotes the operator norm on
$\Mat(n,\C)$ and $\|\cdot\|_{\C^n}$ denotes the canonical norm of
$\C^n$. Set $\cB:=\{1,2,\cdots,b\}$. For $E\in \cE(e_{min},e_{max})$ we
define the free Hamiltonian $\sH_0$ as follows. 
\begin{align*}
\sH_0:=\frac{1}{L^d}\sum_{(\rho,\bx),(\eta,\by)\in\cB\times
 \G}\sum_{\s\in\spin}\sum_{\bk\in\G^*}e^{i\<\bx-\by,\bk\>}E(\bk)(\rho,\eta)\psi_{\rho\bx\s}^*\psi_{\eta\by\s},
\end{align*}
where $\psi_{\rho\bx \s}$ $(\psi_{\rho\bx \s}^*)$ denotes the Fermionic
annihilation (creation) operator for 
$(\rho,\bx,\s)$ $\in \cB\times
\G\times \spin$. 
It follows from \eqref{eq_one_particle_hermitian} that $\sH_0$ is a
self-adjoint operator on the Fermionic Fock space $F_f(L^2(\cB\times
\G\times\spin))$. With the negative coupling constant $U$ $(\in
\R_{<0})$ the interacting part $\sV$ is defined by
\begin{align*}
\sV:=\frac{U}{L^d}\sum_{(\rho,\bx),(\eta,\by)\in\cB\times \G}\psi_{\rho\bx\ua}^*\psi_{\rho\bx\da}^*\psi_{\eta\by\da}\psi_{\eta\by\ua}.
\end{align*}
The whole Hamiltonian $\sH$ is then defined by $\sH:=\sH_0+\sV$. As a
common purpose of this series, we study the infinite-volume limit of the
many-electron system governed by $\sH+i\theta \sS_z$ $(\theta\in\R)$,
where $\sS_z$ is the $z$-component of the spin operator defined by  
\begin{align*}
\sS_z:=\frac{1}{2}\sum_{(\rho,\bx)\in\cB\times
\G}(\psi_{\rho\bx\ua}^*\psi_{\rho\bx\ua}-
 \psi_{\rho\bx\da}^*\psi_{\rho\bx\da}).
\end{align*}
To describe SSB, we need the symmetry breaking external field operator $\sF$
defined by
\begin{align*}
\sF:=\g\sum_{(\rho,\bx)\in\cB\times
 \G}(\psi_{\rho\bx\ua}^*\psi_{\rho\bx\da}^*+\psi_{\rho\bx\da}\psi_{\rho\bx\ua}),\quad
 (\g\in\R).
\end{align*}

One essential difference from the previous works \cite{K_BCS_I},
\cite{K_BCS_II} is the solvability of the gap equation. Let us formulate the
gap equation and see when it is solvable. Take $E\in
\cE(e_{min},e_{max})$ and define the function $g_E:\R_{>0}\times\R\times
\R\to\R$ by
\begin{align*}
&g_E(x,t,z)\\
&:=-\frac{2}{|U|}+D_d\int_{\G_{\infty}^*}d\bk 
\Tr\left(\frac{\sinh(x\sqrt{E(\bk)^2+z^2})}{(\cos(t/2)+\cosh(x\sqrt{E(\bk)^2+z^2}))\sqrt{E(\bk)^2+z^2}}\right),
\end{align*}
where
$$
D_d:=|\det(\hbv_1,\cdots,\hbv_d)|^{-1}(2\pi)^{-d}.
$$
As in \cite{K_BCS_II}, throughout the paper we admit that for any
function $f:\R\backslash \{0\}\to\C$ and $E\in \cE(e_{min},e_{max})$ the
map $f(E(\cdot)):\R^d\to \Mat(b,\C)$ is defined via the spectral
decomposition of $E(\bk)$ for each $\bk\in\R^d$. We should remark that
because of the property \eqref{eq_one_particle_lowest_band}, $f(E(\bk))$
is well-defined for any $\bk\in\R^d$ even if $f(x)$ is not defined at
$x=0$. Our gap equation is to find $\D\in \R_{\ge 0}$ such that
$$ g_E(\beta,\beta\theta, \D)=0. $$
The following lemma can be proved by using the fact that for any
$\eps\in [-1,1]$ the function 
\begin{align}
x\mapsto \frac{\sinh x}{(\eps+\cosh
 x)x}:(0,\infty)\to\R\label{eq_reference_function}
\end{align}
is strictly monotone decreasing in the same way as in the proof of 
\cite[\mbox{Lemma 1.2}]{K_BCS_II}.

\begin{lemma}\label{lem_gap_equation}
The following statements hold for any $(\beta,\theta)\in\R_{>0}\times
 \R$. The equation $g_E(\beta,\beta\theta,\D)=0$ has a solution $\D$ in
 $[0,\infty)$ if and only if $g_E(\beta,\beta \theta,0)\ge 0$. Moreover,
 if a solution exists in $[0,\infty)$, it is unique. 
\end{lemma}

The next lemma tells us that if the interaction in the present model is
weak, there is a critical temperature such that the gap equation has
 a positive solution if and only if the temperature is higher than
the critical temperature. 

\begin{lemma}\label{lem_critical_temperature}
Assume that 
$$|U|<\frac{2e_{min}}{b}.$$
Then, there uniquely exists 
$$
\beta_c\in \left(0,\frac{2}{e_{min}}\tanh^{-1}\left(\frac{b |U|}{2e_{min}}\right)\right]
$$
such that the following statements hold.
\begin{enumerate}[(i)]
\item\label{item_critical_basic} 
     For any $\beta\in \R_{>0}$ $g_E(\beta,\pi,0)<0$.
\item\label{item_critical_positive} 
     For any $\beta\in (0,\beta_c)$ $g_E(\beta,2\pi,0)>0$ and thus there
     exists $\theta\in \R$ such
     that the gap equation $g_E(\beta,\beta\theta,\D)=0$ has a solution
     in $(0,\infty)$.
\item\label{item_critical_zero}
     $g_E(\beta_c,2\pi,0)=0$ and thus there exists $\theta\in \R$ such
     that $g_E(\beta_c,\beta_c\theta,\D)=0$ has the solution $\D=0$.
\item\label{item_critical_negative}
     For any $\beta\in (\beta_c,\infty)$ $g_E(\beta,2\pi,0)<0$ and thus
     for any $\theta\in\R$ the gap equation
     $g_E(\beta,\beta\theta,\D)=0$ has no solution in $[0,\infty)$.
\end{enumerate}
\end{lemma}

\begin{proof}
By the assumption, for any $\beta\in \R_{>0}$
\begin{align*}
g_E(\beta,\pi,0)=-\frac{2}{|U|}+D_d\int_{\G_{\infty}^*}d\bk \Tr
 \left(\frac{\tanh(\beta E(\bk))}{E(\bk)}\right)
\le -\frac{2}{|U|}+ \frac{b}{e_{min}}<0.
\end{align*}
Thus \eqref{item_critical_basic} holds. 

Observe that the function $\beta\mapsto g_E(\beta,2\pi,0):\R_{>0}\to\R$
 is monotone decreasing, 
\begin{align*}
&\lim_{\beta \searrow 0}g_E(\beta,2\pi,0)=\infty,\\
&\lim_{\beta \nearrow \infty}g_E(\beta,2\pi,0)
=-\frac{2}{|U|}+D_d\int_{\G_{\infty}^*}d\bk\Tr\left(\frac{1}{|E(\bk)|}\right)\le 
-\frac{2}{|U|}+\frac{b}{e_{min}}<0.
\end{align*}
Thus, there uniquely exists $\beta_c\in\R_{>0}$ such that 
\begin{align*}
&g_E(\beta,2\pi,0)>0,\quad (\forall \beta \in (0,\beta_c)),\\
&g_E(\beta_c,2\pi,0)=0,\\
&g_E(\beta,2\pi,0)<0,\quad (\forall \beta \in (\beta_c,\infty)).
\end{align*}
Moreover, 
\begin{align*}
0=g_E(\beta_c,2\pi,0)\le -\frac{2}{|U|}+\frac{b}{\tanh(\beta_c
 e_{min}/2)e_{min}},
\end{align*}
which implies that
\begin{align*}
\beta_c\le
 \frac{2}{e_{min}}\tanh^{-1}\left(\frac{b|U|}{2e_{min}}\right).
\end{align*}
The claims \eqref{item_critical_positive}, \eqref{item_critical_zero},
 \eqref{item_critical_negative} follow from these properties.
\end{proof}

To shorten formulas, let us introduce the parameterized matrix-valued
functions $G_{x,y,z}:\R^d\to \Mat(b,\C)$ $((x,y,z)\in \R_{>0}\times
\R\times \R)$ by
\begin{align*}
G_{x,y,z}(\bk):=\frac{\sinh(x\sqrt{E(\bk)^2+z^2})}{(\cos(xy/2)+\cosh(x\sqrt{E(\bk)^2+z^2}))\sqrt{E(\bk)^2+z^2}}.
\end{align*}
Also, for $E\in \cE(e_{min},e_{max})$ let us set 
\begin{align}
c_E:=\sup_{\bk\in \R^d}\sup_{m_j\in \N\cup\{0\}\atop (j=1,\cdots,d)}
\left\|\prod_{j=1}^d\frac{\partial^{m_j}}{\partial k_j^{m_j}}E(\bk)
\right\|_{b\times b}1_{\sum_{j=1}^dm_j\le d+2}.\label{eq_band_spectra_total_derivative}
\end{align}
For any $\sum_{j=1}^dm_j\bv_j\in\G_{\infty}$ there uniquely exists
$\sum_{j=1}^dm_j'\bv_j\in\G$ such that $m_j=m_j'$ $(\text{mod } L)$ for
any $j\in \{1,\cdots,d\}$. This rule defines the map $r_L:\G_{\infty}\to
\G$. For any $(\rho,\bx,\s)\in\cB\times \G_{\infty}\times \spin$ we
identify $\psi_{\rho\bx \s}^*$, $\psi_{\rho \bx \s}$ with $\psi_{\rho
r_L(\bx) \s}^*$, $\psi_{\rho r_L(\bx) \s}$ respectively. 
For clarity of the statements of the main results let us recall a few
more notational rules. For a function
$f:\G_{\infty}\times\G_{\infty}\to\C$ and $a\in\C$ we write
$\lim_{\|\bx-\by\|_{\R^d}\to \infty}f(\bx,\by)=a$ if for any
$\eps\in\R_{>0}$ there exists $\delta \in\R_{>0}$ such that for any
$\bx,\by\in \G_{\infty}$ satisfying $\|\bx-\by\|_{\R^d}>\delta$,
$|f(\bx,\by)-a|<\eps$. Here $\|\cdot\|_{\R^d}$ denotes the Euclidean
norm of $\R^d$. 

\begin{theorem}\label{thm_infinite_volume_limit}
Let $E\in \cE(e_{min},e_{max})$. Let $\D$ $(\in \R_{\ge 0})$ be the
 solution of the gap equation $g_E(\beta, \beta \theta, \D)=0$ if
 $g_E(\beta,\beta\theta,0)\ge 0$. Let $\D:=0$ if
 $g_E(\beta,\beta\theta,0)< 0$. Then, there exists $c'\in (0,1]$
 depending only on $d,b,(\hbv_j)_{j=1}^d,c_E$ such that the following
 statements hold for any 
$$
U\in \left(-\frac{2c'}{b}\min\{e_{min},e_{min}^{d+1}\},0\right),
$$
$\beta\in \R_{>0}$, $\theta\in \R$.
\begin{enumerate}[(i)]
\item\label{item_partition_positivity}
There exists $L_0\in \N$ such that
\begin{align*}
\Tr e^{-\beta (\sH+i\theta \sS_z+\sF)}\in\R_{>0},\quad(\forall
 L\in\N\text{ with }L\ge L_0,\ \g\in[0,1]).
\end{align*}
\item\label{item_free_energy}
\begin{align*}
&\lim_{L\to \infty\atop L\in \N}\left(-\frac{1}{\beta L^d}\log(\Tr
 e^{-\beta(\sH+i\theta \sS_z)})\right)\\
&=\frac{\D^2}{|U|}-\frac{D_d}{\beta}\int_{\G_{\infty}^*}d\bk\Tr\log\Bigg(2\cos\left(\frac{\beta\theta}{2}\right)e^{-\beta
 E(\bk)}\\
&\qquad\qquad\qquad\qquad\qquad\qquad +e^{\beta(\sqrt{E(\bk)^2+\D^2}-E(\bk))}
+e^{-\beta(\sqrt{E(\bk)^2+\D^2}+E(\bk))}\Bigg).
\end{align*}
\item\label{item_SSB}
\begin{align*}
&\lim_{\g\searrow 0\atop \g\in (0,1]}\lim_{L\to \infty\atop L\in\N}\frac{\Tr(e^{-\beta
(\sH+i\theta \sS_z+\sF)}\psi_{\hrho\hbx\ua}^*\psi_{\hrho\hbx\da}^*)}{\Tr e^{-\beta (\sH+i\theta \sS_z+\sF)}}=\lim_{\g\searrow 0\atop \g\in (0,1]}\lim_{L\to \infty\atop L\in\N}\frac{\Tr(e^{-\beta
(\sH+i\theta \sS_z+\sF)}\psi_{\hrho\hbx\da}\psi_{\hrho\hbx\ua})}{\Tr
 e^{-\beta (\sH+i\theta \sS_z+\sF)}}\\
&=-\frac{\D D_d}{2}\int_{\G_{\infty}^*}d\bk G_{\beta, \theta, \D}(\bk)(\hrho,\hrho),\quad
(\forall \hrho\in \cB,\ \hbx\in \G_{\infty}).
\end{align*}
\item\label{item_ODLRO}
If $g_E(\beta,\beta\theta,0)\neq 0$, 
\begin{align*}
&\lim_{\|\hbx-\hby\|_{\R^d}\to\infty}\lim_{L\to \infty\atop L\in\N}\frac{\Tr (e^{-\beta
(\sH+i\theta \sS_z)}\psi_{\hrho\hbx\ua}^*\psi_{\hrho\hbx\da}^*\psi_{\heta\hby\da}\psi_{\heta\hby\ua})}{\Tr
 e^{-\beta (\sH+i\theta \sS_z)}}\\
&=
\D^2\prod_{\rho\in \{\hrho,\heta\}}\left(
\frac{D_d}{2}\int_{\G_{\infty}^*}d\bk G_{\beta,\theta,\D}(\bk)(\rho,\rho)\right),\quad
(\forall\hrho,\heta \in\cB).
\end{align*}
If $g_E(\beta,\beta\theta,0)= 0$, 
\begin{align*}
&\lim_{\|\hbx-\hby\|_{\R^d}\to\infty}\limsup_{L\to \infty\atop L\in\N}\left|\frac{\Tr (e^{-\beta
(\sH+i\theta \sS_z)}\psi_{\hrho\hbx\ua}^*\psi_{\hrho\hbx\da}^*\psi_{\heta\hby\da}\psi_{\heta\hby\ua})}{\Tr
 e^{-\beta (\sH+i\theta \sS_z)}}\right|=0,\quad(\forall\hrho,\heta
 \in\cB).
\end{align*}
\item\label{item_CD}
\begin{align*}
&\lim_{L\to \infty\atop
 L\in\N}\frac{1}{L^{2d}}\sum_{(\hrho,\hbx),(\heta,\hby)\in\cB\times \G}\frac{\Tr (e^{-\beta
(\sH+i\theta \sS_z)}\psi_{\hrho\hbx\ua}^*\psi_{\hrho\hbx\da}^*\psi_{\heta\hby\da}\psi_{\heta\hby\ua})}{\Tr
 e^{-\beta (\sH+i\theta \sS_z)}}=\frac{\D^2}{U^2}.
\end{align*}
\end{enumerate}
\end{theorem}

\begin{remark} 
We should emphasis that $c'$ is independent of $\beta$, $\theta$. Thus, 
once $U$ is fixed, the infinite-volume limits are valid for all
 $(\beta,\theta)\in \R_{>0}\times \R$. This is a notable difference from 
\cite[\mbox{Theorem 1.3}]{K_BCS_I}, \cite[\mbox{Theorem 1.3}]{K_BCS_II}
 where it is assumed that $\beta\theta/2 \notin \pi(2\Z+1)$ and $U$ is
 not independent of $(\beta,\theta)$. Since 
$$
|U|<\frac{2c'}{b}\min\{e_{min},e_{min}^{d+1}\}\le \frac{2e_{min}}{b},
$$
Lemma \ref{lem_critical_temperature} ensures that there exists
 $(\beta,\theta)\in\R_{>0}\times \R$ such that
 $g_E(\beta,\beta\theta,0)>0$ and 
$\D>0$. Thus the claims
 \eqref{item_SSB}, \eqref{item_ODLRO} in particular imply the existence
 of SSB, ODLRO respectively.
\end{remark}

\begin{remark}\label{rem_free_band_spectra}
The smoothness of $\bk\mapsto E(\bk)$ is assumed only for simplicity. All
 the results in this paper can be reconstructed by assuming that
 $\bk\mapsto E(\bk):\R^d\to \Mat(b,\C)$ is continuously differentiable to
 some finite degree depending only on the spatial dimension. The
 symmetry \eqref{eq_time_reversal_symmetry} is assumed to adopt 
\cite[\mbox{Lemma 3.6}]{K_BCS_II} as our formulation. More precisely, we
 used the symmetry \eqref{eq_time_reversal_symmetry} to characterize the covariance
 ``$C(\phi)$'' in \cite[\mbox{Lemma 3.5 (ii)}]{K_BCS_II}. Since the Grassmann
 integral formulation \cite[\mbox{Lemma 3.6}]{K_BCS_II} contains the
 covariance ``$C(\phi)$'', accordingly we assume
 \eqref{eq_time_reversal_symmetry}. The covariance ``$C(\phi)$''
 will be explicitly written in Subsection \ref{subsec_covariances} in
 the same form as in \cite[\mbox{Lemma 5.1}]{K_BCS_II}, which was
 derived from \cite[\mbox{Lemma 3.5 (ii)}]{K_BCS_II}. However, the
 symmetry \eqref{eq_time_reversal_symmetry} itself plays no explicit
 role in this paper. 
\end{remark}

\begin{remark}\label{rem_zero_temperature_limit}
In \cite[\mbox{Corollary 1.11}]{K_BCS_II} we derived the
 zero-temperature limit of the free energy density and the thermal
 expectations. By arguing in parallel with the proof of 
\cite[\mbox{Corollary 1.11}]{K_BCS_II} presented at the end of 
\cite[\mbox{Subsection 5.2}]{K_BCS_II} we can derive the
 zero-temperature limit from Theorem \ref{thm_infinite_volume_limit}. 
Not to lengthen the paper, let us
 state the results in an abbreviated form. 

There exists $c''\in (0,1]$ depending only on $d,b,(\hbv_j)_{j=1}^d,c_E$ 
such that for any 
$$
U\in \left(-\frac{2c''}{b}\min\{e_{min},e_{min}^{d+1}\},0\right)
$$
and $\theta \in \R$ five claims which are same as the claims 
``(i), (ii), (iii), (iv), (v)'' of 
\cite[\mbox{Corollary 1.11}]{K_BCS_II} without the constraint
 $\beta\theta/2\notin \pi (2\Z+1)$ hold.

Here we can drop the constraint $\beta\theta/2\notin \pi (2\Z+1)$ as we do
 not need it throughout this paper thanks to the assumption
 \eqref{eq_one_particle_lowest_band}. After the inequality ``(5.72)''
 in the proof of \cite[\mbox{Corollary 1.11}]{K_BCS_II} a spatial decay
 property of the infinite-volume, zero-temperature limit of the
 covariance was proved in order to study the zero-temperature limit of
 the 4-point correlation function. This part can be replaced by the decay property
 discussed in Remark \ref{rem_decay_for_zero_limit} later.
 The property \eqref{eq_one_particle_lowest_band}
 also helps to shorten the derivation of the zero-temperature limit of
 the free energy density. Apart from these changes, the arguments close to
 the proof of \cite[\mbox{Corollary 1.11}]{K_BCS_II} yield the claims. 
Again the results imply no superconducting order in the
 zero-temperature limit. However, this time the results may not come as
 a surprise, since in low temperatures our gap equation has no solution
 at all as shown in Lemma \ref{lem_critical_temperature}
 \eqref{item_critical_negative}.
\end{remark}

\begin{remark}
Since we do not have any $\beta$-dependent constraint on $U$ in 
Theorem \ref{thm_infinite_volume_limit}, we can also study the
 infinite-temperature limit $\beta\searrow 0$ of the free energy density
 and the thermal expectations. If we set $\D\in\R_{\ge 0}$ by the same
 rule as in Theorem \ref{thm_infinite_volume_limit}, it follows that for
 any $U\in \R_{<0}$, $\theta\in\R$ there exists $\beta_c'\in\R_{>0}$
 such that $\D=0$ for any $\beta\in (0,\beta_c']$. This is because 
$$\lim_{\beta\searrow 0}g_E(\beta,\beta\theta,0)=-\frac{2}{|U|}<0.$$
Let us take $U\in (-\frac{2c'}{b}\min\{e_{min},e_{min}^{d+1}\},0)$ for
 the constant $c'$ introduced in Theorem
 \ref{thm_infinite_volume_limit} and fix any $\theta \in \R$. 
Considering the above property of
 $\D$, we can see from Theorem \ref{thm_infinite_volume_limit}
\eqref{item_free_energy} that
\begin{align*}
\lim_{\beta \searrow 0}\lim_{L\to \infty\atop L\in \N}\left(-\frac{1}{\beta L^d}\log(\Tr
 e^{-\beta(\sH+i\theta \sS_z)})\right)=-\infty.
\end{align*}
Moreover, it is not difficult to modify the proof of 
\cite[\mbox{Corollary 1.11}]{K_BCS_II} to confirm that three claims
 which are same as the claims ``(iii), (iv), (v)'' of 
\cite[\mbox{Corollary 1.11}]{K_BCS_II} apart from having the notation
 $\lim_{\beta \searrow 0}$ in place of 
$$\lim_{\beta\to
 \infty,\beta\in \R_{>0}\atop
\text{with }\frac{\beta\theta}{2}\notin \pi (2\Z+1)}$$
hold. To prove the analogue of the claim ``(iv)'', we need a spatial
 decay property of the covariance in the limit $L\to \infty$,
 $\beta\searrow 0$ in particular. We can explicitly take the limit 
$L\to \infty$, $\beta\searrow 0$ in the characterization 
\cite[\mbox{Lemma 5.11}]{K_BCS_II} and observe that the covariance is in
 fact diagonal with the spatial variables in the limit.
Again the results imply no superconducting order in the limit
 $\beta\searrow 0$.
\end{remark}

Theorem \ref{thm_infinite_volume_limit} \eqref{item_free_energy} gives
the exact formula for the function \eqref{eq_formal_free_energy},
provided $|U|$ is small as required in the theorem. Loss of
analyticity of the function \eqref{eq_formal_free_energy} with $t$ is
considered as an indication of DPT at positive temperature in
contemporary physics (see e.g. \cite{BBD}, \cite{HB}, \cite{AK}), though
the function \eqref{eq_formal_free_energy} with the BCS model has not
been rigorously treated yet, to the author's
knowledge. As one of the main themes of this paper, we focus on the
following questions.
\begin{itemize}
\item At which $(\beta,t)\in\R_{>0}\times \R$ does the function
      \eqref{eq_formal_free_energy} lose analyticity ?
\item What is the regularity of the function
      \eqref{eq_formal_free_energy} when it is not analytic ?
\item What is the shape of the subset of $\R_{>0}\times \R$ where
      the function \eqref{eq_formal_free_energy} is not analytic ?
\end{itemize}
We will study these questions in Section \ref{sec_free_energy}. Answers
to the first and the second question can be found without much
difficulty, since we have already studied similar questions in 
\cite[\mbox{Section 2}]{K_BCS_II}. After studying these two questions, we
will know that the function \eqref{eq_formal_free_energy} is $C^1$-class
in $\R_{>0}\times\R$ and its 2nd order derivatives have jump
discontinuities across a subset of $\R_{>0}\times \R$, which consists of 
periodic copies of one closed curve. To answer the last question, 
we need
constructive arguments. It will turn out that the ratio
$e_{min}/e_{max}$ is the key parameter to classify the shape of the set
of our interest. In particular we will show that 
the lower half of the representative curve of the set has only one local
minimum point, in other words the representative curve does not
oscillate with the temperature, for any $E\in \cE(e_{min},e_{max})$
if and only if $e_{min}/e_{max}$ is larger than the critical value
$\sqrt{17-12\sqrt{2}}$. The result will be officially stated in Theorem 
\ref{thm_boundary_characterization} as the second theorem of this
paper.

\section{Analysis of the free energy density}\label{sec_free_energy}

We assume that $|U|< 2 e_{min} /b$ throughout this section so that we
can refer to the results of Lemma \ref{lem_critical_temperature}. Let $E\in
\cE(e_{min},e_{max})$ and let us define the function $\D:\R_{>0}\times
\R\to \R_{\ge 0}$ as follows. Let $\D(\beta,t)$ be the solution of
$g_E(\beta,t,\D)=0$ if $g_E(\beta,t,0)\ge 0$. Let $\D(\beta,t):=0$
if $g_E(\beta,t,0)<0$. The well-definedness of the function $\D(\cdot)$
is guaranteed by Lemma \ref{lem_gap_equation}. Then we define the
function $F_E:\R_{>0}\times \R\to \R$ by
\begin{align*}
&F_{E}(\beta,t)\\
&:=\frac{\D(\beta,t)^2}{|U|}-\frac{D_d}{\beta}\int_{\G_{\infty}^*}d\bk\Tr\log\Bigg(2\cos\left(\frac{t}{2}\right)e^{-\beta
 E(\bk)}\\
&\qquad\qquad\qquad\qquad\qquad\qquad +e^{\beta(\sqrt{E(\bk)^2+\D(\beta,t)^2}-E(\bk))}
+e^{-\beta(\sqrt{E(\bk)^2+\D(\beta,t)^2}+E(\bk))}\Bigg).
\end{align*}
It follows from Theorem \ref{thm_infinite_volume_limit}
\eqref{item_free_energy} that if $U\in
(-\frac{2c'}{b}\min\{e_{min},e_{min}^{d+1}\},0)$, 
\begin{align*}
F_E(\beta,t)=\lim_{L\to \infty\atop L\in \N}\left(-\frac{1}{\beta
 L^d}\log (\Tr e^{-\beta \sH+it \sS_z})\right),\quad (\forall
 (\beta,t)\in\R_{>0}\times \R).
\end{align*}
Thus the function $F_E(\beta,t)$ can be seen as an extension of the free
energy density with respect to the magnitude of the coupling
constant. In this section we study the regularity of $F_E$ with
$(\beta,t)$ and characterize the subset of $\R_{>0}\times\R$ where the
analyticity is lost. The contents of this section are independent of
Section \ref{sec_derivation}, which is devoted to proving Theorem
\ref{thm_infinite_volume_limit}. The readers can read this section
separately from Section \ref{sec_derivation}.

\subsection{Phase transitions}\label{subsec_phase_transitions}

The domain $\R_{>0}\times \R$ can be decomposed as
follows. $\R_{>0}\times\R=Q_{+}\sqcup Q_{-}\sqcup Q_0$, where 
\begin{align*}
&Q_+:=\left\{(\beta,t)\in\R_{>0}\times\R\ |\ g_E(\beta,t,0)>0\right\},\\
&Q_-:=\left\{(\beta,t)\in\R_{>0}\times\R\ |\ g_E(\beta,t,0)<0\right\},\\
&Q_0:=\left\{(\beta,t)\in\R_{>0}\times\R\ |\ g_E(\beta,t,0)=0\right\}.
\end{align*}
In this subsection we will prove that the function $F_E$ is $C^1$-class in
$\R_{>0}\times \R$, real analytic in $Q_+\cup Q_-$ and non-analytic at
any point of $Q_0$ as a function of two variables. 
More specifically, we will prove that 2nd order
derivatives of $F_E$ have jump discontinuity across $Q_0$, which is a
sign of 2nd order phase transition. Also, we will see that $Q_0$
consists of 
periodic copies of a restriction of a closed curve in $\R^2$. Let us
call the curves making up 
$Q_0$ phase boundaries. In fact the regularity of $F_E$ can be
studied in a way similar to \cite[\mbox{Section 2}]{K_BCS_II}. 
However, we decide not to omit it, since it characterizes the nature of
the phase transitions. 

Let us start by describing universal properties of the phase boundaries,
which hold regardless of $e_{min}$, $e_{max}(\in \R_{>0})$.
We can deduce
from Lemma \ref{lem_critical_temperature} 
\eqref{item_critical_basic},\eqref{item_critical_positive} 
that for any $\beta\in (0,\beta_c)$ there
uniquely exists $\tau(\beta)\in (\pi,2\pi)$ such that
$g_E(\beta,\tau(\beta),0)=0$. This rule defines the function
$\tau:(0,\beta_c)\to (\pi,2\pi)$. The following lemma means little at
this point. However, it will support conclusive parts of our
construction later, or more specifically the proofs of 
Proposition \ref{prop_necessity_equality} and Proposition
\ref{prop_multi_orbital}. Also, it will implicitly support the proof of
Proposition \ref{prop_necessity_inequality}. 

\begin{lemma}\label{lem_tau_implicit_uniqueness}
Assume that $|U|<2e_{min}/b$, $y\in (-1,0)$, $\beta\in \R_{>0}$, $E\in
 \cE(e_{min},e_{max})$ and 
\begin{align*}
-\frac{2}{|U|}+D_d\int_{\G^*_{\infty}}d\bk \Tr\left(\frac{\sinh(\beta
 E(\bk))}{(y+\cosh(\beta E(\bk)))E(\bk)}\right)=0.
\end{align*}
Then $\beta\in (0,\beta_c)$ and $y=\cos(\tau(\beta)/2)$.
\end{lemma}

Basic properties of the function
$\tau(\cdot)$ are summarized as follows. For an open set $O$ of $\R^n$
let $C^{\o}(O)$ denote the set of real analytic functions on $O$.

\begin{lemma}\label{lem_boundary_function}
\begin{enumerate}[(i)]
\item\label{item_boundary_function_smoothness}
$$\tau\in
     C^{\omega}((0,\beta_c)).$$
\item\label{item_boundary_function_boundary} 
$$
\lim_{\beta \nearrow \beta_c}\tau(\beta)=\lim_{\beta\searrow 0}\tau(\beta)=2\pi.
$$ 
\item\label{item_boundary_function_derivative}
$$
\lim_{\beta \nearrow \beta_c}\frac{d \tau}{d\beta}(\beta)=+\infty,\quad \lim_{\beta \searrow 0}\frac{d \tau}{d\beta}(\beta)=-\infty.
$$ 
\end{enumerate}
\end{lemma}

\begin{remark}
In the proofs of Lemma \ref{lem_boundary_function}, Proposition
 \ref{prop_manifold}, Proposition \ref{prop_phase_transition} 
and Proposition \ref{prop_sufficiency}
we will
 apply the implicit function theorem, the inverse function theorem
and the identity theorem for real analytic functions. These theorems are found in e.g. 
\cite[\mbox{Chapter 1, Chapter 2}]{KP}. 
\end{remark}

\begin{proof}[Proof of Lemma \ref{lem_boundary_function}]
\eqref{item_boundary_function_smoothness}: One can see from the
 definition that the function $(x,t)\mapsto g_E(x,t,0):\R_{>0}\times
 \R\to \R$ is real analytic. Since $\frac{\partial g_E}{\partial
 t}(\beta,\tau(\beta),0)\neq 0$ for all $\beta \in (0,\beta_c)$, the
 analytic implicit function theorem ensures the claim.

\eqref{item_boundary_function_boundary}: Suppose that there exists
 $\eps\in \R_{>0}$ such that for any $\delta\in\R_{>0}$ there exists
 $\beta_{\delta}\in (\beta_c-\delta,\beta_c)\cap (0,\beta_c)$ such that
 $\tau(\beta_{\delta})\le 2\pi -\eps$. Then for any $\delta \in \R_{>0}$
\begin{align*}
0=g_E(\beta_{\delta},\tau(\beta_{\delta}),0)\le
 g_E(\beta_{\delta},2\pi-\eps,0)\le \sup_{\beta\in
 (\beta_c-\delta,\beta_c)}g_E(\beta,2\pi-\eps,0).
\end{align*}
By sending $\delta \searrow 0$, $0\le
 g_E(\beta_c,2\pi-\eps,0)<g_E(\beta_c,2\pi,0)=0$, which is a
 contradiction. Thus $\lim_{\beta \nearrow
 \beta_c}\tau(\beta)=2\pi$. 

Suppose that there exists $\eps\in \R_{>0}$ such that for any $\delta
 \in\R_{>0}$ there exists $\beta_{\delta}\in (0,\delta)\cap (0,\beta_c)$ 
such that
 $\tau(\beta_{\delta})\le 2\pi -\eps$. Then for any $\delta \in \R_{>0}$
$$
0=g_E(\beta_{\delta},\tau(\beta_{\delta}),0)\le \sup_{\beta\in
 (0,\delta)}g_E(\beta, 2\pi-\eps,0).
$$
By sending $\delta\searrow 0$, $0\le g_E(0,2\pi
 -\eps,0)=-2/|U|<0$, which is a contradiction. Thus $\lim_{\beta \searrow
 0}\tau(\beta)=2\pi$.

\eqref{item_boundary_function_derivative}: For $\beta\in (0,\beta_c)$ 
\begin{align}
\frac{d \tau}{d\beta}(\beta)&= -\frac{\frac{\partial g_E}{\partial
 x}(\beta,\tau(\beta),0)}{\frac{\partial g_E}{\partial
 t}(\beta,\tau(\beta),0)}\label{eq_implicit_derivative}\\
&=-\frac{2 \int_{\G_{\infty}^*}d\bk \Tr
 \left(\frac{1+\cos(\tau(\beta)/2)\cosh(\beta
 E(\bk))}{(\cos(\tau(\beta)/2)+\cosh(\beta
 E(\bk)))^2}\right)}{\sin\big(\frac{\tau(\beta)}{2}\big)
 \int_{\G_{\infty}^*}d\bk \Tr
 \left(\frac{\sinh(\beta
 E(\bk))}{(\cos(\tau(\beta)/2)+\cosh(\beta
 E(\bk)))^2E(\bk)}\right)}.\notag
\end{align}
Then by using the result of \eqref{item_boundary_function_boundary},
\begin{align*}
\lim_{\beta\nearrow
 \beta_c}\frac{d\tau}{d\beta}(\beta)=\lim_{\beta\nearrow \beta_c}
\frac{2 \int_{\G_{\infty}^*}d\bk \Tr
 \left(\frac{1}{\cosh(\beta_c E(\bk))-1}\right)}{\sin\big(\frac{\tau(\beta)}{2}\big)
 \int_{\G_{\infty}^*}d\bk \Tr
 \left(\frac{\sinh(\beta_c
 E(\bk))}{(\cosh(\beta_c
 E(\bk))-1)^2E(\bk)}\right)}=\infty.
\end{align*}

To study the limit $\lim_{\beta\searrow 0}\frac{d\tau}{d\beta}(\beta)$,
 let us show that 
\begin{align}
\lim_{\beta\searrow
 0}\frac{\cos(\tau(\beta)/2)+1}{\beta}=\frac{b|U|}{2}.\label{eq_key_convergence_zero}
\end{align}
Suppose that there exists $\eps\in \R_{>0}$ such that 
\begin{align*}
\sup_{\beta\in (0,\delta)}\frac{\cos(\tau(\beta)/2)+1}{\beta}\ge
 \frac{b|U|}{2}+\eps,\quad (\forall \delta \in (0,\beta_c)).
\end{align*}
Take any $\delta \in (0,\beta_c)$. Then there exists $\beta_{\delta}\in
 (0,\delta)$ such that 
$$
\frac{\cos(\tau(\beta_{\delta})/2)+1}{\beta_{\delta}}\ge
 \frac{b|U|}{2}+\frac{\eps}{2},
$$
and thus
\begin{align*}
\frac{2}{|U|}&\le D_d\int_{\G_{\infty}^*}d\bk \Tr
 \left(\frac{\sinh(\beta_{\delta}E(\bk))}{\left(\frac{b|U|}{2}+\frac{\eps}{2}+\frac{\cosh(\beta_{\delta}
 E(\bk))-1}{\beta_{\delta}} \right)\beta_{\delta}E(\bk)}
\right)\\
&\le \frac{2}{b|U|+\eps}D_d\int_{\G_{\infty}^*}d\bk
 \Tr\left(\frac{\sinh(\delta E(\bk))}{\delta E(\bk)}\right).
\end{align*}
By sending $\delta \searrow 0$, $\frac{2}{|U|}\le \frac{2
 b}{b|U|+\eps}<\frac{2}{|U|}$, which is a contradiction. Thus for any
 $\eps\in \R_{>0}$ there exists $\delta \in (0,\beta_c)$ such that 
\begin{align*}
\sup_{\beta\in (0,\delta)}\frac{\cos(\tau(\beta)/2)+1}{\beta}<
 \frac{b|U|}{2}+\eps,
\end{align*}
which implies that 
$$
\limsup_{\beta \searrow 0}\frac{\cos(\tau(\beta)/2)+1}{\beta}\le 
 \frac{b|U|}{2}.
$$
On the other hand, suppose that there exists $\eps\in \R_{>0}$ such that 
\begin{align*}
\inf_{\beta\in (0,\delta)}\frac{\cos(\tau(\beta)/2)+1}{\beta}\le
 \frac{b|U|}{2}-\eps,\quad (\forall \delta \in (0,\beta_c)).
\end{align*}
Take any $\delta \in (0,\beta_c)$. Then there exists $\beta_{\delta}\in
 (0,\delta)$ such that
$$
\frac{\cos(\tau(\beta_{\delta})/2)+1}{\beta_{\delta}}\le
 \frac{b|U|}{2}-\frac{\eps}{2},
$$
and thus
\begin{align*}
\frac{2}{|U|}&\ge D_d\int_{\G_{\infty}^*}d\bk \Tr
 \left(\frac{\sinh(\beta_{\delta}E(\bk))}{\left(\frac{b|U|}{2}-\frac{\eps}{2}+\frac{\cosh(\beta_{\delta}
 E(\bk))-1}{\beta_{\delta}} \right)\beta_{\delta}E(\bk)}
\right)\\
&\ge D_d\int_{\G_{\infty}^*}d\bk \Tr
 \left(\frac{1}{\frac{b|U|}{2}-\frac{\eps}{2}+\frac{\cosh(\delta
 E(\bk))-1}{\delta}}\right).
\end{align*}
By sending $\delta \searrow 0$, $\frac{2}{|U|}\ge
 \frac{2b}{b|U|-\eps}>\frac{2}{|U|}$, which is a contradiction. Thus for
 any $\eps\in \R_{>0}$ there exists $\delta \in (0,\beta_c)$ such that 
\begin{align*}
\inf_{\beta\in (0,\delta)}\frac{\cos(\tau(\beta)/2)+1}{\beta}>
 \frac{b|U|}{2}-\eps,
\end{align*}
which implies that 
$$
\liminf_{\beta\searrow 0}\frac{\cos(\tau(\beta)/2)+1}{\beta}\ge 
 \frac{b|U|}{2}.
$$
Therefore, the property \eqref{eq_key_convergence_zero} follows. 

By applying \eqref{eq_key_convergence_zero} we can derive that
\begin{align*}
\lim_{\beta\searrow 0}\frac{\int_{\G_{\infty}^*}d\bk \Tr
 \left(\frac{1+\cos(\tau(\beta)/2)\cosh(\beta
 E(\bk))}{(\cos(\tau(\beta)/2)+\cosh(\beta
 E(\bk)))^2}\right)}{
 \int_{\G_{\infty}^*}d\bk \Tr
 \left(\frac{\sinh(\beta
 E(\bk))}{(\cos(\tau(\beta)/2)+\cosh(\beta
 E(\bk)))^2E(\bk)}\right)}=\frac{b|U|}{2}.
\end{align*}
By combining this with \eqref{eq_implicit_derivative} and the result of
 \eqref{item_boundary_function_boundary} we can deduce the claim on the
 limit $\lim_{\beta\searrow 0}\frac{d\tau}{d\beta}(\beta)$.
\end{proof}

By parity, periodicity and Lemma \ref{lem_critical_temperature}
the set $Q_0$ is characterized as
follows. 
\begin{align}
Q_0=&\{(\beta,\delta \tau(\beta)+4\pi m)\ |\ \beta \in (0,\beta_c),\
 \delta \in \{1,-1\},\ m\in \Z\}\label{eq_profile_boundary}\\
&\cup \{(\beta_c,2\pi + 4\pi m)\ |\ m\in
 \Z\}.\notag
\end{align}
Set 
\begin{align*}
\widehat{Q}_0:=\{(\beta,\tau(\beta)),\ (\beta,4\pi-\tau(\beta))\ |\ \beta\in
 (0,\beta_c)\}\cup \{(0,2\pi),\ (\beta_c,2\pi)\},
\end{align*}
which is a closed curve in $\R^2$ by Lemma \ref{lem_boundary_function}
\eqref{item_boundary_function_boundary}. We can see that $Q_0$ consists
of periodic copies of $\widehat{Q}_0\cap (\R_{>0}\times \R)$. This fact
motivates us to study the curve $\widehat{Q}_0$ as the representative
 of the phase boundaries.

\begin{proposition}\label{prop_manifold}
$\widehat{Q}_0$ is a 1-dimensional real analytic submanifold of $\R^2$.
\end{proposition}

\begin{proof}
By Lemma \ref{lem_boundary_function}
 \eqref{item_boundary_function_smoothness} the maps
\begin{align*}
&\beta\mapsto (\beta,\tau(\beta)):(0,\beta_c)\to \{(\beta,\tau(\beta))\
 |\ \beta\in (0,\beta_c)\},\\
&\beta\mapsto (\beta, 4\pi-\tau(\beta)):(0,\beta_c)\to \{(\beta,4\pi-\tau(\beta))\
 |\ \beta\in (0,\beta_c)\}
\end{align*}
are real analytic homeomorphism. Thus it suffices to prove that there
 exist open intervals $I_1$, $I_2$, an open neighborhood $U_1$ of
 $(0,2\pi)$, an open neighborhood $U_2$ of $(\beta_c,2\pi)$ in $\R^2$
 and real analytic homeomorphisms $f_j:I_j\to U_j\cap \widehat{Q}_0$
 $(j=1,2)$. We can see that  
$\frac{\partial g_E}{\partial x}(\beta_c,2\pi,0)<0$. Thus the analytic
 implicit function theorem ensures that there exists $\eps_1\in (0,\pi)$
 and $\hat{f}\in C^{\o}((2\pi-\eps_1,2\pi+\eps_1))$ such that
 $\hat{f}(2\pi)=\beta_c$, $\hat{f}(t)>0$ and $g_E(\hat{f}(t),t,0)=0$ for
 any $t\in (2\pi-\eps_1,2\pi+\eps_1)$. Thus, $(\beta_c,2\pi)\in
 \{(\hat{f}(t),t)\ |\ t\in (2\pi-\eps_1,2\pi+\eps_1)\}\subset
 \widehat{Q}_0$. Since $\widehat{Q}_0$ is symmetric with respect to the
 line $\{(x,2\pi)\ |\ x\in\R\}$, there exists $\eps_2\in\R_{>0}$ such
 that 
\begin{align*}
&\{(\hat{f}(t),t)\ |\ t\in
 (2\pi-\eps_1,2\pi+\eps_1)\}=(\beta_c-\eps_2,\beta_c+\eps_2)\times 
(2\pi-\eps_1,2\pi+\eps_1)\cap \widehat{Q}_0.
\end{align*}
If we define the map $f_2:(2\pi-\eps_1,2\pi+\eps_1)\to
 (\beta_c-\eps_2,\beta_c+\eps_2)\times (2\pi-\eps_1,2\pi+\eps_1)\cap
 \widehat{Q}_0$ by $f_2(t):=(\hat{f}(t),t)$, we see that the claim on
 $(\beta_c,2\pi)$ holds. 

Let us prove the claim on $(0,2\pi)$. Observe that there exist $\eps_3$,
 $\eps_4\in \R_{>0}$ such that the function 
\begin{align*}
(x,y)\mapsto -\frac{2}{|U|}+ D_d\int_{\G_{\infty}^*}d\bk \Tr
 \left(\frac{\sinh(x E(\bk))}{\big(y+\frac{\cosh(x
 E(\bk))-1}{x} \big)x E(\bk)}\right)
\end{align*}
is real analytic in $(-\eps_3,\eps_3)\times
 (b|U|/2-\eps_4,b|U|/2+\eps_4)$. Let $\phi(x,y)$ denote
 this function. We can check that $\phi(0,b|U|/2)=0$ and $\frac{\partial
 \phi}{\partial y}(0,b|U|/2)<0$. Thus by the analytic implicit function
 theorem there exist $\eps_5\in (0,\eps_3)$ and a real analytic function
 $\eta:(-\eps_5,\eps_5)\to\R_{>0}$ such that 
$\eta(0)=b|U|/2$, $\phi(x,\eta(x))=0$, $(\forall x\in
 (-\eps_5,\eps_5))$. 
Then let us
 define the function $\xi:(-\eps_5,\eps_5)\to\R$ by
 $\xi(x):=x\eta(x)-1$. It follows that 
$\xi\in C^{\o}((-\eps_5,\eps_5))$, $\xi(0)=-1$, $\frac{d\xi}{dx}(0)=b|U|/2>0$.
Thus there exists $\eps_6\in (0,\eps_5)$ such that $\xi(\cdot)$ is
 strictly monotone increasing in $(-\eps_6,\eps_6)$ and 
\begin{align*}
-\frac{2}{|U|}+ D_d\int_{\G_{\infty}^*}d\bk \Tr
 \left(\frac{\sinh(x E(\bk))}{(\xi(x)+\cosh(x
 E(\bk)))E(\bk)}\right)=0,\quad (\forall x\in (-\eps_6,\eps_6)\backslash\{0\}).
\end{align*}
Then by the inverse function theorem there exist $\eps_7\in \R_{>0}$ and
 a real analytic function $\la:(-1-\eps_7,-1+\eps_7)\to
 (-\eps_6,\eps_6)$ such that $\la(\cdot)$ is strictly monotone
 increasing, $\la(-1)=0$, $\xi(\la(y))=y$, $(\forall y\in
 (-1-\eps_7,-1+\eps_7))$. It follows that
\begin{align*}
&-\frac{2}{|U|}+ D_d\int_{\G_{\infty}^*}d\bk \Tr
 \left(\frac{\sinh(\la(y) E(\bk))}{(y+\cosh(\la(y)
 E(\bk)))E(\bk)}\right)=0,\\
&(\forall y\in (-1-\eps_7,-1+\eps_7)\backslash\{-1\}).
\end{align*}
We can take $\eps_8\in (0,\pi)$ so that $\cos(t/2)\in [-1,-1+\eps_7)$,
 $(\forall t\in (2\pi-\eps_8,2\pi+\eps_8))$. Let us define the function
 $\nu:(2\pi-\eps_8,2\pi+\eps_8)\to\R$ by
 $\nu(t):=\la(\cos(t/2))$. Observe that $\nu\in
 C^{\o}((2\pi-\eps_8,2\pi+\eps_8))$, $\nu(2\pi)=0$, $\nu(t)>0$,
 $(\forall t\in (2\pi-\eps_8,2\pi+\eps_8)\backslash \{2\pi\})$,
 $g_E(\nu(t),t,0)=0$, $(\forall t\in
 (2\pi-\eps_8,2\pi+\eps_8)\backslash\{2\pi\})$. Thus $(0,2\pi)\in
 \{(\nu(t),t)\ |\ t\in (2\pi-\eps_8,2\pi+\eps_8)\}\subset
 \widehat{Q}_0$. Since $\widehat{Q}_0$ is symmetric with respect to the
 line $\{(x,2\pi)\ |\ x\in x\in \R\}$, there exists $\eps_9\in \R_{>0}$
 such that 
\begin{align*}
\{(\nu(t),t)\ |\ t\in (2\pi-\eps_8,2\pi+\eps_8)\}=(-\eps_9,\eps_9)\times
 (2\pi-\eps_8,2\pi+\eps_8)\cap \widehat{Q}_0.
\end{align*}
We can define the map $f_1:(2\pi-\eps_8,2\pi+\eps_8)\to (-\eps_9,\eps_9)\times
 (2\pi-\eps_8,2\pi+\eps_8)\cap \widehat{Q}_0$ by $f_1(t):=(\nu(t),t)$ so
 that the claim on $(0,2\pi)$ holds as well. The proof is now complete.
\end{proof}

By taking into account the definition of the function $\D(\cdot)$,
 Lemma \ref{lem_boundary_function} and Proposition
\ref{prop_manifold} we can schematically draw a $\beta-t$ phase diagram
restricted within the plane $\R_{>0}\times (0,4\pi)$ as in Figure
\ref{fig_phase_diagram}.

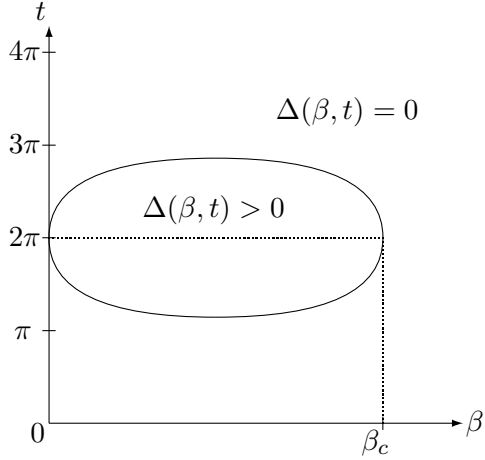
\begin{figure}
\begin{center}
\begin{picture}(180,170)(0,0)

\put(15,10){\vector(1,0){155}}
\put(15,10){\vector(0,1){150}}
\put(8,3){\small 0}
\put(171,7){\small$\beta$}
\put(10,162){\small$t$}
\put(132,0){\small$\beta_c$}

\put(12.5,45){\line(1,0){5}}
\put(12.5,80){\line(1,0){5}}
\put(12.5,115){\line(1,0){5}}
\put(12.5,150){\line(1,0){5}}
\put(140,7.5){\line(0,1){5}}

\qbezier[80](15,80)(80,80)(140,80)
\qbezier[45](140,10)(140,45)(140,80)

\put(2,42){\small$\pi$}
\put(0,77){\small$2\pi$}
\put(0,112){\small$3\pi$}
\put(0,147){\small$4\pi$}

\qbezier(15,80)(15,110)(77.5,110)
\qbezier(15,80)(15,50)(77.5,50)
\qbezier(140,80)(140,110)(77.5,110)
\qbezier(140,80)(140,50)(77.5,50)

\put(100,125){\small$\D(\beta,t)=0$}
\put(50,88){\small$\D(\beta,t)>0$}

\end{picture}
 \caption{The schematic $\beta-t$ phase diagram restricted within $\R_{>0}\times(0,4\pi)$. The curve corresponds to $\widehat{Q}_0$.}\label{fig_phase_diagram}
\end{center}
\end{figure}

Next let us study the regularity of $F_E(\cdot,\cdot)$. In particular
let us show non-analyticity of $F_E(\cdot,\cdot)$ on $Q_0$.

\begin{proposition}\label{prop_phase_transition}
The following statements hold.
\begin{enumerate}[(i)]
\item\label{item_general_regularity}
\begin{align*}
F_E|_{Q_+\cup Q_-}\in C^{\o}(Q_+\cup Q_-),\quad F_E\in C^1(\R_{>0}\times
 \R).
\end{align*}
\item\label{item_nonanalytic_beta}
For any $(\beta_0,t_0)\in Q_0$
$\lim_{(\beta,t)\to(\beta_0,t_0),(\beta,t)\in Q_+}\frac{\partial^2
 F_E}{\partial\beta^2}(\beta,t)$, $\lim_{(\beta,t)\to(\beta_0,t_0),
     (\beta,t)\in Q_-}$ $\frac{\partial^2
 F_E}{\partial\beta^2}(\beta,t)$ 
converge to finite values. Moreover, if $\beta_0\in (0,\beta_c)$ and
     $\frac{d \tau}{d\beta}(\beta_0)\neq 0$ or $\beta_0=\beta_c$, 
\begin{align*}
 \lim_{(\beta,t)\to(\beta_0,t_0)\atop (\beta,t)\in Q_+}\frac{\partial^2
 F_E}{\partial\beta^2}(\beta,t)<
 \lim_{(\beta,t)\to(\beta_0,t_0)\atop (\beta,t)\in Q_-}\frac{\partial^2
 F_E}{\partial\beta^2}(\beta,t).
\end{align*}
If $\beta_0\in (0,\beta_c)$ and $\frac{d \tau}{d\beta}(\beta_0)=0$, 
\begin{align*}
\lim_{(\beta,t)\to(\beta_0,t_0)\atop (\beta,t)\in Q_+}\frac{\partial^2
 F_E}{\partial\beta^2}(\beta,t)=
 \lim_{(\beta,t)\to(\beta_0,t_0)\atop (\beta,t)\in Q_-}\frac{\partial^2
 F_E}{\partial\beta^2}(\beta,t).
\end{align*}
\item\label{item_nonanalytic_time}
For any $(\beta_0,t_0)\in Q_0$
$\lim_{(\beta,t)\to(\beta_0,t_0),(\beta,t)\in Q_+}\frac{\partial^2
 F_E}{\partial t^2}(\beta,t)$, 
$\lim_{(\beta,t)\to(\beta_0,t_0),(\beta,t)\in Q_-}$ $\frac{\partial^2
 F_E}{\partial t^2}(\beta,t)$
converge to finite values. Moreover, if $\beta_0\in (0,\beta_c)$, 
\begin{align*}
 \lim_{(\beta,t)\to(\beta_0,t_0)\atop (\beta,t)\in Q_+}\frac{\partial^2
 F_E}{\partial t^2}(\beta,t)<
 \lim_{(\beta,t)\to(\beta_0,t_0)\atop (\beta,t)\in Q_-}\frac{\partial^2
 F_E}{\partial t^2}(\beta,t).
\end{align*}
If $\beta_0=\beta_c$, 
\begin{align*}
 \lim_{(\beta,t)\to(\beta_0,t_0)\atop (\beta,t)\in Q_+}\frac{\partial^2
 F_E}{\partial t^2}(\beta,t)=
 \lim_{(\beta,t)\to(\beta_0,t_0)\atop (\beta,t)\in Q_-}\frac{\partial^2
 F_E}{\partial t^2}(\beta,t).
\end{align*}
\end{enumerate}
\end{proposition}
\begin{proof}
The claims can be proved in a way similar to the proofs of ``Lemma 2.2'',
 ``Proposition 2.6'' of \cite{K_BCS_II}. However, 
we do not significantly skip the explanations for the readers' convenience. 

\eqref{item_general_regularity}: By using the fact that for $\eps\in
 [-1,1]$ the function \eqref{eq_reference_function} is strictly monotone
 decreasing we can check that $\frac{\partial g_E}{\partial
 z}(\beta,t,\D(\beta,t))<0$ for any $(\beta,t)\in Q_+$. Thus by the
 analytic implicit function theorem $\D|_{Q_+}\in C^{\o}(Q_+)$. Since
 $\D|_{Q_-}\in C^{\o}(Q_-)$ trivially, $\D|_{Q_+\cup Q_-}\in
 C^{\o}(Q_+\cup Q_-)$. Let us prove that $\D\in C(\R_{>0}\times
 \R)$. Let $(\beta_0,t_0)\in Q_0$. Suppose that there exists $\eps \in
 \R_{>0}$ such that for any $\delta \in \R_{>0}$ there exists
 $(\beta_{\delta}, t_{\delta})\in \R_{>0}\times \R$ such that
$\|(\beta_0,t_0)-(\beta_{\delta},t_{\delta})\|_{\R^2}<\delta$ and
 $\D(\beta_{\delta},t_{\delta})\ge \eps$. Then,
\begin{align*}
0=g_E(\beta_{\delta},t_{\delta},\D(\beta_{\delta},t_{\delta}))\le
 \sup_{(\beta,t)\in\R_{>0}\times \R\atop\text{with }\|(\beta,t)-(\beta_0,t_0)\|_{\R^2}<\delta}
g_{E}(\beta,t,\eps).
\end{align*}
By sending $\delta\searrow 0$, $0\le g_E(\beta_0,t_0,\eps)<
 g_E(\beta_0,t_0,0)=0$, which is a contradiction. Thus,
 $\lim_{(\beta,t)\to (\beta_0,t_0)}\D(\beta,t)=0=\D(\beta_0,t_0)$. This
 implies that $\D\in C(\R_{>0}\times\R)$. It readily follows from the
 confirmed regularity of $\D$ that $F_E|_{Q_+\cup Q_-}\in
 C^{\o}(Q_+\cup Q_-)$, $F_E\in C(\R_{>0}\times
 \R)$. Let us define the function $\widehat{F}_E:\R_{>0}\times
 \R\times\R\to \R$ by 
\begin{align}
&\widehat{F}_E(x,t,z):=\frac{z^2}{|U|}-\frac{D_d}{x}\int_{\G_{\infty}^*}d\bk\Tr\log\left(
\cos\left(\frac{t}{2}\right)+\cosh(x\sqrt{E(\bk)^2+z^2})\right).\label{eq_extracted_function}
\end{align}
Observe that the regularity of the function $(\beta,t)\mapsto
 \widehat{F}_E(\beta,t,\D(\beta,t)):\R_{>0}\times\R\to\R$ is same as
 that of $F_E(\beta,t)$. By considering the definition of $\D(\beta,t)$
 we can derive that for any $(\beta,t)\in Q_+\cup Q_-$
\begin{align}
&\frac{\partial}{\partial
 \beta}\widehat{F}_E(\beta,t,\D(\beta,t))=\frac{\partial
 \widehat{F}_E}{\partial x}(\beta,t,\D(\beta,t)),\
\frac{\partial}{\partial
 t}\widehat{F}_E(\beta,t,\D(\beta,t))=\frac{\partial
 \widehat{F}_E}{\partial
 t}(\beta,t,\D(\beta,t)).\label{eq_first_derivative_cancel}
\end{align}
Take any $(\beta_0,t_0)\in Q_0$. The above equalities imply that 
\begin{align}
&\lim_{(\beta,t)\to (\beta_0,t_0)\atop (\beta,t)\in Q_+\cup Q_-}
\frac{\partial}{\partial
 \beta}\widehat{F}_E(\beta,t,\D(\beta,t))=\frac{\partial
 \widehat{F}_E}{\partial
 x}(\beta_0,t_0,\D(\beta_0,t_0)),\label{eq_one_derivative_beta}\\
&\lim_{(\beta,t)\to (\beta_0,t_0)\atop (\beta,t)\in Q_+\cup Q_-}
\frac{\partial}{\partial
 t}\widehat{F}_E(\beta,t,\D(\beta,t))=\frac{\partial
 \widehat{F}_E}{\partial
 t}(\beta_0,t_0,\D(\beta_0,t_0)).\label{eq_one_derivative_time}
\end{align}
We remark that the function $\beta\mapsto g_E(\beta,t_0,0)$ is real analytic in
 $\R_{>0}$. Since 
$$\lim_{\beta\to \infty}g_E(\beta,t_0,0)\le
 -\frac{2}{|U|}+\frac{b}{e_{min}}<0$$
by assumption, this function is not identically zero. Therefore, there
 exists $\eps\in \R_{>0}$ such that for any $\beta\in
 (\beta_0-\eps,\beta_0+\eps)\backslash \{\beta_0\}$
 $g_E(\beta,t_0,0)\neq 0$. Otherwise the identity theorem for real analytic
 functions yields a contradiction. This means that $(\beta,t_0)\in
 Q_+\cup Q_-$ for any $\beta \in
 (\beta_0-\eps,\beta_0+\eps)\backslash\{\beta_0\}$. Thus, it follows
 from \eqref{eq_one_derivative_beta} that $\beta\mapsto
 \widehat{F}_E(\beta,t_0,\D(\beta,t_0))$ is differentiable at
 $\beta=\beta_0$ and 
\begin{align*}
\frac{\partial}{\partial\beta}\widehat{F}_E(\beta,t_0,\D(\beta,t_0))\Big|_{\beta=\beta_0}=
\frac{\partial \widehat{F}_E}{\partial x}(\beta_0,t_0,\D(\beta_0,t_0)).
\end{align*}
By recalling Lemma \ref{lem_boundary_function}
 \eqref{item_boundary_function_boundary},\eqref{item_boundary_function_derivative} and
 \eqref{eq_profile_boundary} we see that there exists $\eps\in \R_{>0}$
 such that $(\beta_0,t)\in Q_+\cup Q_-$ for any $t\in
 (t_0-\eps,t_0+\eps)\backslash \{t_0\}$. Thus by \eqref{eq_one_derivative_time}
$t\mapsto \widehat{F}_E(\beta_0,t,\D(\beta_0,t))$ is differentiable at
 $t=t_0$ and 
\begin{align*}
\frac{\partial}{\partial t}\widehat{F}_E(\beta_0,t,\D(\beta_0,t))\Big|_{t=t_0}=
\frac{\partial \widehat{F}_E}{\partial t}(\beta_0,t_0,\D(\beta_0,t_0)).
\end{align*}
Since $(\beta,t)\mapsto (\frac{\partial \widehat{F}_E}{\partial
 x}(\beta,t,\D(\beta,t)), \frac{\partial \widehat{F}_E}{\partial
 t}(\beta,t,\D(\beta,t)))$ is continuous in $\R_{>0}\times\R$, it
 follows that $(\beta,t)\mapsto \widehat{F}_E(\beta,t,\D(\beta,t))$ is
 $C^1$-class in $\R_{>0}\times \R$ and so is $F_E(\cdot,\cdot)$.

\eqref{item_nonanalytic_beta}: We can derive from
 \eqref{eq_first_derivative_cancel} and the gap equation
 $g_E(\beta,t,\D(\beta,t))=0$ $((\beta,t)\in Q_+)$ that 
\begin{align}
&\frac{\partial^2}{\partial
 \beta^2}\widehat{F}_E(\beta,t,\D(\beta,t))=\frac{\partial^2\widehat{F}_E}{\partial
 x^2}(\beta,t,\D(\beta,t)),\quad (\forall (\beta,t)\in
 Q_-),\label{eq_second_derivative_beta}\\
&\frac{\partial^2}{\partial
 \beta^2}\widehat{F}_E(\beta,t,\D(\beta,t))=\frac{\partial^2\widehat{F}_E}{\partial
 x^2}(\beta,t,\D(\beta,t))+\frac{\partial^2 \widehat{F}_E}{\partial
 x\partial z}(\beta,t,\D(\beta,t))\frac{\partial \D}{\partial
 \beta}(\beta,t)\notag\\
&\qquad\qquad\qquad\qquad\quad =\frac{\partial^2\widehat{F}_E}{\partial
 x^2}(\beta,t,\D(\beta,t))- \frac{\partial g_E}{\partial
 x}(\beta,t,\D(\beta,t))\D(\beta,t)\frac{\partial
 \D}{\partial \beta}(\beta,t)\notag\\
&\qquad\qquad\qquad\qquad\quad =\frac{\partial^2\widehat{F}_E}{\partial
 x^2}(\beta,t,\D(\beta,t))+\D(\beta,t)\frac{\left(\frac{\partial
 g_E}{\partial x}(\beta,t,\D(\beta,t))\right)^2}{\frac{\partial
 g_E}{\partial z}(\beta,t,\D(\beta,t))},\notag\\
&(\forall (\beta,t)\in
 Q_+).\notag
\end{align}
Let us define the function $\hat{g}:\R_{>0}\times \R\times \R_{>0}\to \R$
 by
$$
\hat{g}(x,t,z):=\frac{\sinh(xz)}{\left(\cos(t/2)+\cosh(xz)\right)z}.
$$
Observe that for $(\beta,t)\in Q_+$
\begin{align*}
&\frac{1}{\D(\beta,t)}\cdot\frac{\partial g_E}{\partial
 z}(\beta,t,\D(\beta,t))\\
&=D_d\int_{\G_{\infty}^*}d\bk \Tr \left(
\frac{\partial \hat{g}}{\partial
 z}\big(\beta,t,\sqrt{E(\bk)^2+\D(\beta,t)^2}\big)\cdot\frac{1}{\sqrt{E(\bk)^2+\D(\beta,t)^2}}\right)
\end{align*}
and 
\begin{align}
\lim_{(\beta,t)\to (\beta_0,t_0)\atop (\beta,t)\in Q_+}
\frac{1}{\D(\beta,t)}\cdot\frac{\partial g_E}{\partial
 z}(\beta,t,\D(\beta,t))
&=D_d\int_{\G_{\infty}^*}d\bk \Tr \left(
\frac{\partial \hat{g}}{\partial
 z}(\beta_0,t_0,|E(\bk)|)\cdot\frac{1}{|E(\bk)|}\right)\label{eq_derivative_denominator}\\
&<0.\notag
\end{align}
Here we again used the monotone decreasing property of the function
 \eqref{eq_reference_function}. Let us study the term $\frac{\partial
 g_E}{\partial x}(\beta_0,t_0,0)$. If $\beta_0\in (0,\beta_c)$,
\begin{align*}
\frac{\partial g_E}{\partial x}(\beta_0,t_0,0)=\frac{\partial
 g_E}{\partial x}(\beta_0,\tau(\beta_0),0)=-\frac{d
 \tau}{d\beta}(\beta_0)\frac{\partial g_E}{\partial
 t}(\beta_0,\tau(\beta_0),0).
\end{align*}
Since $\tau(\beta_0)\in (0,2\pi)$ and $t\mapsto g_E(\beta_0,t,0)$ is
 strictly monotone increasing in $(0,2\pi)$, 
\begin{align}
\frac{\partial g_E}{\partial t}(\beta_0,\tau(\beta_0),0)\neq
 0.\label{eq_time_monotonicity}
\end{align}
Thus, $\frac{\partial g_E}{\partial x}(\beta_0,t_0,0)=0$ if and only if
 $\frac{d \tau}{d\beta}(\beta_0)=0$. If $\beta_0=\beta_c$, $t_0=2\pi$
 (mod $4\pi$). In this case we can directly check that 
$\frac{\partial g_E}{\partial x}(\beta_0,t_0,0)<0$. We can conclude the
 claimed convergent properties by combining the above properties of 
$\frac{\partial g_E}{\partial x}(\beta_0,t_0,0)$ with
 \eqref{eq_second_derivative_beta}, \eqref{eq_derivative_denominator}. 

\eqref{item_nonanalytic_time}: In the same way as in the proof of
 \eqref{item_nonanalytic_beta} we have that
\begin{align}
&\frac{\partial^2}{\partial
 t^2}\widehat{F}_E(\beta,t,\D(\beta,t))=\frac{\partial^2\widehat{F}_E}{\partial
 t^2}(\beta,t,\D(\beta,t)),\quad (\forall (\beta,t)\in
 Q_-),\label{eq_second_derivative_time}\\
&\frac{\partial^2}{\partial
 t^2}\widehat{F}_E(\beta,t,\D(\beta,t))=\frac{\partial^2\widehat{F}_E}{\partial
 t^2}(\beta,t,\D(\beta,t))+\D(\beta,t)\frac{\left(\frac{\partial
 g_E}{\partial t}(\beta,t,\D(\beta,t))\right)^2}{\frac{\partial
 g_E}{\partial z}(\beta,t,\D(\beta,t))},\notag\\
&(\forall (\beta,t)\in
 Q_+).\notag
\end{align}
If $\beta_0\in (0,\beta_c)$, $\tau(\beta_0)\in (0,2\pi)$. By periodicity
 and \eqref{eq_time_monotonicity} 
$$\left|\frac{\partial g_E}{\partial
 t}(\beta_0,t_0,0)\right|=\left|\frac{\partial g_E}{\partial
 t}(\beta_0,\tau(\beta_0),0)\right|>0.$$
If $\beta_0=\beta_c$, $t_0=2\pi$ (mod
 $4\pi$), and thus $\frac{\partial g_E}{\partial
 t}(\beta_0,t_0,0)=0$. The claimed convergent properties follow from
 \eqref{eq_derivative_denominator}, \eqref{eq_second_derivative_time}
 and the above properties of $\frac{\partial g_E}{\partial t}(\beta_0,t_0,0)$.
\end{proof}
 
\begin{remark}
For $(\rho,\eta)=(+,-)$ or $(-,+)$ let us set 
\begin{align*}
Q_{\rho,\eta}^{\beta}:=\left\{(\beta_0,t_0)\in Q_0\ \Big|\ \exists \eps \in
 \R_{>0}\text{ s.t. }\begin{array}{l} (\beta,t_0)\in Q_{\rho},\ (\forall
		      \beta\in (\beta_0-\eps,\beta_0)),\\ 
                      (\beta,t_0)\in Q_{\eta},\ (\forall
		      \beta\in (\beta_0,\beta_0+\eps))\end{array}               
\right\}.
\end{align*}
Proposition \ref{prop_phase_transition} \eqref{item_nonanalytic_beta}
 implies that if $(\beta_0,t_0)\in Q_{\rho,\eta}^{\beta}$ satisfies
 $\beta_0\neq \beta_c$ and $\frac{\partial \tau}{\partial
 \beta}(\beta_0)\neq 0$ or $\beta_0=\beta_c$, 
$$
\lim_{\beta \nearrow \beta_0}\frac{\partial^2 F_E}{\partial
 \beta^2}(\beta,t_0)\neq \lim_{\beta \searrow \beta_0}\frac{\partial^2 F_E}{\partial
 \beta^2}(\beta,t_0).
$$
This means that a 2nd order phase transition driven by $\beta$ occurs at
$(\beta,t)=(\beta_0,t_0)$. Assume that $\beta_0\in (0,\beta_c)$, 
$\frac{d \tau}{d\beta}(\beta_0)=0$ and $\beta\mapsto \tau(\beta)$ is
 monotone increasing or decreasing in a neighborhood of $\beta_0$. Then
 $(\beta_0,\tau(\beta_0))\in Q_{+,-}^{\beta}$ or
 $(\beta_0,\tau(\beta_0))\in Q_{-,+}^{\beta}$ respectively. In this case
 Proposition \ref{prop_phase_transition} \eqref{item_nonanalytic_beta}
 implies that $\beta\mapsto
 \frac{\partial^2 F_E}{\partial\beta^2}(\beta,\tau(\beta_0))$ is
 continuous at $\beta=\beta_0$, even though the trajectory $\beta\mapsto
 (\beta,\tau(\beta_0))$ crosses $Q_0$ at $\beta=\beta_0$ from $Q_+$ to
 $Q_-$ or from $Q_-$ to $Q_+$. This interestingly suggests a possibility of
 higher order phase transition with $\beta$ at
 $(\beta,t)=(\beta_0,\tau(\beta_0))$. However,
 as we will see in the following subsections, the monotonicity of
 $\tau(\cdot)$ is sensitive to individual characteristics of $E(\cdot)$ and
 we do not pursue the question whether $\tau(\cdot)$ can satisfy the
 above properties in this paper. On the other hand, if we set 
\begin{align*}
Q_{\rho,\eta}^{t}:=\left\{(\beta_0,t_0)\in Q_0\ \Big|\ \exists \eps \in
 \R_{>0}\text{ s.t. }\begin{array}{l} (\beta_0,t)\in Q_{\rho},\ (\forall
		      t\in (t_0-\eps,t_0)),\\ 
                      (\beta_0,t)\in Q_{\eta},\ (\forall
		      t\in (t_0,t_0+\eps))\end{array}               
\right\}
\end{align*}
for $(\rho,\eta)=(+,-)$ or $(-,+)$, we can see from Lemma
 \ref{lem_boundary_function} and \eqref{eq_profile_boundary} that 
$$
Q_{+,-}^t\sqcup Q_{-,+}^t=\{(\beta_0,t_0)\in Q_0\ |\ \beta_0\neq \beta_c
 \}. 
$$
Thus by Proposition \ref{prop_phase_transition}
 \eqref{item_nonanalytic_time}
\begin{align*}
\lim_{t \nearrow t_0}\frac{\partial^2 F_E}{\partial
 t^2}(\beta_0,t)\neq \lim_{t \searrow t_0}\frac{\partial^2 F_E}{\partial
 t^2}(\beta_0,t),\quad (\forall (\beta_0,t_0)\in Q_{+,-}^t\cup
 Q_{-,+}^t).
\end{align*}
In other words, $t\mapsto \frac{\partial^2 F_E}{\partial t^2}(\beta,t)$
 has jump discontinuity whenever the trajectory $t\mapsto (\beta,t)$
 crosses $Q_0$ from $Q_+$ to $Q_-$ or from $Q_-$ to $Q_+$. This means that
 the phase transitions driven by $t$ in this system are of 2nd order.
\end{remark}

\begin{remark}
The free energy density characterized in Theorem
 \ref{thm_infinite_volume_limit} \eqref{item_free_energy} corresponds to
 $F_E(\beta,\beta\theta)$. In 
\cite[\mbox{Subsection 2.3}]{K_BCS_II} we focused on the properties of
 the function $(\beta,\theta)\mapsto F_E(\beta,\beta\theta)$
 $:\R_{>0}\times\R\to \R$ under different assumptions on $E(\cdot)$. The
 reason why we treated the function $(\beta,t)\mapsto F_E(\beta,t)$ here
 is that it is considered as a dynamical free energy density
 studied in today's physics of DPT. At this point the function
 $F_E(\beta,\beta\theta)$ lacks physical interpretation and its phase
 boundaries are structurally more complicated to analyze than those of
 $F_E(\beta,t)$. Nonetheless, it is possible to study the regularity of
 $(\beta,\theta)\mapsto F_E(\beta,\beta\theta)$ in a manner similar to
 Proposition \ref{prop_phase_transition}. In this case the set 
$$
\{(\beta,\theta)\in\R_{>0}\times \R\ |\ g_E(\beta,\beta\theta,0)=0\}
$$
defines the phase boundaries and it can be shown that 2nd order
 partial derivatives of the function $(\beta,\theta)\mapsto
 F_E(\beta,\beta\theta)$ have jump discontinuities on the phase
 boundaries. However, we do not explicitly present the results for conciseness of
 the paper.
\end{remark}

\subsection{Shape of the phase boundary}\label{subsec_shape}
 
In view of the characterization \eqref{eq_profile_boundary}, we notice 
that the graph of the function $\tau(\cdot)$ determines the
shape of the phase boundaries. So let us study the profile of
$\tau(\cdot)$ more deeply. Its universal properties have already been
summarized in Lemma \ref{lem_boundary_function} and Proposition
\ref{prop_manifold}. As the next step, we should try to reveal geometric properties
which may vary with details of $E(\cdot)$. It will turn out that the ratio
$e_{min}/e_{max}$ is a prime index to classify the shape of
$\tau(\cdot)$. From now on we let $c$ denote a generic positive constant
independent of any parameter. The following proposition tells us when
$\tau(\cdot):(0,\beta_c)\to \R$ is strictly downward convex. 

\begin{proposition}\label{prop_convexity}
There exists $e_0\in (0,1)$ independent of any parameter such that if
 $e_{min}/e_{max}\ge e_0$, for any $U\in
 [-\frac{e_{min}}{\sinh(2)b},0)$, $E\in \cE(e_{min},e_{max})$ and
 $\beta\in (0,\beta_c)$, 
$$\frac{d^2 \tau}{d\beta^2}(\beta)>0.$$
\end{proposition}

\begin{proof}
First of all let us prepare a few quantitative bounds based on the
 assumption 
\begin{align}
|U|\le \frac{e_{min}}{\sinh(2)b}.\label{eq_tighter_coupling}
\end{align}
Observe that $1/\sinh(2)<2\tanh(1)<2$, which implies that
 $|U|<2e_{min}/b$ and combined with Lemma
 \ref{lem_critical_temperature} that
\begin{align}
\beta_c\le \frac{2}{e_{min}}\tanh^{-1}\left(\frac{b|U|}{2
 e_{min}}\right)
<
 \frac{2}{e_{min}}\tanh^{-1}(\tanh(1))=\frac{2}{e_{min}}.\label{eq_critical_beta_bound}
\end{align}
It follows from \eqref{eq_critical_beta_bound} and the equality 
\begin{align}
g_E(\beta,\tau(\beta),0)=0,\quad (\beta\in
 (0,\beta_c))\label{eq_boundary_gap_equation}
\end{align}
that
\begin{align}
\frac{2}{|U|}\le \frac{b\sinh(\beta
 e_{min})}{e_{min}(\cos(\tau(\beta)/2)+\cosh(\beta e_{min}))}\le
 \frac{b\sinh(2)}{e_{min}(\cos(\tau(\beta)/2)+\cosh(\beta e_{min}))},\label{eq_boundary_gap_equation_direct}
\end{align}
or by \eqref{eq_tighter_coupling}
\begin{align}
&\cos\left(\frac{\tau(\beta)}{2}\right)+\cosh(\beta e_{min})\le
 \frac{1}{2},\label{eq_cos_cosh_upper}\\
&-\cos\left(\frac{\tau(\beta)}{2}\right)\ge
 \frac{1}{2},\quad (\forall \beta\in
 (0,\beta_c)).\label{eq_cosign_lower}
\end{align}

By differentiating both sides of \eqref{eq_boundary_gap_equation} twice
 and substituting the first equality of \eqref{eq_implicit_derivative}
 we obtain that for any $\beta\in (0,\beta_c)$ 
\begin{align}
\frac{d^2\tau}{d\beta^2}(\beta)=\frac{1}{\left(\frac{\partial
 g_E}{\partial t}(\beta,\tau(\beta),0)\right)^3}
\Bigg(&2\frac{\partial^2 g_E}{\partial x \partial t}(\beta,\tau(\beta),0)
\frac{\partial g_E}{\partial x}(\beta,\tau(\beta),0)
\frac{\partial g_E}{\partial t}(\beta,\tau(\beta),0)\label{eq_tau_second_derivative_expansion}\\
&-\frac{\partial^2 g_E}{\partial x^2}(\beta,\tau(\beta),0)
\left(\frac{\partial g_E}{\partial t}(\beta,\tau(\beta),0)\right)^2\notag\\
&-\frac{\partial^2 g_E}{\partial t^2}(\beta,\tau(\beta),0)
\left(\frac{\partial g_E}{\partial x}(\beta,\tau(\beta),0)\right)^2
\Bigg).\notag
\end{align}
Define the functions $f_{\beta}^0$,
$f_{\beta}^x$, $f_{\beta}^t$, $f_{\beta}^{xx}$,
 $f_{\beta}^{xt}$, $f_{\beta}^{tt}:\R\to\R$ by 
\begin{align*}
&f_{\beta}^0(y):=\cos\left(\frac{\tau(\beta)}{2}\right)+\cosh(\beta
 y),\\
&f_{\beta}^x(y):=y\left(\cos\left(\frac{\tau(\beta)}{2}\right)\cosh(\beta
 y)+1\right),\\
&f_{\beta}^t(y):=\frac{1}{2}\sin\left(\frac{\tau(\beta)}{2}\right)\sinh(\beta
 y),\\
&f_{\beta}^{xx}(y):=y^2\sinh(\beta
 y)\left(\cos^2\left(\frac{\tau(\beta)}{2}\right)-
 \cos\left(\frac{\tau(\beta)}{2}\right)\cosh(\beta y) -2\right),\\
&f_{\beta}^{xt}(y):=\frac{y}{2}\sin\left(\frac{\tau(\beta)}{2}\right)
\left(\cos\left(\frac{\tau(\beta)}{2}\right)\cosh(\beta
 y)+1-\sinh^2(\beta y)\right),\\
&f_{\beta}^{tt}(y):=\frac{1}{4}\sinh(\beta y)\left(
1+\sin^2\left(\frac{\tau(\beta)}{2}\right)+\cos\left(\frac{\tau(\beta)}{2}\right)\cosh(\beta
 y)\right).
\end{align*}
Then the formula \eqref{eq_tau_second_derivative_expansion} can be
 rewritten as follows. For any $\beta\in (0,\beta_c)$
\begin{align*}
\frac{d^2 \tau}{d\beta^2}(\beta)
=&\frac{1}{\left(\frac{\partial
 g_E}{\partial t}(\beta,\tau(\beta),0)\right)^3}
\prod_{j=1}^3\left(D_d\int_{\G_{\infty}^*}d\bk_j\right)\\
&\cdot\Bigg(
2\Tr
\left(\frac{f_{\beta}^{xt}(|E(\bk_1)|)}{|E(\bk_1)|f_{\beta}^0(E(\bk_1))^3}\right)
\Tr
 \left(\frac{f_{\beta}^{x}(|E(\bk_2)|)}{|E(\bk_2)|f_{\beta}^0(E(\bk_2))^2}\right)\\
&\qquad\cdot \Tr
 \left(\frac{f_{\beta}^{t}(|E(\bk_3)|)}{|E(\bk_3)|f_{\beta}^0(E(\bk_3))^2}\right)\\
&\quad-\Tr
 \left(\frac{f_{\beta}^{xx}(|E(\bk_1)|)}{|E(\bk_1)|
f_{\beta}^0(E(\bk_1))^3}\right)\Tr
 \left(\frac{f_{\beta}^{t}(|E(\bk_2)|)}{|E(\bk_2)|f_{\beta}^0(E(\bk_2))^2}\right)\\
&\qquad\cdot\Tr
 \left(\frac{f_{\beta}^{t}(|E(\bk_3)|)}{|E(\bk_3)|f_{\beta}^0(E(\bk_3))^2}\right)\\
&\quad-\Tr
 \left(\frac{f_{\beta}^{tt}(|E(\bk_1)|)}{|E(\bk_1)|f_{\beta}^0(E(\bk_1))^3}\right)\Tr
 \left(\frac{f_{\beta}^{x}(|E(\bk_2)|)}{|E(\bk_2)|f_{\beta}^0(E(\bk_2))^2}\right)\\
&\qquad\cdot\Tr
 \left(\frac{f_{\beta}^{x}(|E(\bk_3)|)}{|E(\bk_3)|f_{\beta}^0(E(\bk_3))^2}\right)
\Bigg).
\end{align*}
Let us define the function $f_{\beta}:\R^3\to \R$ by
\begin{align*}
f_{\beta}(y_1,y_2,y_3):=2f_{\beta}^{xt}(y_1) f_{\beta}^{x}(y_2)
 f_{\beta}^{t}(y_3)
- f_{\beta}^{xx}(y_1) f_{\beta}^{t}(y_2) f_{\beta}^{t}(y_3)
- f_{\beta}^{tt}(y_1) f_{\beta}^{x}(y_2) f_{\beta}^{x}(y_3).
\end{align*}
We can see from above that if 
\begin{align}
\min\{f_{\beta}(y_1,y_2,y_3)\ |\ y_j\in [e_{min},e_{max}]\
 (j=1,2,3)\}>0,\label{eq_convexity_sufficiency}
\end{align}
then $\frac{d^2\tau}{d\beta^2}(\beta)>0$. 

Let us prove \eqref{eq_convexity_sufficiency}. Observe that 
\begin{align*}
f_{\beta}(y,y,y)=&\frac{y^2}{2}\sin^2\left(\frac{\tau(\beta)}{2}\right)\sinh(\beta
 y)\left(\cos\left(\frac{\tau(\beta)}{2}\right)\cosh(\beta y)+1
\right)^2\\
&-\frac{y^2}{4}\sin^2\left(\frac{\tau(\beta)}{2}\right)\sinh^3(\beta
 y)\left(\cos^2\left(\frac{\tau(\beta)}{2}\right)+\cos\left(\frac{\tau(\beta)}{2}\right)\cosh(\beta
 y)\right)\\
&-f_{\beta}^{tt}(y)f_{\beta}^x(y)^2\\
=&-\frac{y^2}{4}\sinh(\beta
 y)\left(\cos^2\left(\frac{\tau(\beta)}{2}\right)+\cos\left(\frac{\tau(\beta)}{2}\right)\cosh(\beta
 y)\right)\\
&\quad\cdot\left(\cos\left(\frac{\tau(\beta)}{2}\right)\cosh(\beta
 y)+1\right)^2\\
&-\frac{y^2}{4}\sin^2\left(\frac{\tau(\beta)}{2}\right)\sinh^3(\beta
 y)\left(\cos^2\left(\frac{\tau(\beta)}{2}\right)+\cos\left(\frac{\tau(\beta)}{2}\right)\cosh(\beta
 y)\right)\\
=&-\frac{y^2}{4}\sinh(\beta
 y)\cos\left(\frac{\tau(\beta)}{2}\right)f_{\beta}^0(y)^3,
\end{align*}
which combined with \eqref{eq_cosign_lower} implies that
\begin{align}
\min_{y\in [e_{min},e_{max}]}f_{\beta}(y,y,y)\ge 
\frac{e_{min}^2}{8}\sinh(\beta e_{min})f_{\beta}^0(e_{min})^3.\label{eq_whole_function_lower}
\end{align}
For a continuous function $f:[e_{min},e_{max}]\to\R$ let
 $\|f\|_{\infty}$ denote $\sup_{y\in [e_{min},e_{max}]}|f(y)|$ in the
 following. For any $y_j\in [e_{min},e_{max}]$ $(j=1,2,3)$
\begin{align}
&|f_{\beta}(y_1,y_1,y_1)-f_{\beta}(y_1,y_2,y_3)|\label{eq_whole_function_difference}\\
&\le
 |f_{\beta}(y_1,y_1,y_1)-f_{\beta}(y_1,y_2,y_1)|+|f_{\beta}(y_1,y_2,y_1)-f_{\beta}(y_1,y_2,y_3)|\notag\\
&\le 2(e_{max}-e_{min})\Bigg(\|f_{\beta}^{xt}\|_{\infty}\left(\left\|\frac{d}{dy}f_{\beta}^x\right\|_{\infty}\|f_{\beta}^t\|_{\infty}+
 \|f_{\beta}^x\|_{\infty}\left\|\frac{d}{dy}f_{\beta}^t\right\|_{\infty}\right)\notag\\
&\quad +\|f_{\beta}^{xx}\|_{\infty}\|f_{\beta}^t\|_{\infty}\left\|\frac{d}{dy}f_{\beta}^t\right\|_{\infty}+
 \|f_{\beta}^{tt}\|_{\infty}\|f_{\beta}^x\|_{\infty}\left\|\frac{d}{dy}f_{\beta}^x\right\|_{\infty}\Bigg).\notag
\end{align}
To estimate the right-hand side of the above inequality, let us prepare
 necessary bounds.
\begin{align*}
&\|f_{\beta}^x\|_{\infty}\le
 ce_{max}f_{\beta}^0(e_{max}),\\
&\left\|\frac{d}{dy}f_{\beta}^x\right\|_{\infty}\le 
c\left(f_{\beta}^0(e_{max})+\sinh^2(\beta e_{max})\right),\\
&\|f_{\beta}^t\|_{\infty}\le c
 f_{\beta}^0(e_{max})^{\frac{1}{2}}\sinh(\beta e_{max}),\\
&\left\|\frac{d}{dy}f_{\beta}^t\right\|_{\infty}\le \beta \left|
\sin\left(\frac{\tau(\beta)}{2}\right)f_{\beta}^0(e_{max})\right|+\beta
 \left|\sin\left(\frac{\tau(\beta)}{2}\right)\cos\left(\frac{\tau(\beta)}{2}\right)\right|\\
&\qquad\qquad\ \le c\beta (f_{\beta}^0(e_{max})^{\frac{3}{2}}+ f_{\beta}^0(e_{max})^{\frac{1}{2}}),\\
&\|f_{\beta}^{xx}\|_{\infty}\le c e_{max}^2\sinh(\beta e_{max})
f_{\beta}^0(e_{max}),\\
&\|f_{\beta}^{xt}\|_{\infty}\le c e_{max}
 f_{\beta}^0(e_{max})^{\frac{1}{2}}
\left( f_{\beta}^0(e_{max}) + \sinh^2(\beta e_{max})\right),\\
&\|f_{\beta}^{tt}\|_{\infty}\le c\sinh(\beta e_{max})
 f_{\beta}^0(e_{max}),
\end{align*}
which lead to that 
\begin{align*}
&\|f_{\beta}^{xt}\|_{\infty}\left\|\frac{d}{dy}f_{\beta}^x\right\|_{\infty}\|f_{\beta}^t\|_{\infty}\\
&\le
 ce_{max}\sinh(\beta
 e_{max}) f_{\beta}^0(e_{max})^3 + ce_{max}\sinh^5(\beta
 e_{max})f_{\beta}^0(e_{max}),\\
&\|f_{\beta}^{xt}\|_{\infty}\|f_{\beta}^x\|_{\infty}
\left\|\frac{d}{dy}f_{\beta}^t\right\|_{\infty}\\
&\le
 ce_{max}\sinh(\beta
 e_{max})\sum_{j=3}^4f_{\beta}^0(e_{max})^j + ce_{max}\sinh^3(\beta
 e_{max})\sum_{j=2}^3f_{\beta}^0(e_{max})^j,\\
&\|f_{\beta}^{xx}\|_{\infty}\|f_{\beta}^t\|_{\infty}
\left\|\frac{d}{dy}f_{\beta}^t\right\|_{\infty}\le
 ce_{max}\sinh^3(\beta
 e_{max})\sum_{j=2}^3f_{\beta}^0(e_{max})^j,\\
&\|f_{\beta}^{tt}\|_{\infty}\|f_{\beta}^x\|_{\infty}
\left\|\frac{d}{dy}f_{\beta}^x\right\|_{\infty}\\
&\le c e_{max}\sinh(\beta e_{max}) f_{\beta}^0(e_{max})^3+ ce_{max}\sinh^3(\beta
 e_{max})f_{\beta}^0(e_{max})^2.
\end{align*}
By combining these inequalities with
 \eqref{eq_whole_function_difference} we obtain that
\begin{align}
&|f_{\beta}(y_1,y_1,y_1)-f_{\beta}(y_1,y_2,y_3)|\label{eq_whole_function_difference_application}\\
&\le c(e_{max}-e_{min})e_{max}\sinh(\beta e_{max})\notag\\
&\quad \cdot\left(
\sum_{j=3}^4f_{\beta}^0(e_{max})^j+\sinh^2(\beta
 e_{max})\sum_{j=2}^3f_{\beta}^0(e_{max})^j+\sinh^4(\beta
 e_{max})f_{\beta}^0(e_{max})\right).\notag
\end{align}
Let us bound the right-hand side of this inequality by that of
 \eqref{eq_whole_function_lower}. Let us prepare a few more inequalities
 for this purpose. We can use \eqref{eq_critical_beta_bound} to derive
 that 
\begin{align}
&\sinh(\beta e_{max})=\left(\frac{\sinh(\beta e_{max})-\sinh(\beta
 e_{min})}{\sinh(\beta e_{min})}+1\right)\sinh(\beta e_{min})\label{eq_sinh_sinh}\\
&\qquad\qquad\quad\  \le \left(\frac{\beta (e_{max}-e_{min})\cosh(\beta e_{max})}{\sinh(\beta
 e_{min})}+1\right)\sinh(\beta e_{min})\notag\\
&\qquad\qquad\quad\  \le
 \left(\left(\frac{e_{max}}{e_{min}}-1\right)\cosh\left(\frac{2e_{max}}{e_{min}}\right)+1\right)\sinh(\beta
 e_{min}),\notag\\
&f_{\beta}^0(e_{max})\le \cos\left(\frac{\tau(\beta)}{2}\right)+1+(\beta
 e_{max})^2\cosh(\beta e_{max})\label{eq_cos_cosh}\\
&\qquad\qquad\le
 \cos\left(\frac{\tau(\beta)}{2}\right)+1+2\left(\frac{e_{max}}{e_{min}}\right)^2\cosh(\beta
 e_{max})(\cosh(\beta e_{min})-1)\notag\\
&\qquad\qquad \le 2
 \left(\frac{e_{max}}{e_{min}}\right)^2\cosh\left(\frac{2e_{max}}{e_{min}}\right)f_{\beta}^0(e_{min}).\notag
\end{align}
Moreover, by \eqref{eq_critical_beta_bound} and \eqref{eq_cos_cosh}
\begin{align}
\sinh^2(\beta e_{max})&=\cosh^2(\beta e_{max})-1\label{eq_sinh_cos_cosh}\\
&\le (\cosh(\beta e_{max})+1)f_{\beta}^0(e_{max})\notag\\
&\le 2\left(\frac{e_{max}}{e_{min}}
\right)^2\cosh\left(\frac{2e_{max}}{e_{min}}\right)\left(\cosh\left(\frac{2e_{max}}{e_{min}}
\right)+1\right)f_{\beta}^0(e_{min}).\notag
\end{align}
Substitution of \eqref{eq_cos_cosh_upper}, \eqref{eq_sinh_sinh},
 \eqref{eq_cos_cosh}, \eqref{eq_sinh_cos_cosh} into
 \eqref{eq_whole_function_difference_application} yields the following
 inequality.
We especially use \eqref{eq_sinh_sinh} to bound $\sinh(\beta e_{max})$ in
 front of the large parenthesis and \eqref{eq_sinh_cos_cosh} to bound 
$\sinh^2(\beta e_{max})$, $\sinh^4(\beta e_{max})$ inside the large parenthesis.
\begin{align*}
&|f_{\beta}(y_1,y_1,y_1)-f_{\beta}(y_1,y_2,y_3)|\\
&\le
 c\left(\frac{e_{max}}{e_{min}}-1\right)\frac{e_{max}}{e_{min}}\left(\left(\frac{e_{max}}{e_{min}}-1\right)\cosh\left(\frac{2e_{max}}{e_{min}}\right)+1\right)\\
&\quad \cdot
 \Bigg(\left(\frac{e_{max}}{e_{min}}\right)^8\left(\cosh\left(\frac{2e_{max}}{e_{min}}\right)\right)^4\\
&\qquad\quad +\left(\cosh\left(\frac{2e_{max}}{e_{min}}\right)+1\right)
\left(\frac{e_{max}}{e_{min}}\right)^8
\left(\cosh\left(\frac{2e_{max}}{e_{min}}\right)\right)^4\\
&\qquad\quad +\left(\cosh\left(\frac{2e_{max}}{e_{min}}\right)+1\right)^2
\left(\frac{e_{max}}{e_{min}}\right)^6
\left(\cosh\left(\frac{2e_{max}}{e_{min}}\right)\right)^3\Bigg)\\
&\quad\cdot e_{min}^2\sinh(\beta
 e_{min}) f_{\beta}^0(e_{min})^3,\\
&(\forall y_j\in [e_{min},e_{max}]\ (j=1,2,3)).
\end{align*}
We can see that there exists $e_0\in (0,1)$ independent of any parameter
 such that if $e_{min}/e_{max}\ge e_0$, 
\begin{align}
&|f_{\beta}(y_1,y_1,y_1)-f_{\beta}(y_1,y_2,y_3)|
\le \frac{e_{min}^2}{16}\sinh(\beta
 e_{min})f_{\beta}^0(e_{min})^3,\label{eq_whole_function_difference_upper}\\
&(\forall y_j\in [e_{min},e_{max}]\ (j=1,2,3)). \notag
\end{align}
The inequalities \eqref{eq_whole_function_lower},
 \eqref{eq_whole_function_difference_upper} imply
 \eqref{eq_convexity_sufficiency} and thus the claim holds true.
\end{proof}

Proposition \ref{prop_convexity} together with Lemma
\ref{lem_boundary_function} means in particular that under the
assumptions of Proposition \ref{prop_convexity} $\tau(\cdot)$ has one
and only one local minimum point in $(0,\beta_c)$. We will see that this
property does not always hold if $e_{min}/e_{max}$ is small. To describe
the profile of $\tau(\cdot)$ in terms of number of local minimum
points, let us make clear the definition. 

\begin{definition}
Let $f$ be a real-valued function on an open interval $(a,b)$ and $c\in
 (a,b)$. The point $c$ is said to be a local minimum point of $f$ if
 there exists $\eps\in\R_{>0}$ such that $f(c)\le f(x)$ for any $x\in
 (c-\eps,c+\eps)$. 
\end{definition}

Our main goal in this section is to give a necessary and sufficient
condition for $\tau(\cdot)$ to have only one local minimum point for any
choice of $E\in \cE(e_{min},e_{max})$. The next proposition gives a
sufficient condition. 

\begin{proposition}\label{prop_sufficiency}
Assume that $e_{min}/e_{max}>\sqrt{17-12\sqrt{2}}$. Then there exists
 $U_0(b,e_{min}, e_{max})\in (0,\frac{e_{min}}{\sinh(2)b}]$ depending
 only on $b$, $e_{min}$, $e_{max}$ such that for any $U\in [-U_0(b,
 e_{min}, e_{max}), 0)$ and $E\in \cE(e_{min},e_{max})$ $\tau(\cdot)$
 has one and only one local minimum point in $(0,\beta_c)$. 
\end{proposition}

\begin{remark}
According to the proof of the proposition, $U_0(b,e_{min}, e_{max})$ is
 equal to 
\begin{align*}
\frac{c'
 \frac{e_{min}^2}{e_{max}}\big(\big(\frac{e_{min}}{e_{max}}\big)^2-17+12\sqrt{2}\big)}{\sinh(2)b\cosh^2\big(2c''
 \frac{e_{max}}{e_{min}}\big)\cosh^2\big(c''\frac{e_{max}}{e_{min}}\big)}
\end{align*}
with generic constants $c'\in (0,1]$, $c''\in \R_{>0}$. 
More specifically, $U_0(b,e_{min},e_{max})$ is given by the right-hand side of \eqref{eq_coupling_closeness_to_one}.
\end{remark}

Let us prepare an essential part of the proof of Proposition
\ref{prop_sufficiency} separately in the next lemma. Define the function
$u:\R_{>0}\times [-1,1]\times \R_{>0}\to \R$ by 
\begin{align}
u(x,y,z):=\frac{\sinh(xz)}{(y+\cosh(xz))z}.\label{eq_temporal_function}
\end{align}

\begin{lemma}\label{lem_sufficiency_lemma}
Assume that $\sqrt{17-12\sqrt{2}}<e_{min}/e_{max}<1$. Then there exists
 $c_1\in \R_{>0}$ independent of any parameter such that for any
 $(x,y)\in \R_{>0}\times (-1,0)$ satisfying 
\begin{align*}
\frac{|y+1|}{1-|y+1|}<c_1
\frac{\frac{e_{min}}{e_{max}}\big(\big(\frac{e_{min}}{e_{max}}\big)^2-17+12\sqrt{2}\big)}{\cosh^2(2x)\cosh^2(x)}
\end{align*}
and $e_1$, $e_2\in \R_{>0}$ satisfying $e_{max}\ge e_1 > e_2\ge
 e_{min}$, 
\begin{align}
&\frac{\partial u}{\partial x}\left(\sqrt{y+1}\cdot\frac{x}{e_1},y,e_1\right)
\frac{\partial^2 u}{\partial x^2}\left(\sqrt{y+1}\cdot\frac{x}{e_1},y,e_2\right)\label{eq_sufficiency_key_inequality}\\
&-\frac{\partial^2 u}{\partial
 x^2}\left(\sqrt{y+1}\cdot\frac{x}{e_1},y,e_1\right)
\frac{\partial u}{\partial
 x}\left(\sqrt{y+1}\cdot\frac{x}{e_1},y,e_2\right)>0.\notag
\end{align}
\end{lemma}

\begin{proof}
Define the function $v:\R_{>0}\times (-1,0)\times \R_{>0}\to \R$ by 
\begin{align*}
v(x,y,z)
&:=\frac{1}{x(y+1)^{\frac{7}{2}}}\Bigg(
 z\sinh(\sqrt{y+1}\cdot xz)\big(y^2-\cosh(\sqrt{y+1}\cdot xz)y
 -2\big)\\
&\qquad\qquad\qquad\quad\cdot \big(\cosh(\sqrt{y+1}\cdot x)y+1\big)\big(\cosh(\sqrt{y+1}\cdot x)+y\big)\\
&\qquad\qquad\qquad\quad-\sinh(\sqrt{y+1}\cdot x)\big(y^2-\cosh(\sqrt{y+1}\cdot x)y
 -2\big)\\
&\qquad\qquad\qquad\qquad\cdot\big(\cosh(\sqrt{y+1}\cdot xz)y+1\big)\big(\cosh(\sqrt{y+1}\cdot xz)+y\big)
\Bigg).
\end{align*}
Let us observe that for any $(x,y)\in \R_{>0}\times (-1,0)$
\begin{align}
&(\text{L.H.S of
 }\eqref{eq_sufficiency_key_inequality})=\frac{e_1x(y+1)^{\frac{7}{2}}}{\prod_{j=1}^2\big(\cosh\big(\sqrt{y+1}\cdot
 x\frac{e_j}{e_1}\big)+y\big)^3}\cdot v\left(x,y,\frac{e_2}{e_1}\right).\label{eq_key_inequality_transform}
\end{align}
We can also derive that 
\begin{align}
&v(x,y,z)\label{eq_analytic_continuation_entire}\\
&=z\left(z+\sum_{n=1}^{\infty}\frac{1}{(2n+1)!}(y+1)^nz^{2n+1}x^{2n}\right)
\left(y-2-y\sum_{n=1}^{\infty}\frac{1}{(2n)!}(y+1)^{n-1}z^{2n}x^{2n}\right)\notag\\
&\qquad\cdot\left(1+y\sum_{n=1}^{\infty}\frac{1}{(2n)!}(y+1)^{n-1}x^{2n}\right)
\left(1+\sum_{n=1}^{\infty}\frac{1}{(2n)!}(y+1)^{n-1}x^{2n}\right)\notag\\
&\quad -\left(1+\sum_{n=1}^{\infty}\frac{1}{(2n+1)!}(y+1)^nx^{2n}\right)
\left(y-2-y\sum_{n=1}^{\infty}\frac{1}{(2n)!}(y+1)^{n-1}x^{2n}\right)\notag\\
&\qquad\quad \cdot \left(1+y\sum_{n=1}^{\infty}\frac{1}{(2n)!}(y+1)^{n-1}z^{2n}x^{2n}\right)
\left(1+\sum_{n=1}^{\infty}\frac{1}{(2n)!}(y+1)^{n-1}z^{2n}x^{2n}\right).\notag
\end{align}
This expansion implies that the function $v(\cdot,\cdot,\cdot)$ can be
 analytically continued into $\C^3$. By abusing notation we let
 $v(x,y,z)$ denote the entire function defined by the right-hand side of
 \eqref{eq_analytic_continuation_entire} as well. It follows from the
 assumption that for any $x\in \R$
\begin{align}
&v\left(x,-1,\frac{e_2}{e_1}\right)\label{eq_key_core_lower}\\
&=3\left(\frac{e_2}{e_1}\right)^2\Bigg(1-\left(\frac{e_2}{e_1}\right)^2\Bigg)
\Bigg(
\Bigg(\frac{x^2}{2}-\frac{1+\big(\frac{e_2}{e_1}\big)^2}{6\big(\frac{e_2}{e_1}\big)^2}\Bigg)^2+\frac{-\big(\frac{e_2}{e_1}\big)^4+34\big(\frac{e_2}{e_1}\big)^2-1}{36\big(\frac{e_2}{e_1}\big)^4}
\Bigg)\notag\\
&\ge
 \frac{1-\big(\frac{e_2}{e_1}\big)^2}{12\big(\frac{e_2}{e_1}\big)^2}
\Bigg(17+12\sqrt{2}-\left(\frac{e_2}{e_1}\right)^2\Bigg)
\Bigg(\left(\frac{e_2}{e_1}\right)^2-17+12\sqrt{2}\Bigg)\notag\\
&\ge
 \left(1-\frac{e_2}{e_1}\right)\Bigg(\left(\frac{e_{min}}{e_{max}}\right)^2-17+12\sqrt{2}\Bigg)>0.\notag
\end{align}
Also, the Taylor expansion and the Cauchy formula yield that for any
 $x\in\C$, $y\in (-1,0)$
\begin{align*}
v\left(x,y,\frac{e_2}{e_1}\right)
&=v\left(x,-1,\frac{e_2}{e_1}\right)+\sum_{m=1}^{\infty}\frac{1}{2\pi
 i}\oint_{|\zeta+1|=1}d\zeta\frac{v\big(x,\zeta,\frac{e_2}{e_1}\big)}{(\zeta+1)^{m+1}}(y+1)^m\\
&=v\left(x,-1,\frac{e_2}{e_1}\right)\\
&\quad +\sum_{m=1}^{\infty}\sum_{n=1}^{\infty}
\frac{1}{(2\pi
 i)^2}\oint_{|\zeta+1|=1}d\zeta\oint_{|\xi-1|=1}d\xi 
\frac{v(x,\zeta,\xi)}{(\zeta+1)^{m+1}(\xi-1)^{n+1}}\\
&\qquad\qquad\quad \cdot (y+1)^m\left(\frac{e_2}{e_1}-1\right)^n.
\end{align*}
In the second equality we used the fact that $v(x,y,1)=0$ for any $x$,
 $y\in \C$. Moreover, by considering \eqref{eq_analytic_continuation_entire}
we can see that for any $x\in \R_{>0}$,
 $y\in (-1,0)$
\begin{align*}
\Bigg|v\left(x,y,\frac{e_2}{e_1}\right)-
 v\left(x,-1,\frac{e_2}{e_1}\right)\Bigg|&\le
 c\cosh^2(2x)\cosh^2(x)\sum_{m=1}^{\infty}|y+1|^m\sum_{n=1}^{\infty}\left|\frac{e_2}{e_1}-1\right|^n\\
&\le c
 \cosh^2(2x)\cosh^2(x)\frac{e_{max}}{e_{min}}\left|\frac{e_2}{e_1}-1\right|\frac{|y+1|}{1-|y+1|},
\end{align*}
which combined with \eqref{eq_key_core_lower} implies that for any $x\in
 \R_{>0}$, $y\in (-1,0)$
\begin{align}
&v\left(x,y,\frac{e_2}{e_1}\right)\label{eq_key_core_final_lower}\\
&\ge \left(1-\frac{e_2}{e_1}\right)\Bigg(
\left(\frac{e_{min}}{e_{max}}\right)^2-17+12\sqrt{2}-c\frac{e_{max}}{e_{min}}\cosh^2(2x)\cosh^2(x)\frac{|y+1|}{1-|y+1|}
\Bigg).\notag
\end{align}
We can deduce the claim from \eqref{eq_key_inequality_transform},
 \eqref{eq_key_core_final_lower}. 
\end{proof}

In the following we let $\cosh^{-1}$ $(:\R_{\ge 1}\to \R_{\ge 0})$
denote the inverse function of $\cosh|_{\R_{\ge 0}}:\R_{\ge 0}\to
\R_{\ge 1}$. 
\begin{proof}[Proof of Proposition \ref{prop_sufficiency}]
Let us fix $L\in \N$ and $y\in (-1,-1/2]$. Define the function
 $F_L:\R\to \R$ by 
\begin{align*}
F_L(x):=\frac{1}{L^d}\sum_{\bk\in\G^*}\Tr\left(\frac{\sinh(x
 E(\bk))}{(y+\cosh(x E(\bk)))E(\bk)}\right).
\end{align*}
There are $e_j\in [e_{min},e_{max}]$ $(j=1,2,\cdots,bL^d)$ such that 
$e_{max}\ge e_1\ge e_2\ge \cdots \ge e_{bL^d}\ge e_{min}$ and 
$$
F_L(x)=\frac{1}{L^d}\sum_{j=1}^{bL^d}u(x,y,e_j),
$$
where $u(\cdot)$ is the function defined in
 \eqref{eq_temporal_function}. Let us prove that
\begin{align}
&\exists x_0\in
 \left[\frac{1}{e_{max}}\cosh^{-1}(|y|^{-1}),\frac{1}{e_{min}}\cosh^{-1}(|y|^{-1})\right]\label{eq_finite_volume_claim}\\
&\quad \text{ s.t. }\begin{array}{l} \frac{d}{dx}F_L(x)>0,\quad (\forall x\in (0,x_0)),\\
                               \frac{d}{dx}F_L(x_0)=0,\\
                               \frac{d}{dx}F_L(x)<0,\quad (\forall x\in
				(x_0,\infty)). 
\end{array}
\notag
\end{align}
We can check by calculation that for any $z\in \R_{>0}$
\begin{align}
&\frac{\partial u}{\partial x}(x,y,z)>0,\quad \left(\forall x\in
 \left(0,\frac{1}{z}\cosh^{-1}(|y|^{-1})\right)\right),\label{eq_temporal_function_derivative}\\
&\frac{\partial u}{\partial
 x}\left(\frac{1}{z}\cosh^{-1}(|y|^{-1}),y,z\right)=0,\notag\\
&\frac{\partial u}{\partial x}(x,y,z)<0,\quad \left(\forall x\in
 \left(\frac{1}{z}\cosh^{-1}(|y|^{-1}),\infty\right)\right),\notag\\
&\frac{\partial^2 u}{\partial
 x^2}\left(\frac{1}{z}\cosh^{-1}(|y|^{-1}),y,z\right)<0.\notag
\end{align}
Thus, if $e_1=e_{bL^d}$, the claim \eqref{eq_finite_volume_claim} holds
 with $x_0=\frac{1}{e_1}\cosh^{-1}(|y|^{-1})$. Let us assume that
 $e_1>e_{bL^d}$. This obviously implies that $e_{max}>e_{min}$. We can
 deduce from \eqref{eq_temporal_function_derivative} that
\begin{align*}
&\frac{d}{dx}F_L(x)>0,\quad \left(\forall x\in
 \left(0,\frac{1}{e_1}\cosh^{-1}(|y|^{-1})\right]\right),\\
&\frac{d}{dx}F_L(x)<0,\quad \left(\forall x\in
 \left[\frac{1}{e_{bL^d}}\cosh^{-1}(|y|^{-1}),\infty\right)\right).
\end{align*}
Thus there exists $x_0\in
 (\frac{1}{e_1}\cosh^{-1}(|y|^{-1}),\frac{1}{e_{bL^d}}\cosh^{-1}(|y|^{-1}))$
 such that $\frac{d}{dx}F_L(x_0)=0$. Set
\begin{align}
c_{max}:=\sup_{y\in
 (-1,-\frac{1}{2}]}\frac{\cosh^{-1}(|y|^{-1})}{\sqrt{y+1}}.\label{eq_definition_c_max}
\end{align}
By using the equality
\begin{align}
\cosh^{-1}(|y|^{-1})=\log(|y|^{-1}+\sqrt{|y|^{-2}-1}),\label{eq_arc_cosh}
\end{align}
one can confirm that $0<c_{max}<\infty$. It follows that 
\begin{align*}
\left|\frac{e_j}{\sqrt{y+1}}x_0\right|\le
 c_{max}\frac{e_{max}}{e_{min}},\quad (\forall j\in \{1,\cdots, bL^d\}).
\end{align*}
Then we can apply Lemma \ref{lem_sufficiency_lemma} to conclude that if
\begin{align}
&|y+1|<\frac{c_1}{2}\cdot
 \frac{\frac{e_{min}}{e_{max}}\big(\big(\frac{e_{min}}{e_{max}}\big)^2-17+12\sqrt{2}\big)}{\cosh^2\big(2c_{max}\frac{e_{max}}{e_{min}}\big)\cosh^2\big(c_{max}\frac{e_{max}}{e_{min}}\big)}\label{eq_closeness_to_minus_one}
\end{align}
and $e_j>e_{bL^d}$, 
\begin{align*}
\frac{\partial u}{\partial x}(x_0,y,e_j)\frac{\partial^2 u}{\partial
 x^2}(x_0,y,e_{bL^d})- \frac{\partial^2 u}{\partial
 x^2}(x_0,y,e_j)\frac{\partial u}{\partial x}(x_0,y,e_{bL^d})>0.
\end{align*}
Since $e_1>e_{bL^d}$, this implies that 
\begin{align*}
&\frac{\partial u}{\partial
 x}(x_0,y,e_{bL^d})\frac{d^2}{dx^2}F_L(x_0)=\frac{1}{L^d}\sum_{j=1}^{bL^d}\frac{\partial^2
 u}{\partial x^2}(x_0,y,e_j)\frac{\partial u}{\partial
 x}(x_0,y,e_{bL^d})\\
&<\frac{1}{L^d}\sum_{j=1}^{bL^d}\frac{\partial u}{\partial
 x}(x_0,y,e_j)\frac{\partial^2 u}{\partial
 x^2}(x_0,y,e_{bL^d})=\frac{d}{dx}F_L(x_0)\frac{\partial^2 u}{\partial
 x^2}(x_0,y,e_{bL^d})=0.
\end{align*}
Since $x_0 \in (0,\frac{1}{e_{bL^d}}\cosh^{-1}(|y|^{-1}))$, $\frac{\partial u}{\partial
 x}(x_0,y,e_{bL^d})>0$ by \eqref{eq_temporal_function_derivative}. Thus
 we obtain that $\frac{d^2}{dx^2}F_L(x_0)<0$. It follows from the above
 argument that if $e_1>e_{bL^d}$ and 
\eqref{eq_closeness_to_minus_one} holds, the claim
 \eqref{eq_finite_volume_claim} holds. 
This can be confirmed as follows. Suppose that $x_1,x_2\in
 [\frac{1}{e_{max}}\cosh^{-1}(|y|^{-1}),\frac{1}{e_{min}}\cosh^{-1}(|y|^{-1})]$,
 $x_1<x_2$ and $\frac{d}{dx}F_L(x_j)=0$ for $j=1,2$. Since the function
 $\frac{d}{dx}F_L(\cdot)$ is non-constant and real analytic in
 $\R_{>0}$, 
$$
\sharp \left\{x\in [x_1,x_2]\ \big|\ \frac{d}{dx}F_L(x)=0\right\}<\infty.
$$
Thus, there exists $x_3\in (x_1,x_2]$ such that $\frac{d}{dx}F_L(x_3)=0$
 and $\frac{d}{dx}F_L(x)\neq 0$ for any $x\in (x_1,x_3)$. Since
 $\frac{d^2}{dx^2}F_L(x_j)<0$ for $j=1,3$, there exists $x_4\in
 (x_1,x_3)$ such that $\frac{d}{dx}F_L(x_4)=0$, which is a
 contradiction. 
Now we can conclude that 
under the assumption
 \eqref{eq_closeness_to_minus_one} the claim
 \eqref{eq_finite_volume_claim} holds.

Define the function $F_{\infty}:\R\times (-1,0)\to \R$ by
\begin{align}
F_{\infty}(x,y):=D_d\int_{\G_{\infty}^*}d\bk \Tr\left(\frac{\sinh(x
 E(\bk))}{(y+\cosh(x
 E(\bk)))E(\bk)}\right).\label{eq_infinite_volume_trace}
\end{align}
Since $E\in C^{\infty}(\R^d,\Mat(b,\C))$, for any $y\in
 (-1,-\frac{1}{2}]$ $\frac{d}{dx}F_L(\cdot)$ converges to
 $\frac{\partial F_{\infty}}{\partial x}(\cdot,y)$ locally uniformly as
 $L\to \infty$. Therefore if $y\in (-1,-\frac{1}{2}]$ satisfies
 \eqref{eq_closeness_to_minus_one}, there exists $\hat{x}\in
 [\frac{1}{e_{max}}\cosh^{-1}(|y|^{-1}),\frac{1}{e_{min}}\cosh^{-1}(|y|^{-1})]$
 such that 
\begin{align}
&\frac{\partial F_{\infty}}{\partial x}(x,y)\ge 0,\quad (\forall x\in
 (0,\hat{x})),\label{eq_infinite_volume_claim}\\
&\frac{\partial F_{\infty}}{\partial x}(\hat{x},y)= 0,\notag\\
&\frac{\partial F_{\infty}}{\partial x}(x,y)\le 0,\quad (\forall x\in
 (\hat{x},\infty)).\notag
\end{align}
Let us recall that the assumption \eqref{eq_tighter_coupling} implies
 \eqref{eq_boundary_gap_equation_direct} and \eqref{eq_cosign_lower}. 
If we assume that 
\begin{align}
&|U|\le 
 \frac{\min\left\{1,\frac{c_1}{2}\right\}\frac{e_{min}^2}{e_{max}}\big(\big(\frac{e_{min}}{e_{max}}\big)^2-17+12\sqrt{2}\big)}{\sinh(2)b\cosh^2\big(2c_{max}\frac{e_{max}}{e_{min}}\big)\cosh^2\big(c_{max}\frac{e_{max}}{e_{min}}\big)},\label{eq_coupling_closeness_to_one}
\end{align}
\eqref{eq_tighter_coupling} holds. Thus, by \eqref{eq_cosign_lower}
$\cos(\tau(\beta)/2)\in (-1,-1/2]$ for all $\beta\in
 (0,\beta_c)$. Moreover, \eqref{eq_boundary_gap_equation_direct} and
 \eqref{eq_coupling_closeness_to_one} again ensure that 
\eqref{eq_closeness_to_minus_one} holds with $y=\cos(\tau(\beta)/2)$
 for any $\beta\in (0,\beta_c)$. Let us note that the right-hand side of
 \eqref{eq_coupling_closeness_to_one} does not depend on $E$ $(\in
 \cE(e_{min},e_{max}))$.
These properties combined with
 \eqref{eq_infinite_volume_claim} imply that on the assumption
 \eqref{eq_coupling_closeness_to_one} for any $E\in
 \cE(e_{min},e_{max})$, 
$\beta\in (0,\beta_c)$
 there exists $\tilde{x}\in\R_{>0}$ such that 
\begin{align}
&\frac{\partial g_E}{\partial x}(x,\tau(\beta),0)\ge 0,\quad (\forall
 x\in (0,\tilde{x})),\label{eq_infinite_volume_application}\\
&\frac{\partial g_E}{\partial x}(\tilde{x},\tau(\beta),0)= 0,\notag\\
&\frac{\partial g_E}{\partial x}(x,\tau(\beta),0)\le 0,\quad (\forall
 x\in (\tilde{x},\infty)).\notag
\end{align}

Finally let us prove that $\tau(\cdot)$ has one and only one local
 minimum point in $(0,\beta_c)$. Suppose that
 $0<\beta_1<\beta_2<\beta_c$ and $\beta_1$, $\beta_2$ are local minimum
 points. If $\tau(\beta_1)\le \tau(\beta_2)$, there exist $\beta_1'$,
 $\beta_2'$, $\beta_3'\in (0,\beta_2]$ such that
 $\beta_1'<\beta_2'<\beta_3'$ and
 $\tau(\beta_1')=\tau(\beta_2')=\tau(\beta_3')$. If
 $\tau(\beta_1)>\tau(\beta_2)$, we can take such $\beta_1'$,
 $\beta_2'$, $\beta_3'$ from $[\beta_1,\beta_c)$. It follows that
 $g_E(\beta_j',\tau(\beta_1'),0)=0$ for all $j\in \{1,2,3\}$. By
 \eqref{eq_infinite_volume_application} there exists $\tilde{x}\in
 \R_{>0}$ such that 
\begin{align*}
&\frac{\partial g_E}{\partial x}(x,\tau(\beta_1'),0)\ge 0,\quad (\forall
 x\in (0,\tilde{x})),\\ 
&\frac{\partial g_E}{\partial x}(\tilde{x},\tau(\beta_1'), 0)=0,\\
&\frac{\partial g_E}{\partial x}(x,\tau(\beta_1'),0)\le 0,\quad (\forall
 x\in (\tilde{x},\infty)).
\end{align*}
If $\tilde{x}\in (0,\beta_2']$, the function $x\mapsto
 g_E(x,\tau(\beta_1'),0)$ must be identically zero in
 $[\beta_2',\beta_3']$. Since this function is real analytic in
 $\R_{>0}$,
the identity theorem 
ensures that this function is identically zero in $\R_{>0}$,
 which is a contradiction. If $\tilde{x}\in (\beta_2',\infty)$, this
 function must be identically zero in $[\beta_1',\beta_2']$, which also
 leads to a contradiction. Therefore, if $\tau(\cdot)$ has a local
 minimum point in $(0,\beta_c)$, it must be unique. Let us define the
 function $\hat{\tau}(\cdot):[0,\beta_c]\to\R$ as
 follows. $\hat{\tau}(x):=2\pi$ for $x\in \{0,\beta_c\}$,
 $\hat{\tau}(x):=\tau(x)$ for $x\in (0,\beta_c)$. By Lemma
 \ref{lem_boundary_function}, $\hat{\tau}\in C([0,\beta_c])$ and
 $\hat{\tau}(x)\le \hat{\tau}(0)=\hat{\tau}(\beta_c)$ for any $x\in
 [0,\beta_c]$. Thus $\hat{\tau}(\cdot)$ attains its global minimum in
 $(0,\beta_c)$, which implies that $\tau(\cdot)$ has a local minimum
 point in $(0,\beta_c)$. The proof is complete. 
\end{proof}

Next we will prove that the conclusion of Proposition \ref{prop_sufficiency}
does not hold if $e_{min}/e_{max}\le \sqrt{17-12\sqrt{2}}$. We divide
the problem into two cases, $e_{min}/e_{max}= \sqrt{17-12\sqrt{2}}$ or
$e_{min}/e_{max}<\sqrt{17-12\sqrt{2}}$. The following proposition states
the result for the case that the equality holds.

\begin{proposition}\label{prop_necessity_equality}
Assume that $e_{min}/e_{max}= \sqrt{17-12\sqrt{2}}$. Then for any $d$,
 $b\in \N$, basis $(\hbv_j)_{j=1}^d$ of $\R^d$, $U_0\in (0,2e_{min}/b)$
 there exist $U\in [-U_0,0)$ and $E\in\cE(e_{min}$, $e_{max})$ such that 
$\tau(\cdot)$ has more than one local minimum points in $(0,\beta_c)$. 
\end{proposition}

\begin{remark}
We should stress that in our proof we construct such $E(\in
 \cE(e_{min},$ $e_{max}))$ depending on $U_0$. On the contrary, we will
 construct  $E(\in
 \cE(e_{min},e_{max}))$ independently of the magnitude of the coupling
 constant when we deal with the case
 $e_{min}/e_{max}<\sqrt{17-12\sqrt{2}}$ in Proposition
 \ref{prop_necessity_inequality}.
\end{remark}

Let us show a lemma which we need to prove the above proposition. Set
\begin{align}
D:=\left\{(x,y,z)\in\R_{>0}\times(-1,0)\times\R_{>0}\ \Big|\ x
 <\frac{1}{2z(y+1)}
(\cosh^{-1}(|y|^{-1}))^2\right\}.
\label{eq_domain_D}
\end{align}
Define the function $w:D\to \R$ by
\begin{align}
w(x,y,z):=-\frac{(1+y\cosh(\sqrt{y+1}\sqrt{2x}))(y+\cosh(\sqrt{y+1}\sqrt{2zx}))^2}
{(1+y\cosh(\sqrt{y+1}\sqrt{2zx}))(y+\cosh(\sqrt{y+1}\sqrt{2x}))^2}.
\label{eq_internal_function}
\end{align}
The necessary lemma concerns properties of the function $w$. For
$(x,y,z)\in D$ we can rewrite as follows. 
\begin{align}
w(x,y,z)=-\frac{\left(1+y\sum_{m=1}^{\infty}\frac{(y+1)^{m-1}}{(2m)!}2^mx^m\right)\left(
1+\sum_{n=1}^{\infty}\frac{(y+1)^{n-1}}{(2n)!}2^nz^nx^n\right)^2}{
\left(1+y\sum_{m=1}^{\infty}\frac{(y+1)^{m-1}}{(2m)!}2^mz^mx^m\right)\left(
1+\sum_{n=1}^{\infty}\frac{(y+1)^{n-1}}{(2n)!}2^nx^n\right)^2}.
\label{eq_internal_function_continuation}
\end{align}
Define the open set $\tilde{D}$ of $\C^3$ by
\begin{align}
\tilde{D}:=\Bigg\{&(x,y,z)\in\C^3\ \Bigg|\ \label{eq_continued_domain_D}\\
&\left|1+y\sum_{m=1}^{\infty}\frac{(y+1)^{m-1}}{(2m)!}2^mz^mx^m\right|
\left|1+\sum_{n=1}^{\infty}\frac{(y+1)^{n-1}}{(2n)!}2^nx^n\right|^2>0
\Bigg\}.\notag
\end{align}
Then we can define the analytic function $\tilde{w}:\tilde{D}\to \C$ by
the right-hand side of \eqref{eq_internal_function_continuation}. It
follows that $\tilde{w}|_D=w$. It will often be more convenient to deal
with $\tilde{w}$ than $w$ during our construction. Note that for $z\in
\R_{>0}$ and $x\in (0,z^{-1})$, $(x,-1,z)\in \tilde{D}$. We 
will particularly use the following equalities. For $z\in
\R_{>0}$ and $x\in (0,z^{-1})$
\begin{align}
&\tilde{w}(x,-1,z)=\frac{(x-1)(1+zx)^2}{(1-zx)(1+x)^2},\label{eq_extended_function}\\
&\frac{\partial\tilde{w}}{\partial
 x}(x,-1,z)=\frac{3z(1-z)(1+zx)}{(1-zx)^2(1+x)^3}\left(x^2-\frac{z+1}{3z}x+\frac{1}{z}\right),\label{eq_extended_function_one_derivative}\\
&\frac{\partial\tilde{w}}{\partial
 y}(x,-1,z)=
 -\frac{x(1+zx)}{6(1-zx)^2(x+1)^3}\label{eq_extended_function_two_derivative}\\
&\qquad\qquad\qquad\qquad \cdot ((6+3x+x^2)(1-z^2x^2)+z(x^2-1)(6+3zx +
 z^2x^2)).\notag
\end{align}
To shorten subsequent formulas, let us set $a_0:=3+2\sqrt{2}$,
$\eta_0:=17-12\sqrt{2}$.
\begin{lemma}\label{lem_necessity_equality_lemma}
There exists $y_0\in (-1,0)$ such that for any $y\in (-1,y_0]$ 
\begin{align}
&\frac{1}{2(y+1)}(\cosh^{-1}(|y|^{-1}))^2<a_0<\frac{1}{2\eta_0(y+1)}(\cosh^{-1}(|y|^{-1}))^2,\label{eq_internal_location}\\
&0< w(a_0,y,\eta_0)<1.\label{eq_internal_function_positive}
\end{align}
Moreover, there exist
\begin{align*}
&x_1(y)\in
 \left(\frac{1}{2(y+1)}(\cosh^{-1}(|y|^{-1}))^2,a_0\right),\\
&x_2(y)\in \left(a_0,\frac{1}{2\eta_0(y+1)}(\cosh^{-1}(|y|^{-1}))^2\right)
\end{align*}
such that
\begin{align*}
&w(x_1(y),y,\eta_0)=w(a_0,y,\eta_0)=w(x_2(y),y,\eta_0),\\
&w(x,y,\eta_0)>w(a_0,y,\eta_0),\quad (\forall x\in (x_1(y),a_0)),\\
&w(x,y,\eta_0)<w(a_0,y,\eta_0),\quad (\forall x\in (a_0,x_2(y))).
\end{align*}
\end{lemma}

\begin{proof}
The following equalities are useful. 
\begin{align}
&\eta_0^2-34\eta_0+1=0,\label{eq_eta_quadratic}\\
&a_0=\frac{\eta_0+1}{6\eta_0},\label{eq_a_eta_relation}\\
&a_0(\eta_0+1)=6,\label{eq_a_eta_product}\\
&a_0^2=\frac{1}{\eta_0},\label{eq_a_quadratic}\\
&a_0^2-\frac{\eta_0+1}{3\eta_0}a_0+\frac{1}{\eta_0}=0.\label{eq_a_eta_quadratic}
\end{align}
We can deduce from \eqref{eq_arc_cosh} that 
\begin{align}
&\lim_{y\searrow
 -1}\frac{1}{2(y+1)}(\cosh^{-1}(|y|^{-1}))^2=1<a_0<\frac{1}{\eta_0}=\lim_{y\searrow -1}\frac{1}{2\eta_0(y+1)}(\cosh^{-1}(|y|^{-1}))^2.\label{eq_arc_cosh_limit_relation}
\end{align}
This implies that there exists $\eps\in\R_{>0}$ such that
$(a_0,y,\eta_0)\in D$ for any $y\in (-1,-1+\eps)$. Moreover,
 $(a_0,-1,\eta_0)\in \tilde{D}$. By multiplying both the denominator
 and the numerator of \eqref{eq_extended_function} by $a_0^2$ and using
 \eqref{eq_a_quadratic} we can derive that 
\begin{align}
\tilde{w}(a_0,-1,\eta_0)=\frac{1}{a_0}.\label{eq_extended_function_a}
\end{align}
Thus, there exists $y_1\in (-1,-1+\eps)$ such that for any $y\in
 (-1,y_1]$ \eqref{eq_internal_location} and
 \eqref{eq_internal_function_positive} hold. Also, by
 \eqref{eq_extended_function_one_derivative} and \eqref{eq_a_eta_quadratic}
$\frac{\partial \tilde{w}}{\partial x}(a_0,-1,\eta_0)=0$. Next let us
 compute $\frac{\partial^2\tilde{w}}{\partial x\partial
 y}(a_0,-1,\eta_0)$. The computation can be quite complicated if we
 follow a wrong way. Let us present right steps leading to a concise
 formula, though this would not be the only approach. 
Let us decompose the right-hand side of
 \eqref{eq_extended_function_two_derivative} as follows. 
\begin{align*}
&\frac{\partial \tilde{w}}{\partial y}(x,-1,\eta_0)=w_1(x)w_2(x),\\
&w_1(x):=-\frac{x(1+\eta_0x)}{6(1-\eta_0x)^2(x+1)^3},\\
&w_2(x):=(6+3x +
 x^2)(1-\eta^2_0x^2)+\eta_0(x^2-1)(6+3\eta_0x+\eta_0^2x^2).
\end{align*}
Using \eqref{eq_a_quadratic}, \eqref{eq_a_eta_relation}, \eqref{eq_a_eta_product},
\eqref{eq_a_quadratic} in this order, we obtain
 that
\begin{align}
\frac{d
 w_1}{dx}(a_0)&=-\frac{1+(3\eta_0-2)a_0+3\eta_0a_0^2+3\eta_0^2a_0^3}{6(1-\eta_0a_0)^3(a_0+1)^4}=-\frac{5+\eta_0-2a_0}{6(1-\eta_0a_0)^3(a_0+1)^4}\label{eq_first_term_derivative}\\
&=\frac{1-\eta_0}{6(1-\eta_0a_0)^3(a_0+1)^3}=-\frac{w_1(a_0)}{a_0}.\notag
\end{align}
By using \eqref{eq_eta_quadratic} and \eqref{eq_a_quadratic} repeatedly
\begin{align}
w_2(a_0)=(1-\eta_0)(46+3(1+\eta_0)a_0).\label{eq_second_term_direct}
\end{align}
By using \eqref{eq_a_quadratic} only, 
\begin{align*}
\frac{d w_2}{d x}(a_0)=(1-\eta_0)(2(\eta_0^2+5\eta_0+1)a_0+3(1+\eta_0)).
\end{align*}
Then by using \eqref{eq_eta_quadratic} and \eqref{eq_a_quadratic} again
\begin{align}
&a_0\frac{d
 w_2}{dx}(a_0)=(1-\eta_0)(78+3(1+\eta_0)a_0).\label{eq_second_term_derivative}
\end{align}
By combining \eqref{eq_first_term_derivative},
 \eqref{eq_second_term_direct}, \eqref{eq_second_term_derivative} and
 using \eqref{eq_a_quadratic} once
\begin{align*}
\frac{\partial^2\tilde{w}}{\partial x\partial
 y}(a_0,-1,\eta_0)&=-\frac{w_1(a_0)}{a_0}\left(
w_2(a_0)-a_0\frac{dw_2}{dx}(a_0)\right)=32(1-\eta_0)\frac{w_1(a_0)}{a_0}\\
&=-\frac{16(1-\eta_0)^2}{3(1-\eta_0a_0)^3(a_0+1)^3}.
\end{align*}

Since $\frac{\partial^2\tilde{w}}{\partial x\partial
 y}(a_0,-1,\eta_0)<0$, there exists $y_2\in (-1,y_1]$ such that 
\begin{align*}
\frac{\partial^2\tilde{w}}{\partial x\partial
 y}(a_0,-1,\eta_0)+\sup_{t\in
 [-1,y_1]}\left|\frac{\partial^3\tilde{w}}{\partial x\partial
 y^2}(a_0,t,\eta_0)\right|(y_2+1)<0.
\end{align*}
Since $\frac{\partial\tilde{w}}{\partial x}(a_0,-1,\eta_0)=0$, this
 estimate ensures that for any $y\in (-1,y_2]$
\begin{align}
\frac{\partial\tilde{w}}{\partial x}(a_0,y,\eta_0)&=
\frac{\partial^2\tilde{w}}{\partial x\partial
 y}(a_0,-1,\eta_0)(y+1)+\int_{-1}^ydt (y-t)\frac{\partial^3\tilde{w}}{\partial x\partial
 y^2}(a_0,t,\eta_0)\label{eq_internal_field_dependent}\\
&\le \left(\frac{\partial^2\tilde{w}}{\partial x\partial
 y}(a_0,-1,\eta_0)+\sup_{t\in [-1,y_1]}\left|\frac{\partial^3\tilde{w}}{\partial x\partial
 y^2}(a_0,t,\eta_0)\right|(y_2+1)\right)(y+1)\notag\\
&<0.\notag
\end{align}
Let us fix $y\in (-1,y_2]$. Observe that 
\begin{align*}
&\lim_{x\searrow
 \frac{1}{2(y+1)}(\cosh^{-1}(|y|^{-1}))^2}w(x,y,\eta_0)=0,\quad 
\lim_{x\nearrow
 \frac{1}{2\eta_0(y+1)}(\cosh^{-1}(|y|^{-1}))^2}w(x,y,\eta_0)=\infty,
\end{align*}
which combined with the inequality \eqref{eq_internal_field_dependent}
 imply the existence of $x_1(y)$, $x_2(y)$ with the claimed properties.
\end{proof}

Define the function $W:\R_{>0}\times (-1,0)\times \R_{>0}\times
\R_{>0}\to \R$ by 
\begin{align}
W(x,y,z,s):=\frac{\sinh(x)}{y+\cosh(x)}+s\frac{\sinh(zx)}{(y+\cosh(zx))z}.\label{eq_large_function}
\end{align}
We will use this function and the functions $w:D\to \R$,
$\tilde{w}:\tilde{D}\to\C$ in the rest of this section mainly for
organizing proofs. 

\begin{proof}[Proof of Proposition \ref{prop_necessity_equality}]
By Lemma \ref{lem_necessity_equality_lemma} there exists $y_0\in (-1,0)$
 such that for any $y\in (-1,y_0]$ 
\begin{align}
&\frac{1}{\sqrt{y+1}}\cosh^{-1}(|y|^{-1})<\sqrt{2
 a_0}<\frac{1}{\sqrt{\eta_0(y+1)}}\cosh^{-1}(|y|^{-1}),\label{eq_critical_a_interval}\\
&0<w(a_0,y,\eta_0)<1.\notag
\end{align}
Observe that by \eqref{eq_critical_a_interval} and the inequality
 $\sinh(x)\ge x$ $(\forall x\in \R_{>0})$,
\begin{align}
&\frac{b}{w(a_0,y,\eta_0)+1}W(\sqrt{2a_0(y+1)},y,\sqrt{\eta_0},w(a_0,y,\eta_0))\label{eq_blow_up_lower_bound}\\
&\ge 
\frac{b\cosh^{-1}(|y|^{-1})}{y+\cosh(\eta_0^{-\frac{1}{2}}\cosh^{-1}(|y|^{-1}))}\notag
\end{align}
for any $y\in (-1,y_0]$. Take any $U_0\in (0,2e_{min}/b)$. By using
 \eqref{eq_arc_cosh} one can check that the right-hand side of
 \eqref{eq_blow_up_lower_bound} diverges to $+\infty$ as $y\searrow
 -1$. Thus, there exists $y_1\in (-1,y_0]$ such that for any $y\in
 (-1,y_1]$
\begin{align}
\frac{b}{w(a_0,y,\eta_0)+1}W(\sqrt{2a_0(y+1)},y,\sqrt{\eta_0},w(a_0,y,\eta_0))>
\frac{2}{U_0}.\label{eq_large_function_lower_specific}
\end{align}
Note that for $(x,y,z)\in \R_{>0}\times (-1,0)\times \R_{>0}$ satisfying
 $x< \frac{1}{\sqrt{z(y+1)}}\cosh^{-1}(|y|^{-1})$, 
\begin{align}
&\frac{\partial W}{\partial x}(\sqrt{y+1}\cdot
 x,y,\sqrt{z},s)=\frac{1+y\cosh(\sqrt{z(y+1)}\cdot x)}{(y+\cosh(\sqrt{z(y+1)}\cdot x))^2}
\left(s-w\left(\frac{x^2}{2},y,z\right)\right).\label{eq_large_small_relation}
\end{align}
Let us fix $y\in (-1,y_1]$. Lemma \ref{lem_necessity_equality_lemma}
 ensures that there exist
\begin{align*}
&\hat{x}_1(y)\in \left(\frac{1}{\sqrt{y+1}}\cosh^{-1}(|y|^{-1}),\sqrt{2a_0}
\right),\\
&\hat{x}_2(y)\in \left(\sqrt{2a_0},
\frac{1}{\sqrt{\eta_0(y+1)}}\cosh^{-1}(|y|^{-1})\right)
\end{align*}
such that
\begin{align*}
&\frac{\partial W}{\partial x}(\sqrt{y+1}\cdot
 \hat{x}_j(y),y,\sqrt{\eta_0},w(a_0,y,\eta_0))\\
&=\frac{\partial W}{\partial
 x}(\sqrt{2a_0(y+1)},y,\sqrt{\eta_0},w(a_0,y,\eta_0))=0,\quad (j=1,2),\\
&\frac{\partial W}{\partial x}(\sqrt{y+1}\cdot
 x,y,\sqrt{\eta_0},w(a_0,y,\eta_0))<0,\quad (\forall x\in
 (\hat{x}_1(y),\sqrt{2a_0})),\\
&\frac{\partial W}{\partial x}(\sqrt{y+1}\cdot
 x,y,\sqrt{\eta_0},w(a_0,y,\eta_0))>0,\quad (\forall x\in
 (\sqrt{2a_0},\hat{x}_2(y))).
\end{align*}
These imply that 
\begin{align}
&\min_{j\in \{1,2\}} W(\sqrt{y+1}\cdot
 \hat{x}_j(y),y,\sqrt{\eta_0},w(a_0,y,\eta_0))\label{eq_large_local_minimum}\\
&> W(\sqrt{2a_0(y+1)},y,\sqrt{\eta_0},w(a_0,y,\eta_0)).\notag
\end{align}
By \eqref{eq_large_function_lower_specific},
 \eqref{eq_large_local_minimum} we can take small $\delta\in\R_{>0}$ so
 that 
\begin{align}
&\frac{1}{w(a_0,y,\eta_0)+1}-\delta>0,\notag\\
&\frac{2}{U_0}<b\left(\frac{1}{w(a_0,y,\eta_0)+1}-\delta\right)\label{eq_large_small_mixed}\\
&\qquad\quad\cdot W\left(
\sqrt{2a_0(y+1)},y,\sqrt{\eta_0},\frac{w(a_0,y,\eta_0)+\delta(w(a_0,y,\eta_0)+1)}{1-\delta
 (w(a_0,y,\eta_0)+1)}
\right)\notag\\
&\quad\ <\frac{b}{w(a_0,y,\eta_0)+1}\min_{j\in \{1,2\}} W(\sqrt{y+1}\cdot
 \hat{x}_j(y),y,\sqrt{\eta_0},w(a_0,y,\eta_0)).\notag
\end{align}
Here let us apply Lemma \ref{lem_special_function} proved in Appendix
 \ref{app_special_function} with $e_{min}=\sqrt{\eta_0}$, $e_{max}=1$, 
$s=\frac{1}{w(a_0,y,\eta_0)+1}-\delta$, $t=\frac{1}{w(a_0,y,\eta_0)+1}$.
By substituting the matrix-valued function $E$ into the function
 \eqref{eq_infinite_volume_trace} and recalling the monotone decreasing
 property of the function \eqref{eq_reference_function} we observe that
 for any $x\in \R_{>0}$ 
\begin{align}
&F_{\infty}(x,y)\ge
 bs\frac{\sinh(x)}{y+\cosh(x)}+b(t-s)\frac{\sinh(x)}{y+\cosh(x)}+b(1-t) 
\frac{\sinh(x\sqrt{\eta_0})}{(y+\cosh(x\sqrt{\eta_0}))\sqrt{\eta_0}}\label{eq_special_function_application}\\
&=\frac{b}{w(a_0,y,\eta_0)+1}W(x,y,\sqrt{\eta_0},w(a_0,y,\eta_0)),\notag\\
&F_{\infty}(x,y)\notag\\
&\le bs \frac{\sinh(x)}{y+\cosh(x)} +
 b(t-s)\frac{\sinh(x\sqrt{\eta_0})}{(y+\cosh(x\sqrt{\eta_0}))\sqrt{\eta_0}}+b(1-t)\frac{\sinh(x\sqrt{\eta_0})}{(y+\cosh(x\sqrt{\eta_0}))\sqrt{\eta_0}}\notag\\
&=b\left(\frac{1}{w(a_0,y,\eta_0)+1}-\delta\right)
W\left(x,y,\sqrt{\eta_0},\frac{w(a_0,y,\eta_0)+\delta(w(a_0,y,\eta_0)+1)}{1-\delta
 (w(a_0,y,\eta_0)+1)}\right).\notag
\end{align}
By combining these inequalities with \eqref{eq_large_small_mixed} we
 have that 
\begin{align*}
&F_{\infty}(\sqrt{2a_0(y+1)},y)<\min_{j\in
 \{1,2\}}F_{\infty}(\sqrt{y+1}\cdot \hat{x}_j(y),y),\\
&\frac{2}{U_0}<\min_{j\in \{1,2\}}F_{\infty}(\sqrt{y+1}\cdot
 \hat{x}_j(y),y).
\end{align*}
This implies that there exists $U\in [-U_0,0)$ such that
\begin{align*}
F_{\infty}(\sqrt{2a_0(y+1)},y)<\frac{2}{|U|}<\min_{j\in
 \{1,2\}}F_{\infty}(\sqrt{y+1}\cdot\hat{x}_j(y),y).
\end{align*}
Therefore, by taking into account the fact $F_{\infty}(0,y)=0$ we see
 that there exist
\begin{align*}
&\beta_1\in \big(0,\sqrt{y+1}\cdot\hat{x}_1(y)\big),\\ 
&\beta_2\in \big(\sqrt{y+1}\cdot\hat{x}_1(y), \sqrt{2a_0(y+1)}\big),\\
&\beta_3\in \big(\sqrt{2a_0(y+1)}, \sqrt{y+1}\cdot \hat{x}_2(y)\big)
\end{align*}
 such that 
$-2/|U|+F_{\infty}(\beta_j,y)=0$ for all $j\in \{1,2,3\}$. Moreover, it
 follows from Lemma \ref{lem_tau_implicit_uniqueness} that
 $0<\beta_1<\beta_2<\beta_3<\beta_c$, 
 $y=\cos(\tau(\beta_j)/2)$ for all $j\in \{1,2,3\}$, and thus
 $\tau(\beta_1)=\tau(\beta_2)=\tau(\beta_3)$. 

Finally let us prove that there exist $\hat{\beta}_1$, $\hat{\beta}_2\in
 (0,\beta_c)$ such that $\hat{\beta}_1<\hat{\beta}_2$ and these are
 local minimum points of $\tau(\cdot)$. If $\tau(\beta)=\tau(\beta_2)$
 $(\forall \beta\in (\beta_1,\beta_2))$ or $\tau(\beta)=\tau(\beta_2)$
 $(\forall \beta\in (\beta_2,\beta_3))$, such $\hat{\beta}_1$,
 $\hat{\beta}_2$ obviously exist. If there exists $\beta'\in
 (\beta_1,\beta_2)$ such that $\tau(\beta')>\tau(\beta_2)$, since
 $\lim_{\beta\searrow 0}\tau(\beta)=\lim_{\beta\nearrow
 \beta_c}\tau(\beta)=2\pi>\tau(\beta_1)=\tau(\beta_2)$, local minimum points  
$\hat{\beta}_1$, $\hat{\beta}_2$ exist in $(0,\beta')$,
 $(\beta',\beta_c)$ respectively. The same conclusion holds if there
 exists $\beta'\in (\beta_2,\beta_3)$ such that
 $\tau(\beta')>\tau(\beta_2)$. It remains to study the case that there
 are $\beta_1'\in (\beta_1,\beta_2)$, $\beta_2'\in (\beta_2,\beta_3)$
 such that $\tau(\beta_j')<\tau(\beta_2)$ for $j\in \{1,2\}$. In
 this case local minimum points $\hat{\beta}_1$, $\hat{\beta}_2$ exist
 in $(\beta_1,\beta_2)$, $(\beta_2,\beta_3)$ respectively. The
 proposition has been proved.
\end{proof}

A stronger conclusion than Proposition \ref{prop_necessity_equality}
holds when $e_{min}/e_{max}<\sqrt{17-12\sqrt{2}}$.

\begin{proposition}\label{prop_necessity_inequality}
Assume that $e_{min}/e_{max}<\sqrt{17-12\sqrt{2}}$. 
Then for any $d$, $b\in \N$, basis $(\hbv_j)_{j=1}^d$ of $\R^d$ 
there exist $E\in \cE(e_{min},e_{max})$ and $U_0\in (0,2e_{min}/b)$
such that for any $U\in [-U_0,0)$ $\tau(\cdot)$ has more than one local 
minimum points in $(0,\beta_c)$. 
\end{proposition}

\begin{remark}
The difference from the conclusion of Proposition
 \ref{prop_necessity_equality} is that here $E$ $(\in
 \cE(e_{min},e_{max}))$ is independent of the choice of small $U$. This
 conclusion implies the conclusion of Proposition
 \ref{prop_necessity_equality}. 
\end{remark}

Observe that for $\eta \in (0,17-12\sqrt{2}]$,
$(\frac{1+\eta}{6\eta})^2-\frac{1}{\eta}\ge 0$. 
This allows us to define the real numbers $a_{+}(\eta)$, $a_-(\eta)$,
$\hat{a}(\eta)$ by
\begin{align}
&a_+(\eta):=\frac{1+\eta}{6\eta}+\Big(\Big(\frac{1+\eta}{6\eta}\Big)^2-\frac{1}{\eta}
\Big)^{\frac{1}{2}},\label{eq_eta_dependent_variable}\\
&a_-(\eta):=\frac{1+\eta}{6\eta}-\Big(\Big(\frac{1+\eta}{6\eta}\Big)^2-\frac{1}{\eta}
\Big)^{\frac{1}{2}},\notag\\
&\hat{a}(\eta):=a_-(\eta)+\frac{a_+(\eta)-a_-(\eta)}{2}.\notag
\end{align}
Let us summarize basic properties concerning these numbers, which can be
deduced from \eqref{eq_extended_function_one_derivative}, \eqref{eq_a_eta_relation} and will be
used not only in the proof of Proposition
\ref{prop_necessity_inequality} but also in Sub-subsection
\ref{subsubsec_multi_orbital}. 

\begin{lemma}\label{lem_extended_function_properties}
If $\eta=17-12\sqrt{2}$ $(=\eta_0)$, 
\begin{align}
1<a_+(\eta)=a_-(\eta)=\hat{a}(\eta)=a_0=3+2\sqrt{2}<\eta^{-1}.\label{eq_eta_function_critical_order}
\end{align}
For any $\eta\in (0,17-12\sqrt{2})$ 
\begin{align}
&1<a_-(\eta)<\hat{a}(\eta)<a_+(\eta)<\eta^{-1},\label{eq_eta_function_first_order}\\
&\frac{\partial \tilde{w}}{\partial x}(x,-1,\eta)>0,\quad (\forall x\in
 (0,a_-(\eta))),\label{eq_extended_function_slope}\\
&\frac{\partial \tilde{w}}{\partial x}(a_-(\eta),-1,\eta)=0,\notag\\
&\frac{\partial \tilde{w}}{\partial x}(x,-1,\eta)<0,\quad (\forall x\in
 (a_-(\eta),a_+(\eta))),\notag\\
&\frac{\partial \tilde{w}}{\partial x}(a_+(\eta),-1,\eta)=0,\notag\\
&\frac{\partial \tilde{w}}{\partial x}(x,-1,\eta)>0,\quad (\forall x\in
 (a_+(\eta),\eta^{-1})),\notag\\
&0<\tilde{w}(a_+(\eta),-1,\eta)< \tilde{w}(\hat{a}(\eta),-1,\eta)
 <\tilde{w}(a_-(\eta),-1,\eta).\label{eq_extended_function_profile}
\end{align}
\end{lemma}

\begin{proof}[Proof of Proposition \ref{prop_necessity_inequality}]
Define the function $\widehat{W}:\R_{>0}\times \R_{>0}\times\R_{>0}\to
 \R$ by 
\begin{align*}
\widehat{W}(x,z,s):=\frac{x}{1+\frac{x^2}{2}}+s\frac{x}{1+z^2\frac{x^2}{2}}.
\end{align*}
Let us observe that for $(x,z)\in \R_{>0}\times\R_{>0}$ satisfying
 $x<\sqrt{2/z}$
\begin{align*}
&\frac{\partial \widehat{W}}{\partial
 x}(x,\sqrt{z},s)=\frac{1-z\frac{x^2}{2}}{(1+z\frac{x^2}{2})^2}\left(s-\tilde{w}\left(\frac{x^2}{2},
-1,z\right)\right).
\end{align*}
Fix $\eta\in (0,17-12\sqrt{2})$. On the basis of 
 \eqref{eq_extended_function_slope} and the facts
 $\tilde{w}(1,-1,\eta)=0$, $\lim_{x\nearrow
 \eta^{-1}}\tilde{w}(x,-1,\eta)=+\infty$,  we conclude that there exist
 $x_1\in (\sqrt{2},\sqrt{2\hat{a}(\eta)})$, $x_2\in
 (\sqrt{2\hat{a}(\eta)},\sqrt{2\eta^{-1}})$ such that 
\begin{align*}
&\frac{\partial \widehat{W}}{\partial
 x}(x_j,\sqrt{\eta},\tilde{w}(\hat{a}(\eta),-1,\eta))\\
&=
\frac{\partial \widehat{W}}{\partial
 x}(\sqrt{2\hat{a}(\eta)},\sqrt{\eta},\tilde{w}(\hat{a}(\eta),-1,\eta))=0,\quad(\forall
 j\in \{1,2\}),\\
&\frac{\partial \widehat{W}}{\partial
 x}(x,\sqrt{\eta},\tilde{w}(\hat{a}(\eta),-1,\eta))<0,\quad (\forall
 x\in (x_1,\sqrt{2\hat{a}(\eta)})),\\
&\frac{\partial \widehat{W}}{\partial
 x}(x,\sqrt{\eta},\tilde{w}(\hat{a}(\eta),-1,\eta))>0,\quad (\forall
 x\in (\sqrt{2\hat{a}(\eta)},x_2)).
\end{align*}
These imply that 
\begin{align*}
\min_{j\in
 \{1,2\}}\widehat{W}(x_j,\sqrt{\eta},\tilde{w}(\hat{a}(\eta),-1,\eta))
>\widehat{W}(\sqrt{2\hat{a}(\eta)},\sqrt{\eta},\tilde{w}(\hat{a}(\eta),-1,\eta)).
\end{align*}
We can choose small $\delta\in\R_{>0}$ so that 
\begin{align}
&\frac{1}{\tilde{w}(\hat{a}(\eta),-1,\eta)+1}-\delta>0,\notag\\
&\frac{b}{\tilde{w}(\hat{a}(\eta),-1,\eta)+1}\min_{j\in \{1,2\}}
 \widehat{W}(x_j,\sqrt{\eta},\tilde{w}(\hat{a}(\eta),-1,\eta))\label{eq_special_function_application_next}\\
&> b\left(\frac{1}{\tilde{w}(\hat{a}(\eta),-1,\eta)+1}-\delta\right)\notag\\
&\quad\cdot \widehat{W}\left(
\sqrt{2\hat{a}(\eta)},\sqrt{\eta},\frac{\tilde{w}(\hat{a}(\eta),-1,\eta)+\delta(\tilde{w}(\hat{a}(\eta),-1,\eta)+1)}{1-\delta
 (\tilde{w}(\hat{a}(\eta),-1,\eta)+1)}
\right).\notag
\end{align}
Here we apply Lemma \ref{lem_special_function} with 
$e_{min}=\sqrt{\eta}$, $e_{max}=1$,
 $s=\frac{1}{\tilde{w}(\hat{a}(\eta),-1,\eta)+1}-\delta$,
 $t=\frac{1}{\tilde{w}(\hat{a}(\eta),-1,\eta)+1}$. With the
 matrix-valued function $E$ we define the function
 $\widehat{F}_{\infty}:\R\to \R$ by 
\begin{align*}
\widehat{F}_{\infty}(x):=D_d\int_{\G_{\infty}^*}d\bk\Tr\Bigg(\frac{x}{1+\frac{x^2}{2}E(\bk)^2}\Bigg).
\end{align*}
By arguing in the same way as in \eqref{eq_special_function_application}
 we can derive from \eqref{eq_special_function_application_next} that
$\widehat{F}_{\infty}(\sqrt{2\hat{a}(\eta)})$ $<\min_{j\in
 \{1,2\}}\widehat{F}_{\infty}(x_j)$. Check that $\lim_{y\searrow
 -1}\sqrt{y+1}F_{\infty}(\sqrt{y+1}\cdot x,y)=\widehat{F}_{\infty}(x)$
 for all $x\in \R$, where $F_{\infty}(\cdot)$ is the function defined in
 \eqref{eq_infinite_volume_trace}. Thus, there exists $y_1(\eta)\in
 (-1,0)$ such that for any $y\in (-1,y_1(\eta)]$ 
\begin{align}
F_{\infty}(\sqrt{2\hat{a}(\eta)(y+1)},y)<\min_{j\in
 \{1,2\}}F_{\infty}(\sqrt{y+1}\cdot
 x_j,y).\label{eq_infinity_function_relation}
\end{align}
By recalling the monotone decreasing property of the function
 \eqref{eq_reference_function} we have that for any $y\in
 (-1,y_1(\eta)]$ 
\begin{align}
\frac{b\sinh(\sqrt{2\hat{a}(\eta)(y+1)})}{y+\cosh(\sqrt{2\hat{a}(\eta)(y+1)})}&\le 
F_{\infty}(\sqrt{2\hat{a}(\eta)(y+1)},y)\label{eq_infinity_function_above_below}\\
&\le
 \frac{b\sinh(\sqrt{2\hat{a}(\eta)(y+1)\eta})}{(y+\cosh(\sqrt{2\hat{a}(\eta)(y+1)\eta}))\sqrt{\eta}}.\notag
\end{align}
Set 
\begin{align*}
U_0:=\min\left\{
\frac{e_{min}}{b},\frac{(y_1(\eta)+\cosh(\sqrt{2\hat{a}(\eta)(y_1(\eta)+1)\eta}))\sqrt{\eta}}{b\sinh(\sqrt{2\hat{a}(\eta)(y_1(\eta)+1)\eta})}\right\}.
\end{align*}
It follows that $U_0\in (0,2e_{min}/b)$. Take any $U\in
 [-U_0,0)$. By \eqref{eq_infinity_function_above_below} 
\begin{align*}
F_{\infty}(\sqrt{2\hat{a}(\eta)(y_1(\eta)+1)},y_1(\eta))<\frac{2}{U_0}\le
 \frac{2}{|U|}.
\end{align*}
Set
\begin{align*}
S:=\left\{
y\in (-1,y_1(\eta)]\ \Big|\ F_{\infty}(\sqrt{2\hat{a}(\eta)(y+1)},y)=\frac{2}{|U|}
\right\}.
\end{align*}
By considering the fact that the left-hand side of
 \eqref{eq_infinity_function_above_below} diverges to $+\infty$ as
 $y\searrow -1$ we see that $S\neq \emptyset$. Set $y_2(\eta,U):=\sup
 S$. Then, $-1<y_2(\eta,U)<y_1(\eta)$ and by
 \eqref{eq_infinity_function_relation}
\begin{align*}
&F_{\infty}(\sqrt{2\hat{a}(\eta)(y_2(\eta,U)+1)},y_2(\eta,U))=\frac{2}{|U|}\\
&<\min_{j\in
 \{1,2\}}F_{\infty}(\sqrt{y_2(\eta,U)+1}\cdot x_j,y_2(\eta,U)),\\
&F_{\infty}(\sqrt{2\hat{a}(\eta)(y+1)},y)<\frac{2}{|U|},\quad (\forall
 y\in (y_2(\eta,U),y_1(\eta)]).
\end{align*}
This implies that if we take $y_3(\eta,U)\in (y_2(\eta,U),y_1(\eta)]$
 sufficiently close to $y_2(\eta,U)$,
\begin{align*}
&F_{\infty}(\sqrt{2\hat{a}(\eta)(y_3(\eta,U)+1)},y_3(\eta,U))<\frac{2}{|U|}\\
&<\min_{j\in
 \{1,2\}}F_{\infty}(\sqrt{y_3(\eta,U)+1}\cdot x_j,y_3(\eta,U)).
\end{align*}
Then we only need to repeat the same argument as in the last part of the
 proof of Proposition \ref{prop_necessity_equality} to conclude that
 $\tau(\cdot)$ has at least two local minimum points.
\end{proof}

By combining Proposition \ref{prop_sufficiency}, Proposition
\ref{prop_necessity_equality}, Proposition
\ref{prop_necessity_inequality} we reach the following theorem.

\begin{theorem}\label{thm_boundary_characterization}
For any $d,b\in \N$, basis $(\hbv_j)_{j=1}^d$ of $\R^d$ and $e_{min}$,
 $e_{max}\in \R_{>0}$ satisfying $e_{min}\le e_{max}$ the following
 statements are equivalent to each other.
\begin{enumerate}[(i)]
\item There exists $U_0\in (0,2e_{min}/b)$ such that
for any $U\in [-U_0,0)$ and $E\in \cE(e_{min},$ $e_{max})$ $\tau(\cdot)$
      has one and only one local minimum point in $(0,\beta_c)$. 
\item 
$$
\frac{e_{min}}{e_{max}}>\sqrt{17-12\sqrt{2}}.
$$
\end{enumerate}
\end{theorem}

\begin{remark}
According to Theorem \ref{thm_infinite_volume_limit}, we have to take
 small $U$ depending on $E\in \cE(e_{min},e_{max})$ in order to justify
 the derivation of the infinite-volume limit of the
 free energy density and the thermal expectations from the finite-volume
 lattice Fermion system. The graph $\{(\beta,\tau(\beta))\ |\ \beta\in
 (0,\beta_c)\}$ can be rigorously considered as the representative
 curve of the phase boundaries of the phase transition happening in
 our system only if the derivation of the infinite-volume limit is
 justified. Here let us summarize what we can conclude by
 combining the results obtained in this section with the sufficient condition for
 justifying the derivation.

By Proposition \ref{prop_convexity} for any $E\in \cE(e_{min},e_{max})$
 with $e_{min}/e_{max}\ge e_0$ there exists $U_0\in (0,2e_{min}/b)$ such
 that for any $U\in [-U_0,0)$ the derivation is justified and $\frac{d^2
 \tau}{d\beta^2}(\beta)>0$ for any $\beta\in (0,\beta_c)$.

By Proposition \ref{prop_sufficiency} for any $E\in
 \cE(e_{min},e_{max})$ with $e_{min}/e_{max}>\sqrt{17-12\sqrt{2}}$
there exists $U_0\in (0,2e_{min}/b)$ such that for any $U\in [-U_0,0)$
 the derivation is justified and $\tau(\cdot)$ has only one local
 minimum point in $(0,\beta_c)$.

By Proposition \ref{prop_necessity_inequality} for any $e_{min}$,
 $e_{max}\in\R_{>0}$ with $e_{min}/e_{max}<\sqrt{17-12\sqrt{2}}$ there
 exist $E\in\cE(e_{min},e_{max})$ and $U_0\in (0,2e_{min}/b)$ such that 
for any $U\in [-U_0,0)$ the derivation is justified and $\tau(\cdot)$
 has more than one local minimum points in $(0,\beta_c)$. 

However, in Proposition \ref{prop_necessity_equality} we do not have
 freedom to choose small $U$. The coupling constant $U$ was chosen
 depending on $E$ in the proof and it is not clear whether for such $U$
 the derivation is justified by Theorem
 \ref{thm_infinite_volume_limit}. Thus, strictly speaking, in the case
 $e_{min}/e_{max}=\sqrt{17-12\sqrt{2}}$ we cannot claim that
 $\tau(\cdot)$ has more than one local minimum points while justifying
 the derivation. 
\end{remark}

\begin{remark}
In view of Proposition \ref{prop_convexity}, we can propose a problem to
 find a necessary and sufficient condition in terms of $e_{min}/e_{max}$
 for that $\tau(\cdot)$ is downward convex in $(0,\beta_c)$ for any
 $E\in \cE(e_{min},e_{max})$. However, we are unable to solve the problem at
 present.
\end{remark}

\subsection{Study of specific models}\label{subsec_specific_models}

In the proofs of Proposition \ref{prop_necessity_equality} and Proposition
\ref{prop_necessity_inequality} we constructed particular examples of
$E$ $(\in \cE(e_{min},e_{max}))$ for which $\tau(\cdot)$ has more than
one local minimum points. However, these results do not tell us whether
$\tau(\cdot)$ can have more than one local minimum points when we
change the value $e_{min}/e_{max}$ within a one-particle
Hamiltonian explicitly parameterized by $e_{min}$, $e_{max}$, though we
know that $\tau(\cdot)$ must have only one local minimum point for small
$U$ when
$e_{min}/e_{max}>\sqrt{17-12\sqrt{2}}$ by Proposition
\ref{prop_sufficiency}. In this subsection we study this question for
the following two models. Let $I_n$ denote the $n\times n$ unit matrix
for $n\in\N$.
\begin{enumerate}[(1)]
\item\label{item_model_multi_orbital} 
Let $d\in \N$, $b\in \N_{\ge 2}$, $b'\in \{1,2,\cdots,b-1\}$ and
     $(\hbv_j)_{j=1}^d$ be any basis of $\R^d$. Let us define $E_b\in
     \cE(e_{min},e_{max})$ by 
\begin{align*}
&E_b(\bk):=((e_{max}1_{i\le b'}+ e_{min}1_{i> b'})1_{i=j})_{1\le i,j\le
 b}=\left(\begin{array}{cc} e_{max}I_{b'} & 0 \\
                           0 & e_{min}I_{b-b'}\end{array}\right),\\
&(\bk\in \R^d),
\end{align*}
which is a $b$-orbital model without hopping.
\item\label{item_model_one_band}
Let $d=b=1$ and $\hbv_1=1$. For $t\in \R_{\ge 0}$, $e_{min}\in \R_{>0}$
     let us define $E_1\in\cE(e_{min},2t+e_{min})$ by $E_1(k):=t(\cos
     k +1)+e_{min}$, $(k\in \R)$. The function $E_1(\cdot)$ is the
     dispersion relation of nearest-neighbor hopping free electron on
     the 1-dimensional lattice $\Z$.
\end{enumerate}

It will turn out that the uniqueness of local minimum points is
sensitive to the ratios $e_{min}/e_{max}$, $(b-b')/b'$ in the model
\eqref{item_model_multi_orbital}, while the uniqueness holds for any
$t\in \R_{\ge 0}$, $e_{min}\in \R_{>0}$ in the model
\eqref{item_model_one_band}. 

\begin{remark}
For $t$, $\mu\in \R$ let us define the function $e_1:\R\to \R$ by
$e_1(k):=t\cos k+\mu$. The function $e_1(\cdot)$ satisfying the condition
 $\inf_{k\in \R}|e_1(k)|>0$ is the most general form of a
 non-vanishing dispersion relation of nearest-neighbor hopping free
 electron on $\Z$. We can check that
\begin{align*}
&\int_0^{2\pi}dk \frac{\sinh(\beta e_1(k))}{(y+\cosh(\beta
 e_1(k)))e_1(k)}\\
&=\int_0^{2\pi}dk \frac{\sinh(\beta (|t|\cos k
 +|\mu|))}{(y+\cosh(\beta(|t|\cos k + |\mu|)))(|t|\cos k + |\mu|)},\\
&(\forall \beta \in\R_{>0},\ y\in (-1,0)).
\end{align*}
By using the above equality and the fact that
 $\inf_{k\in \R}|e_1(k)|>0$ is equivalent to $|\mu|>|t|$
we can reduce the
 problem with $e_1(\cdot)$ to that with $E_1(\cdot)$ defined in
 \eqref{item_model_one_band}. This means that the results we will obtain
 in Sub-subsection \ref{subsubsec_one_band} for $E_1(\cdot)$ also
 hold for $e_1(\cdot)$ satisfying 
 $\inf_{k\in\R}|e_1(k)|>0$.
\end{remark}

\subsubsection{The multi-orbital model without hopping}\label{subsubsec_multi_orbital}

Here let us study the profile of $\tau(\cdot)$ in the model defined in
\eqref{item_model_multi_orbital}. Our central question is when
$\tau(\cdot)$ has only one local minimum point. The answer is given in
the next proposition. 

\begin{proposition}\label{prop_multi_orbital}
Set the condition ($\star$) as follows. 
\begin{itemize}
\item[($\star$)] There exists $U_0\in (0,2 e_{min}/b)$ such that for any
		 $U\in [-U_0,0)$ $\tau(\cdot)$ has one and only one
		 local minimum point in $(0,\beta_c)$. 
\end{itemize}
Then the following statements hold. 
\begin{enumerate}[(i)]
\item\label{item_band_first}
Assume that $\frac{b-b'}{b'}\in [3-2\sqrt{2},\infty)$. Then for any
     $e_{min}$, $e_{max}\in \R_{>0}$ with $e_{min}\le e_{max}$ ($\star$)
     holds. 
\item\label{item_band_middle}
Assume that $\frac{b-b'}{b'}\in (1/8,3-2\sqrt{2})$. Then there
     exist $e_1$, $e_2\in (0,\sqrt{17-12\sqrt{2}})$ such that $e_1<e_2$
     and ($\star$) holds if $e_{min}/e_{max}\in (e_2,1]$,
 ($\star$) does not hold if $e_{min}/e_{max}\in (e_1,e_2]$,
 ($\star$) holds if $e_{min}/e_{max}\in (0,e_1]$.
\item\label{item_band_end}
Assume that $\frac{b-b'}{b'}\in (0,1/8]$. Then there
     exists $e_1\in (0,\sqrt{17-12\sqrt{2}})$ such that ($\star$) holds if $e_{min}/e_{max}\in (e_1,1]$,
 ($\star$) does not hold if $e_{min}/e_{max}\in (0,e_1]$.
\end{enumerate}
\end{proposition}

Again the proof of this proposition is based on some properties of the
function $w$ defined in \eqref{eq_internal_function}. Let us set two
conditions concerning the function $w$. Let $\eta\in (0,1)$, $s\in \R_{>0}$.
\begin{itemize}
\item[$(\text{i})_{\eta,s}$] There exists $y_0\in (-1,0)$ such that for
			     any $y\in (-1,y_0]$ there exists 
\begin{align*}
x_0(y)\in \left(\frac{1}{2(y+1)}(\cosh^{-1}(|y|^{-1}))^2, 
                \frac{1}{2\eta (y+1)}(\cosh^{-1}(|y|^{-1}))^2\right)
\end{align*}
such that 
\begin{align*}
&w(x,y,\eta)<s,\quad \left(\forall x\in \left(\frac{1}{2(y+1)}(\cosh^{-1}(|y|^{-1}))^2, 
                x_0(y)\right)\right),\\
&w(x_0(y),y,\eta)=s,\\
&w(x,y,\eta)>s,\quad \left(\forall x\in \left(x_0(y),\frac{1}{2\eta
 (y+1)}(\cosh^{-1}(|y|^{-1}))^2\right)\right).
\end{align*}
\item[$(\text{ii})_{\eta,s}$] 
There exists $y_0\in (-1,0)$ such that for
			     any $y\in (-1,y_0]$ there exist 
\begin{align*}
x_j(y)\in \left(\frac{1}{2(y+1)}(\cosh^{-1}(|y|^{-1}))^2, 
                \frac{1}{2\eta (y+1)}(\cosh^{-1}(|y|^{-1}))^2\right),\ (j=1,2,3)
\end{align*}
such that $x_1(y)<x_2(y)<x_3(y)$,
\begin{align*}
&w(x_j(y),y,\eta)=s,\quad (\forall j\in \{1,2,3\}),\\
&w(x,y,\eta)>s,\quad (\forall x\in (x_1(y),x_2(y)),\\
&w(x,y,\eta)<s,\quad (\forall x\in (x_2(y),x_3(y))).
\end{align*}
\end{itemize}

We summarize sufficient conditions for $(\text{i})_{\eta,s}$ (or
$(\text{ii})_{\eta,s}$) to hold in the next lemma. To understand the
statements, we should recall the inequalities
\eqref{eq_extended_function_profile}. 

\begin{lemma}\label{lem_internal_function_classification}
\begin{enumerate}[(i)]
\item\label{item_up_eta}
Assume that $\eta=17-12\sqrt{2}$. Then for any $s\in (0,\infty)$
     $(\text{i})_{\eta,s}$ holds. 
\item\label{item_down_eta}
Assume that $\eta\in (0,17-12\sqrt{2})$. Then for any $s\in
     [\tilde{w}(a_-(\eta),-1,\eta),\infty)$ $(\text{i})_{\eta,s}$
     holds. For any $s\in
     [\tilde{w}(a_+(\eta),-1,\eta),\tilde{w}(a_-(\eta),-1,\eta))$
     $(\text{ii})_{\eta,s}$ holds. For any $s\in
     (0,\tilde{w}(a_+(\eta),-1,\eta))$ $(\text{i})_{\eta,s}$ holds.    
\end{enumerate}
\end{lemma}

\begin{proof}
Assume that $\eta\in (0,17-12\sqrt{2}]$. Let us prepare necessary basic
 properties related to the function $w$. 
The preparation continues until we prove the claim \eqref{eq_internal_inside_edge}.
Observe that there exists
 $y_0\in (-1,0)$ such that
 $$
\frac{1}{\sqrt{2(y+1)}}\cosh^{-1}(|y|^{-1})>1,\quad (\forall y\in
 (-1,y_0]).$$
 This claim can be proved efficiently by proving the
 equivalent statement that there exists $y_0\in (-1,0)$ such that
$|y|^{-1}>\cosh(\sqrt{2(y+1)})$ for any $y\in
 (-1,y_0]$. Moreover, by \eqref{eq_eta_function_critical_order}, 
\eqref{eq_eta_function_first_order}
there exists $y_0(\eta)\in (-1,y_0]$ such that 
$|y|a_+(\eta)>1$ for any $y\in (-1,y_0(\eta)]$. 
We can see from \eqref{eq_internal_function_continuation}
and the limit in the left-hand side of
 \eqref{eq_arc_cosh_limit_relation} 
that for any $y\in (-1,y_0(\eta)]$, $\eps
 \in (0,\eta^{-1}-a_+(\eta))$ and 
$$
x\in \left[\eta^{-1}-\eps,\frac{1}{2\eta
 (y+1)}(\cosh^{-1}(|y|^{-1}))^2\right),
$$ 
\begin{align}
&w(x,y,\eta)\ge
 \frac{|y|(\eta^{-1}-\eps)-1}{(1+y\eta(\eta^{-1}-\eps))
\left(1+\sum_{n=1}^{\infty}\frac{(y+1)^{n-1}}{(2n)!}2^n\big(\frac{1}{2\eta(y+1)}(\cosh^{-1}(|y|^{-1}))^2\big)^n\right)^2},\label{eq_internal_temporal_lower}\\
&\lim_{y\searrow -1}(\text{R.H.S of
 }\eqref{eq_internal_temporal_lower})=\frac{\eta^{-1}-\eps-1}{\eta
 \eps(1+\eta^{-1})^2}.\notag
\end{align}
Take any $s\in \R_{>0}$. Note that there exists
 $\eps(s,\eta)\in (0,\eta^{-1}-a_+(\eta))$ such that 
$$
\frac{\eta^{-1}-\eps(s,\eta)-1}{\eta \eps(s,\eta)(1+\eta^{-1})^2}\ge 2s.
$$
Then it follows from the above claims that there exists $y_1(s,\eta)\in
 (-1,y_0(\eta)]$ such that for any $y\in (-1,y_1(s,\eta)]$ 
\begin{align}
1&<\frac{1}{2
 (y+1)}(\cosh^{-1}(|y|^{-1}))^2<a_-(\eta)\le
 a_+(\eta)<\eta^{-1}-\eps(s,\eta)\label{eq_eta_dependent_y_order}\\
&<\frac{1}{2\eta
 (y+1)}(\cosh^{-1}(|y|^{-1}))^2\notag
\end{align}
and for any 
$$
x\in \left[\eta^{-1}-\eps(s,\eta),\frac{1}{2\eta
 (y+1)}(\cosh^{-1}(|y|^{-1}))^2\right),
$$ 
\begin{align}
w(x,y,\eta)>s.\label{eq_internal_permanent_lower}
\end{align}
Recall \eqref{eq_continued_domain_D}. To justify the subsequent argument,
 let us check that there exists $\delta(s,\eta)\in (0,\eps(s,\eta))$
 such that
 \begin{align*}
(1-\delta(s,\eta),\eta^{-1}-\eps(s,\eta)+\delta(s,\eta))\times
 (-1-\delta(s,\eta),-1+\delta(s,\eta))\times \{\eta\}\subset
 \tilde{D}.
\end{align*}
 We can deduce from
 \eqref{eq_extended_function_two_derivative} that $\sup_{x\in
 [1,\eta^{-1}-\eps(s,\eta)]}\frac{\partial \tilde{w}}{\partial
 y}(x,-1,\eta)<0$. For $(x,y)\in [1,\eta^{-1}-\eps(s,\eta)]\times
 (-1,-1+\delta(s,\eta)/2]$
\begin{align*}
\tilde{w}(x,y,\eta)&=\tilde{w}(x,-1,\eta)+\frac{\partial \tilde{w}}{\partial
 y}(x,-1,\eta)(y+1)+\int_{-1}^yd\xi (y-\xi)\frac{\partial^2
 \tilde{w}}{\partial y^2}(x,\xi,\eta)\\
&\le \tilde{w}(x,-1,\eta)+\sup_{\zeta\in
 [1,\eta^{-1}-\eps(s,\eta)]}\frac{\partial \tilde{w}}{\partial
 y}(\zeta,-1,\eta)(y+1)\\
&\quad +\sup_{\zeta\in [1,\eta^{-1}-\eps(s,\eta)]\atop
 \xi\in [-1,-1+\delta(s,\eta)/2]}\left|\frac{\partial^2
 \tilde{w}}{\partial y^2}(\zeta,\xi,\eta)\right|(y+1)^2.
\end{align*}
Recall \eqref{eq_domain_D}. These imply that there exists
 $y_2(s,\eta)\in (-1,y_1(s,\eta)]$ such that
 $[1,\eta^{-1}-\eps(s,\eta)]\times (-1,y_2(s,\eta)]\times
 \{\eta\}\subset D$ and 
\begin{align}
w(x,y,\eta)<\tilde{w}(x,-1,\eta),\quad (\forall (x,y)\in
 [1,\eta^{-1}-\eps(s,\eta)]\times
 (-1,y_2(s,\eta)]).\label{eq_uniform_lower_order}
\end{align}
We will also refer to the basic fact that
\begin{align}
w\left(\frac{1}{2(y+1)}(\cosh^{-1}(|y|^{-1}))^2,y,\eta
\right)=0,\quad (\forall y\in
 (-1,y_2(s,\eta)]).\label{eq_basic_vanishing}
\end{align}
Let us define the function $\underline{w}:\R^3\to \R$ by
\begin{align*}
\underline{w}(x,y,z):=&\left(
1+y\sum_{m=1}^{\infty}\frac{(y+1)^{m-1}}{(2m)!}2^mx^m
\right)
\left(
1+\sum_{n=1}^{\infty}\frac{(y+1)^{n-1}}{(2n)!}2^nz^nx^n
\right)\\
&\cdot\Bigg(
y\sum_{m=1}^{\infty}\frac{m(y+1)^{m-1}}{(2m)!}2^mz^mx^{m-1}
\left(1+\sum_{n=1}^{\infty}\frac{(y+1)^{n-1}}{(2n)!}2^nx^n
\right)\\
&\qquad+2\left(
1+y\sum_{m=1}^{\infty}\frac{(y+1)^{m-1}}{(2m)!}2^mz^mx^m
\right)\sum_{n=1}^{\infty}\frac{n(y+1)^{n-1}}{(2n)!}2^nx^{n-1}
\Bigg)\\
&-\Bigg(
y\sum_{m=1}^{\infty}\frac{m(y+1)^{m-1}}{(2m)!}2^mx^{m-1}
\left(1+\sum_{n=1}^{\infty}\frac{(y+1)^{n-1}}{(2n)!}2^nz^nx^n
\right)\\
&\qquad+2\left(
1+y\sum_{m=1}^{\infty}\frac{(y+1)^{m-1}}{(2m)!}2^mx^m
\right)\sum_{n=1}^{\infty}\frac{n(y+1)^{n-1}}{(2n)!}2^nz^nx^{n-1}
\Bigg)\\
&\quad\cdot\left(
1+y\sum_{m=1}^{\infty}\frac{(y+1)^{m-1}}{(2m)!}2^mz^mx^m
\right)
\left(
1+\sum_{n=1}^{\infty}\frac{(y+1)^{n-1}}{(2n)!}2^nx^n
\right).
\end{align*}
By differentiating \eqref{eq_internal_function_continuation} we can
 derive that
\begin{align}
&\frac{\partial w}{\partial
 x}(x,y,\eta)\label{eq_internal_inside_relation}\\
&=\frac{1+\sum_{n=1}^{\infty}\frac{(y+1)^{n-1}}{(2n)!}2^n\eta^nx^n}
{\left(1+y\sum_{m=1}^{\infty}\frac{(y+1)^{m-1}}{(2m)!}2^m\eta^mx^m\right)^2
 \left(1+\sum_{n=1}^{\infty}\frac{(y+1)^{n-1}}{(2n)!}2^nx^n\right)^3}\underline{w}(x,y,\eta),\notag\\
&(\forall (x,y)\in [1,\eta^{-1}-\eps(s,\eta)]\times
 (-1,y_2(s,\eta)]).\notag
\end{align}
Let us observe that 
\begin{align}
&\underline{w}(x,-1,\eta)
=3\eta(1-\eta)(x-a_+(\eta))(x-a_-(\eta)),\label{eq_internal_inside_minus_one}\\
&\frac{\partial^2 \underline{w}}{\partial x^2}(x,-1,\eta)=6\eta
 (1-\eta)>0.\notag
\end{align}
The above inequality implies that there exists $y_3(s,\eta)\in
 (-1,y_2(s,\eta)]$ such that 
\begin{align}
&\frac{\partial^2\underline{w}}{\partial x^2}(x,y,\eta)>0,
\quad (\forall (x,y)\in [1,\eta^{-1}-\eps(s,\eta)]\times
 (-1,y_3(s,\eta)]).\label{eq_internal_inside_convex}
\end{align}
Also, by \eqref{eq_eta_dependent_y_order} and
 \eqref{eq_internal_inside_minus_one}
$\underline{w}(\eta^{-1}-\eps(s,\eta),-1,\eta)>0$.
Thus, there exists $y_4(s,\eta)\in (-1,y_3(s,\eta)]$ such that 
\begin{align}
\underline{w}(\eta^{-1}-\eps(s,\eta),y,\eta)>0,\quad (\forall y\in
 (-1,y_4(s,\eta)]).\label{eq_internal_inside_edge}
\end{align}

As we have prepared necessary tools, let us start proving the claims of
 the lemma case by case.

\eqref{item_up_eta}: Assume that $\eta=17-12\sqrt{2}$. Recall the
 relation \eqref{eq_eta_function_critical_order}. Assume that $s\in
 (\tilde{w}(a_-(\eta),-1,\eta),\infty)$. We can deduce from
 \eqref{eq_extended_function},
 \eqref{eq_extended_function_one_derivative} that 
\begin{align}
&x\mapsto \tilde{w}(x,-1,\eta):[1,\eta^{-1})\to\R\text{ is strictly monotone
 increasing, }\label{eq_extended_function_monotonicity}\\
&\tilde{w}(1,-1,\eta)=0\text{ and }\lim_{x\nearrow
 \eta^{-1}}\tilde{w}(x,-1,\eta)=\infty.\notag
\end{align}
Thus there uniquely exists $a_1\in (a_-(\eta),\eta^{-1})$ such that 
$s=\tilde{w}(a_1,-1,\eta)$. If $a_1\in
 (\eta^{-1}-\eps(s,\eta),\eta^{-1})$, by \eqref{eq_uniform_lower_order}
 and \eqref{eq_extended_function_monotonicity} $w(x,y,\eta)<s$,
 $(\forall (x,y)\in [1,\eta^{-1}-\eps(s,\eta)]\times
 (-1,y_2(s,\eta)])$. This contradicts
 \eqref{eq_internal_permanent_lower}. Thus, $a_1\in
 (a_-(\eta),\eta^{-1}-\eps(s,\eta)]$. By \eqref{eq_uniform_lower_order}
 and \eqref{eq_extended_function_monotonicity} again 
$w(x,y,\eta)<s$ for all $(x,y)\in [1,a_1]\times (-1,y_2(s,\eta)]$. This
 property coupled with \eqref{eq_internal_permanent_lower} ensures that
 for any $y\in (-1,y_4(s,\eta)]$ 
\begin{align}
\emptyset \neq &\left\{x\in
 \left(\frac{1}{2(y+1)}(\cosh^{-1}(|y|^{-1}))^2,\frac{1}{2\eta(y+1)}(\cosh^{-1}(|y|^{-1}))^2\right)\ \Big|\
 w(x,y,\eta)=s\right\}\label{eq_level_set_inclusion}\\
&\subset
 [a_1,\eta^{-1}-\eps(s,\eta)].\notag
\end{align}
By \eqref{eq_internal_inside_minus_one} $\underline{w}(x,-1,\eta)\ge
 3\eta(1-\eta)(a_1-a_-(\eta))^2>0$ for any $x\in
 [a_1,\eta^{-1}-\eps(s,\eta)]$. Thus there exists $y_5(s,\eta)\in
 (-1,y_4(s,\eta)]$ such that for any $(x,y)\in
 [a_1,\eta^{-1}-\eps(s,\eta)]\times (-1,y_5(s,\eta)]$
 $\underline{w}(x,y,\eta)>0$ and by \eqref{eq_internal_inside_relation}
 $\frac{\partial w}{\partial x}(x,y,\eta)>0$. This property combined with
 \eqref{eq_level_set_inclusion} implies that $(\text{i})_{\eta,s}$
 holds. 

Assume that $s\in (0,\tilde{w}(a_-(\eta),-1,\eta))$. By
 \eqref{eq_extended_function_monotonicity} there uniquely exists $a_1\in
 (1,a_-(\eta))$ such that $\tilde{w}(a_1,-1,\eta)=s$.   
Since $\tilde{w}(x,-1,\eta)>s$ for any $x\in
 [a_1+\frac{1}{2}(a_-(\eta)-a_1), \eta^{-1}-\eps(s,\eta)]$, there exists
 $y_6(s,\eta)\in (-1,y_4(s,\eta)]$ such that 
\begin{align}
w(x,y,\eta)>s,\quad \left(\forall (x,y)\in
 \left[a_1+\frac{1}{2}(a_-(\eta)-a_1),\eta^{-1}-\eps(s,\eta)\right]\times
 (-1,y_6(s,\eta)]\right).\label{eq_internal_function_monotone_up}
\end{align}
By \eqref{eq_internal_inside_minus_one} $\underline{w}(x,-1,\eta)>0$ for
 any $x\in [1,a_1+\frac{1}{2}(a_-(\eta)-a_1)]$. Thus there exists
 $y_7(s,\eta)\in (-1,y_6(s,\eta)]$ such that for any $(x,y)\in
 [1,a_1+\frac{1}{2}(a_-(\eta)-a_1)]\times (-1,y_7(s,\eta)]$
 $\underline{w}(x,y,\eta)>0$ and thus $\frac{\partial w}{\partial
 x}(x,y,\eta)>0$. This property together with
 \eqref{eq_eta_dependent_y_order}, \eqref{eq_internal_permanent_lower},
 \eqref{eq_basic_vanishing}, \eqref{eq_internal_function_monotone_up}
 implies that $(\text{i})_{\eta,s}$ holds.

Assume that $s=\tilde{w}(a_-(\eta),-1,\eta)$. Since $\eta=\eta_0$,
 $a_-(\eta)=\frac{\eta_0+1}{6\eta_0}=a_0$ by \eqref{eq_a_eta_relation},
 we can apply \eqref{eq_internal_field_dependent} to ensure that there
 exists $y_8(s,\eta)\in (-1,y_4(s,\eta)]$ such that for any $y\in
 (-1,y_8(s,\eta)]$ $\frac{\partial w}{\partial
 x}(a_-(\eta),y,\eta)<0$. This combined with
 \eqref{eq_internal_inside_relation} implies that
 $\underline{w}(a_-(\eta),y,\eta)<0$ for any $y\in
 (-1,y_8(s,\eta)]$. Therefore, by \eqref{eq_internal_inside_convex},
 \eqref{eq_internal_inside_edge} for any $y\in (-1,y_8(s,\eta)]$ 
there exists $x_1(y)\in (a_-(y),\eta^{-1}-\eps(s,\eta))$ such that 
\begin{align*}
&\underline{w}(x,y,\eta)<0,\quad (\forall x\in [a_-(\eta),x_1(y))),\\
&\underline{w}(x_1(y),y,\eta)=0,\\
&\underline{w}(x,y,\eta)>0,\quad (\forall x\in
 (x_1(y),\eta^{-1}-\eps(s,\eta)]),
\end{align*}
or by \eqref{eq_internal_inside_relation}
\begin{align}
&\frac{\partial w}{\partial x}(x,y,\eta)<0,\quad (\forall x\in [a_-(\eta),x_1(y))),\label{eq_local_internal_convexity}\\
&\frac{\partial w}{\partial x}(x_1(y),y,\eta)=0,\notag\\
&\frac{\partial w}{\partial x}(x,y,\eta)>0,\quad (\forall x\in
 (x_1(y),\eta^{-1}-\eps(s,\eta)]).\notag
\end{align}
We can see from \eqref{eq_uniform_lower_order} and
 \eqref{eq_extended_function_monotonicity} that 
\begin{align}
w(x,y,\eta)<s,\quad (\forall (x,y)\in [1,a_-(\eta)]\times
 (-1,y_8(s,\eta)]).\label{eq_uniform_lower_order_specific}
\end{align}
Considering \eqref{eq_eta_dependent_y_order},
 \eqref{eq_internal_permanent_lower},
 \eqref{eq_local_internal_convexity} and
 \eqref{eq_uniform_lower_order_specific}, we can conclude that
 $(\text{i})_{\eta,s}$ holds in this case. 

\eqref{item_down_eta}: Assume that $\eta\in (0,17-12\sqrt{2})$ and $s\in
 [\tilde{w}(a_-(\eta),-1,\eta),\infty)$. The properties
\eqref{eq_eta_function_first_order}, 
 \eqref{eq_extended_function_slope},
 \eqref{eq_extended_function_profile} tell us the profile of the
 function $\tilde{w}(\cdot,-1,\eta)$, which together with
 \eqref{eq_uniform_lower_order} implies that
\begin{align}
w(x,y,\eta)<s,\quad (\forall (x,y)\in [1,a_+(\eta)]\times
 (-1,y_2(s,\eta)]).\label{eq_uniform_lower_type_second}
\end{align}
Since
 $\underline{w}(a_-(\eta)+(a_+(\eta)-a_-(\eta))/2,-1,\eta)<0$
 by \eqref{eq_internal_inside_minus_one}, there exists $y_9(s,\eta)\in
 (-1,y_4(s,\eta)]$ such that 
$\underline{w}(a_-(\eta)+(a_+(\eta)-a_-(\eta))/2,y,\eta)<0$ 
for any $y\in (-1,$ $y_9(s,\eta)]$. By taking this property,
 \eqref{eq_internal_inside_convex} and \eqref{eq_internal_inside_edge}
 into account we can prove the following statement. For any $y\in
 (-1,y_{9}(s,\eta)]$ there exists $x_2(y)\in
 (a_-(\eta)+(a_+(\eta)-a_-(\eta))/2,\eta^{-1}-\eps(s,\eta))$
 such that 
\begin{align*}
&\underline{w}(x,y,\eta)<0,\quad \left(\forall x\in
 \left(a_-(\eta)+\frac{1}{2}(a_+(\eta)-a_-(\eta)),x_2(y)\right)\right),\\
&\underline{w}(x_2(y),y,\eta)=0,\\
&\underline{w}(x,y,\eta)>0,\quad (\forall x\in
 (x_2(y),\eta^{-1}-\eps(s,\eta)]),
\end{align*}
or by \eqref{eq_internal_inside_relation}
\begin{align*}
&\frac{\partial w}{\partial x}(x,y,\eta)<0,\quad \left(\forall x\in
 \left(a_-(\eta)+\frac{1}{2}(a_+(\eta)-a_-(\eta)),x_2(y)\right)\right),\\
&\frac{\partial w}{\partial x}(x_2(y),y,\eta)=0,\\
&\frac{\partial w}{\partial x}(x,y,\eta)>0,\quad (\forall x\in
 (x_2(y),\eta^{-1}-\eps(s,\eta)]).
\end{align*}
By this property, \eqref{eq_internal_permanent_lower} and 
\eqref{eq_uniform_lower_type_second} for any $y\in (-1,y_{9}(s,\eta)]$ 
there exists $x_3(y)\in (a_+(\eta),\eta^{-1}-\eps(s,\eta))$ such that 
\begin{align*}
&w(x,y,\eta)<s,\quad (\forall x\in [1,x_3(y))),\\
&w(x_3(y),y,\eta)=s,\\
&w(x,y,\eta)>s,\quad \left(\forall x\in
 \left(x_3(y),\frac{1}{2\eta(y+1)}(\cosh^{-1}(|y|^{-1}))^2\right)\right).
\end{align*}
Thus, the property $(\text{i})_{\eta,s}$ holds.

Assume that $s\in [\tilde{w}(a_+(\eta),-1,\eta),
 \tilde{w}(a_-(\eta),-1,\eta))$. By \eqref{eq_uniform_lower_order} there
 exists $y_{10}(s,\eta)\in (-1,y_4(s,\eta)]$ such that
\begin{align}
w(a_-(\eta),y,\eta)>s>w(a_+(\eta),y,\eta),\quad (\forall y\in
 (-1,y_{10}(s,\eta)]).\label{eq_internal_winding}
\end{align}
Since $\underline{w}(1+(a_-(\eta)-1)/2,-1,\eta)>0$ and
 $\underline{w}(a_-(\eta)+(a_+(\eta)-a_-(\eta))/2,-1,\eta)<0$ by \eqref{eq_internal_inside_minus_one}, there
 exists $y_{11}(s,\eta)\in (-1,y_{10}(s,\eta)]$ such that 
$\underline{w}(1+(a_-(\eta)-1)/2,y,\eta)>0$,
 $\underline{w}(a_-(\eta)+(a_+(\eta)-a_-(\eta))/2,y,\eta)<0$ for any
 $y\in (-1,y_{11}(s,\eta)]$. This property combined with
 \eqref{eq_internal_inside_convex}, \eqref{eq_internal_inside_edge}
 implies the following statement. For any $y\in (-1,y_{11}(s,\eta)]$ there
 exist $x_4(y)\in (1+(a_-(y)-1)/2,a_-(y)+(a_+(y)-a_-(y))/2)$, $x_5(y)\in
 (a_-(y)+(a_+(y)-a_-(y))/2, \eta^{-1}-\eps(s,\eta))$ such that
\begin{align*}
&\underline{w}(x_4(y),y,\eta)=\underline{w}(x_5(y),y,\eta)=0,\\
&\underline{w}(x,y,\eta)>0,\quad (\forall x\in [1,x_4(y))),\\
&\underline{w}(x,y,\eta)<0,\quad (\forall x\in (x_4(y),x_5(y))),\\
&\underline{w}(x,y,\eta)>0,\quad (\forall x\in
 (x_5(y),\eta^{-1}-\eps(s,\eta)]),
\end{align*}
or by \eqref{eq_internal_inside_relation}
\begin{align*}
&\frac{\partial w}{\partial x}(x_4(y),y,\eta)=\frac{\partial w}{\partial x}(x_5(y),y,\eta)=0,\\
&\frac{\partial w}{\partial x}(x,y,\eta)>0,\quad (\forall x\in [1,x_4(y))),\\
&\frac{\partial w}{\partial x}(x,y,\eta)<0,\quad (\forall x\in (x_4(y),x_5(y))),\\
&\frac{\partial w}{\partial x}(x,y,\eta)>0,\quad (\forall x\in
 (x_5(y),\eta^{-1}-\eps(s,\eta)]).
\end{align*}
By considering these properties we can picture the profile of the function
 $w(\cdot,y,\eta)$.
Take any $y\in (-1,y_{11}(s,\eta)]$. Suppose that $s\ge
 w(x_4(y),y,\eta)$. Then, by the profile of $w(\cdot,y,\eta)$ in $[1,\eta^{-1}-\eps(s,\eta)]$, 
if $a\in [1,\eta^{-1}-\eps(s,\eta)]$ and
 $w(a,y,\eta)>s$, $w(a',y,\eta)>s$ for any $a'\in
 [a,\eta^{-1}-\eps(s,\eta)]$. This claim contradicts
 \eqref{eq_internal_winding}. Suppose that $s\le
 w(x_5(y),y,\eta)$. Then, if $a\in [1,\eta^{-1}-\eps(s,\eta)]$ and
 $w(a,y,\eta)>s$, $w(a',y,\eta)\ge s$ for any $a'\in
 [a,\eta^{-1}-\eps(s,\eta)]$. This claim contradicts
 \eqref{eq_internal_winding} as well. Therefore, 
$w(x_4(y),y,\eta)>s>w(x_5(y),y,\eta)$. Moreover, by
 \eqref{eq_internal_permanent_lower}, \eqref{eq_basic_vanishing}
and the profile of $w(\cdot,y,\eta)$,  
$x_4(y)>\frac{1}{2(y+1)}(\cosh^{-1}(|y|^{-1}))^2$
and there exist
\begin{align*} 
&x_6(y)\in
 \left(\frac{1}{2(y+1)}(\cosh^{-1}(|y|^{-1}))^2,x_4(y)\right),\\
&x_7(y)\in (x_4(y),x_5(y)),\\
&x_8(y)\in (x_5(y),\eta^{-1}-\eps(s,\eta))\left(
\subset \left(x_5(y),\frac{1}{2\eta(y+1)}(\cosh^{-1}(|y|^{-1}))^{2}\right)
\right)
\end{align*}
such that
\begin{align*}
&w(x_6(y),y,\eta)=w(x_7(y),y,\eta)=w(x_8(y),y,\eta)=s,\\
&w(x,y,\eta)>s,\quad (\forall x\in (x_6(y),x_7(y))),\\
&w(x,y,\eta)<s,\quad (\forall x\in (x_7(y),x_8(y))).
\end{align*}
This means that $(\text{ii})_{\eta,s}$ holds.

Finally let us assume that $s\in (0,\tilde{w}(a_+(\eta),-1,\eta))$. We can see from the
 profile of $\tilde{w}(\cdot,-1,\eta)$ that there uniquely exists
 $a_2\in (1,a_-(\eta))$ such that $s=\tilde{w}(a_2,-1,\eta)$. Moreover,
 there exists $a_3\in (a_2,a_-(\eta))$ such that
 $\tilde{w}(x,-1,\eta)\ge \tilde{w}(a_3,-1,\eta)>s$ for all $x\in
 [a_3,\eta^{-1})$. Thus we can take $y_{12}(s,\eta)\in (-1,y_4(s,\eta)]$
 so that $w(x,y,\eta)>s$ for any $(x,y)\in
 [a_3,\eta^{-1}-\eps(s,\eta)]\times (-1,y_{12}(s,\eta)]$. Since
 $\underline{w}(x,-1,\eta)\ge \underline{w}(a_3,-1,\eta)>0$ for any
 $x\in [1,a_3]$ by \eqref{eq_internal_inside_minus_one},
there exists $y_{13}(s,\eta)\in (-1,y_{12}(s,\eta)]$
 such that for any $(x,y)\in [1,a_3]\times (-1,y_{13}(s,\eta)]$
 $\underline{w}(x,y,\eta)>0$, and thus by
 \eqref{eq_internal_inside_relation} $\frac{\partial w}{\partial
 x}(x,y,\eta)>0$. These combined with \eqref{eq_eta_dependent_y_order},
 \eqref{eq_basic_vanishing} imply that for any $y\in
 (-1,y_{13}(s,\eta)]$ $a_3>\frac{1}{2(y+1)}(\cosh^{-1}(|y|^{-1}))^2$ and 
there exists $x_9(y)\in
 (\frac{1}{2(y+1)}(\cosh^{-1}(|y|^{-1}))^2,a_3)$ such that 
\begin{align*}
&w(x,y,\eta)<s,\quad \left(\forall x\in
 \left(\frac{1}{2(y+1)}(\cosh^{-1}(|y|^{-1}))^2, x_9(y)
\right)\right),\\
&w(x_9(y),y,\eta)=s,\\
&w(x,y,\eta)>s,\quad (\forall x\in (x_9(y),\eta^{-1}-\eps(s,\eta)]).
\end{align*}
Taking into account \eqref{eq_internal_permanent_lower}, we conclude that
 $(\text{i})_{\eta,s}$ holds.
\end{proof}

By applying Lemma \ref{lem_internal_function_classification} we can
prove Proposition \ref{prop_multi_orbital}. 

\begin{proof}[Proof of Proposition \ref{prop_multi_orbital}]
Recalling \eqref{eq_large_function}, we see that 
\begin{align}
g_{E_b}(x,t,0)=\frac{b'}{e_{max}}\left(
-\frac{2e_{max}}{b'|U|}+W\left(e_{max}x,\cos\left(\frac{t}{2}\right),
 \frac{e_{min}}{e_{max}},\frac{b-b'}{b'}\right)\right).\label{eq_initial_gap_large_relation}
\end{align}
Theorem \ref{thm_boundary_characterization} implies that if
 $e_{min}/e_{max}>\sqrt{17-12\sqrt{2}}$, for any $b\in \N_{\ge 2}$,
 $b'\in \{1,\cdots,b-1\}$ ($\star$) holds. Let us assume that 
$e_{min}/e_{max}\le \sqrt{17-12\sqrt{2}}$ in the following. For
 $\eta=(e_{min}/e_{max})^2$, $s=(b-b')/b'$ let us prove the following
 statements. 
\begin{itemize}
\item If the condition $(\text{i})_{\eta,s}$ holds, ($\star$) holds.
\item If the condition $(\text{ii})_{\eta,s}$ holds, ($\star$) does not
      hold.
\end{itemize}

Assume that $(\text{i})_{\eta,s}$ holds. Then by
 \eqref{eq_large_small_relation} there exists $y_0\in (-1,0)$ such that
 for any $y\in (-1,y_0]$ there exists 
\begin{align*}
&x_0(y)\in \left(\frac{1}{2(y+1)}(\cosh^{-1}(|y|^{-1}))^2, \frac{1}{2\eta(y+1)}(\cosh^{-1}(|y|^{-1}))^2\right)
\end{align*}
such that 
\begin{align*}
&\frac{\partial W}{\partial x}(\sqrt{y+1}\cdot
 x,y,\sqrt{\eta},s)>0,\quad \left(\forall x\in
 \left(\frac{1}{\sqrt{y+1}}\cosh^{-1}(|y|^{-1}), \sqrt{2x_0(y)}
\right)
\right),\\
&\frac{\partial W}{\partial x}(\sqrt{2(y+1)x_0(y)},y,\sqrt{\eta},s)=0,\\
&\frac{\partial W}{\partial x}(\sqrt{y+1}\cdot
 x,y,\sqrt{\eta},s)<0,\quad \left(\forall x\in
 \left(\sqrt{2x_0(y)},\frac{1}{\sqrt{\eta(y+1)}}\cosh^{-1}(|y|^{-1})\right)
\right).
\end{align*}
Let $\cos^{-1}:[-1,1]\to [0,\pi]$ denote the inverse function of
 $\cos|_{[0,\pi]}$. The above property and
 \eqref{eq_initial_gap_large_relation} ensure that 
for any $t\in [2\cos^{-1}y_0,2\pi)(\subset (\pi,2\pi))$ there exists 
\begin{align*}
\hat{x}(t)\in \left(
\frac{1}{e_{max}}\cosh^{-1}\left(\left|\cos\left(\frac{t}{2}\right)
\right|^{-1}\right), 
\frac{1}{e_{min}}\cosh^{-1}\left(\left|\cos\left(\frac{t}{2}\right)
\right|^{-1}\right)\right)
\end{align*}
such that
\begin{align*}
&\frac{\partial g_{E_b}}{\partial x}(x,t,0)>0,\quad \left(\forall x\in 
 \left(
\frac{1}{e_{max}}\cosh^{-1}\left(\left|\cos\left(\frac{t}{2}\right)
\right|^{-1}\right), \hat{x}(t)\right)\right),\\
&\frac{\partial g_{E_b}}{\partial x}(\hat{x}(t),t,0)=0,\\
&\frac{\partial g_{E_b}}{\partial x}(x,t,0)<0,\quad \left(\forall x\in 
 \left(\hat{x}(t),
\frac{1}{e_{min}}\cosh^{-1}\left(\left|\cos\left(\frac{t}{2}\right)
\right|^{-1}\right)\right)\right),
\end{align*}
or by taking into account \eqref{eq_temporal_function_derivative}
\begin{align}
&\frac{\partial g_{E_b}}{\partial x}(x,t,0)>0,\quad (\forall x\in
 (0,\hat{x}(t))),\label{eq_unique_mountain_peak}\\
&\frac{\partial g_{E_b}}{\partial x}(\hat{x}(t),t,0)=0,\notag\\
&\frac{\partial g_{E_b}}{\partial x}(x,t,0)<0,\quad (\forall x\in
 (\hat{x}(t),\infty)).\notag
\end{align}
Suppose that ($\star$) does not hold. Then for any $U_0\in
 (0,2e_{min}/b)$ there exists $U\in [-U_0,0)$ such that $\tau(\cdot)$
 has more than one local minimum points in $(0,\beta_c)$. By
 \eqref{eq_boundary_gap_equation_direct}
 $\cos(\tau(\beta)/2)+1\le \frac{\sinh(2)b}{2e_{min}}U_0$. Thus
 if we take $U_0$ sufficiently small, $\tau(\beta)\in
 [2\cos^{-1}y_0,2\pi)$ for any $\beta\in (0,\beta_c)$. Now there are
 $\beta_j\in (0,\beta_c)$ $(j=1,2,3)$ such that
 $\beta_1<\beta_2<\beta_3$ and
 $\tau(\beta_1)=\tau(\beta_2)=\tau(\beta_3)$. Thus,
 $g_{E_b}(\beta_j,\tau(\beta_1),0)=0$ for $j\in \{1,2,3\}$, which 
implies that there exist $x_1\in (\beta_1,\beta_2)$, $x_2\in
 (\beta_2,\beta_3)$ such that 
$\frac{\partial g_{E_b}}{\partial x}(x_j,\tau(\beta_1),0)=0$ for $j\in \{1,2\}$. This contradicts
 \eqref{eq_unique_mountain_peak} with $t=\tau(\beta_1)$. Therefore,
 ($\star$) must hold.

Assume that $(\text{ii})_{\eta,s}$ holds. Take any $U_0\in
 (0,2e_{min}/b)$. The limit in the left-hand side of \eqref{eq_arc_cosh_limit_relation} tells
 us that there exists $\eps\in \R_{>0}$ such that 
\begin{align*}
\left(
\frac{1}{\sqrt{y+1}}\cosh^{-1}(|y|^{-1}),
\frac{1}{\sqrt{\eta(y+1)}}\cosh^{-1}(|y|^{-1})
\right)\subset \left[1,\frac{2}{\sqrt{\eta}}\right]
\end{align*}
for any $y\in (-1,-1+\eps)$. Thus, 
\begin{align*}
&W(\sqrt{y+1}\cdot x,y,\sqrt{\eta},s)\ge \inf_{\xi\in
 [1,\frac{2}{\sqrt{\eta}}]}\frac{\sinh(\sqrt{y+1}\cdot\xi)}{y+\cosh(\sqrt{y+1}\cdot
 \xi)},\\
&\left(\forall y\in (-1,-1+\eps),\ x\in
 \left(\frac{1}{\sqrt{y+1}}\cosh^{-1}(|y|^{-1}),
 \frac{1}{\sqrt{\eta(y+1)}}\cosh^{-1}(|y|^{-1})\right)\right).
\end{align*}
Since the right-hand side of the above inequality diverges to $+\infty$
 as $y\searrow -1$, there exists $y_1\in (-1,0)$ such that 
\begin{align}
&W(\sqrt{y+1}\cdot x,y,\sqrt{\eta},s)>
 \frac{2e_{max}}{b'U_0},\label{eq_large_lower_bound_specific}\\
&\left(\forall y\in (-1,y_1],\ x\in
 \left(\frac{1}{\sqrt{y+1}}\cosh^{-1}(|y|^{-1}),
 \frac{1}{\sqrt{\eta(y+1)}}\cosh^{-1}(|y|^{-1})\right)\right).\notag
\end{align}
By the assumption and \eqref{eq_large_small_relation} there exist
 $y\in (-1,y_1]$, 
\begin{align*}
x_j(y)\in \left(
\frac{1}{2(y+1)}(\cosh^{-1}(|y|^{-1}))^2,
\frac{1}{2\eta(y+1)}(\cosh^{-1}(|y|^{-1}))^2
\right),\quad (j=1,2,3)
\end{align*}
such that $x_1(y)<x_2(y)<x_3(y)$, 
\begin{align*}
&\frac{\partial W}{\partial
 x}\big(\sqrt{2(y+1)x_j(y)},y,\sqrt{\eta},s\big)=0,\quad (\forall j\in
 \{1,2,3\}),\\
&\frac{\partial W}{\partial
 x}(\sqrt{y+1}\cdot x,y,\sqrt{\eta},s)<0,\quad (\forall x\in
 (\sqrt{2x_1(y)},\sqrt{2x_2(y)})),\\
&\frac{\partial W}{\partial
 x}(\sqrt{y+1}\cdot x,y,\sqrt{\eta},s)>0,\quad (\forall x\in
 (\sqrt{2x_2(y)},\sqrt{2x_3(y)})).
\end{align*}
By combining this with \eqref{eq_large_lower_bound_specific} we have that
\begin{align*}
\min_{j\in \{1,3\}}W\big(\sqrt{2(y+1)x_j(y)},y,\sqrt{\eta},s\big)>
W(\sqrt{2(y+1)x_2(y)},y,\sqrt{\eta},s)>\frac{2e_{max}}{b'U_0}.
\end{align*}
Thus, there exists $U\in [-U_0,0)$ such that
\begin{align*}
\min_{j\in
 \{1,3\}}W(\sqrt{2(y+1)x_j(y)},y,\sqrt{\eta},s)>\frac{2e_{max}}{b'|U|}>
W(\sqrt{2(y+1)x_2(y)},y,\sqrt{\eta},s).
\end{align*}
Therefore, there exist
\begin{align*}
&\hat{\beta}_1\in\left(0,\frac{1}{e_{max}}\sqrt{2(y+1)x_1(y)}\right),\\ 
&\hat{\beta}_2\in \left(\frac{1}{e_{max}}\sqrt{2(y+1)x_1(y)},
 \frac{1}{e_{max}}\sqrt{2(y+1)x_2(y)}\right),\\
&\hat{\beta}_3\in \left(\frac{1}{e_{max}}\sqrt{2(y+1)x_2(y)},
 \frac{1}{e_{max}}\sqrt{2(y+1)x_3(y)}\right)
\end{align*}
such that
 $W(e_{max}\hat{\beta}_j,y,\sqrt{\eta},s)=\frac{2e_{max}}{b'|U|}$ for
 $j\in \{1,2,3\}$,
or by \eqref{eq_initial_gap_large_relation} and Lemma \ref{lem_tau_implicit_uniqueness}
\begin{align*}
\hat{\beta}_j\in (0,\beta_c),\quad
\cos\Big(\frac{\tau(\hat{\beta}_j)}{2}\Big)=y,\quad (\forall j\in \{1,2,3\}).
\end{align*}
Then by repeating the same argument as the final part of the proof of
 Proposition \ref{prop_necessity_equality} we can reach the conclusion
 that $\tau(\cdot)$ has more than one local minimum points. This means
 that ($\star$) does not hold. 

Now we know that it suffices to determine for which  $(\eta,s)$
 $(\text{i})_{\eta,s}$ (or $(\text{ii})_{\eta,s}$) holds. In fact for
 this purpose we have prepared Lemma
 \ref{lem_internal_function_classification}. We still need more
 information about how the function $\tilde{w}(\cdot,-1,\eta)$ behaves
 when $\eta$ varies. We can derive from \eqref{eq_extended_function}
 that for $z\in \R_{>0}$, $x\in (0,z^{-1})$
\begin{align*}
\frac{\partial \tilde{w}}{\partial
 z}(x,-1,z)=\frac{(x-1)x(1+zx)(3-zx)}{(x+1)^2(1-zx)^2}.
\end{align*}
Then we can see from this equality and \eqref{eq_eta_function_first_order}
 that for $\eta\in (0,17-12\sqrt{2})$, $\delta \in \{+,-\}$,
 $\frac{\partial \tilde{w}}{\partial z}(a_{\delta}(\eta),-1,\eta)>0$. 
Combination of this inequality and \eqref{eq_extended_function_slope}
 implies that for $\eta\in (0,17-12\sqrt{2})$, $\delta\in \{+,-\}$ 
\begin{align*}
&\frac{d}{d\eta}\tilde{w}(a_{\delta}(\eta),-1,\eta)=\frac{\partial
 \tilde{w}}{\partial x}(a_{\delta}(\eta),-1,\eta)\frac{d
 a_{\delta}}{d\eta}(\eta) + \frac{\partial \tilde{w}}{\partial
 z}(a_{\delta}(\eta),-1,\eta)>0.
\end{align*}
By this and \eqref{eq_extended_function_slope} again
we can understand that both the local maximum and the local minimum of 
the function $x\mapsto \tilde{w}(x,-1,\eta)$ $:(0,\eta^{-1})\to\R$ are
 strictly monotone increasing with $\eta\in
 (0,17-12\sqrt{2})$. Moreover, by \eqref{eq_extended_function},
\eqref{eq_extended_function_a},
\eqref{eq_eta_dependent_variable}, \eqref{eq_eta_function_critical_order} 
\begin{align*}
&\lim_{\eta\nearrow 17-12\sqrt{2}}a_+(\eta)=\lim_{\eta\nearrow 17-12\sqrt{2}}a_-(\eta)=3+2\sqrt{2},\\
&\lim_{\eta\nearrow 17-12\sqrt{2}}\tilde{w}(a_+(\eta),-1,\eta)=
\lim_{\eta\nearrow
 17-12\sqrt{2}}\tilde{w}(a_-(\eta),-1,\eta)=3-2\sqrt{2},\\
&\lim_{\eta\searrow 0}a_-(\eta)=3,\quad \lim_{\eta\searrow
 0}a_+(\eta)=+\infty,\\
&\lim_{\eta\searrow 0}\tilde{w}(a_-(\eta),-1,\eta)=\frac{1}{8},\quad 
  \lim_{\eta\searrow 0}\tilde{w}(a_+(\eta),-1,\eta)=0.
\end{align*}

In the following we let $\eta =(e_{min}/e_{max})^2$, $s=(b-b')/b'$.
If $e_{min}/e_{max}=\sqrt{17-12\sqrt{2}}$, by Lemma
 \ref{lem_internal_function_classification} \eqref{item_up_eta} for
 any $b\in \N_{\ge 2}$, $b'\in \{1,\cdots,b-1\}$ the condition
 $(\text{i})_{\eta,s}$ holds and thus ($\star$) holds. 

Assume that $e_{min}/e_{max}\in (0,\sqrt{17-12\sqrt{2}})$. In this
 situation Lemma \ref{lem_internal_function_classification}
 \eqref{item_down_eta} is applicable. If $\frac{b-b'}{b'}\in
 [3-2\sqrt{2},\infty)$,  $\frac{b-b'}{b'}\in
 [\tilde{w}(a_-(\eta),-1,\eta),\infty)$ and thus $(\text{i})_{\eta,s}$
 holds. Thus ($\star$) holds. If $\frac{b-b'}{b'}\in
 (1/8,3-2\sqrt{2})$, there exist $e_1$, $e_2\in
 (0,\sqrt{17-12\sqrt{2}})$ such that $e_1<e_2$, 
\begin{align*}
&\frac{b-b'}{b'}\in (0, \tilde{w}(a_+(\eta),-1,\eta))\ \text{if
 }\frac{e_{min}}{e_{max}}\in \left(e_2,\sqrt{17-12\sqrt{2}}\right),\\
&\frac{b-b'}{b'}\in [\tilde{w}(a_+(\eta),-1,\eta), \tilde{w}(a_-(\eta),-1,\eta))
\ \text{if
 }\frac{e_{min}}{e_{max}}\in (e_1,e_2],\\
&\frac{b-b'}{b'}\in [\tilde{w}(a_-(\eta),-1,\eta), \infty)
\ \text{if
 }\frac{e_{min}}{e_{max}}\in (0,e_1],
\end{align*}
or
\begin{align*}
&(\text{i})_{\eta,s}\text{ holds and thus }(\star)\text{ holds if
 }\frac{e_{min}}{e_{max}}\in \left(e_2,\sqrt{17-12\sqrt{2}}\right),\\
&(\text{ii})_{\eta,s}\text{ holds and thus }(\star)\text{ does not hold if
 }\frac{e_{min}}{e_{max}}\in (e_1,e_2],\\
&(\text{i})_{\eta,s}\text{ holds and thus }(\star)\text{ holds if
 }\frac{e_{min}}{e_{max}}\in (0,e_1].
\end{align*}
If $\frac{b-b'}{b'}\in (0,1/8]$, there exists $e_1\in
 (0,\sqrt{17-12\sqrt{2}})$ such that 
\begin{align*}
&\frac{b-b'}{b'}\in (0, \tilde{w}(a_+(\eta),-1,\eta))\ \text{if
 }\frac{e_{min}}{e_{max}}\in \left(e_1,\sqrt{17-12\sqrt{2}}\right),\\
&\frac{b-b'}{b'}\in [\tilde{w}(a_+(\eta),-1,\eta), \tilde{w}(a_-(\eta),-1,\eta))
\ \text{if
 }\frac{e_{min}}{e_{max}}\in (0,e_1],
\end{align*}
or 
\begin{align*}
&(\text{i})_{\eta,s}\text{ holds and thus }(\star)\text{ holds if
 }\frac{e_{min}}{e_{max}}\in \left(e_1,\sqrt{17-12\sqrt{2}}\right),\\
&(\text{ii})_{\eta,s}\text{ holds and thus }(\star)\text{ does not hold if
 }\frac{e_{min}}{e_{max}}\in (0,e_1].
\end{align*}
These can be summarized as in the statements of the proposition.
\end{proof}

In fact in this model $\tau(\beta)$ can be exactly computed. Remind us 
that $\cos^{-1}:[-1,1]\to [0,\pi]$ denotes the inverse function
of $\cos|_{[0,\pi]}$. 

\begin{proposition}\label{prop_exact_solution}
Set 
\begin{align*}
&D_0:=\cosh(\beta e_{max})\cosh(\beta e_{min})\\
&\qquad\quad -\frac{|U|}{2}\left(
\frac{b'}{e_{max}}\sinh(\beta e_{max})\cosh(\beta e_{min})+
 \frac{b-b'}{e_{min}}\cosh(\beta e_{max})\sinh(\beta e_{min})\right),\\
&D_1:=\cosh(\beta e_{max})+\cosh(\beta e_{min})
-\frac{|U|}{2}\left(
\frac{b'}{e_{max}}\sinh(\beta e_{max})+
 \frac{b-b'}{e_{min}}\sinh(\beta e_{min})\right).
\end{align*}
Assume that $U\in (-2e_{min}/b,0)$. Then for any $\beta\in
 (0,\beta_c)$, $D_1^2-4D_0>0$, $\frac{1}{2}(-D_1+\sqrt{D_1^2-4D_0})\in
 (-1,0)$ and 
$$\tau(\beta)=2\cos^{-1}\left(\frac{-D_1+\sqrt{D_1^2-4D_0}}{2}\right).$$ 
\end{proposition}

\begin{proof} 
The statements of Lemma \ref{lem_critical_temperature}
 \eqref{item_critical_basic},\eqref{item_critical_positive} imply the
 following basic fact. On the assumption $|U|<2e_{min}/b$ for any
 $\beta\in (0,\beta_c)$ there uniquely exists $y\in (-1,0)$ such that 
\begin{align}
&-\frac{2}{|U|}+b'\frac{\sinh(\beta e_{max})}{(y+\cosh(\beta
 e_{max}))e_{max}}+(b-b')\frac{\sinh(\beta e_{min})}{(y+\cosh(\beta
 e_{min}))e_{min}}=0.\label{eq_basic_rational_equation}
\end{align}
Moreover, for $y\in [0,\infty)$ \eqref{eq_basic_rational_equation} does
 not hold. Observe that $y\in (-1,0)$ and $y$ solves
 \eqref{eq_basic_rational_equation} if and only if $y\in (-1,0)$ and $y$
 solves $y^2+D_1y+D_0=0$. Setting 
\begin{align*}
&X_1:=\cosh(\beta e_{max}),\quad X_2:=\cosh(\beta e_{min}),\\
&Y_1:=\frac{|U|b'}{2e_{max}}\sinh(\beta e_{max}),\quad 
Y_2:=\frac{|U|(b-b')}{2e_{min}}\sinh(\beta e_{min}),
\end{align*}
we can derive that
\begin{align*}
D_1^2-4D_0=(X_1-X_2-Y_1+Y_2)^2+4Y_1Y_2>0.
\end{align*}
Set $y_+:=\frac{1}{2}(-D_1+\sqrt{D_1^2-4D_0})$, 
   $y_-:=\frac{1}{2}(-D_1-\sqrt{D_1^2-4D_0})$. These are the roots of
 $y^2+D_1y+D_0$. The unique solution to
 \eqref{eq_basic_rational_equation} in $(-1,0)$ must be one of them.
If $y_+\ge 0$, \eqref{eq_basic_rational_equation} has a
 non-negative solution, which is a contradiction. Thus $y_+<0$. If
 $y_->-1$, \eqref{eq_basic_rational_equation} has the 2 different
 solutions $y_+$, $y_-\in (-1,0)$, which is again a contradiction. Thus
 $y_-\le -1$. Therefore the solution to
 \eqref{eq_basic_rational_equation} in $(-1,0)$ must be $y_+$, and thus the claims follow.
\end{proof}

Let $b=8$, $b'=7$, $e_{min}=1$. In this case
$\frac{b-b'}{b'}=1/7\in (1/8,3-2\sqrt{2})$. Proposition
\ref{prop_multi_orbital} \eqref{item_band_middle} implies that there
exist $U\in (-2e_{min}/b,0)$ ($=(-1/4,0)$) and
$e_{max,1}$, $e_{max,2}$, $e_{max,3}\in
(1/\sqrt{17-12\sqrt{2}},\infty)$ $(\approx (5.83,\infty))$ such that
$e_{max,1}<e_{max,2}<e_{max,3}$ and for $U$ $\tau(\cdot)$ has
only one local minimum point if $e_{max}=e_{max,1}$, $\tau(\cdot)$ has
more than one local minimum points if $e_{max}=e_{max,2}$,   
$\tau(\cdot)$ has
only one local minimum point if $e_{max}=e_{max,3}$.
Figure \ref{fig_exact_solution} shows the graph $\{(\beta,\tau(\beta))\
|\ \beta\in (0,\beta_c) \}$ for $U=-1/8$, $e_{max}=6,7,9$.
In these cases $U\in (-2e_{min}/b,0)$, 
$e_{max}\in (1/\sqrt{17-12\sqrt{2}},\infty)$. 
The figure demonstrates the properties described above. The graph was drawn by
implementing the exact solution obtained in Proposition
\ref{prop_exact_solution}. 

\begin{figure}
\begin{center}
\begin{tabular}{c}
\includegraphics[width=7cm]{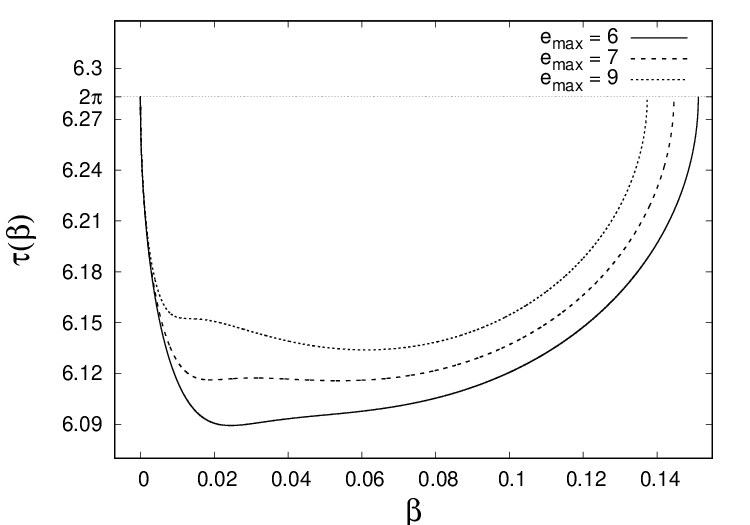}\\
(a)\\
\includegraphics[width=7cm]{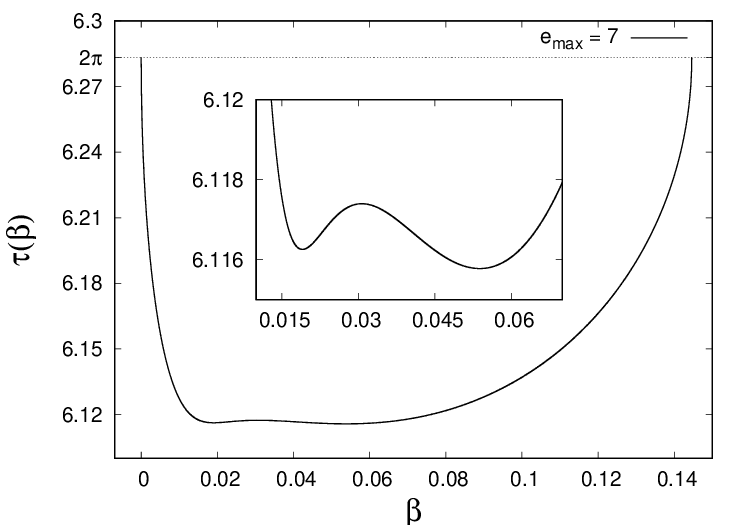}\ 
\includegraphics[width=7cm]{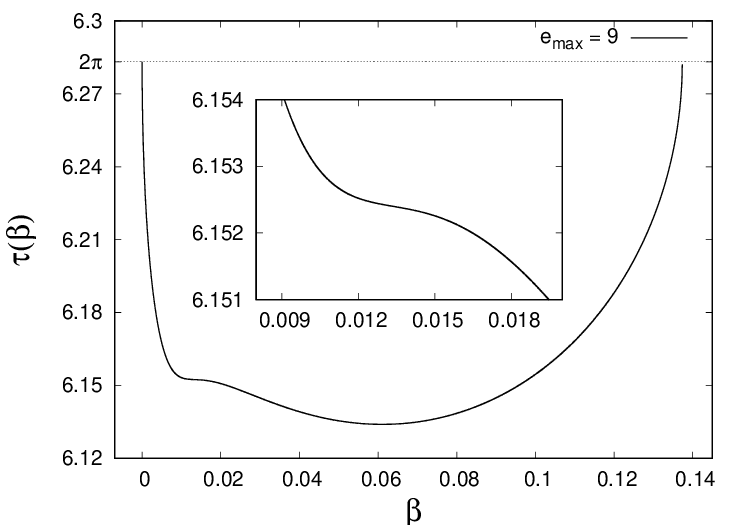}\\
(b)\qquad\qquad\qquad\qquad\qquad\qquad\qquad\qquad (c)
\end{tabular}
\caption{The graph $\{(\beta,\tau(\beta))\ |\ \beta\in (0,\beta_c)\}$
 drawn by implementing the exact solution for $b=8$, $b'=7$,
 $U=-1/8$, $e_{min}=1$ and $e_{max}=6,7,9$. Picture (a) shows the graphs
 for $e_{max}=6,7,9$. We can see that $\tau(\cdot)$ has only one
 local minimum point when $e_{max}=6$. Picture (b) shows the graph for $e_{max}=7$.
By magnifying we can see that $\tau(\cdot)$ has two local
 minimum points. Picture (c) shows the graph for $e_{max}=9$. By
 magnifying we can see that $\tau(\cdot)$ has only one local minimum point.}\label{fig_exact_solution}
\end{center}
\end{figure}

\subsubsection{The one-dimensional model with nearest-neighbor
   hopping}\label{subsubsec_one_band}

As for the model defined in \eqref{item_model_one_band}, we find a
simpler result as follows.

\begin{proposition}\label{prop_one_band_result}
For any $t\in \R_{\ge 0}$, $e_{min}\in\R_{>0}$ there exists $U_0\in
 (0,2e_{min})$ such that for any $U\in [-U_0,0)$ $\tau(\cdot)$ has one
 and only one local minimum point in $(0,\beta_c)$.
\end{proposition}

\begin{proof}
Let us assume that $e_{min}=1$ for the moment. We will see that the
 other case can be deduced from this special case. It follows that
 $e_{max}=2t+1$. Define the open set $\cO$ of $\R^2$ by
\begin{align*}
\cO:=\left\{(x,y)\in\R^2\ \Big|\
 \sum_{n=2}^{\infty}\frac{x^{2n}}{(2n)!}|y+1|^{n-1}e_{max}^{2n}<1\text{
 or }y> -1\right\}.
\end{align*}
We define the function $P:\cO\to\R$ as follows. 
\begin{align*}
P(x,y):=\frac{1}{2\pi}\int_0^{2\pi}dk\frac{x+\sum_{n=1}^{\infty}\frac{x^{2n+1}}{(2n+1)!}(y+1)^nE_1(k)^{2n}}{1+\frac{x^2}{2}E_1(k)^2+\sum_{n=2}^{\infty}\frac{x^{2n}}{(2n)!}(y+1)^{n-1}E_1(k)^{2n}}.
\end{align*}
The function $P$ is real analytic in $\cO$. Let us observe that for
 $(x,y)\in \R_{>0}\times (-1,\infty)$
\begin{align}
P(x,y)=\frac{\sqrt{y+1}}{2\pi}\int_0^{2\pi}dk\frac{\sinh(\sqrt{y+1}\cdot
 x E_1(k))}{(y+\cosh(\sqrt{y+1}\cdot x
 E_1(k)))E_1(k)}.\label{eq_scaling_explicit_relation}
\end{align}
We can apply Lemma \ref{lem_definite_integral} proved in Appendix
 \ref{app_integral} to derive that for any $x\in \R$
\begin{align*}
P(x,-1)^2=\frac{e_{max}^{-1}\left(\frac{x^2}{2}\right)}{\left(\frac{x^2}{2}+1
\right)\left(\frac{x^2}{2}+e_{max}^{-2}\right)}
\left(\left(\frac{x^2}{2}+1
\right)^{\frac{1}{2}}
\left(\frac{x^2}{2}+e_{max}^{-2}
\right)^{\frac{1}{2}}-\frac{x^2}{2}+e_{max}^{-1}\right).
\end{align*}
To facilitate the derivation of the above equality from Lemma \ref{lem_definite_integral}, let us add that we multiplied both the
 numerator and the denominator of $P(x,-1)^2$ by
$$
\left(\frac{x^2}{2}+1
\right)^{\frac{1}{2}}
\left(e_{max}^{2}\frac{x^2}{2}+1\right)^{\frac{1}{2}}
-e_{max}\frac{x^2}{2}-1
$$
at the beginning.
Moreover, setting 
\begin{align*}
&P_1(x):=\left(\frac{x^2}{2}+1
\right)^{\frac{1}{2}}
\left(\frac{x^2}{2}+e_{max}^{-2}
\right)^{\frac{1}{2}}\left((1+e_{max}^{-2})\frac{x^2}{2}+2e_{max}^{-2}\right),\\
&P_2(x):=2(e_{max}^{-2}+e_{max}^{-1}+1)\left(\frac{x^2}{2}\right)^2+4e_{max}^{-2}\left(\frac{x^2}{2}\right)-2e_{max}^{-3},\\
&P_3(x):=\frac{2e_{max}}{x}\left(\frac{x^2}{2}+1
\right)^{2}\left(\frac{x^2}{2}+e_{max}^{-2}
\right)^{2},
\end{align*}
we see that for any $x\in \R_{>0}$
$$\frac{d}{dx}P(x,-1)^2=\frac{P_1(x)-P_2(x)}{P_3(x)}.$$
 If we assume that
 $\hat{x}\in \R_{>0}$, $y\in (-1,-\frac{1}{2}]$ and $\frac{\partial
 P}{\partial x}(\hat{x},y)=0$, 
it follows from \eqref{eq_temporal_function_derivative} that 
\begin{align*}
&\hat{x}\in
 \left[\frac{1}{e_{max}\sqrt{y+1}}\cosh^{-1}(|y|^{-1}),\frac{1}{\sqrt{y+1}}\cosh^{-1}(|y|^{-1})\right],\\
&\frac{d}{dx}P(x,-1)^2\Big|_{x=\hat{x}}+\frac{\partial}{\partial
 x}(P(x,y)^2-P(x,-1)^2)\Big|_{x=\hat{x}}=0.
\end{align*}
Let us recall the definition \eqref{eq_definition_c_max} of
 $c_{max}$. We can also deduce from \eqref{eq_arc_cosh} that
if we set 
$$
c_{min}:=\inf_{y\in (-1,-\frac{1}{2}]}\frac{\cosh^{-1}(|y|^{-1})}{\sqrt{y+1}},
$$ 
$0<c_{min}<\infty$. Then the above properties lead to that
\begin{align*}
&\hat{x}\in \left[\frac{c_{min}}{e_{max}},c_{max}\right],\\
&P_1(\hat{x})^2-P_2(\hat{x})^2+2P_2(\hat{x})P_3(\hat{x})\frac{\partial}{\partial
 x}(P(x,y)^2-P(x,-1)^2)\Big|_{x=\hat{x}}\\
& -P_3(\hat{x})^2\left(
\frac{\partial}{\partial x}(P(x,y)^2-P(x,-1)^2)\Big|_{x=\hat{x}}
\right)^2=0.
\end{align*}

Let us define the function $Q:\R_{>0}\times (-1,\infty)\to \R$ by 
\begin{align*}
Q(x,y):=&P_1(x)^2-P_2(x)^2+2P_2(x)P_3(x)\frac{\partial}{\partial
 x}(P(x,y)^2-P(x,-1)^2)\\
&-P_3(x)^2\left(\frac{\partial}{\partial
 x}(P(x,y)^2-P(x,-1)^2)\right)^2.
\end{align*}
We will prove the following statement. 
\begin{align}
&\text{There exists }y_0(e_{max})\in
 \left(-1,-\frac{1}{2}\right]\text{ depending only on }e_{max}\text{
 such that}\label{eq_logical_target}\\
&\text{if for }y\in
 (-1,y_0(e_{max})]\text{ a solution to }Q(x,y)=0\text{ exists in
 }\left[\frac{c_{min}}{e_{max}},c_{max}\right],\notag\\
&\text{then it is
 unique.}\notag
\end{align}
We can expand $P_1(x)^2-P_2(x)^2$ as
 follows. 
\begin{align*}
P_1(x)^2-P_2(x)^2=\sum_{j=1}^4a_j(e_{max})\left(\frac{x^2}{2}\right)^j,
\end{align*}
where $a_j(e_{max})$ $(j=1,\cdots,4)$ are real coefficients depending
 only on $e_{max}$. We can check that 
\begin{align}
a_1(e_{max})>0,\quad a_2(e_{max})>0,\quad
 a_4(e_{max})<0.\label{eq_key_coefficient_sign}
\end{align}
We do not need to deal with $a_3(e_{max})$, since the term involving
 $a_3(e_{max})$ will be subsequently canceled. Though it is not
 essential to make explicit, $a_2(e_{max})$ is computed as
 follows.
 $a_2(e_{max})=5e_{max}^{-6}+8e_{max}^{-5}+6e_{max}^{-4}+8e_{max}^{-3}+5
 e_{max}^{-2}$. Assume that $(x_0,y)\in
 [c_{min}/e_{max},c_{max}]\times (-1,-1/2]$
 and $Q(x_0,y)=0$. We can derive that 
\begin{align*}
&x_0\frac{\partial Q}{\partial x}(x_0,y)\\
&=\sum_{j=1}^42ja_j(e_{max})\left(\frac{x_0^2}{2}\right)^j\\
&\quad 
+x_0\frac{\partial}{\partial
 x}\Bigg(2P_2(x)P_3(x)\frac{\partial}{\partial
 x}(P(x,y)^2-P(x,-1)^2)\\
&\qquad\qquad\quad -P_3(x)^2\left(\frac{\partial}{\partial x}(P(x,y)^2-P(x,-1)^2)
\right)^2\Bigg)\Bigg|_{x=x_0}\\
&=\sum_{j\in\{1,2,4\}}(2j-6)a_j(e_{max})\left(\frac{x_0^2}{2}\right)^j\\
&\quad -12P_2(x_0)P_3(x_0)
\frac{\partial}{\partial
 x}(P(x,y)^2-P(x,-1)^2)\Big|_{x=x_0}\\
&\quad +6P_3(x_0)^2\left(\frac{\partial}{\partial x}(P(x,y)^2-P(x,-1)^2)\Big|_{x=x_0}
\right)^2\\
&\quad 
+x_0\frac{\partial}{\partial
 x}\Bigg(2P_2(x)P_3(x)\frac{\partial}{\partial
 x}(P(x,y)^2-P(x,-1)^2)\\
&\qquad\qquad\quad -P_3(x)^2\left(\frac{\partial}{\partial x}(P(x,y)^2-P(x,-1)^2)
\right)^2\Bigg)\Bigg|_{x=x_0}\\
&\le -2a_2(e_{max})\left(\frac{c_{min}^2}{2e_{max}^2}\right)^2\\
&\quad +c\sup_{x\in
 [\frac{c_{min}}{e_{max}},c_{max}]}\Bigg((1+c_{max})|P_2(x)P_3(x)|+
 (1+c_{max})|P_3(x)^2|\\
&\qquad\qquad\qquad +c_{max}\left|\frac{dP_2}{dx}(x)P_3(x)\right|+c_{max}\left|P_2(x)\frac{dP_3}{dx}(x)\right|+c_{max}\left|\frac{dP_3}{dx}(x)P_3(x)\right|
\Bigg)\\
&\quad\cdot \left(1+\sum_{i,j\in \{0,1,2\}}1_{1\le i+j\le 2}\sup_{x\in
 [\frac{c_{min}}{e_{max}},c_{max}]}\sup_{\eta\in
 [-1,-\frac{1}{2}]}\left|\frac{\partial^i P}{\partial
 x^i}(x,\eta)\frac{\partial^{j+1} P}{\partial x^j\partial y}(x,\eta)
\right|\right)^2(y+1),
\end{align*}
where $c$ is a positive constant independent of any parameter. In the
 second equality we used the equality $Q(x_0,y)=0$ to erase the term
 $a_3(e_{max})(x_0^2/2)^3$. In the last inequality we took
 \eqref{eq_key_coefficient_sign} into account. The above inequality
 implies that there exists $y_0(e_{max})\in (-1,-\frac{1}{2}]$ depending
 only on $e_{max}$ such that if $y\in (-1,y_0(e_{max})]$,
 $\frac{\partial Q}{\partial x}(x_0,y)<0$. We can sum up the above
 arguments to conclude that if $(x_0,y)\in
 [c_{min}/e_{max},c_{max}]\times (-1,y_0(e_{max})]$
 and $Q(x_0,y)=0$, then $\frac{\partial Q}{\partial x}(x_0,y)<0$. This
 ensures that the claim \eqref{eq_logical_target} holds true. 

If for $y\in (-1,y_0(e_{max})]$ $\hat{x}$ is a solution to
 $\frac{\partial P}{\partial x}(x,y)=0$ in $\R_{>0}$, then
 $\hat{x}\in[c_{min}/e_{max},c_{max}]$ and $Q(\hat{x},y)=0$ and
 thus it must be unique by \eqref{eq_logical_target}. We can deduce from
 \eqref{eq_temporal_function_derivative} that
\begin{align*}
&\frac{\partial P}{\partial x}(x,y)>0,\quad \left(\forall x\in
 \left(0,\frac{1}{e_{max}\sqrt{y+1}}\cosh^{-1}(|y|^{-1})\right)\right),\\
&\frac{\partial P}{\partial x}(x,y)<0,\quad \left(\forall x\in
 \left(\frac{1}{\sqrt{y+1}}\cosh^{-1}(|y|^{-1}),\infty\right)\right),
\end{align*}
which means that a solution to $\frac{\partial P}{\partial x}(x,y)=0$
 actually exists in $\R_{>0}$. Thus we have proved that for any $y\in
 (-1,y_0(e_{max})]$ a solution to $\frac{\partial P}{\partial x}(x,y)=0$
 uniquely exists in $\R_{>0}$. Therefore, by
 \eqref{eq_scaling_explicit_relation} for any $y\in (-1,y_0(e_{max})]$
 there uniquely exists $\tilde{x}\in \R_{>0}$ such that 
\begin{align}
&\frac{d}{dx}\left(\frac{1}{2\pi}\int_0^{2\pi}dk\frac{\sinh(xE_1(k))}{(y+\cosh(xE_1(k)))E_1(k)}\right)\Bigg|_{x=\tilde{x}}=0.\label{eq_full_integral_zero}
\end{align}
Now let us lift the condition $e_{min}=1$. Since
 $E_1(k)=e_{min}(\frac{t}{e_{min}}(\cos k+1)+1)$, the above result
 implies that there exists $y_0(t/e_{min})\in (-1,-1/2]$
 depending only on $t/e_{min}$ such that for any $y\in (-1,
 y_0(t/e_{min})]$ there uniquely exists $\tilde{x}\in \R_{>0}$
 such that \eqref{eq_full_integral_zero} with this $E_1$ holds. This
 further implies that for any $y\in
 [2\cos^{-1}(y_0(t/e_{min})),2\pi)$ there exists
 $\hat{x}(y)\in \R_{>0}$ such that 
\begin{align*}
&\frac{\partial g_{E_1}}{\partial x}(x,y,0)>0,\quad (\forall x\in (0,\hat{x}(y))),\\
&\frac{\partial g_{E_1}}{\partial x}(\hat{x}(y),y,0)=0,\\
&\frac{\partial g_{E_1}}{\partial x}(x,y,0)<0,\quad (\forall x\in
 (\hat{x}(y),\infty)).
\end{align*}
Then by repeating the same proof by contradiction as that after
 \eqref{eq_unique_mountain_peak} in the proof of Proposition
 \ref{prop_multi_orbital} we can conclude that the claim holds true.
\end{proof}

\begin{remark}
One natural question is whether the same result holds for the model in
 higher spatial dimensions 
\begin{align}
E(\bk)=t\left(\sum_{j=1}^d\cos k_j+d\right)+e_{min},\quad (t,e_{min}\in
 \R_{>0},\ d\in \N).\label{eq_model_general_dimension}
\end{align}
In the above proof we relied on the exact formula Lemma
 \ref{lem_definite_integral}. Since we do not have a useful formula of the
 definite integral for the model \eqref{eq_model_general_dimension} with
 $d\ge 2$, we cannot find an answer to this question by this approach at
 present.
\end{remark}

\section{Derivation of the infinite-volume limit}\label{sec_derivation}

In this section we will prove Theorem
\ref{thm_infinite_volume_limit}. As in the previous work \cite{K_BCS_I},
\cite{K_BCS_II}, the proof is based on multi-scale analysis of Grassmann
integral formulations of the free energy density and the thermal
expectations. In this approach qualitative bound properties of the
covariance matrices are the essential ingredients. This time we decide to
prepare them in the first subsection (Subsection
\ref{subsec_covariances}). 
The focus
of this part is to find optimal upper bounds on norms of the covariances
with respect to dependency on the inverse temperature $\beta$ and the
magnitude of the imaginary magnetic field $\theta$. Then in Subsections
\ref{subsec_general_estimation}-\ref{subsec_double_scale_integration} we
will develop a general double-scale integration scheme by assuming only
generic bounds of the covariances. In Subsection
\ref{subsec_infinite_volume} we combine the proved bound properties of
the real covariances with the general integration scheme to complete the
proof of Theorem \ref{thm_infinite_volume_limit}. The index set of the
finite-dimensional Grassmann algebra is exactly same as that in
\cite{K_BCS_II}. Accordingly, concerning the Grassmann integration, we
can use the same notations as in \cite{K_BCS_II}. We will sometimes
refer to the definitions presented in \cite{K_BCS_II} or \cite{K_BCS_I}
instead of restating them in order not to lengthen the paper. We will
also skip proofs of lemmas if they straightforwardly follow from lemmas
presented in \cite{K_BCS_I}, \cite{K_BCS_II}. To support the
readers, we illustrate the dependencies between the following subsections
and the previous constructions in Figure \ref{fig_dependency}.

\begin{figure}
\begin{center}
\begin{picture}(355,230)(0,0)

\put(20,210){\line(1,0){100}}
\put(20,230){\line(1,0){100}}
\put(20,210){\line(0,1){20}}
\put(120,210){\line(0,1){20}}
\put(33,216){Subsection \ref{subsec_covariances}}

\put(20,220){\line(-1,0){20}}
\put(0,220){\line(0,-1){210}}
\put(0,10){\vector(1,0){20}}

\put(150,210){\line(1,0){95}}
\put(150,230){\line(1,0){95}}
\put(150,210){\line(0,1){20}}
\put(245,210){\line(0,1){20}}
\put(150,220){\vector(-1,0){30}}
\put(160,218){\scriptsize \cite{K_BCS_II} Lemma 3.5 (iii)}

\put(150,180){\line(1,0){85}}
\put(150,200){\line(1,0){85}}
\put(150,180){\line(0,1){20}}
\put(235,180){\line(0,1){20}}
\put(150,190){\vector(-4,3){30}}
\put(160,188){\scriptsize \cite{PS} Theorem 1.3}
\put(192.5,200){\vector(0,1){10}} 

\put(20,160){\line(1,0){100}}
\put(20,140){\line(1,0){100}}
\put(20,140){\line(0,1){20}}
\put(120,140){\line(0,1){20}}
\put(70,140){\vector(0,-1){50}}
\put(33,145){Subsection \ref{subsec_general_estimation}}

\put(150,160){\line(1,0){165}}
\put(150,115){\line(1,0){165}}
\put(150,115){\line(0,1){45}}
\put(315,115){\line(0,1){45}}
\put(150,137){\vector(-1,-2){30}}
\put(150,137){\vector(-3,2){30}}
\put(155,135){\scriptsize \cite{K_BCS_I} $\left\{\begin{array}{l} 
\text{Lemma } 3.1,\\
\text{Lemma } 3.2,\\
\text{Lemma } 3.3,
\end{array}\right.$ 
\cite{K_BCS_II} $\left\{\begin{array}{l} 
\text{Lemma } 4.1,\\
\text{Lemma } 4.2,\\
\text{Lemma } 4.4
\end{array}\right.$ }

\put(20,90){\line(1,0){100}}
\put(20,70){\line(1,0){100}}
\put(20,70){\line(0,1){20}}
\put(120,70){\line(0,1){20}}
\put(70,70){\vector(0,-1){50}}
\put(33,76){Subsection \ref{subsec_double_scale_integration}}

\put(20,0){\line(1,0){100}}
\put(20,20){\line(1,0){100}}
\put(20,0){\line(0,1){20}}
\put(120,0){\line(0,1){20}}

\put(33,6){Subsection \ref{subsec_infinite_volume}}

\put(150,0){\line(1,0){205}}
\put(150,105){\line(1,0){205}}
\put(150,0){\line(0,1){105}}
\put(355,0){\line(0,1){105}}
\put(150,57){\vector(-2,-3){30}}
\put(155,51){\scriptsize \cite{K_BCS_I} $\left\{\begin{array}{l} 
\text{Lemma } 4.13,\\
\text{Proposition } 4.16,
\end{array}\right.$
\cite{K_BCS_II} $\left\{\begin{array}{l} 
\text{Lemma } 3.1,\\
\text{Lemma } 3.2,\\
\text{Lemma } 3.6,\\
\text{Proposition } 5.9,\\
\text{Proposition } 5.10,\\
\text{Lemma } 5.11,\\
\text{Lemma } A.1,\\
\text{Lemma } A.2,\\
\text{Lemma } A.3,\\
\text{Lemma } A.4
\end{array}
\right.$}

\end{picture}
 \caption{Dependencies between Subsections
 \ref{subsec_covariances}-\ref{subsec_infinite_volume}, results of
 \cite{K_BCS_I}, \cite{K_BCS_II} and \cite[\mbox{Theorem 1.3}]{PS}.}
\label{fig_dependency}
\end{center}
\end{figure}
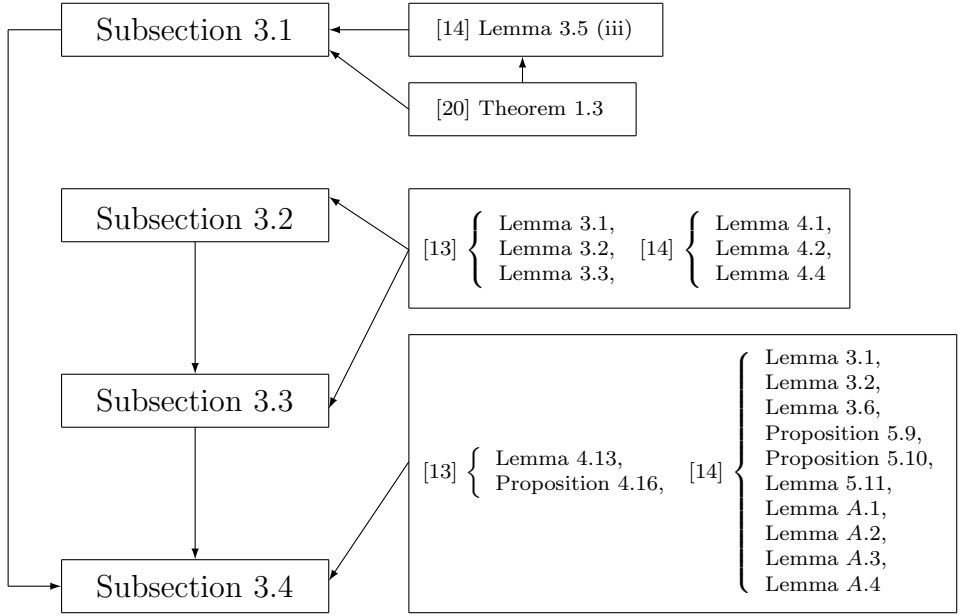

One important difference from the previous construction is that here the
parameter $\theta$ is allowed to take any real value thanks to the
gapped property of band spectra \eqref{eq_one_particle_lowest_band},
while it could not belong to $\frac{2\pi}{\beta}(2\Z+1)$ in
\cite{K_BCS_I}, \cite{K_BCS_II}. This affects the allowed value of
$\theta(\beta)$ as well. To make clear, we should state the definition of
$\theta(\beta)$ here. For any $\beta\in \R_{>0}$, $\theta\in\R$ there
uniquely exists $\theta'\in (-2\pi/\beta,2\pi/\beta]$
such that $\theta=\theta'$ (mod $4\pi/\beta$). We define the
number $\theta(\beta)\in [0,2\pi/\beta]$ by $\theta(\beta):=|\theta'|$.

\subsection{Properties of covariances}\label{subsec_covariances}

With the artificial parameter $h\in \frac{2}{\beta}\N$, we set 
$[0,\beta)_{h}:=\left\{0,1/h,2/h,\cdots,\beta-1/h\right\}$
as already stated in Subsection \ref{subsec_introduction}. Define the sets
$I_0$, $I$ by
\begin{align*}
I_0:=\{1,2\}\times \cB \times \G \times [0,\beta)_h,\quad I:=I_0\times
 \{1,-1\}.
\end{align*}
As we have seen in \cite[\mbox{Section 3}]{K_BCS_II}, our many-electron
system is formulated into the (imaginary) time-continuum limit $h\to \infty$
of the Grassmann Gaussian integral, which has the covariance $C(\phi):I^2_0\to
\C$ $(\phi\in \C)$ defined by 
\begin{align*}
&C(\phi)(\orho\rho \bx s, \oeta \eta \by t)\\
&:= 
\frac{1}{\beta L^d}\sum_{\bk\in \G^*}\sum_{\o\in
 \cM_h}e^{i\<\bk,\bx-\by\>+i\o(s-t)}\\
&\quad\cdot h^{-1}(I_{2b}-e^{-\frac{i}{h}(\o-\frac{\theta(\beta)}{2})I_{2b}+\frac{1}{h}E(\phi)(\bk)})^{-1}((\orho-1)b+\rho,(\oeta-1)b+\eta).
\end{align*}
Here $\cM_h$ is the set of the Matsubara frequencies with cut-off
$$
\left\{\o\in \frac{\pi}{\beta}(2\Z+1)\ \Big|\ |\o|<\pi h
\right\}
$$
and 
\begin{align*}
E(\phi)(\bk):=\left(\begin{array}{cc} E(\bk) & \overline{\phi}I_b \\
                                      \phi I_b & -E(\bk) \end{array}
\right)\in \Mat(2b,\C)
\end{align*}
for $\phi\in \C$. 
In fact $C(\phi)$ was originally defined as the free 2-point correlation
function in \cite[\mbox{Section 3}]{K_BCS_II} and was rewritten in the
above form in \cite[\mbox{Lemma 5.1}]{K_BCS_II}. 
As explained in Remark \ref{rem_free_band_spectra}, the symmetry
\eqref{eq_time_reversal_symmetry} was used in the derivation of
$C(\phi)$. Apart from the necessity to adopt the previous derivation, 
we do not use the symmetry \eqref{eq_time_reversal_symmetry} in this
paper. Our double-scale integration
regime is based on the following decomposition of the covariance. 
\begin{align}
&e^{-i\frac{\pi}{\beta}(s-t)}C(\phi)(\orho\rho\bx s, \oeta \eta \by t)=
C_0(\orho\rho \bx s, \oeta \eta \by t)+ C_1(\orho\rho \bx s, \oeta \eta \by t),\label{eq_covariance_double_decomposition}\\
&((\orho,\rho,\bx,s),(\oeta,\eta,\by,t)\in I_0,\ \phi\in \C),\notag
\end{align}
where the covariances $C_0$, $C_1:I_0^2\to \C$ are defined by 
\begin{align*}
&C_0(\orho\rho \bx s, \oeta \eta \by t)\\
&:= 
\frac{1}{\beta L^d}\sum_{\bk\in \G^*}e^{i\<\bk,\bx-\by\>}\\
&\quad\cdot
 h^{-1}(I_{2b}-e^{-\frac{i}{h}(\frac{\pi}{\beta}-\frac{\theta(\beta)}{2})I_{2b}+\frac{1}{h}E(\phi)(\bk)})^{-1}((\orho-1)b+\rho,(\oeta-1)b+\eta),\\
&C_1(\orho\rho \bx s, \oeta \eta \by t)\\
&:= 
\frac{1}{\beta L^d}\sum_{\bk\in \G^*}\sum_{\o\in
 \cM_h\backslash \{\frac{\pi}{\beta}\}}e^{i\<\bk,\bx-\by\>+i(\o-\frac{\pi}{\beta})(s-t)}\\
&\quad\cdot
 h^{-1}(I_{2b}-e^{-\frac{i}{h}(\o-\frac{\theta(\beta)}{2})I_{2b}+\frac{1}{h}E(\phi)(\bk)})^{-1}((\orho-1)b+\rho,(\oeta-1)b+\eta).
\end{align*}
Our aim here is to establish necessary bound properties of $C(\phi)$,
$C_0$, $C_1$. The bounds must be so sharp that the resulting multi-scale
analysis does not require any $(\beta,\theta)$-dependent condition on
the coupling constant $U$. First let us present bound properties which
can be proved by standard arguments. In the following we use the norms
$\|\cdot\|_{1,\infty}$, $\|\cdot\|_{1,\infty}'$ defined in 
\cite[\mbox{Subsection 4.1}]{K_BCS_II}. Let $\<\cdot,\cdot\>_{\C^m}$
denote the canonical inner product of $\C^m$. 
More precisely, for $\bu=(u_1,\cdots,u_m)$, $\bv=(v_1,\cdots,v_m)\in
\C^m$ $\<\bu,\bv\>_{\C^m}:=\sum_{j=1}^m\overline{u_j}v_j$.
Moreover, for any
$f:I_0^2\to \C$ let $\tilde{f}:I^2\to \C$ denote the anti-symmetric
extension of $f$ defined by 
\begin{align}
&\tilde{f}((X,\xi),(Y,\zeta)):=\frac{1}{2}(1_{(\xi,\zeta)=(1,-1)}f(X,Y)-1_{(\xi,\zeta)=(-1,1)}f(Y,X)),\label{eq_anti_symmetric_extension}\\
&(\forall X,Y\in I_0,\ \xi,\zeta\in \{1,-1\}).\notag
\end{align}
From here for any objects $\alpha_1,\cdots,\alpha_m$ we let
$c(\alpha_1,\cdots,\alpha_m)$ denote a positive constant depending only
on $\alpha_1,\cdots,\alpha_m$.

\begin{lemma}\label{lem_standard_covariance_bounds}
Assume that 
\begin{align}
h\ge \max\{\sqrt{e_{max}^2+|\phi|^2},1\}.\label{eq_h_domination}
\end{align}
Then there exists $c(d,b,(\hat{\bv}_j)_{j=1}^d,c_E)\in \R_{>0}$
 depending only on $d,b,(\hat{\bv}_j)_{j=1}^d,c_E$ such that the
 following statements hold. 
\begin{enumerate}[(i)]
\item\label{item_full_determinant}
\begin{align}
&|\det(\<\bu_i,\bw_j\>_{\C^m}C(\phi)(X_i,Y_j))_{1\le i,j\le n}|\le
 (c(d,b,(\hat{\bv}_j)_{j=1}^d,c_E)(1+\beta^{-1}e_{min}^{-1}))^n,\label{eq_full_determinant_bound}\\
&(\forall m,n\in \N,\ \bu_i,\bw_i\in \C^m\text{ with
 }\|\bu_i\|_{\C^m},\|\bw_i\|_{\C^m}\le 1,\ X_i,Y_i\in I_0\
 (i=1,\cdots,n)).\notag
\end{align}
\item\label{item_zero_determinant}
\begin{align*}
&|\det(\<\bu_i,\bw_j\>_{\C^m}C_0(X_i,Y_j))_{1\le i,j\le n}|\le
 (c(d,b,(\hat{\bv}_j)_{j=1}^d,c_E)\beta^{-1}e_{min}^{-1})^n,\\
&(\forall m,n\in \N,\ \bu_i,\bw_i\in \C^m\text{ with
 }\|\bu_i\|_{\C^m},\|\bw_i\|_{\C^m}\le 1,\ X_i,Y_i\in I_0\
 (i=1,\cdots,n)).\notag
\end{align*}
\item\label{item_zero_decay_bound}
\begin{align*}
&\|\tilde{C}_0\|_{1,\infty}\le  c(d,b,(\hat{\bv}_j)_{j=1}^d,c_E)
 \max\{e_{min}^{-1}, e_{min}^{-d-1}\},\\
&\|\tilde{C}_0\|_{1,\infty}'\le  c(d,b,(\hat{\bv}_j)_{j=1}^d,c_E)
 \beta^{-1}\max\{e_{min}^{-1}, e_{min}^{-d-1}\}.
\end{align*}
\item\label{item_one_decay_bound}
\begin{align*}
&\|\tilde{C}_1\|_{1,\infty}\le  c(d,b,(\hat{\bv}_j)_{j=1}^d,c_E)
 \max\{e_{min}^{-1}, e_{min}^{-d-1}\},\\
&\|\tilde{C}_1\|_{1,\infty}'\le  c(d,b,(\hat{\bv}_j)_{j=1}^d,c_E)
(e_{min}+\beta^{-1}+\beta^{-1}e_{min}^{-1}+1)
 \max\{e_{min}^{-1}, e_{min}^{-d-1}\}.
\end{align*}
\end{enumerate}
\end{lemma}

\begin{remark}
The bound \eqref{eq_full_determinant_bound} is not directly used in our
 multi-scale integration process, so its dependency on $\beta$ does not
 affect the magnitude of the coupling constant. The upper bounds on
 $\|\tilde{C}_0\|_{1,\infty}'$, $\|\tilde{C}_1\|_{1,\infty}'$ depend on
 $\beta$. However, they are to be multiplied by $L^{-d}$ during the
 multi-scale integration and thus do not yield a $\beta$-dependent
 condition on the coupling constant. Our essential problem is to prevent 
the $\beta$-dependent determinant bound of $C_0$ 
from affecting the
 magnitude of the coupling constant. Solving this problem is the main
 novelty of the present double-scale integration scheme.
\end{remark}
 
\begin{proof}[Proof of Lemma \ref{lem_standard_covariance_bounds}]
We fix $\phi\in \C$ during the proof. Resulting bounds will be
 independent of $\phi$, mainly due to the assumption
 \eqref{eq_h_domination}. First of all let us list useful estimates. 
For $(\o,\bk)\in \R^{d+1}$, set  
$$
B(\o,\bk):=h(I_{2b}-e^{-\frac{i}{h}(\o-\frac{\theta(\beta)}{2})I_{2b}+\frac{1}{h}E(\phi)(\bk)}).$$
We should recall the definition \eqref{eq_band_spectra_total_derivative}
 of $c_E$ beforehand.
\begin{align}
&\inf_{\bk\in \R^d}\inf_{\bu\in \C^{2b}\atop\text{with
 }\|\bu\|_{\C^{2b}}=1}\|E(\phi)(\bk)\bu\|_{\C^{2b}}=\sqrt{e_{min}^2+|\phi|^2},\label{eq_extended_big_band}\\
&\sup_{\bk\in \R^d}\|E(\phi)(\bk)\|_{2b\times
 2b}=\sqrt{e_{max}^2+|\phi|^2},\label{eq_extended_small_band}\\
&\|B(\o,\bk)^{-1}\|_{2b\times 2b}\le c\left(
h^2\sin^2\left(\frac{1}{2h}\left(\o-\frac{\theta(\beta)}{2}
\right)\right)+e_{min}^2\right)^{-\frac{1}{2}},\label{eq_matrix_inverse}\\
&\left\|
\left(\frac{\partial}{\partial \o}\right)^mB(\o,\bk)
\right\|_{2b\times 2b}\le ch^{-m+1},\label{eq_matrix_derivative}\\
&\left\|
\left(\frac{\partial}{\partial \hat{k}_j}\right)^mB\left(\o,\sum_{i=1}^d\hat{k}_i\hat{\bv}_i\right)
\right\|_{2b\times 2b}\le
 c(d,(\hat{\bv}_j)_{j=1}^d,c_E),\label{eq_matrix_derivative_space}\\
&(\forall m\in \{1,\cdots,d+2\},\ j\in \{1,\cdots,d\},\ \o\in \R,\
 \bk,\hat{\bk}\in \R^d).\notag
\end{align}
In the derivation of \eqref{eq_matrix_inverse},
 \eqref{eq_matrix_derivative}, \eqref{eq_matrix_derivative_space} we use
 \eqref{eq_h_domination}, \eqref{eq_extended_big_band},
 \eqref{eq_extended_small_band}. Also, to derive
 \eqref{eq_matrix_derivative_space}, one can repeatedly use the formula
\begin{align*}
&\frac{\partial}{\partial
 k_j}e^{\frac{1}{h}E(\phi)(\bk)}=\frac{1}{h}\int_0^1ds
 e^{\frac{s}{h}E(\phi)(\bk)}\frac{\partial}{\partial k_j}E(\phi)(\bk)
 e^{\frac{1-s}{h}E(\phi)(\bk)},\quad (j\in \{1,\cdots,d\}).
\end{align*}

\eqref{item_full_determinant}: It was proved in \cite[\mbox{Lemma 3.5
 (iii)}]{K_BCS_II}, which is based on the general determinant bound
 \cite[\mbox{Theorem 1.3}]{PS}, that 
\begin{align*}
&|\det(\<\bu_i,\bw_j\>_{\C^m}C(\phi)(X_i,Y_j))_{1\le i,j\le n}|\\
&\le\left(\frac{2^4b}{L^d}\sum_{\bk\in \G^*}\Tr\left(
1+2\cos\left(\frac{\beta\theta(\beta)}{2}\right)e^{-\beta\sqrt{E(\bk)^2+|\phi|^2}}
+e^{-2\beta\sqrt{E(\bk)^2+|\phi|^2}}
\right)^{-\frac{1}{2}}\right)^n,\\
&(\forall m,n\in \N,\ \bu_i,\bw_i\in \C^m\text{ with
 }\|\bu_i\|_{\C^m},\|\bw_i\|_{\C^m}\le 1,\ X_i,Y_i\in I_0\
 (i=1,\cdots,n)).
\end{align*}
Observe that
\begin{align*}
&\Tr\left(
1+2\cos\left(\frac{\beta\theta(\beta)}{2}\right)e^{-\beta\sqrt{E(\bk)^2+|\phi|^2}}
+e^{-2\beta\sqrt{E(\bk)^2+|\phi|^2}}
\right)^{-\frac{1}{2}}\\
&\le b(1-e^{-\beta e_{min}})^{-1}\le cb(1+\beta^{-1}e_{min}^{-1}).
\end{align*}
Thus the claimed bound holds.

\eqref{item_zero_determinant}: Let $L^2(\{1,2\}\times \cB\times \G^*
 \times \cM_h)$ be the Hilbert space whose inner product is defined by 
\begin{align*}
&\<f,g\>_{L^2}:=\frac{1}{\beta L^d}\sum_{K\in \{1,2\}\times \cB\times
 \G^*\times \cM_h}\overline{f(K)}g(K).
\end{align*}
We derive the claimed bound by applying the Gram inequality in the
 Hilbert space $\C^m\otimes L^2(\{1,2\}\times \cB\times \G^*\times
 \cM_h)$. Let us define the vectors $f_X$, $g_X\in L^2(\{1,2\}\times
 \cB\times \G^*\times \cM_h)$ $(X\in I_0)$ by
\begin{align*}
&f_{\orho\rho\bx
 s}(\otau,\tau,\bk,\o):=e^{-i\<\bk,\bx\>}1_{\o=\frac{\pi}{\beta}}1_{(\orho,\rho)=(\otau,\tau)}e_{min}^{-\frac{1}{2}},\\
&g_{\orho\rho\bx
 s}(\otau,\tau,\bk,\o):=e^{-i\<\bk,\bx\>}1_{\o=\frac{\pi}{\beta}}e_{min}^{\frac{1}{2}}B\left(\frac{\pi}{\beta},\bk\right)^{-1}((\otau-1)b+\tau,(\orho-1)b+\rho).
\end{align*}
It follows that $C_0(X,Y)=\<f_X,g_Y\>_{L^2}$ for any $X$, $Y\in
 I_0$. We can apply \eqref{eq_matrix_inverse} to verify that
\begin{align*}
\|f_X\|^2_{L^2}\le \beta^{-1}e_{min}^{-1},\quad 
\|g_X\|^2_{L^2}\le c(b)\beta^{-1}e_{min}^{-1},\quad (\forall X\in I_0).
\end{align*}
Therefore by the Gram inequality
\begin{align*}
&|\det(\<\bu_i,\bw_j\>_{\C^m}C_0(X_i,Y_j))_{1\le i,j\le n}|\le
 \prod_{i=1}^n\|\bu_i\|_{\C^m}\|\bw_i\|_{\C^m}\|f_{X_i}\|_{L^2}
\|g_{Y_i}\|_{L^2}\\
&\qquad\qquad\qquad\qquad\qquad\qquad\qquad\ \ \le (c(b)\beta^{-1}e_{min}^{-1})^n,\\
&(\forall m,n\in \N,\ \bu_i,\bw_i\in \C^m\text{ with
 }\|\bu_i\|_{\C^m},\|\bw_i\|_{\C^m}\le 1,\ X_i,Y_i\in I_0\
 (i=1,\cdots,n)).
\end{align*}

\eqref{item_zero_decay_bound}:
By applying e.g. the formula \cite[\mbox{(C.1)}]{K_RG} we can derive the
 following inequality.
\begin{align*}
&\left\|\left(\frac{\partial}{\partial \hat{k}_j}\right)^n
B\left(\o,\sum_{i=1}^d\hat{k}_i\hat{\bv}_i\right)^{-1}\right\|_{2b\times
 2b}\\
&\le
 c(d)\sum_{m=1}^n\prod_{u=1}^m\left(\sum_{l_u=1}^n\right)1_{\sum_{u=1}^ml_u=n}\\
&\quad\cdot \prod_{p=1}^m\left\|
B\left(\o,\sum_{i=1}^d\hat{k}_i\hat{\bv}_i
\right)^{-1}\left(\frac{\partial}{\partial \hat{k}_j}\right)^{l_p}B\left(\o,\sum_{i=1}^d\hat{k}_i\hat{\bv}_i\right)
\right\|_{2b\times 2b}\\
&\quad\cdot \left\|
B\left(\o,\sum_{i=1}^d\hat{k}_i\hat{\bv}_i
\right)^{-1}\right\|_{2b\times 2b},\\
&(\forall n\in \{1,\cdots,d+2\},\ j\in\{0,\cdots,d\},\ \o\in \R,\
 \hat{\bk}\in\R^d),
\end{align*}
where $\frac{\partial}{\partial \hat{k}_0}$ denotes
 $\frac{\partial}{\partial \o}$. 
Combination of this inequality and \eqref{eq_matrix_inverse},
 \eqref{eq_matrix_derivative}, \eqref{eq_matrix_derivative_space} yields
 that
\begin{align}
&\left\|\left(\frac{\partial}{\partial \hat{k}_j}\right)^n
B\left(\o,\sum_{i=1}^d\hat{k}_i\hat{\bv}_i\right)^{-1}\right\|_{2b\times
 2b}\label{eq_convenient_core_inequality}\\
&\le c(d,(\hat{\bv}_j)_{j=1}^d,c_E)\notag\\
&\quad\cdot \sum_{m=1}^n\left(
h^2\sin^2\left(\frac{1}{2h}\left(\o-\frac{\theta(\beta)}{2}\right)\right)+e_{min}^2\right)^{-\frac{m+1}{2}}(1_{j=0}h^{-n+m}+1_{j\ge
 1}),\notag\\
&(\forall n\in \{1,\cdots,d+2\},\ j\in\{0,\cdots,d\},\ \o\in \R,\
 \hat{\bk}\in\R^d).\notag
\end{align}
By periodicity we can perform integration by parts to derive that for
 any $\bx$, $\by\in \G$, $s,t\in [0,\beta)_h$, $j\in \{1,\cdots,d\}$
\begin{align*}
&\left(\frac{L}{2\pi}(e^{-i\frac{2\pi}{L}\<\bx-\by,\hat{\bv}_j\>}-1)\right)^{d+1}C_0(\cdot\bx
 s,\cdot \by t)\\
&=\frac{1}{\beta L^d}\sum_{\bk\in
 \G^*}e^{i\<\bk,\bx-\by\>}\\
&\quad \cdot\prod_{m=1}^{d+1}\left(\frac{L}{2\pi}\int_0^{\frac{2\pi}{L}}dp_m\right)
\left(\frac{\partial}{\partial
 \hat{k}_j}\right)^{d+1}B\left(\frac{\pi}{\beta},\bk+\hat{k}_j\hat{\bv}_j\right)^{-1}\Big|_{\hat{k}_j=\sum_{m=1}^{d+1}p_m}.
\end{align*}
Substitution of \eqref{eq_matrix_inverse},
 \eqref{eq_convenient_core_inequality} gives that
\begin{align*}
&\left|\left(\frac{L}{2\pi}(e^{-i\frac{2\pi}{L}\<\bx-\by,\hat{\bv}_j\>}-1)\right)^{d+1}\right|\|C_0(\cdot\bx
 s,\cdot \by t)\|_{2b\times 2b}\\
&\le
 c(d,(\hat{\bv}_j)_{j=1}^d,c_E)\beta^{-1}\sum_{m=1}^{d+1}e_{min}^{-m-1}\le
 c(d,(\hat{\bv}_j)_{j=1}^d,c_E)\beta^{-1}\max\{e_{min}^{-2},e_{min}^{-d-2}\},\\
&\|C_0(\cdot\bx
 s,\cdot \by t)\|_{2b\times 2b}\le c\beta^{-1}e_{min}^{-1},\\
&(\forall \bx,\by\in \G,\ s,t\in [0,\beta)_h).
\end{align*}
These bounds lead to that
\begin{align*}
\|\tilde{C}_0\|_{1,\infty}&\le\sum_{\bx\in\G}\frac{c(d,b,(\hat{\bv}_j)_{j=1}^d,c_E)e_{min}^{-1}}{1+(\max\{e_{min}^{-1},e_{min}^{-d-1}\})^{-1}\sum_{j=1}^d|\frac{L}{2\pi}(e^{i\frac{2\pi}{L}\<\bx,\hat{\bv}_j\>}-1)|^{d+1}}\\
&\le c(d,b,(\hat{\bv}_j)_{j=1}^d,c_E)e_{min}^{-1}
\Bigg(\sum_{\bx\in \G}\frac{1_{e_{min}\ge
 1}}{1+\sum_{j=1}^d|\frac{L}{2\pi}(e^{i\frac{2\pi}{L}\<\bx,\hat{\bv}_j\>}-1)|^{d+1}}\\
&\qquad\qquad\qquad\qquad\qquad\quad +
\sum_{\bx\in \G}\frac{1_{e_{min}<
 1}}{1+e_{min}^{d+1}\sum_{j=1}^d|\frac{L}{2\pi}(e^{i\frac{2\pi}{L}\<\bx,\hat{\bv}_j\>}-1)|^{d+1}}\Bigg)\\
&\le
 c(d,b,(\hat{\bv}_j)_{j=1}^d,c_E)\max\{e_{min}^{-1},e_{min}^{-d-1}\}.
\end{align*}
The claimed bound on $\|\tilde{C}_0\|_{1,\infty}'$ is proved in the same
 way. 

\eqref{item_one_decay_bound}: Let us apply a standard method of slicing
 the covariance. Let us take a function $\chi\in C^{\infty}(\R,\R)$
 satisfying that 
\begin{align*}
&\chi(x)=1,\quad (\forall x\in (-\infty,1]),\\
&\chi(x)\in (0,1),\quad (\forall x\in (1,2)),\\
&\chi(x)=0,\quad (\forall x\in [2,\infty)),\\
&\frac{d}{dx}\chi(x)\le 0,\quad (\forall x\in \R).
\end{align*}
Set 
\begin{align*}
N_h:=\left\lfloor \frac{\log h}{\log 2}\right\rfloor +1,\quad N_0:=
\left\lfloor \frac{\log (\max\{e_{min},\beta^{-1}\})}{\log
 2}\right\rfloor,
\end{align*}
where $\lfloor x \rfloor$ denotes the largest integer less than or equal
 to $x$ for $x\in \R$. By \eqref{eq_h_domination} and the definition of
 $h$, $h\ge \max\{e_{min},\beta^{-1}\}$ and thus $N_0<N_h$. Then we
 define the functions $\chi_l\in C^{\infty}(\R)$
 $(l=N_0,N_0+1,\cdots,N_h)$ by 
\begin{align*}
&\chi_{N_0}(\o):=\chi\left(
2^{-N_0}h\left|
\sin\left(\frac{\o - \theta(\beta)/2}{2h}\right)
\right|\right),\\
&\chi_{l}(\o):=\chi\left(
2^{-l}h\left|
\sin\left(\frac{\o - \theta(\beta)/2}{2h}\right)
\right|\right)
-\chi\left(
2^{-(l-1)}h\left|
\sin\left(\frac{\o - \theta(\beta)/2}{2h}\right)
\right|\right),\\
&(l=N_0+1,\cdots,N_h).
\end{align*}
These functions behave as follows.
\begin{align}
&\chi_{N_0}(\o)=\left\{\begin{array}{ll} 1 & \text{if }h\left|
\sin\left(\frac{\o - \theta(\beta)/2}{2h}\right)
\right|\le 2^{N_0},\\
\in (0,1) & \text{if } 2^{N_0}< h\left|
\sin\left(\frac{\o - \theta(\beta)/2}{2h}\right)
\right| < 2^{N_0+1},\\
0 & \text{if }h\left|
\sin\left(\frac{\o - \theta(\beta)/2}{2h}\right)
\right|\ge 2^{N_0+1},\end{array}
\right.\label{eq_cut_off_profile}\\
&\chi_{l}(\o)=\left\{\begin{array}{ll} 0 & \text{if }h\left|
\sin\left(\frac{\o - \theta(\beta)/2}{2h}\right)
\right|\le 2^{l-1},\\
\in (0,1] & \text{if } 2^{l-1}< h\left|
\sin\left(\frac{\o - \theta(\beta)/2}{2h}\right)
\right| < 2^{l+1},\\
0 & \text{if }h\left|
\sin\left(\frac{\o - \theta(\beta)/2}{2h}\right)
\right|\ge 2^{l+1},\end{array}
\right.\notag\\
&(l=N_0+1,\cdots,N_h).\notag
\end{align}
Moreover, there exists $c(d,\chi)\in \R_{>0}$ depending only on $d$,
 $\chi$ such that the following statements hold.
\begin{itemize}
\item 
\begin{align}
\sum_{l=N_0}^{N_h}\chi_l(\o)=1,\quad (\forall \o\in
 \R).\label{eq_sum_cut_off}
\end{align}
\item 
\begin{align}
&\left|
\left(\frac{\partial}{\partial \o}\right)^n\chi_l(\o)\right|\le
 c(d,\chi)2^{-nl},\label{eq_derivative_cut_off}\\
&(\forall n\in \{1,\cdots,d+2\},\ l\in
 \{N_0,\cdots,N_h\},\ \o\in \R).\notag
\end{align}
\item 
\begin{align}
\frac{1}{\beta}\sup_{x\in \R}\sum_{\o\in \cM_h}1_{\chi_l(\o+x)\neq 0}\le
 c(d,\chi)2^l,\quad(\forall l\in
 \{N_0,\cdots,N_h\}).\label{eq_support_size_cut_off}
\end{align}
\end{itemize}
To prove \eqref{eq_sum_cut_off}, \eqref{eq_derivative_cut_off}, we use
 that 
\begin{align}
2^{N_h-1}\le h\le 2^{N_h}.\label{eq_h_smallest}
\end{align}
 To prove
 \eqref{eq_support_size_cut_off}, we use that $\beta^{-1}\le c
 2^{N_0}$. Then let us define the covariances $C_l':I_0^2\to \C$
 $(l=N_0,N_0+1,\cdots,N_h)$ by 
\begin{align*}
C_l'(\cdot \bx s, \cdot \by t):=\frac{1}{\beta L^d}\sum_{\bk\in
 \G^*}\sum_{\o\in
 \cM_h}e^{i\<\bk,\bx-\by\>+i\o(s-t)}\chi_l(\o)B(\o,\bk)^{-1}.
\end{align*}
It follows from \eqref{eq_sum_cut_off} that 
\begin{align}
&\sum_{l=N_0}^{N_h}C_l'(\cdot \bx s, \cdot \by t)=C(\phi)(\cdot\bx
 s,\cdot \by t),\quad (\forall \bx,\by \in\G,\ s,t\in [0,\beta)_h).\label{eq_covariance_multi_decomposition}
\end{align}

Our strategy is as follows. We first find upper bounds on
 $\|\tilde{C}(\phi)\|_{1,\infty}$,  $\|\tilde{C}(\phi)\|_{1,\infty}'$ by
 estimating each $C_l'$ and summing up them. Then we derive the claimed
 bounds on $\|\tilde{C}_1\|_{1,\infty}$,  $\|\tilde{C}_1\|_{1,\infty}'$
 by using the relation \eqref{eq_covariance_double_decomposition} and
 the results of \eqref{item_zero_decay_bound}. 
By \eqref{eq_matrix_inverse}, \eqref{eq_cut_off_profile}, \eqref{eq_support_size_cut_off}
\begin{align}
&\|C_l'(\cdot\bx s,\cdot \by t)\|_{2b\times 2b} \le
 c(d,\chi)2^l(1_{l=N_0}e_{min}^{-1}+1_{l\ge
 N_0+1}(2^l+e_{min})^{-1}),\label{eq_sliced_full_bound}\\
&(\forall l\in \{N_0,\cdots,N_h\},\ \bx,\by\in \G,\ s,t\in
 [0,\beta)_h).\notag
\end{align}
Integrating by parts based on periodicity yields that 
\begin{align}
&\left(\frac{\beta}{2\pi}(e^{-i\frac{2\pi}{\beta}(s-t)}-1)\right)^{n}C_l'(\cdot\bx
 s,\cdot \by t)\label{eq_integration_by_parts_time}\\
&=\frac{1}{\beta L^d}\sum_{\bk\in
 \G^*}\sum_{\o\in \cM_h}
e^{i\<\bk,\bx-\by\>+i\o(s-t)}\notag\\
&\quad \cdot\prod_{m=1}^{n}\left(\frac{\beta}{2\pi}\int_0^{\frac{2\pi}{\beta}}dr_m\right)
\left(\frac{\partial}{\partial
 r}\right)^{n}\chi_l(r)B(r,\bk)^{-1}\Big|_{r=\o+\sum_{m=1}^nr_m},\notag\\
&\left(\frac{L}{2\pi}(e^{-i\frac{2\pi}{L}\<\bx-\by,\hat{\bv}_j\>}-1)\right)^{n}C_l'(\cdot\bx s,\cdot \by t)\label{eq_integration_by_parts_full}\\
&=\frac{1}{\beta L^d}\sum_{\bk\in
 \G^*}\sum_{\o\in \cM_h}
e^{i\<\bk,\bx-\by\>+i\o(s-t)}\notag\\
&\quad \cdot\prod_{m=1}^{n}\left(\frac{L}{2\pi}\int_0^{\frac{2\pi}{L}}dp_m\right)
\chi_l(\o)\left(\frac{\partial}{\partial
 \hat{k}_j}\right)^{n}B(\o,\bk+\hat{k}_j\hat{\bv}_j)^{-1}\Big|_{\hat{k}_j=\sum_{m=1}^{n}p_m},\notag\\
&(\forall l\in \{N_0,\cdots,N_h\},\ \bx,\by\in \G,\ s,t\in
 [0,\beta)_h,\ j\in\{1,\cdots,d\},\ n\in \{1,\cdots,d+2\}).\notag
\end{align}

Assume that $l\ge N_0+1$. By \eqref{eq_matrix_inverse},
 \eqref{eq_convenient_core_inequality}, \eqref{eq_cut_off_profile},
 \eqref{eq_derivative_cut_off}, \eqref{eq_support_size_cut_off} and
 \eqref{eq_integration_by_parts_time}
\begin{align}
&\left|\frac{\beta}{2\pi}(e^{-i\frac{2\pi}{\beta}(s-t)}-1)\right|^{d+2}\|C_l'(\cdot\bx
 s,\cdot \by t)\|_{2b\times 2b}\label{eq_sliced_time_up}\\
&\le
 \prod_{m=1}^{d+2}\left(\frac{\beta}{2\pi}\int_{0}^{\frac{2\pi}{\beta}}dr_m\right)\notag\\
&\quad\cdot \frac{1}{\beta
 L^d}\sum_{\bk\in \G^*}\sum_{\o\in
 \cM_h}1_{\chi_l(\o+\sum_{m=1}^nr_m)\neq 0} \sup_{r\in [-\pi h,\pi
 h]}\left\|\left(\frac{\partial}{\partial
 r}\right)^{d+2}\chi_l(r)B(r,\bk)^{-1}\right\|_{2b\times 2b}\notag\\
&\le c(d,(\hat{\bv}_j)_{j=1}^d,c_E,\chi)2^l\notag\\
&\quad \cdot\Bigg(
\sum_{p=0}^{d+1}2^{-pl}\sum_{m=1}^{d+2-p}h^{-(d+2-p)+m}(2^{2l}+e_{min}^2)^{-\frac{m+1}{2}}
+2^{-(d+2)l}(2^{2l}+e_{min}^2)^{-\frac{1}{2}}\Bigg)\notag\\
&\le c(d,(\hat{\bv}_j)_{j=1}^d,c_E,\chi) 2^{-(d+2)l}.\notag
\end{align}
In the last inequality we also used \eqref{eq_h_smallest}.
On the other hand, by \eqref{eq_convenient_core_inequality},
 \eqref{eq_cut_off_profile}, \eqref{eq_support_size_cut_off} and
 \eqref{eq_integration_by_parts_full} for $j\in \{1,\cdots,d\}$, $n\in
 \{1,\cdots,d+2\}$
\begin{align}
&\left|\frac{L}{2\pi}(e^{-i\frac{2\pi}{L}\<\bx-\by,\hat{\bv}_j\>}-1)\right|^n
\|C_l'(\cdot \bx s,\cdot \by t)\|_{2b\times
 2b}\label{eq_sliced_space_up}\\
&\le c(d,(\hat{\bv}_j)_{j=1}^d,c_E,\chi) 2^l\sum_{m=1}^n(2^{2l}+e_{min}^2)^{-\frac{m+1}{2}}\notag\\
&\le
 c(d,(\hat{\bv}_j)_{j=1}^d,c_E,\chi)\max\{e_{min}^{-1},e_{min}^{-n}\}.\notag
\end{align}
By combining \eqref{eq_sliced_full_bound},
 \eqref{eq_sliced_time_up} and \eqref{eq_sliced_space_up} for $n=d+2$
\begin{align*}
&\|C_l'(\cdot\bx s,\cdot \by t)\|_{2b\times 2b}\\
&\le c(d,(\hat{\bv}_j)_{j=1}^d,c_E,\chi)\Bigg/\Bigg(1+2^{(d+2)l}\left|\frac{\beta}{2\pi}(e^{i\frac{2\pi}{\beta}(s-t)}-1)\right|^{d+2}\\
&\qquad\qquad\qquad\qquad\qquad\quad +(\max\{e_{min}^{-1},e_{min}^{-d-2}\})^{-1}\sum_{j=1}^d\left|\frac{L}{2\pi}(e^{i\frac{2\pi}{L}\<\bx-\by,\hat{\bv}_j\>}-1)\right|^{d+2}\Bigg),\\
&(\forall \bx,\by\in\G,\ s,t\in [0,\beta)_h),
\end{align*}
which together with \eqref{eq_h_smallest} implies that
\begin{align}
\|\tilde{C}_l'\|_{1,\infty}&\le
 c(d,b,(\hat{\bv}_j)_{j=1}^d,c_E,\chi)2^{-l}(1_{e_{min}\ge
 1}+1_{e_{min}<1}e_{min}^{-d})\label{eq_sliced_norm_up}\\
&\le
 c(d,b,(\hat{\bv}_j)_{j=1}^d,c_E,\chi)2^{-l}\max\{1,e_{min}^{-d}\}.\notag
\end{align}
Also by \eqref{eq_sliced_space_up} for $n=d+1$ 
\begin{align}
&\sum_{\bx\in \G\atop \bx\neq \b0}(\|C_l'(\cdot \bx s,\cdot \b0
 t)\|_{2b\times 2b}+
\|C_l'(\cdot \b0 t,\cdot \bx
 s)\|_{2b\times 2b})\label{eq_modified_norm_up}\\
&\le
 c(d,(\hat{\bv}_j)_{j=1}^d,c_E,\chi)\max\{2^{-l},2^{-(d+1)l}\}\sum_{\bx\in\G\atop
 \bx\neq
 \b0}\frac{1}{\sum_{j=1}^d|\frac{L}{2\pi}(e^{i\frac{2\pi}{L}\<\bx,\hat{\bv}_j\>}-1)|^{d+1}}\notag\\
&\le
 c(d,(\hat{\bv}_j)_{j=1}^d,c_E,\chi)\max\{2^{-l},2^{-(d+1)l}\},\quad
 (\forall s,t\in [0,\beta)_h).\notag
\end{align}

Let us derive necessary bounds for $l=N_0$. By
 \eqref{eq_convenient_core_inequality}, \eqref{eq_support_size_cut_off},
 \eqref{eq_integration_by_parts_full}
\begin{align}
&\left|\frac{L}{2\pi}(e^{-i\frac{2\pi}{L}\<\bx-\by,\hat{\bv}_j\>}-1)\right|^n\|C_{N_0}'(\cdot\bx
 s,\cdot\by t)\|_{2b\times 2b}\label{eq_sliced_space_down}\\
&\le  c(d,(\hat{\bv}_j)_{j=1}^d,c_E,\chi)
 2^{N_0}\max\{e_{min}^{-2},e_{min}^{-n-1}\},\notag\\
&(\forall \bx,\by\in\G,\ s,t\in [0,\beta)_h,\ j\in \{1,\cdots,d\},\ n\in
 \{1,\cdots,d+2\}).\notag
\end{align}
Assume that $e_{min}\le \beta^{-1}$. It follows from
 \eqref{eq_sliced_full_bound}, \eqref{eq_sliced_space_down} for $n=d+1$
 that
\begin{align*}
&\|C_{N_0}'(\cdot\bx s,\cdot \by t)\|_{2b\times 2b}\le
 \frac{c(d,(\hat{\bv}_j)_{j=1}^d,c_E,\chi)2^{N_0}e_{min}^{-1}}{1+(\max\{e_{min}^{-1},e_{min}^{-d-1}\})^{-1}\sum_{j=1}^d|\frac{L}{2\pi}(e^{i\frac{2\pi}{L}\<\bx-\by,\hat{\bv}_j\>}-1)|^{d+1}},\\
&(\forall \bx,\by\in \G,\ s,t\in [0,\beta)_h),
\end{align*}
and thus 
\begin{align*}
\|\tilde{C}_{N_0}'\|_{1,\infty}&\le
 c(d,b,(\hat{\bv}_j)_{j=1}^d,c_E,\chi)\beta
 2^{N_0}e_{min}^{-1}(1_{e_{min}\ge 1}+1_{e_{min}<1}e_{min}^{-d})\\
&\le
 c(d,b,(\hat{\bv}_j)_{j=1}^d,c_E,\chi)\max\{e_{min}^{-1},e_{min}^{-d-1}\},
\end{align*}
where we used that $2^{N_0}\le \beta^{-1}$. On the other hand, let us
 assume that $e_{min}>\beta^{-1}$. By \eqref{eq_matrix_inverse}, 
 \eqref{eq_convenient_core_inequality}, \eqref{eq_derivative_cut_off},
 \eqref{eq_support_size_cut_off} and
 \eqref{eq_integration_by_parts_time}
\begin{align}
&\left|\frac{\beta}{2\pi}(e^{-i\frac{2\pi}{\beta}(s-t)}-1)\right|^{d+2}
\|C_{N_0}'(\cdot\bx s,\cdot \by t)\|_{2b\times 2b}\label{eq_sliced_time_down}
\\
&\le c(d,(\hat{\bv}_j)_{j=1}^d,c_E,\chi)2^{N_0}\left(
\sum_{p=0}^{d+1}2^{-pN_0}\sum_{m=1}^{d+2-p}h^{-(d+2-p)+m}e_{min}^{-m-1}+2^{-(d+2)N_0}e_{min}^{-1}\right)\notag\\
&\le c(d,(\hat{\bv}_j)_{j=1}^d,c_E,\chi)2^{N_0}\left(
\sum_{p=0}^{d+1}\sum_{m=1}^{d+2-p}2^{-N_0(d+2-m)}e_{min}^{-m-1}+2^{-(d+2)N_0}e_{min}^{-1}\right)\notag\\
&\le
 c(d,(\hat{\bv}_j)_{j=1}^d,c_E,\chi)2^{-(d+1)N_0}e_{min}^{-1},\notag\\
&(\forall \bx,\by\in \G,\ s,t\in [0,\beta)_h).\notag
\end{align}
In the second inequality we used \eqref{eq_h_smallest}. In the last
 inequality we used that $2^{N_0}\le e_{min}$. By using
 \eqref{eq_sliced_full_bound}, \eqref{eq_sliced_space_down} for $n=d+2$
 and \eqref{eq_sliced_time_down} we have that
\begin{align*}
&\|C_{N_0}'(\cdot\bx s,\cdot \by t)\|_{2b\times 2b}\\
&\le c(d,(\hat{\bv}_j)_{j=1}^d,c_E,\chi)2^{N_0}e_{min}^{-1}\Bigg/\Bigg(1+2^{(d+2)N_0}\left|\frac{\beta}{2\pi}(e^{i\frac{2\pi}{\beta}(s-t)}-1)\right|^{d+2}\\
&\qquad\qquad\qquad\qquad\qquad\quad +(\max\{e_{min}^{-1},e_{min}^{-d-2}\})^{-1}\sum_{j=1}^d\left|\frac{L}{2\pi}(e^{i\frac{2\pi}{L}\<\bx-\by,\hat{\bv}_j\>}-1)\right|^{d+2}\Bigg),
\end{align*}
and thus by using \eqref{eq_h_smallest} 
\begin{align*}
\|\tilde{C}_{N_0}'\|_{1,\infty}&\le
 c(d,b,(\hat{\bv}_j)_{j=1}^d,c_E,\chi)e_{min}^{-1}(1_{e_{min}\ge
 1}+1_{e_{min}<1}e_{min}^{-d})\\
&\le
 c(d,b,(\hat{\bv}_j)_{j=1}^d,c_E,\chi)\max\{e_{min}^{-1},e_{min}^{-d-1}\}.
\end{align*}
In both cases we have derived that
\begin{align}
\|\tilde{C}_{N_0}'\|_{1,\infty}\le
 c(d,b,(\hat{\bv}_j)_{j=1}^d,c_E,\chi)\max\{e_{min}^{-1},e_{min}^{-d-1}\}.\label{eq_sliced_norm_down}
\end{align}
Moreover, it follows from \eqref{eq_sliced_space_down} for $n=d+1$ that 
\begin{align}
&\sum_{\bx\in \G\atop \bx\neq \b0}(\|C_{N_0}'(\cdot\bx s,\cdot \b0
 t)\|_{2b\times 2b}+\|C_{N_0}'(\cdot\b0 t,\cdot \bx
 s)\|_{2b\times 2b})\label{eq_modified_norm_down}\\
&\le
 c(d,(\hat{\bv}_j)_{j=1}^d,c_E,\chi)2^{N_0}\max\{e_{min}^{-2},e_{min}^{-d-2}\},\quad
 (\forall s,t\in [0,\beta)_h).\notag
\end{align}

Let us sum up the above estimates. By
 \eqref{eq_covariance_multi_decomposition}, \eqref{eq_sliced_norm_up} and \eqref{eq_sliced_norm_down}
\begin{align}
\|\tilde{C}(\phi)\|_{1,\infty}&\le
 \sum_{l=N_0}^{N_h}\|\tilde{C}_l'\|_{1,\infty}\label{eq_full_covariance_norm}\\
&\le
 c(d,b,(\hat{\bv}_j)_{j=1}^d,c_E,\chi)(\max\{e_{min}^{-1},e_{min}^{-d-1}\}+2^{-N_0}\max\{1,e_{min}^{-d}\})\notag\\
&\le
 c(d,b,(\hat{\bv}_j)_{j=1}^d,c_E,\chi)\max\{e_{min}^{-1},e_{min}^{-d-1}\}.
\notag
\end{align}
Also, we can apply \eqref{eq_full_determinant_bound},
 \eqref{eq_covariance_multi_decomposition}, \eqref{eq_modified_norm_up}
 and \eqref{eq_modified_norm_down} to deduce that
\begin{align}
&\|\tilde{C}(\phi)\|_{1,\infty}'\label{eq_full_covariance_modified_norm}\\
&\le c(b)\sup_{s,t\in
 [0,\beta)_h}\|C(\phi)(\cdot \b0 s,\cdot \b0 t)\|_{2b\times 2b}\notag\\
&\quad 
+c(b)\sup_{s,t\in
 [0,\beta)_h}\sum_{l=N_0}^{N_h}\sum_{\bx\in \G\atop \bx\neq \b0}(
\|C_l'(\cdot \bx s,\cdot \b0 t)\|_{2b\times 2b}+
\|C_l'(\cdot \b0 t,\cdot \bx s)\|_{2b\times
 2b})\notag\\
&\le
 c(d,b,(\hat{\bv}_j)_{j=1}^d,c_E,\chi)\notag\\
&\quad\cdot \left(1+\beta^{-1}e_{min}^{-1}+2^{N_0}\max\{e_{min}^{-2},e_{min}^{-d-2}\}+\sum_{l=N_0+1}^{N_h}\max\{2^{-l},2^{-(d+1)l}\}
\right)\notag\\
&\le  c(d,b,(\hat{\bv}_j)_{j=1}^d,c_E,\chi)\notag\\
&\quad\cdot (1+\beta^{-1}e_{min}^{-1}+(e_{min}+\beta^{-1})\max\{e_{min}^{-2},
 e_{min}^{-d-2}\} + e_{min}^{-1}+ e_{min}^{-d-1})\notag\\
&\le  c(d,b,(\hat{\bv}_j)_{j=1}^d,c_E,\chi)
 (e_{min}+\beta^{-1}+\beta^{-1}e_{min}^{-1}+1)\max\{e_{min}^{-1},
 e_{min}^{-d-1}\}.\notag
\end{align}
Observe that by \eqref{eq_covariance_double_decomposition} 
\begin{align*}
\|\tilde{C}_1\|_{1,\infty}\le
 \|\tilde{C}(\phi)\|_{1,\infty}+\|\tilde{C}_0\|_{1,\infty},\quad 
\|\tilde{C}_1\|_{1,\infty}'\le
 \|\tilde{C}(\phi)\|_{1,\infty}'+\|\tilde{C}_0\|_{1,\infty}'.
\end{align*}
Then, substitution of \eqref{eq_full_covariance_norm},
 \eqref{eq_full_covariance_modified_norm} and the results of 
 \eqref{item_zero_decay_bound} yields the claimed inequalities. 
\end{proof}

\begin{remark}\label{rem_decay_for_zero_limit}
Assume that $\beta\ge e_{min}^{-1}$. Then it follows from 
\eqref{eq_covariance_multi_decomposition}, 
\eqref{eq_sliced_space_up},
\eqref{eq_sliced_space_down} that 
\begin{align*}
&\sum_{j=1}^d\left|
\frac{L}{2\pi}(e^{i\frac{2\pi}{L}\<\bx-\by,\hbv_j\>}-1)
\right|\|C(\phi)(\cdot\bx 0,\cdot \by 0)\|_{2b\times 2b}\le
 c(d,(\hbv_j)_{j=1}^d,c_E,\chi)e_{min}^{-1},\\
&(\forall \bx,\by\in\G,\ \phi\in \C).
\end{align*}
The above inequality holds for any $\phi\in\C$ due to 
the fact that $C(\phi)(\cdot \bx 0,\cdot \by 0)$ is
 independent of $h$ (see \cite[\mbox{(3.2)}]{K_BCS_II}). 
As explained in Remark \ref{rem_zero_temperature_limit}, the above
 spatial decay property can be used to study the zero-temperature limit
 of the 4-point correlation function.
\end{remark}

Lemma \ref{lem_standard_covariance_bounds} does not include a
determinant bound of $C_1$, which crucially affects the possible
magnitude of the coupling constant in our double-scale integration
scheme. A determinant bound of $C_1$ can be useful only if it is optimal
with respect to the dependency on $(\beta,\theta)$. Let us derive a
desirable bound in the next lemma. Again we will essentially apply not
only the general bound \cite[\mbox{Theorem 1.3}]{PS} but the
representation techniques presented in \cite[\mbox{Subsection 4.1}]{PS}
by de Siqueira Pedra and Salmhofer as in our previous derivation of
determinant bound \cite[\mbox{Proposition 4.2}]{K_BCS_I}. 
We should remark more
specifically that the decompositions \eqref{eq_P_S_type_decomposition}, 
\eqref{eq_P_S_type_decomposition_simpler} below are influenced by the
techniques of \cite[\mbox{Subsection 4.1}]{PS}. However, the choice of
the Hilbert space, which will be denoted by $\cH$, and
the construction of necessary vectors belonging to the Hilbert space
 are much more complicated than the corresponding parts of the previous
 papers. 
 The essential idea here 
is to replace the sum over $\cM_h\backslash\{\pi/\beta\}$ by a
contour integral plus an extra term by means of the residue theorem. 

\begin{lemma}\label{lem_covariance_determinant_bound}
Assume that 
\begin{align}\label{eq_h_new_domination}
h\ge \sqrt{e_{max}^2+|\phi|^2}+\frac{1}{\beta}(3\pi+2).
\end{align}
Then there exists $c(b)\in\R_{>0}$ depending only on $b$ such that
\begin{align*}
&|\det(\<\bu_i,\bw_j\>_{\C^m}C_1(X_i,Y_j))_{1\le i,j\le n}|\le c(b)^n,\\
&(\forall m,n\in \N,\ \bu_i,\bw_i\in \C^m\text{ with
 }\|\bu_i\|_{\C^m},\|\bw_i\|_{\C^m}\le 1,\ X_i,Y_i\in I_0\
 (i=1,\cdots,n)).
\end{align*}
\end{lemma}

\begin{proof}
Let us fix $\phi\in \C$ throughout the proof. We will need to assume
 that $h$ is large depending on $\phi$ on several occasions. We will
 eventually see
 that the assumption \eqref{eq_h_new_domination} is sufficient. 
Let $\s(E(\bk))$, $\s(E(\phi)(\bk))$ denote the set of
 eigenvalues of $E(\bk)$, $E(\phi)(\bk)$ respectively. For any $\bk\in
 \G^*$ there exist $e_{\rho}(\bk)\in\R$ $(\rho=1,\cdots,b)$ such that
$e_1(\bk)\le e_2(\bk)\le \cdots \le e_b(\bk)$ and
 $\s(E(\bk))=\{e_{\rho}(\bk)\}_{\rho\in \cB}$. Set 
$$
\hat{e}_{\rho}(\bk):=\sqrt{e_{\rho}(\bk)^2+|\phi|^2}
$$
for $\rho\in\cB$. Observe that $\s(E(\phi)(\bk))=\{\pm
 \hat{e}_{\rho}(\bk)\}_{\rho\in \cB}$. For any $\bk\in \G^*$ there
 exists $x_{\bk}\in [1/\beta,2/\beta]$ such that 
\begin{align}
\left[x_{\bk}-\frac{1}{2(b+1)\beta},x_{\bk}+\frac{1}{2(b+1)\beta}\right)\cap
 \s(E(\phi)(\bk))=\emptyset.\label{eq_interval_intersection_empty}
\end{align}
This claim can be proved as follows. Suppose that 
\begin{align*}
\left[\frac{1}{\beta}+\frac{2m+1}{2(b+1)\beta},\frac{1}{\beta}+\frac{2m+3}{2(b+1)\beta}
\right)\cap \s(E(\phi)(\bk))\neq \emptyset
\end{align*}
for any $m\in \{0,1,\cdots,b\}$. Since these $b+1$ intervals are
 disjoint, it implies that $\sharp \s(E(\phi)(\bk))\cap \R_{\ge 0}\ge
 b+1$. However, $\sharp \s(E(\phi)(\bk))\cap \R_{\ge 0}\le b$, which is a
 contradiction. Thus the claim holds with some $x_{\bk}\in
 \{\frac{1}{\beta}+\frac{m+1}{(b+1)\beta}\}_{m=0}^b$. Fix such
 $\{x_{\bk}\}_{\bk\in\G^*}$. For $\bk\in \G^*$ let us set 
\begin{align*}
&\cB(\bk):=\left\{\rho\in \cB\ \Big|\ \hat{e}_{\rho}(\bk)\ge
 x_{\bk}+\frac{1}{2(b+1)\beta}\right\},\\
&P_1:=\{z\in \C\ |\ |z|=\pi h\},\\
&P_2(\bk):=\left\{x+i\frac{2\pi}{\beta}\ \Big|\ -x_{\bk}\le x\le
 x_{\bk}\right\}\cup \left\{x_{\bk}+iy\ \Big|\ -\frac{\pi}{2\beta}\le
 y\le \frac{2\pi}{\beta}\right\}\\
&\qquad\qquad \cup \left\{x-i\frac{\pi}{2\beta}\ \Big|\ -x_{\bk}\le x\le
 x_{\bk}\right\}\cup \left\{-x_{\bk}+iy\ \Big|\ -\frac{\pi}{2\beta}\le
 y\le \frac{2\pi}{\beta}\right\}.
\end{align*}
By the assumption \eqref{eq_h_new_domination},
 $\sqrt{x_{\bk}^2+(2\pi/\beta)^2}<\pi h$. This implies that $P_1\cap
 P_2(\bk)=\emptyset$. We consider $P_1$ as a contour oriented
 counter-clockwise and $P_2(\bk)$ as a contour oriented clockwise. 
Let us admit a convention that for $A$, $B\in \Mat(b,\C)$ $A\oplus B$
 denotes the $2b\times 2b$ matrix 
$$\left(\begin{array}{cc} A & 0 \\ 0 & B \end{array}
\right).$$
For any $\bk\in \G^*$ there exists a $2b\times 2b$ unitary matrix
 $U(\bk)$ such that 
\begin{align}
U(\bk)^*E(\phi)(\bk)U(\bk)=(\delta_{\rho,\eta}\hat{e}_{\rho}(\bk))_{1\le
 \rho,\eta\le b}\oplus (-\delta_{\rho,\eta}\hat{e}_{\rho}(\bk))_{1\le
 \rho,\eta\le b}.\label{eq_energy_matrix_diagonalization}
\end{align}
It follows that 
\begin{align}
&C_1(\cdot\bx s,\cdot \by t)\label{eq_one_covariance_diagonalization}\\
&=\frac{1}{L^d}\sum_{\bk\in
 \G^*}e^{i\<\bk,\bx-\by\>-i\frac{\pi}{\beta}(s-t)}\notag\\
&\quad \cdot U(\bk)\left(
\frac{\delta_{\rho,\eta}}{\beta}\sum_{\o\in \cM_h\backslash
 \{\frac{\pi}{\beta}\}}e^{i\o(s-t)}h^{-1}(1-e^{-\frac{i}{h}(\o-\frac{\theta(\beta)}{2})+\frac{1}{h}\hat{e}_{\rho}(\bk)})^{-1}\right)_{1\le \rho,\eta\le b}\notag\\
&\qquad \oplus \left(
\frac{\delta_{\rho,\eta}}{\beta}\sum_{\o\in \cM_h\backslash
 \{\frac{\pi}{\beta}\}}e^{i\o(s-t)}h^{-1}(1-e^{-\frac{i}{h}(\o-\frac{\theta(\beta)}{2})-\frac{1}{h}\hat{e}_{\rho}(\bk)})^{-1}\right)_{1\le \rho,\eta\le b}U(\bk)^*,\notag\\
&(\forall \bx,\by\in \G,\ s,t\in [0,\beta)_h).\notag
\end{align}
The assumption \eqref{eq_h_new_domination} implies that
$|i\theta(\beta)/2+\delta \hat{e}_{\rho}(\bk)|<\pi h$ for any
 $\bk\in \G^*$, $\rho\in \cB$, $\delta\in \{1,-1\}$. Based on this fact
 and the property \eqref{eq_interval_intersection_empty}, the residue
 theorem ensures that for any $r\in \R$, $\bk\in \G^*$, $\rho\in \cB$,
 $\delta\in \{1,-1\}$ 
\begin{align}
&\frac{1}{2\pi i}\oint_{P_1\cup P_2(\bk)}dz \frac{e^{zr}}{1+e^{\beta
 z}}h^{-1}(1-e^{-\frac{1}{h}(z-i\frac{\theta(\beta)}{2})+\frac{\delta}{h}\hat{e}_{\rho}(\bk)})^{-1}\label{eq_residue_theorem}\\
&=-\frac{1}{\beta}\sum_{\o\in \cM_h\backslash \{\frac{\pi}{\beta}\}}
e^{i\o
 r}h^{-1}(1-e^{-\frac{i}{h}(\o-\frac{\theta(\beta)}{2})+\frac{\delta}{h}\hat{e}_{\rho}(\bk)})^{-1}+1_{\rho\in
 \cB(\bk)}\frac{e^{(i\frac{\theta(\beta)}{2}+\delta \hat{e}_{\rho}(\bk))r}}{1+e^{\beta(i\frac{\theta(\beta)}{2}+\delta \hat{e}_{\rho}(\bk))}}.\notag
\end{align}
Let us define the functions $C_{1-1}^{\ge}$, $C_{1-2}^{\ge}$,
 $C_{1-1}^{<}$, $C_{1-2}^{<}$, $C_{1-1}$,
 $C_{1-2}:(\{1,2\}\times\cB\times \G\times [0,\beta))^2\to \C$ as
 follows. 
\begin{align*}
&C_{1-1}^{\ge}(\cdot \bx s,\cdot \by t)\\
&:=\frac{1}{L^d}\sum_{\bk\in
 \G^*}e^{i\<\bk,\bx-\by\>}\frac{-1}{2\pi i} \oint_{P_1\cup P_2(\bk)}dz \frac{e^{z(s-t)}}{1+e^{\beta
 z}}h^{-1}(I_{2b}-e^{-\frac{1}{h}(z-i\frac{\theta(\beta)}{2})I_{2b}+\frac{1}{h}E(\phi)(\bk)})^{-1},\\
&C_{1-2}^{\ge}(\cdot\bx s,\cdot \by t)\\
&:=\frac{1}{L^d}\sum_{\bk\in
 \G^*}e^{i\<\bk,\bx-\by\>}
 U(\bk)\left(\frac{\delta_{\rho,\eta}1_{\rho\in\cB(\bk)}e^{(i\frac{\theta(\beta)}{2}+\hat{e}_{\rho}(\bk))(s-t)}}{1+e^{\beta(i\frac{\theta(\beta)}{2}+\hat{e}_{\rho}(\bk))}}\right)_{1\le \rho,\eta\le b}\\
&\qquad\qquad\qquad\qquad \oplus 
\left(\frac{\delta_{\rho,\eta}1_{\rho\in\cB(\bk)}e^{(i\frac{\theta(\beta)}{2}-\hat{e}_{\rho}(\bk))(s-t)}}{1+e^{\beta(i\frac{\theta(\beta)}{2}-\hat{e}_{\rho}(\bk))}}\right)_{1\le
 \rho,\eta\le b}U(\bk)^*,\\
&C_{1-1}^{<}(\cdot\bx s,\cdot \by t)\\
&:=\frac{1}{L^d}\sum_{\bk\in
 \G^*}e^{i\<\bk,\bx-\by\>}\frac{-1}{2\pi i} \oint_{P_1\cup P_2(\bk)}dz \frac{e^{z(s-t+\beta)}}{1+e^{\beta
 z}}h^{-1}(I_{2b}-e^{-\frac{1}{h}(z-i\frac{\theta(\beta)}{2})I_{2b}+\frac{1}{h}E(\phi)(\bk)})^{-1},\\
&C_{1-2}^{<}(\cdot\bx s,\cdot \by t)\\
&:=\frac{1}{L^d}\sum_{\bk\in
 \G^*}e^{i\<\bk,\bx-\by\>}
 U(\bk)\left(\frac{\delta_{\rho,\eta}1_{\rho\in\cB(\bk)}e^{(i\frac{\theta(\beta)}{2}+\hat{e}_{\rho}(\bk))(s-t+\beta)}}{1+e^{\beta(i\frac{\theta(\beta)}{2}+\hat{e}_{\rho}(\bk))}}\right)_{1\le \rho,\eta\le b}\\
&\qquad\qquad\qquad\qquad \oplus 
\left(\frac{\delta_{\rho,\eta}1_{\rho\in\cB(\bk)}e^{(i\frac{\theta(\beta)}{2}-\hat{e}_{\rho}(\bk))(s-t+\beta)}}{1+e^{\beta(i\frac{\theta(\beta)}{2}-\hat{e}_{\rho}(\bk))}}\right)_{1\le
 \rho,\eta\le b}U(\bk)^*,\\
&C_{1-1}(\cdot \bx s,\cdot \by t):=1_{s\ge t}C_{1-1}^{\ge}(\cdot\bx
 s,\cdot \by t)-1_{s< t}C_{1-1}^{<}(\cdot\bx
 s,\cdot \by t),\\
&C_{1-2}(\cdot \bx s,\cdot \by t):=1_{s\ge t}C_{1-2}^{\ge}(\cdot\bx
 s,\cdot \by t)-1_{s< t}C_{1-2}^{<}(\cdot\bx
 s,\cdot \by t),\\
&(\forall \bx,\by\in \G,\ s,t\in [0,\beta)).\notag
\end{align*}
By combining these with \eqref{eq_energy_matrix_diagonalization},
\eqref{eq_one_covariance_diagonalization},
 \eqref{eq_residue_theorem} we have that 
\begin{align}
e^{i\frac{\pi}{\beta}(s-t)}C_1(\cdot \bx s,\cdot \by t)=C_{1-1}(\cdot
 \bx s, \cdot \by t)+C_{1-2}(\cdot
 \bx s, \cdot \by t),\quad (\forall \bx,\by\in \G,\ s,t\in [0,\beta)_h).
\label{eq_one_covariance_bare_decomposition}
\end{align}
Let us find suitable determinant bounds of $C_{1-1}$, $C_{1-2}$ so that
 the claimed determinant
 bound of $C_1$ can be derived from them. 

Let us consider $C_{1-1}$ first. Let $\cH$ denote the Hilbert space 
$L^2(\{1,2\}\times \cB\times \G^*\times \R\times [0,1]\times
 \{1,2,3,4,5\})$ whose inner product is given by 
\begin{align*}
&\<f,g\>_{\cH}:=\sum_{(\otau,\tau)\in \{1,2\}\times
 \cB}\frac{1}{L^d}\sum_{\bk\in \G^*}\int_{-\infty}^{\infty}du \int^1_0dv
 \sum_{j=1}^5\overline{f(\otau,\tau,\bk,u,v,j)}g(\otau,\tau,\bk,u,v,j).
\end{align*}
Let us define the vectors $f_{X}^a$, $g_{X}^a\in \cH$ $(X\in
 \{1,2\}\times\cB\times \G\times \R,\ a\in \{1,-1\})$ in the following
 arguments.  
For $(\orho,\rho,\bx,s)\in \{1,2\}\times \cB\times \G\times \R$, $a\in
 \{1,-1\}$, $(\otau,\tau,\bk,u)\in \{1,2\}\times \cB\times \G^*\times
 \R$, $z\in P_1\cup P_2(\bk)$, set
\begin{align*}
&f_{\orho\rho\bx s}^a(\otau,\tau,\bk,u)(z)\\
&:=\frac{1}{\sqrt{2\pi}}1_{a\Re z >0}e^{-i\<\bk,\bx\>-is (a\Im z
 -u)}1_{(\otau,\tau)=(\orho,\rho)}\frac{1+e^{-\beta a z}}{|1+e^{-\beta a
 z}|^{\frac{3}{2}}}\\
&\quad\cdot\sqrt{\frac{|\Re z|}{\pi}}\frac{\|h^{-1}(I_{2b}-e^{-\frac{1}{h}(z-i\frac{\theta(\beta)}{2})I_{2b}+\frac{1}{h}E(\phi)(\bk)})^{-1}\|_{2b\times
 2b}^{\frac{1}{2}}}{iu +\Re z},\\
&g_{\orho\rho\bx s}^a(\otau,\tau,\bk,u)(z)\\
&:=\frac{1}{\sqrt{2\pi}i}1_{a\Re z >0}e^{-i\<\bk,\bx\>-is (a\Im z
 -u)}\frac{1}{|1+e^{-\beta a
 z}|^{\frac{1}{2}}}\\
&\quad\cdot \sqrt{\frac{|\Re z|}{\pi}}\frac{\|h^{-1}(I_{2b}-e^{-\frac{1}{h}(z-i\frac{\theta(\beta)}{2})I_{2b}+\frac{1}{h}E(\phi)(\bk)})^{-1}\|_{2b\times
 2b}^{-\frac{1}{2}}}{iu +\Re z}\\
&\quad\cdot
 h^{-1}(I_{2b}-e^{-\frac{1}{h}(z-i\frac{\theta(\beta)}{2})I_{2b}+\frac{1}{h}E(\phi)(\bk)})^{-1}((\otau-1)b+\tau,(\orho-1)b+\rho).
\end{align*}
Then, for $(\otau,\tau,\bk,u,v,j)\in \{1,2\}\times \cB\times \G^*\times
 \R\times [0,1]\times \{1,2,3,4,5\}$, set 
\begin{align*}
f_{\orho\rho \bx s}^a(\otau,\tau,\bk, u,v,j)
&:=1_{j=1}\sqrt{2h}\pi f_{\orho\rho\bx s}^a(\otau,\tau,\bk,u)(\pi h
 e^{i2\pi v})\\
&\quad\ +1_{j=2}\sqrt{2x_{\bk}} f_{\orho\rho\bx
 s}^a(\otau,\tau,\bk,u)\left(2x_{\bk}v-x_{\bk}+i\frac{2\pi}{\beta}\right)\\
&\quad\ +1_{j=3}\sqrt{\frac{5\pi}{2\beta}}
 f_{\orho\rho\bx
 s}^a(\otau,\tau,\bk,u)\left(x_{\bk}-i\frac{5\pi}{2\beta}v+i\frac{2\pi}{\beta}\right)\\
&\quad\ +1_{j=4}\sqrt{2x_{\bk}}
 f_{\orho\rho\bx
 s}^a(\otau,\tau,\bk,u)\left(-2x_{\bk}v+x_{\bk}-i\frac{\pi}{2\beta}\right)\\
&\quad\ +1_{j=5}\sqrt{\frac{5\pi}{2\beta}}
 f_{\orho\rho\bx
 s}^a(\otau,\tau,\bk,u)\left(-x_{\bk}+i\frac{5\pi}{2\beta}v-i\frac{\pi}{2\beta}\right),\\
g_{\orho\rho \bx s}^a(\otau,\tau,\bk, u,v,j)
&:=1_{j=1}i\sqrt{2h}\pi e^{i2\pi v}g_{\orho\rho\bx s}^a(\otau,\tau,\bk,u)(\pi h
 e^{i2\pi v})\\
&\quad\ +1_{j=2}\sqrt{2x_{\bk}} g_{\orho\rho\bx
 s}^a(\otau,\tau,\bk,u)\left(2x_{\bk}v-x_{\bk}+i\frac{2\pi}{\beta}\right)\\
&\quad\ +1_{j=3}\left(-i\sqrt{\frac{5\pi}{2\beta}}\right)
 g_{\orho\rho\bx
 s}^a(\otau,\tau,\bk,u)\left(x_{\bk}-i\frac{5\pi}{2\beta}v+i\frac{2\pi}{\beta}\right)\\
&\quad\ +1_{j=4}(-\sqrt{2x_{\bk}})
 g_{\orho\rho\bx
 s}^a(\otau,\tau,\bk,u)\left(-2x_{\bk}v+x_{\bk}-i\frac{\pi}{2\beta}\right)\\
&\quad\ +1_{j=5}i\sqrt{\frac{5\pi}{2\beta}}
 g_{\orho\rho\bx
 s}^a(\otau,\tau,\bk,u)\left(-x_{\bk}+i\frac{5\pi}{2\beta}v-i\frac{\pi}{2\beta}\right).
\end{align*}
Moreover, using the vectors $f_X^1$, $f_{X}^{-1}$, $g_X^1$, $g_X^{-1}\in
 \cH$ defined above, 
we define the vectors $f_X^{\ge}$, $f_X^{<}$, $g_X^{\ge}$,
 $g_X^{<}\in \cH$ $(X\in \{1,2\}\times\cB\times \G\times \R)$ as
 follows. For $(\orho,\rho,\bx,s)\in \{1,2\}\times\cB\times\G\times \R$
\begin{align}
&f_{\orho\rho\bx s}^{\ge}=f_{\orho \rho \bx s}^{<}:=f_{\orho\rho \bx
 s}^1+f_{\orho\rho \bx (-s)}^{-1},\label{eq_P_S_type_decomposition}\\
&g_{\orho\rho \bx s}^{\ge}:=-g_{\orho\rho\bx (\beta+s)}^1-g_{\orho\rho\bx
 (-s)}^{-1},\quad g_{\orho\rho \bx s}^{<}:=-g_{\orho\rho\bx s}^1-g_{\orho\rho\bx
 (\beta-s)}^{-1}.\notag
\end{align}
By using the formula 
\begin{align}
e^{-tD}=\frac{D}{\pi}\int_{-\infty}^{\infty}du\frac{e^{itu}}{u^2+D^2},\quad
 (\forall t\in \R_{\ge 0},\ D\in
 \R_{>0})\label{eq_exponential_integral_formula}
\end{align}
one can verify that for any $(\orho,\rho,\bx,s)$,
 $(\oeta,\eta,\by,t)\in\{1,2\}\times \cB\times \G\times [0,\beta)$
\begin{align*}
1_{s\ge t}\<f_{\orho\rho\bx s}^{\ge},g_{\oeta\eta\by t}^{\ge}\>_{\cH}
&=-1_{s\ge t}(\<f_{\orho\rho\bx s}^1,g_{\oeta\eta\by
 (\beta+t)}^1\>_{\cH}+\<f_{\orho\rho\bx (-s)}^{-1},g_{\oeta\eta\by
 (-t)}^{-1}\>_{\cH})\\
&=-1_{s\ge t}\sum_{(\otau,\tau)\in \{1,2\}\times
 \cB}\frac{1}{L^d}\sum_{\bk\in \G^*}\int_{-\infty}^{\infty}du
 \oint_{P_1\cup P_2(\bk)}dz\\
&\qquad\cdot (\overline{f_{\orho\rho \bx
 s}^1(\otau,\tau,\bk,u)(z)}g_{\oeta\eta \by (\beta + t)}^1(\otau,\tau,\bk,u)(z)\\
&\qquad\quad+ \overline{f_{\orho\rho \bx
 (-s)}^{-1}(\otau,\tau,\bk,u)(z)}g_{\oeta\eta \by
 (-t)}^{-1}(\otau,\tau,\bk,u)(z))\\
&=1_{s\ge t}C_{1-1}^{\ge}(\orho\rho\bx s,\oeta \eta \by t),\\
1_{s< t}\<f_{\orho\rho\bx s}^{<},g_{\oeta\eta\by t}^{<}\>_{\cH}
&=-1_{s< t}(\<f_{\orho\rho\bx s}^1,g_{\oeta\eta\by
 t}^1\>_{\cH}+\<f_{\orho\rho\bx (-s)}^{-1},g_{\oeta\eta\by
 (\beta-t)}^{-1}\>_{\cH})\\ 
&=1_{s< t}C_{1-1}^{<}(\orho\rho\bx s,\oeta \eta \by t),
\end{align*}
and thus
\begin{align}
&C_{1-1}(\orho\rho\bx s, \oeta \eta \by t)=1_{s\ge t}\<f_{\orho\rho \bx
 s}^{\ge}, g_{\oeta\eta \by t}^{\ge}\>_{\cH}-1_{s< t}\<f_{\orho\rho \bx
 s}^{<}, g_{\oeta\eta \by
 t}^{<}\>_{\cH}.\label{eq_covariance_residue_gram}
\end{align}

To apply \cite[\mbox{Theorem 1.3}]{PS}, we need to estimate
 $\|f_{X}^{\ge }\|_{\cH}$,  $\|g_{X}^{\ge }\|_{\cH}$,
$\|f_X^{<}\|_{\cH}$, $\|g_X^{<}\|_{\cH}$, $(X\in \{1,2\}\times\cB\times
 \G\times [0,\beta))$. These can be expanded as follows. 
For any $(\orho,\rho,\bx,s)\in \{1,2\}\times\cB\times
 \G\times [0,\beta)$ and $A\in \{f,g\}$
\begin{align}
\|A_{\orho\rho \bx s}^{\ge}\|_{\cH}^2 = \|A_{\orho\rho \bx
 s}^{<}\|_{\cH}^2
=&\sum_{(\otau,\tau)\in \{1,2\}\times \cB}\frac{1}{L^d}\sum_{\bk\in
 \G^*}\int_{-\infty}^{\infty}du\int_0^1dv\sum_{a\in \{1,-1\}}\label{eq_core_function_bare_expansion}\\
&\quad\cdot \Bigg(
2h\pi^2 |A_{\orho\rho\b0 0}^a(\otau,\tau,\bk,u)(\pi h
 e^{i2\pi v})|^2\notag\\
&\qquad+2x_{\bk} \left|A_{\orho\rho\b0
 0}^a(\otau,\tau,\bk,u)\left(2x_{\bk}v-x_{\bk}+i\frac{2\pi}{\beta}\right)\right|^2\notag\\
&\qquad+\frac{5\pi}{2\beta}\left|A_{\orho\rho\b0
 0}^a(\otau,\tau,\bk,u)\left(x_{\bk}-i\frac{5\pi}{2\beta}v+i\frac{2\pi}{\beta}\right)\right|^2\notag\\
&\qquad+2x_{\bk}\left|A_{\orho\rho\b0
 0}^a(\otau,\tau,\bk,u)\left(-2x_{\bk}v+x_{\bk}-i\frac{\pi}{2\beta}\right)\right|^2\notag\\
&\qquad+\frac{5\pi}{2\beta}\left|A_{\orho\rho\b0
 0}^a(\otau,\tau,\bk,u)\left(-x_{\bk}+i\frac{5\pi}{2\beta}v-i\frac{\pi}{2\beta}\right)\right|^2\Bigg).\notag
\end{align}

As the next step, let us fix $\bk\in \G^*$ and estimate
\begin{align*}
\inf_{z\in P_1\cup P_2(\bk)}|1+e^{\beta z}|,\quad 
\inf_{z\in P_1\cup P_2(\bk)}|1+e^{-\beta z}|.
\end{align*}
For $z\in P_1$ there exists $t\in [-1,1]$ such that
\begin{align*}
|1+e^{\beta z}|^2=1+2\cos(\pi\beta h\sqrt{1-t^2})e^{\pi\beta h
 t}+e^{2\pi\beta h t}.
\end{align*}
There exists $m\in \N$ such that $h=2m/\beta$. Then there exist
 $n\in \{0,1,\cdots,m-1\}$, $\theta\in [0,2\pi]$ such that 
$\pi\beta h\sqrt{1-t^2}=\theta+2n\pi$. If $\theta\in
 [0,\pi/2]\cup [3\pi/2,2\pi]$, $|1+e^{\beta z}|^2\ge 1+
 e^{2\pi \beta h t}\ge 1$. If $\theta\in
 (\pi/2,3\pi/2)$, $(\pi\beta h
 t)^2=(2m\pi-\theta-2n\pi)(2m\pi+\theta+2n\pi)\ge \pi^2/4$, and
 thus $|1+e^{\beta z}|^2\ge (1-e^{\pi\beta h t})^2\ge
 (1-e^{-\frac{\pi}{2}})^2$. We have proved that 
$$
\inf_{z\in P_1}|1+e^{\beta z}|=\inf_{z\in P_1}|1+e^{-\beta z}|\ge 1- e^{-\frac{\pi}{2}}.
$$
If $z\in \{x_{\bk}+iy,-x_{\bk}+iy\ |\ -\frac{\pi}{2\beta}\le y\le
 \frac{2\pi}{\beta}\}$, 
$\min\{|1+e^{\beta z}|,|1+e^{-\beta z}|\}\ge 1-e^{-\beta x_{\bk}}\ge
 1-e^{-1}$, where we used that $x_{\bk}\ge 1/\beta$. 
If $z\in \{x+i\frac{2\pi}{\beta}, x-i\frac{\pi}{2\beta}\ |\ -x_{\bk}\le
 x\le x_{\bk}\}$, $\min\{|1+e^{\beta z}|, |1+e^{-\beta z}|\}\ge 1$. Thus 
$$
\inf_{z\in P_2(\bk)}\min \{|1+e^{\beta z}|,\ |1+e^{-\beta z}|\}\ge 1-e^{-1}.
$$
Now we can see that
\begin{align}
\inf_{z\in P_1\cup P_2(\bk)}|1+e^{a\beta z}|\ge 1-e^{-1},\quad (\forall
 a\in \{1,-1\}).\label{eq_absolute_path}
\end{align}
We also need to find upper bounds on 
\begin{align*}
&\sup_{z\in
 P_1}\|h^{-1}(I_{2b}-e^{-\frac{1}{h}(z-i\frac{\theta(\beta)}{2})I_{2b}+\frac{1}{h}E(\phi)(\bk)})^{-1}\|_{2b\times
 2b},\\
&\sup_{z\in
 P_2(\bk)}\|h^{-1}(I_{2b}-e^{-\frac{1}{h}(z-i\frac{\theta(\beta)}{2})I_{2b}+\frac{1}{h}E(\phi)(\bk)})^{-1}\|_{2b\times
 2b}.
\end{align*}
On the assumption \eqref{eq_h_new_domination}
\begin{align*}
&\sup_{z\in P_1}\sup_{\alpha\in \s(E(\phi)(\bk))}\left|
-\frac{1}{h}\left(z-i\frac{\theta(\beta)}{2}\right)+\frac{1}{h}\alpha
\right|\le \pi
 +\frac{1}{h}\left(\frac{\pi}{\beta}+\sqrt{e_{max}^2+|\phi|^2}\right)\le \frac{3\pi}{2},\\
&\inf_{z\in P_1}\inf_{\alpha\in \s(E(\phi)(\bk))}\left|
-\frac{1}{h}\left(z-i\frac{\theta(\beta)}{2}\right)+\frac{1}{h}\alpha
\right|\ge \pi
 -\frac{1}{h}\left(\frac{\pi}{\beta}+\sqrt{e_{max}^2+|\phi|^2}\right)\ge
 \frac{\pi}{2},
\end{align*}
which imply that
\begin{align*}
&\inf_{z\in P_1}\inf_{\alpha \in
 \s(E(\phi)(\bk))}|1-e^{-\frac{1}{h}(z-i\frac{\theta(\beta)}{2})+\frac{1}{h}\alpha}|\ge
 \inf_{z\in \C\atop\text{with }\frac{\pi}{2}\le |z|\le
 \frac{3\pi}{2}}|1-e^{z}|>0,
\end{align*}
and thus
\begin{align}
\sup_{z\in
 P_1}\|h^{-1}(I_{2b}-e^{-\frac{1}{h}(z-i\frac{\theta(\beta)}{2})I_{2b}+\frac{1}{h}E(\phi)(\bk)})^{-1}\|_{2b\times
 2b}\le ch^{-1}.\label{eq_absolute_one_path}
\end{align}
On the other hand, since $x_{\bk}\in [1/\beta,2/\beta]$, 
\begin{align*}
&\sup_{z\in P_2(\bk)}\sup_{\alpha \in \s(E(\phi)(\bk))}\left|
\Re\left(-\frac{z}{h}+i\frac{\theta(\beta)}{2h}+\frac{\alpha}{h}
\right)
\right|\le
 \frac{1}{h}\left(\frac{2}{\beta}+\sqrt{e_{max}^2+|\phi|^2}\right),\\
&\sup_{z\in P_2(\bk)}\sup_{\alpha \in \s(E(\phi)(\bk))}\left|
\Im\left(-\frac{z}{h}+i\frac{\theta(\beta)}{2h}+\frac{\alpha}{h}
\right)
\right|\le \frac{3\pi}{\beta h}.
\end{align*}
By the assumption \eqref{eq_h_new_domination} 
\begin{align*}
\sup_{z\in P_2(\bk)}\sup_{\alpha \in \s(E(\phi)(\bk))}\left|
-\frac{z}{h}+i\frac{\theta(\beta)}{2h}+\frac{\alpha}{h}
\right|\le 1,
\end{align*}
and thus for any $z\in P_2(\bk)$, $\alpha \in \s(E(\phi)(\bk))$
\begin{align*}
|1-e^{-\frac{z}{h}+i\frac{\theta(\beta)}{2h}+\frac{\alpha}{h}}|\ge \left(1-\sum_{n=2}^{\infty}\frac{1}{n!}\right)\left|
-\frac{z}{h}+i\frac{\theta(\beta)}{2h}+\frac{\alpha}{h}
\right|.
\end{align*}
Therefore
\begin{align}
&\sup_{z\in P_2(\bk)}\|h^{-1}(I_{2b}-e^{-\frac{1}{h}(z-i\frac{\theta(\beta)}{2})I_{2b}+\frac{1}{h}E(\phi)(\bk)})^{-1}\|_{2b\times
 2b}\label{eq_absolute_two_path}\\
&\le c\sup_{z\in P_2(\bk)}\sup_{\alpha \in
 \s(E(\phi)(\bk))}\left(|\Re z-\alpha|+\left|\Im
 z-\frac{\theta(\beta)}{2}\right|\right)^{-1}\notag\\
&\le c\max\left\{\sup_{\alpha\in
 \s(E(\phi)(\bk))}|x_{\bk}-\alpha|^{-1},\
 \left|\frac{2\pi}{\beta}-\frac{\theta(\beta)}{2}\right|^{-1},\ \left|
\frac{\pi}{2\beta}+\frac{\theta(\beta)}{2}\right|^{-1}\right\}\le
 c(b)\beta,\notag
\end{align}
where we used that 
$$\inf_{\alpha\in \s(E(\phi)(\bk))}|x_{\bk}-\alpha|\ge
 \frac{1}{2(b+1)\beta},$$ 
which is ensured by
 \eqref{eq_interval_intersection_empty}.

By substituting \eqref{eq_absolute_path}, \eqref{eq_absolute_one_path},
 \eqref{eq_absolute_two_path} into
 \eqref{eq_core_function_bare_expansion} and using
 \eqref{eq_exponential_integral_formula} and $x_{\bk}\le
 2/\beta$ $(\forall \bk\in\G^*)$ we observe that for any $X\in \{1,2\}\times \cB\times
 \G\times [0,\beta)$ and $A\in \{f,g\}$ 
\begin{align*}
&\|A_X^{\ge}\|_{\cH}^2=\|A_X^{<}\|_{\cH}^2\\
&\le \frac{c(b)}{L^d}\sum_{\bk\in\G^*}\int_{-\infty}^{\infty}du\int_0^1dv\sum_{a\in \{1,-1\}}\\
&\quad \cdot\Bigg(
\frac{|\pi h\cos(2\pi v)|1_{a\pi h\cos(2\pi v)>0}}{u^2+(\pi h \cos(2\pi
 v))^2}
+\frac{|2x_{\bk}v-x_{\bk}|1_{a(2x_{\bk}v-x_{\bk})>0}}{u^2+(2x_{\bk}v-x_{\bk})^2}
+\frac{x_{\bk}1_{a x_{\bk}>0}}{u^2+x_{\bk}^2}\Bigg)\\
&\le c(b).
\end{align*}
Now we can apply the extended Gram inequality \cite[\mbox{Theorem
 1.3}]{PS} in the representation \eqref{eq_covariance_residue_gram} to
 derive that
\begin{align}
&|\det(\<\bu_i,\bw_j\>_{\C^m}C_{1-1}(X_i,Y_j))_{1\le i,j\le n}|\le
 c(b)^n,\label{eq_determinant_bound_residue_part}\\
&(\forall m,n\in \N,\ \bu_i,\bw_i\in \C^m\text{ with
 }\|\bu_i\|_{\C^m},\|\bw_i\|_{\C^m}\le 1,\ X_i,Y_i\in I_0\
 (i=1,\cdots,n)).\notag
\end{align}
The readers can refer to \cite[\mbox{Remark 5.2}]{K_2010} for a minor
 necessary modification of  \cite[\mbox{Theorem
 1.3}]{PS} concerning the factor $\<\bu_i,\bw_j\>_{\C^m}$
 $(i,j=1,\cdots,n)$, as it was originally claimed only for $m=n$ in 
 \cite[\mbox{Theorem 1.3}]{PS}.

Let us treat $C_{1-2}$. In fact the procedure to find a determinant
 bound on $C_{1-2}$ is simpler than that on $C_{1-1}$. Let
 $\widehat{\cH}$ denote the Hilbert space $L^2(\{1,2\}\times \cB\times
 \G^*\times \R)$ whose inner product is defined by 
\begin{align*}
&\<f,g\>_{\widehat{\cH}}:=\sum_{(\otau,\tau)\in \{1,2\}\times
 \cB}\frac{1}{L^d}\sum_{\bk\in \G^*}\int_{-\infty}^{\infty}du
 \overline{f(\otau,\tau,\bk,u)}g(\otau,\tau,\bk,u).
\end{align*}
Define the vectors $\hat{f}_X^{\bar{a}}$, $\hat{g}_X^{\bar{a}}\in
 \widehat{\cH}$ $(X\in \{1,2\}\times \cB\times \G\times \R,$ $\bar{a}\in
 \{1,2\})$ as follows. 
\begin{align*}
&\hat{f}_{\orho\rho\bx
 s}^{\bar{a}}(\otau,\tau,\bk,u)\\
&:=1_{\otau=\bar{a}}e^{-i\<\bk,\bx\>-is((-1)^{\bar{a}+1}\frac{\theta(\beta)}{2}-u)}1_{\tau\in
 \cB(\bk)}\overline{U(\bk)((\orho-1)b+\rho,(\otau-1)b+\tau)}\\
&\quad\cdot\frac{1+e^{-\beta
 ((-1)^{\bar{a}+1}i\frac{\theta(\beta)}{2}+\hat{e}_{\tau}(\bk))}}
{|1+e^{-\beta
 ((-1)^{\bar{a}+1}i\frac{\theta(\beta)}{2}+\hat{e}_{\tau}(\bk))}|^{\frac{3}{2}}}\sqrt{\frac{\hat{e}_{\tau}(\bk)}{\pi}}\frac{1}{iu
 + \hat{e}_{\tau}(\bk)},\\
&\hat{g}_{\orho\rho\bx
 s}^{\bar{a}}(\otau,\tau,\bk,u)\\
&:=1_{\otau=\bar{a}}e^{-i\<\bk,\bx\>-is((-1)^{\bar{a}+1}\frac{\theta(\beta)}{2}-u)}1_{\tau\in
 \cB(\bk)}U(\bk)^*((\otau-1)b+\tau,(\orho-1)b+\rho)\\
&\quad\cdot \frac{1}
{|1+e^{-\beta
 ((-1)^{\bar{a}+1}i\frac{\theta(\beta)}{2}+\hat{e}_{\tau}(\bk))}|^{\frac{1}{2}}}\sqrt{\frac{\hat{e}_{\tau}(\bk)}{\pi}}\frac{1}{iu
 + \hat{e}_{\tau}(\bk)}.
\end{align*}
Then let us define $\hat{f}_{X}^{\ge}$, $\hat{f}_{X}^{<}$,
 $\hat{g}_{X}^{\ge}$, $\hat{g}_{X}^{<}\in \widehat{\cH}$ $(X\in
 \{1,2\}\times \cB\times \G\times [0,\beta))$ by
\begin{align}
&\hat{f}_{\orho\rho \bx s}^{\ge}=\hat{f}_{\orho\rho \bx
 s}^{<}:=\hat{f}_{\orho\rho \bx s}^1+\hat{f}_{\orho\rho \bx
 (-s)}^2,\label{eq_P_S_type_decomposition_simpler}\\
&\hat{g}_{\orho\rho \bx s}^{\ge}:= \hat{g}_{\orho\rho \bx
 (\beta+s)}^1+\hat{g}_{\orho\rho \bx(-s)}^2,\quad   
\hat{g}_{\orho\rho \bx s}^{<}:= \hat{g}_{\orho\rho \bx
 s}^1+\hat{g}_{\orho\rho \bx(\beta-s)}^2.\notag
\end{align}
By applying \eqref{eq_exponential_integral_formula} repeatedly we can
 confirm that for any $(\orho,\rho,\bx,s)$, $(\oeta,\eta,\by,t)\in
 \{1,2\}\times \cB\times \G\times [0,\beta)$
\begin{align*}
1_{s\ge t}\<\hat{f}_{\orho\rho\bx s}^{\ge}, \hat{g}_{\oeta\eta\by
 t}^{\ge} \>_{\widehat{\cH}}&=1_{s\ge t}(\<\hat{f}_{\orho\rho\bx s}^{1}, \hat{g}_{\oeta\eta\by
 (\beta+t)}^{1} \>_{\widehat{\cH}}+
\<\hat{f}_{\orho\rho\bx (-s)}^{2}, \hat{g}_{\oeta\eta\by
 (-t)}^{2}\>_{\widehat{\cH}})\\
&=1_{s\ge t}C_{1-2}^{\ge}(\orho\rho\bx s,\oeta\eta \by t),\\
1_{s< t}\<\hat{f}_{\orho\rho\bx s}^{<}, \hat{g}_{\oeta\eta\by
 t}^{<} \>_{\widehat{\cH}}&=1_{s< t}(\<\hat{f}_{\orho\rho\bx s}^{1}, \hat{g}_{\oeta\eta\by
 t}^{1} \>_{\widehat{\cH}}+
\<\hat{f}_{\orho\rho\bx (-s)}^{2}, \hat{g}_{\oeta\eta\by
 (\beta-t)}^{2}\>_{\widehat{\cH}})\\
&=1_{s< t}C_{1-2}^{<}(\orho\rho\bx s,\oeta\eta \by t),
\end{align*}
and thus
\begin{align}
C_{1-2}(\orho\rho\bx s,\oeta \eta \by t)=1_{s\ge t} \<\hat{f}_{\orho\rho\bx s}^{\ge}, \hat{g}_{\oeta\eta\by
 t}^{\ge} \>_{\widehat{\cH}}
-1_{s< t}\<\hat{f}_{\orho\rho\bx s}^{<}, \hat{g}_{\oeta\eta\by
 t}^{<} \>_{\widehat{\cH}}.\label{eq_covariance_rest_gram}
\end{align}
To estimate the norms of $\hat{f}_{X}^{\ge}$, $\hat{f}_{X}^{<}$,
 $\hat{g}_{X}^{\ge}$, $\hat{g}_{X}^{<}$ $(X\in \{1,2\}\times \cB\times
 \G\times [0,\beta))$, let us observe that
for $\bk\in \G^*$, $\tau\in \cB(\bk)$, $\bar{a}\in \{1,2\}$
\begin{align}
&|1+e^{-\beta
 ((-1)^{\bar{a}+1}i\frac{\theta(\beta)}{2}+\hat{e}_{\tau}(\bk))}|^2\ge 
(1-e^{-\beta \hat{e}_{\tau}(\bk)})^2\ge
(1-e^{-1})^2,\label{eq_core_condition_assisted_bound}
\end{align}
where we used the fact that $\hat{e}_{\tau}(\bk)\ge
 x_{\bk}+\frac{1}{2(b+1)\beta}\ge 1/\beta$.   
Taking into account \eqref{eq_core_condition_assisted_bound} and the
 unitary property of $U(\bk)$ and using
 \eqref{eq_exponential_integral_formula}, we can derive that for any
 $(\orho,\rho,\bx,s)\in \{1,2\}\times \cB\times \G\times [0,\beta)$ and
 $A\in \{f,g\}$
\begin{align*}
\|\hat{A}_{\orho\rho \bx
 s}^{\ge}\|_{\widehat{\cH}}^2&=\|\hat{A}_{\orho\rho \bx
 s}^{<}\|_{\widehat{\cH}}^2=\|\hat{A}^1_{\orho \rho \b0 0}\|^2_{\widehat{\cH}}+ 
\|\hat{A}^2_{\orho \rho \b0 0}\|^2_{\widehat{\cH}}\\
&\le \frac{c}{L^d}\sum_{(\otau,\tau)\in \{1,2\}\times\cB}\sum_{\bk\in
 \G^*}|U(\bk)((\orho-1)b+\rho,(\otau-1)b+\tau))|^2=c.
\end{align*}
With these bounds we can apply \cite[\mbox{Theorem
 1.3}]{PS} in \eqref{eq_covariance_rest_gram} and conclude that
\begin{align}
&|\det(\<\bu_i,\bw_j\>_{\C^m}C_{1-2}(X_i,Y_j))_{1\le i,j\le n}|\le
 c^n,\label{eq_determinant_bound_rest_part}\\
&(\forall m,n\in \N,\ \bu_i,\bw_i\in \C^m\text{ with
 }\|\bu_i\|_{\C^m},\|\bw_i\|_{\C^m}\le 1,\ X_i,Y_i\in I_0\
 (i=1,\cdots,n)).\notag
\end{align}

Since we have \eqref{eq_determinant_bound_residue_part} and 
\eqref{eq_determinant_bound_rest_part}, we can apply the Cauchy-Binet
 formula in a standard way (see e.g. \cite[\mbox{Lemma
 A.1}]{K_BCS_I}) in \eqref{eq_one_covariance_bare_decomposition} to
 obtain the claimed determinant bound.
\end{proof}

\subsection{General estimation}\label{subsec_general_estimation}

Let $\cV$ denote the complex vector space spanned by the abstract basis
$\{\psi_{X}\}_{X\in I}$. Then let $\bigwedge \cV$ be the Grassmann 
algebra generated by $\{\psi_{X}\}_{X\in I}$ and $\bigwedge_{even} \cV$
be the subspace of $\bigwedge \cV$ spanned by even monomials. These
Grassmann algebras are exactly same as those defined in \cite{K_BCS_II}. The
grand canonical partition function and the thermal expectations are
formulated into a hybrid of Gaussian integral with real variables and
Grassmann Gaussian integral over $\bigwedge \cV$ in the same way as 
\cite[\mbox{Lemma 3.6}]{K_BCS_II}. As in the previous papers, the proof of Theorem
\ref{thm_infinite_volume_limit} relies on analysis of the Grassmann
Gaussian integral appearing in the hybrid formulation. The aim of this
subsection is to summarize necessary estimates of the output of the Grassmann
Gaussian integral in a generalized setting. Here we do not introduce
concrete model-dependent Grassmann polynomials or covariances. We only
assume generic properties of Grassmann polynomials and a covariance. The
estimates can be used as tools to analyze the Grassmann integral
formulation if the real Grassmann polynomials and the real covariances
stemming from the model are
substituted. In fact all the inequalities claimed below are
straightforward variants of the results of \cite[\mbox{Subsection
3.2}]{K_BCS_I}, \cite[\mbox{Subsection 4.2}]{K_BCS_II}. We only provide
minimum sketches of the proofs rather than fully repeat parallel
arguments. However, the resulting inequalities themselves will be stated
without omission. We will see that seemingly subtle changes from the
previous estimates constitute the essence of the proof of Theorem
\ref{thm_infinite_volume_limit}.  

In this subsection we assume that the covariance $\cC:I_0^2\to\C$
satisfies with a constant $D\in \R_{>0}$ that 
\begin{align}
&\cC(\orho\rho\bx s,\oeta\eta\by t)=\cC(\orho\rho\bx 0, \oeta\eta \by
 0),\quad (\forall (\orho,\rho,\bx,s),(\oeta,\eta,\by,t)\in
 I_0),\label{eq_covariance_time_independence}\\
&|\det(\<\bu_i,\bw_j\>_{\C^m}\cC(X_i,Y_j))_{1\le i,j\le n}|\le
 D^n,\notag\\
&(\forall m,n\in \N,\ \bu_i,\bw_i\in \C^m\text{ with
 }\|\bu_i\|_{\C^m},\|\bw_i\|_{\C^m}\le 1,\ X_i,Y_i\in I_0\
 (i=1,\cdots,n)).\notag
\end{align}
A common property satisfied by kernels of Grassmann polynomials in the
following analysis is the invariance 
\begin{align}
F(\cR_{\beta}(\bX+s))=F(\bX),\quad \left(\forall \bX\in I^m,\ s\in
 \frac{1}{h}\Z\right),\label{eq_time_translation_invariance}
\end{align}
where $F:I^m\to \C$. Let us refer to \cite[\mbox{Subsection
4.2}]{K_BCS_II} for the definition of the map $\cR_{\beta}$. Also, the
meaning of the notation $\bX+s$ is explained in \cite[\mbox{Subsection
3.1}]{K_BCS_I} in a parallel situation. The property
\eqref{eq_covariance_time_independence} implies that its extension
$\tilde{\cC}:I^2\to\C$ defined as in \eqref{eq_anti_symmetric_extension}
satisfies \eqref{eq_time_translation_invariance}. In the following we
assume that $F^j(\psi)$ $(j\in \N)$, $F(\psi)\in \bigwedge_{even}\cV$
and the anti-symmetric kernels $F_m^j:I^m\to\C$, $F_m:I^m\to \C$
$(m=2,4,\cdots,N)$ satisfy \eqref{eq_time_translation_invariance}. Here
$N$ denotes $4b\beta h L^d$, the cardinality of $I$. We use these
Grassmann polynomials as input to the tree expansions. As another input,
we take $G\in \bigwedge_{even}\cV$ having the form 
\begin{align*}
G(\psi)=\sum_{p,q=2}^N1_{p,q\in
 2\N}\left(\frac{1}{h}\right)^{p+q}\sum_{\bX\in I^p\atop \bY\in
 I^q}G_{p,q}(\bX,\bY)\psi_{\bX}\psi_{\bY}
\end{align*}
with the bi-anti-symmetric kernels $G_{p,q}:I^p\times I^q\to \C$
$(p,q=2,4,\cdots,N)$ satisfying \eqref{eq_time_translation_invariance}
and the vanishing property
\begin{align}
&\sum_{(s_1,\cdots,s_p)\in
 [0,\beta)_h^p}G_{p,q}((\orho_1\rho_1\bx_1s_1\xi_1,\cdots,\orho_p\rho_p\bx_ps_p\xi_p),\bY)f(s_1,\cdots,s_p)=0,\label{eq_time_vanishing_property}\\
&(\forall
 (\orho_1,\rho_1,\bx_1,\xi_1),\cdots,(\orho_p,\rho_p,\bx_p,\xi_p)\in\{1,2\}\times\cB\times\G\times\{1,-1\},\
 \bY\in I^q),\notag\\
&\sum_{(t_1,\cdots,t_q)\in
 [0,\beta)_h^q}G_{p,q}(\bX,(\oeta_1\eta_1\by_1t_1\zeta_1,\cdots,\oeta_q\eta_q\by_q
 t_q\zeta_q))g(t_1,\cdots,t_q)=0,\notag\\
&(\forall \bX\in I^p,\ 
 (\oeta_1,\eta_1,\by_1,\zeta_1),\cdots,(\oeta_q,\eta_q,\by_q,\zeta_q)\in\{1,2\}\times\cB\times\G\times\{1,-1\}),\notag
\end{align}
for any $f:[0,\beta)_h^p\to \C$, $g:[0,\beta)_h^q\to \C$ satisfying that
\begin{align*}
&f(r_{\beta}(s_1+s),\cdots,r_{\beta}(s_p+s))=f(s_1,\cdots,s_p),\  \left(\forall (s_1,\cdots,s_p)\in [0,\beta)_h^p,\ s\in
 \frac{1}{h}\Z\right),\\
&g(r_{\beta}(s_1+s),\cdots,r_{\beta}(s_q+s))=g(s_1,\cdots,s_q),\ 
 \left(\forall (s_1,\cdots,s_q)\in [0,\beta)_h^q,\ s\in
 \frac{1}{h}\Z\right).
\end{align*}
Recall that for any $s\in \frac{1}{h}\Z$, $r_{\beta}(s)\in [0,\beta)_h$
and $r_{\beta}(s)=s$ in $\frac{1}{h}\Z/\beta\Z$. The definition of
the map $r_{\beta}:\frac{1}{h}\Z\to [0,\beta)_h$ was originally given in 
\cite[\mbox{Subsection 3.2}]{K_BCS_I}. We also introduce $G^j\in
\bigwedge_{even}\cV$ $(j\in \N)$, assuming that $G^j$ has the
bi-anti-symmetric kernels $G_{p,q}^j:I^p\times I^q\to \C$ 
$(p,q=2,4,\cdots,N)$ 
satisfying
\eqref{eq_time_translation_invariance} and
\eqref{eq_time_vanishing_property}.

For $n\in \N_{\ge 2}$, $l\in \{0,1,\cdots,n\}$ we define
$A^{(n,l)}(\psi)\in \bigwedge_{even}\cV$ by
\begin{align*}
&A^{(n,l)}(\psi):=Tree(\{1,2,\cdots,n\},\cC)\prod_{j=1}^lF^j(\psi^j+\psi)\prod_{k=l+1}^nG^k(\psi^k+\psi)\Bigg|_{\psi^i=0\atop
 (\forall i\in \{1,\cdots,n\})}.
\end{align*}
The definition of the operator ``$Tree(\{1,\cdots,n\},\cC)$'' is
written in \cite[\mbox{Subsection 3.1}]{K_BCS_I}. It applies to the
present case if we add the set $\cB$ to the index set ``$I$''
of \cite{K_BCS_I}. In fact the current version of
$Tree(\{1,\cdots,n\},\cC)$ is exactly same as that used in 
\cite[\mbox{Subsection 4.2}]{K_BCS_II}. 
In the first lemma we summarize necessary bound properties of the
anti-symmetric kernels of $A^{(n,l)}(\psi)$. Let us refer to 
\cite[\mbox{Subsection 4.1}]{K_BCS_II} for the definition of the norm
$\|\cdot\|_1$.

\begin{lemma}\label{lem_general_estimation_direct}
For any $m\in \{2,4,\cdots,N\}$, $n\in \N_{\ge 2}$, $l\in
 \{0,1,\cdots,n\}$ the anti-symmetric kernel $A_m^{(n,l)}(\cdot)$
 satisfies \eqref{eq_time_translation_invariance}. Moreover, the
 following inequalities hold for any $m\in \{0,2,\cdots,N\}$, $n\in
 \N_{\ge 2}$, $l\in \{0,1,\cdots,n\}$, $l'\in \{1,2,\cdots,n\}$. 
\begin{align}
&\|A_m^{(n,l)}\|_{1,\infty}\le \left(\frac{N}{h}\right)^{1_{m=0}}
(n-2)!D^{-n+1-\frac{m}{2}}2^{-2m}\|\tilde{\cC}\|_{1,\infty}^{n-1}\label{eq_general_estimation_direct}\\
&\qquad\qquad\qquad\cdot\prod_{j=1}^l\left(\sum_{p_j=2}^N2^{3p_j}D^{\frac{p_j}{2}}\|F_{p_j}^j\|_{1,\infty}
\right)\prod_{k=l+1}^n\left(\sum_{p_k=4}^N2^{3p_k}D^{\frac{p_k}{2}}\|G_{p_k}^k\|_{1,\infty}
\right)\notag\\
&\qquad\qquad\qquad\cdot 1_{\sum_{j=1}^np_j-2(n-1)\ge m\ge 2(n-l)}.\notag\\
&\|A_m^{(n,l')}\|_{1}\le (n-2)!
 D^{-n+1-\frac{m}{2}}2^{-2m}\|\tilde{\cC}\|_{1,\infty}^{n-1}
\sum_{p_1=2}^N2^{3p_1}D^{\frac{p_1}{2}}\|F_{p_1}^1\|_{1}\label{eq_general_estimation_direct_one}\\
&\qquad\qquad\quad\cdot\prod_{j=2}^{l'}\left(\sum_{p_j=2}^N2^{3p_j}D^{\frac{p_j}{2}}\|F_{p_j}^j\|_{1,\infty}
\right)
\prod_{k=l'+1}^n\left(\sum_{p_k=4}^N2^{3p_k}D^{\frac{p_k}{2}}\|G_{p_k}^k\|_{1,\infty}
\right)\notag\\
&\qquad\qquad\quad\cdot 1_{\sum_{j=1}^np_j-2(n-1)\ge m\ge 2(n-l')}.\notag
\end{align}
\end{lemma}

\begin{proof}
The statement concerning the property
 \eqref{eq_time_translation_invariance}
is essentially implied by 
\cite[\mbox{Lemma 3.1}]{K_BCS_I}. Let us define the map $P_0:I^m\to I^m$
 by 
\begin{align*}
&P_0((\orho_1,\rho_1,\bx_1,s_1,\xi_1),\cdots,(\orho_m,\rho_m,\bx_m,s_m,\xi_m))\\
&:=((\orho_1,\rho_1,\bx_1,0,\xi_1),\cdots,(\orho_m,\rho_m,\bx_m,0,\xi_m)),\\
&(\forall (\orho_j,\rho_j,\bx_j,s_j,\xi_j)\in I\ (j=1,\cdots,m)).
\end{align*}
Let us use the notation $P_0$ for different $m$ for simplicity. 
Then by taking into account anti-symmetry and the time-independent property
 \eqref{eq_covariance_time_independence} we observe that for $m\in
 \{0,2,\cdots,N\}$, $n\in \N_{\ge 2}$, $l\in \{0,1,\cdots,n\}$ 
\begin{align*}
&A_m^{(n,l)}(\psi)\\
&=Tree(\{1,\cdots,n\},\cC)\\
&\quad\cdot\prod_{j=1}^l\Bigg(
\sum_{n_j=2}^N\sum_{m_j=0}^{n_j-1}\left(\begin{array}{c}n_j \\
					m_j\end{array}\right)
\left(\frac{1}{h}\right)^{n_j}\sum_{\bX_j\in I^{m_j}}\sum_{\bY_j\in
 I^{n_j-m_j}}F_{n_j}^j(\bY_j,\bX_j)\psi_{P_0(\bY_j)}^j\psi_{\bX_j}\Bigg)\\
&\quad\cdot\prod_{k=l+1}^n\Bigg(
\sum_{n_k=4}^N\sum_{m_k=0}^{n_k-1}\left(\begin{array}{c}n_k \\
					m_k\end{array}\right)
\left(\frac{1}{h}\right)^{n_k}\sum_{\bX_k\in I^{m_k}}\sum_{\bY_k\in
 I^{n_k-m_k}}G_{n_k}^k(\bY_k,\bX_k)\psi_{P_0(\bY_k)}^k\psi_{\bX_k}\Bigg)\\
&\quad\cdot \Bigg|_{\psi^i=0\atop (\forall i\in \{1,\cdots,n\})}1_{\sum_{j=1}^nm_j=m}.
\end{align*}
By the uniqueness of an anti-symmetric kernel, for any
 $\bX=(X_1,\cdots,X_{n_k})\in I^{n_k}$ 
\begin{align*}
G_{n_k}^k(\bX)=&\sum_{p,q=2}^N1_{p,q\in
 2\N}1_{p+q=n_k}\\
&\cdot\frac{1}{n_k!}\sum_{\s\in
 \S_{n_k}}\sgn(\s)G_{p,q}^k((X_{\s(1)},\cdots,X_{\s(p)}),(X_{\s(p+1)},\cdots,X_{\s(p+q)})),
\end{align*}
where $\S_{n_k}$ is the set of permutations of $\{1,\cdots,n_k\}$ and
 $\sgn(\s)$ is the sign of $\s\in \S_{n_k}$. If $m_k\le 1$, the property
 \eqref{eq_time_vanishing_property} implies that
\begin{align*}
\sum_{\bY_k\in I^{n_k-m_k}}G_{n_k}^k(\bY_k,\bX_k)\psi_{P_0(\bY_k)}^k=0
\end{align*}
for any $\bX_k\in I^{m_k}$. Therefore 
\begin{align*}
&A_m^{(n,l)}(\psi)\\
&=Tree(\{1,\cdots,n\},\cC)\\
&\quad\cdot\prod_{j=1}^l\Bigg(
\sum_{n_j=2}^N\sum_{m_j=0}^{n_j-1}\left(\begin{array}{c}n_j \\
					m_j\end{array}\right)
\left(\frac{1}{h}\right)^{n_j}\sum_{\bX_j\in I^{m_j}}\sum_{\bY_j\in
 I^{n_j-m_j}}F_{n_j}^j(\bY_j,\bX_j)\psi_{\bY_j}^j\psi_{\bX_j}\Bigg)\\
&\quad\cdot\prod_{k=l+1}^n\Bigg(
\sum_{n_k=4}^N\sum_{m_k=0}^{n_k-1}\left(\begin{array}{c}n_k \\
					m_k\end{array}\right)
\left(\frac{1}{h}\right)^{n_k}\sum_{\bX_k\in I^{m_k}}\sum_{\bY_k\in
 I^{n_k-m_k}}G_{n_k}^k(\bY_k,\bX_k)\psi_{\bY_k}^k\psi_{\bX_k}\Bigg)\\
&\quad\cdot \Bigg|_{\psi^i=0\atop (\forall i\in
 \{1,\cdots,n\})}1_{\sum_{j=1}^nm_j=m\ge 2(n-l)}\\
&=1_{m\ge 2(n-l)}A_m^{(n,l)}(\psi).
\end{align*}
We can apply the inequality ``(3.16)'' of \cite[\mbox{Lemma
 3.1}]{K_BCS_I} or ``(4.8)'' of 
\cite[\mbox{Lemma
 4.1}]{K_BCS_II} to estimate the anti-symmetric kernel of
 $A_m^{(n,l)}(\psi)$. Multiplying the result by $1_{m\ge 2(n-l)}$ yields
 \eqref{eq_general_estimation_direct}. Now we have 
$A_m^{(n,l')}(\psi)=1_{m\ge 2(n-l')}A_m^{(n,l')}(\psi)$.
We can apply ``(3.17)'' of \cite[\mbox{Lemma
 3.1}]{K_BCS_I} or ``(4.9)'' of \cite[\mbox{Lemma
 4.1}]{K_BCS_II} to bound $\|A_m^{(n,l')}\|_1$ and multiply the result
 by $1_{m\ge 2(n-l')}$ to obtain \eqref{eq_general_estimation_direct_one}.
\end{proof}

Next we consider the Grassmann polynomials $B^{(n)}(\psi)$,
$\widehat{B}^{(n')}(\psi)\in \bigwedge_{even}\cV$ $(n\in \N,\ n'\in
\N_{\ge 2})$ defined as below. 
\begin{align*}
B^{(n)}(\psi)
:=&\sum_{p,q=2}^{N}1_{p,q\in 2\N}\frah^{p+q}\sum_{\bX\in I^p\atop\bY\in
 I^q}G_{p,q}(\bX,\bY)Tree(\{1,\cdots,n+1\},\cC)\\
&\cdot(\psi^1+\psi)_{\bX}(\psi^2+\psi)_{\bY}\prod_{j=3}^{n+1}G^j(\psi^j+\psi)\Bigg|_{\psi^i=0\atop(\forall
 i\in\{1,\cdots,n+1\})},\\
\widehat{B}^{(n')}(\psi)
:=&\sum_{p,q=2}^{N}1_{p,q\in 2\N}\frah^{p+q}\sum_{\bX\in I^p\atop\bY\in
 I^q}G_{p,q}(\bX,\bY)Tree(\{1,\cdots,n'+1\},\cC)\\
&\cdot(\psi^1+\psi)_{\bX}(\psi^2+\psi)_{\bY}\prod_{j=3}^{n'}G^j(\psi^j+\psi)
F(\psi^{n'+1}+\psi)
\Bigg|_{\psi^i=0\atop(\forall
 i\in\{1,\cdots,n'+1\})}.
\end{align*}
The anti-symmetric kernels of these polynomials can be estimated as
follows. See \cite[\mbox{Subsection 4.1}]{K_BCS_II} for the definition of
the measurement $[\cdot,\cdot]_{1,\infty}$.

\begin{lemma}\label{lem_general_estimation_double}
For any $m\in \{2,4,\cdots,N\}$, $n\in\N$, $n'\in \N_{\ge 2}$ the anti-symmetric
 kernels $B_m^{(n)}(\cdot)$, $\widehat{B}_m^{(n')}(\cdot)$ satisfy
 \eqref{eq_time_translation_invariance}. Moreover, 
the following inequalities hold for any 
 $m\in\{0,2,\cdots,N\}$, $n\in \N_{\ge 2}$.
\begin{align}
&\|B_m^{(1)}\|_{1,\infty}
\le
 D^{-1-\frac{m}{2}}\sum_{p_1,p_2=2}^{N}1_{p_1,p_2\in
 2\N}2^{2p_1+2p_2}D^{\frac{p_1+p_2}{2}}[G_{p_1,p_2},\tilde{\cC}]_{1,\infty}1_{p_1+p_2-2\ge
 m\ge 2}.\label{eq_general_estimation_double_one_field}\\
&\|B_m^{(n)}\|_{1,\infty}\le 
(n-1)!D^{-n-\frac{m}{2}}2^{-2m}\|\tilde{\cC}\|_{1,\infty}^{n-1}\label{eq_general_estimation_double}\\
&\qquad\qquad\quad \cdot\sum_{p_1,p_2=2}^{N}1_{p_1,p_2\in2\N}2^{3p_1+3p_2}D^{\frac{p_1+p_2}{2}}[G_{p_1,p_2},\tilde{\cC}]_{1,\infty}
\prod_{j=3}^{n+1}\left(\sum_{p_j=4}^N2^{3p_j}D^{\frac{p_j}{2}}\|G_{p_j}^j\|_{1,\infty}
\right)\notag\\
&\qquad\qquad\quad\cdot 1_{\sum_{j=1}^{n+1}p_j-2n\ge m\ge 2n}.\notag\\
&\|\widehat{B}_m^{(n)}\|_{1}\le 
(n-1)!D^{-n-\frac{m}{2}}2^{-2m}\|\tilde{\cC}\|_{1,\infty}^{n-1}\label{eq_general_estimation_double_one_norm}\\
&\qquad\qquad\ \cdot\sum_{p_1,p_2=2}^{N}1_{p_1,p_2\in2\N}2^{3p_1+3p_2}D^{\frac{p_1+p_2}{2}}[G_{p_1,p_2},\tilde{\cC}]_{1,\infty}
\prod_{j=3}^{n}\left(\sum_{p_j=4}^N2^{3p_j}D^{\frac{p_j}{2}}\|G_{p_j}^j\|_{1,\infty}
\right)\notag\\
&\qquad\qquad\ \cdot
\sum_{p_{n+1}=2}^N2^{3p_{n+1}}D^{\frac{p_{n+1}}{2}}\|F_{p_{n+1}}^{n+1}\|_1
1_{\sum_{j=1}^{n+1}p_j-2n\ge m\ge 2n-2}.\notag
\end{align}
\end{lemma}

\begin{proof}
The first statement of the lemma is essentially proved in 
\cite[\mbox{Lemma 3.2}]{K_BCS_I}. By the same consideration based on
 anti-symmetry and the properties
 \eqref{eq_covariance_time_independence},
 \eqref{eq_time_vanishing_property} as in the proof of Lemma
 \ref{lem_general_estimation_direct} we can deduce that for any $n\in
 \N_{\ge 2}$, $m\in \{0,2,\cdots,N\}$ 
\begin{align*}
&\widehat{B}_m^{(n)}(\psi)\\
&=\sum_{p_1,p_2=2}^N\left(\frac{1}{h}\right)^{p_1+p_2}\sum_{m_1=0}^{p_1-1}\sum_{m_2=0}^{p_2-1}\left(\begin{array}{c}p_1\\
															     m_1\end{array}\right)\left(\begin{array}{c}p_2\\
																			m_2\end{array}\right)\\
&\quad\cdot \sum_{\bX_1\in I^{m_1}\atop \bY_1\in I^{p_1-m_1}}
\sum_{\bX_2\in I^{m_2}\atop \bY_2\in
 I^{p_2-m_2}}G_{p_1,p_2}((\bY_1,\bX_1),(\bY_2,\bX_2))\\
&\quad\cdot Tree(\{1,\cdots,n+1\},\cC)\psi_{\bY_1}^1\psi_{\bX_1}\psi_{\bY_2}^2\psi_{\bX_2}\\
&\quad\cdot \prod_{j=3}^n\Bigg(\sum_{p_j=4}^N\left(
\frac{1}{h}\right)^{p_j}\sum_{m_j=0}^{p_j-1}\left(\begin{array}{c} p_j\\
					    m_j \end{array}\right)
\sum_{\bX_j\in I^{m_j}\atop \bY_j\in
 I^{p_j-m_j}}G_{p_j}^j(\bY_j,\bX_j)\psi_{\bY_j}^j\psi_{\bX_j}
\Bigg)\\
&\quad\cdot \sum_{p_{n+1}=2}^N\left(
\frac{1}{h}\right)^{p_{n+1}}\sum_{m_{n+1}=0}^{p_{n+1}-1}\left(\begin{array}{c} p_{n+1}\\
					    m_{n+1} \end{array}\right)
\sum_{\bX_{n+1}\in I^{m_{n+1}}\atop \bY_{n+1}\in
 I^{p_{n+1}-m_{n+1}}}F_{p_{n+1}}(\bY_{n+1},\bX_{n+1})\psi_{\bY_{n+1}}^{n+1}\psi_{\bX_{n+1}}\\
&\quad\cdot \Bigg|_{\psi^i=0\atop(\forall
 i\in\{1,\cdots,n+1\})}1_{\sum_{j=1}^{n+1}m_j=m\ge 2n-2}\\
&=1_{m\ge 2n-2}\widehat{B}_m^{(n)}(\psi).
\end{align*}
In the first equality we took into account the constraints $m_1\ge 1$,
 $m_2\ge 1$, $m_j\ge 2$ $(j=3,\cdots,n)$. Then we can apply ``(3.27)''
 of \cite[\mbox{Lemma 3.2}]{K_BCS_I} or ``(4.14)''
 of \cite[\mbox{Lemma 4.2}]{K_BCS_II} to derive
 \eqref{eq_general_estimation_double_one_norm}. In the same way as above
 we have that for any $m\in \{0,2,\cdots,N\}$, $n\in \N_{\ge 2}$
$B_m^{(1)}(\psi)=1_{m\ge 2}B_m^{(1)}(\psi)$, $B_m^{(n)}(\psi)=1_{m\ge
 2n}B_m^{(n)}(\psi)$. Then we can apply  ``(3.24)''
 of \cite[\mbox{Lemma 3.2}]{K_BCS_I} or ``(4.11)''
 of \cite[\mbox{Lemma 4.2}]{K_BCS_II} to derive
 \eqref{eq_general_estimation_double_one_field} and ``(3.26)''
 of \cite[\mbox{Lemma 3.2}]{K_BCS_I} or ``(4.13)''
 of \cite[\mbox{Lemma 4.2}]{K_BCS_II} to derive
 \eqref{eq_general_estimation_double}.
\end{proof}

Assume that $n\in \N$, $m\in \{0,1,\cdots,n-1\}$, 
\begin{align*}
&1=s_1<s_2<\cdots <s_{m+1}\le n,\quad 1=t_1<t_2<\cdots <t_{n-m}\le n,\\
&\{s_j\}_{j=2}^{m+1}\cup \{t_k\}_{k=2}^{n-m}=\{2,3,\cdots,n\},\quad
 \{s_j\}_{j=2}^{m+1}\cap \{t_k\}_{k=2}^{n-m}=\emptyset.
\end{align*}
Finally let us study the Grassmann polynomials $E^{(n)}(\psi)$,
$\widehat{E}^{(n)}(\psi)\in \bigwedge_{even}\cV$ defined as follows. 
\begin{align*}
E^{(n)}(\psi)
&:=\sum_{p,q=2}^{N}1_{p,q\in 2\N}\frah^{p+q}\sum_{\bX\in I^p\atop
 \bY\in I^q}G_{p,q}(\bX,\bY)\\
&\quad\cdot
 Tree(\{s_j\}_{j=1}^{m+1},\cC)(\psi^1+\psi)_{\bX}\prod_{j=2}^{m+1}G^{s_j}(\psi^{s_j}+\psi)\Bigg|_{\psi^{s_j}=0\atop(\forall
 j\in\{1,2,\cdots,m+1\})}\\
&\quad\cdot
 Tree(\{t_k\}_{k=1}^{n-m},\cC)(\psi^1+\psi)_{\bY}\prod_{k=2}^{n-m}G^{t_k}(\psi^{t_k}+\psi)\Bigg|_{\psi^{t_k}=0\atop(\forall
 k\in\{1,2,\cdots,n-m\})},\\
\widehat{E}^{(n)}(\psi)
&:=\sum_{p,q=2}^{N}1_{p,q\in 2\N}\frah^{p+q}\sum_{\bX\in I^p\atop
 \bY\in I^q}G_{p,q}(\bX,\bY)\\
&\quad\cdot
 Tree(\{s_j\}_{j=1}^{m+1},\cC)(\psi^1+\psi)_{\bX}\\
&\quad\cdot \prod_{j=2}^{m+1}
(1_{s_j\neq n}G^{s_j}(\psi^{s_j}+\psi)+ 1_{s_j= n}F(\psi^{s_j}+\psi))
\Bigg|_{\psi^{s_j}=0\atop(\forall
 j\in\{1,2,\cdots,m+1\})}\\
&\quad\cdot
 Tree(\{t_k\}_{k=1}^{n-m},\cC)(\psi^1+\psi)_{\bY}\\
&\quad\cdot \prod_{k=2}^{n-m}(1_{t_k\neq n}
G^{t_k}(\psi^{t_k}+\psi)+1_{t_k=n}F(\psi^{t_k}+\psi))\Bigg|_{\psi^{t_k}=0\atop(\forall
 k\in\{1,2,\cdots,n-m\})}.
\end{align*}
These Grassmann polynomials are special examples of those studied in 
\cite[\mbox{Lemma 4.4}]{K_BCS_II} and also close to those studied in
\cite[\mbox{Lemma 3.3}]{K_BCS_I}. The properties we need for later
application are summarized in the next lemma. 
The definition of the measurement $[\cdot,\cdot]_1$ is found in 
\cite[\mbox{Subsection 4.1}]{K_BCS_II}.

\begin{lemma}\label{lem_general_estimation_divided}
For any $n\in \N$, $a,b\in \{2,4,\cdots,N\}$ there exist functions 
$E_{a,b}^{(n)}$, $\widehat{E}_{a,b}^{(n)}:I^a\times I^b\to \C$ such that
 they are bi-anti-symmetric, satisfy
 \eqref{eq_time_translation_invariance},
 \eqref{eq_time_vanishing_property} and 
\begin{align*}
&E^{(n)}(\psi)=\sum_{a,b=2}^N1_{a,b\in
 2\N}\left(\frac{1}{h}\right)^{a+b}\sum_{\bX\in I^a\atop \bY\in
 I^b}E_{a,b}^{(n)}(\bX,\bY)\psi_{\bX}\psi_{\bY},\\
&\widehat{E}^{(n)}(\psi)=\sum_{a,b=2}^N1_{a,b\in
 2\N}\left(\frac{1}{h}\right)^{a+b}\sum_{\bX\in I^a\atop \bY\in
 I^b}\widehat{E}_{a,b}^{(n)}(\bX,\bY)\psi_{\bX}\psi_{\bY}.
\end{align*}
Moreover, the following inequalities hold for any $a,b\in
 \{2,4,\cdots,N\}$, $n\in \N_{\ge 2}$ and anti-symmetric function
 $g:I^2\to\C$. 
\begin{align}
&\|E_{a,b}^{(1)}\|_{1,\infty}\le
 \sum_{p=a}^{N}\sum_{q=b}^{N}1_{p,q\in 2\N}
\left(\begin{array}{c} p \\ a \end{array}\right)
\left(\begin{array}{c} q \\ b \end{array}\right)
D^{\frac{1}{2}(p+q-a-b)}\|G_{p,q}\|_{1,\infty}.\label{eq_general_estimation_divided_single}\\&[E_{a,b}^{(1)},g]_{1,\infty}\le
 \sum_{p=a}^{N}\sum_{q=b}^{N}1_{p,q\in 2\N}
\left(\begin{array}{c} p \\ a \end{array}\right)
\left(\begin{array}{c} q \\ b \end{array}\right)
D^{\frac{1}{2}(p+q-a-b)}[G_{p,q},g]_{1,\infty}.\label{eq_general_estimation_divided_single_coupled}\\
&\|E_{a,b}^{(n)}\|_{1,\infty}\label{eq_general_estimation_divided}\\
&\le (1_{m\neq 0}(m-1)!+1_{m=0})(1_{m\neq
 n-1}(n-m-2)!+1_{m=n-1})\notag\\
&\quad\cdot 2^{-2a-2b}D^{-n+1-\frac{1}{2}(a+b)}\|\tilde{\cC}\|_{1,\infty}^{n-1}\sum_{p_1,q_1=2}^{N}1_{p_1,q_1\in
 2\N}2^{3p_1+3q_1}D^{\frac{p_1+q_1}{2}}\|G_{p_1,q_1}\|_{1,\infty}\notag\\
&\quad\cdot\prod_{j=2}^{m+1}\Bigg(
\sum_{p_j=4}^N2^{3p_j}D^{\frac{p_j}{2}}\|G_{p_j}^{s_j}\|_{1,\infty}\Bigg)
\prod_{k=2}^{n-m}\Bigg(
\sum_{q_k=4}^N2^{3q_k}D^{\frac{q_k}{2}}
\|G_{q_k}^{t_k}\|_{1,\infty}\Bigg)\notag\\
&\quad\cdot 1_{\sum_{j=1}^{m+1}p_j-2m\ge a\ge 2m+2}
1_{\sum_{k=1}^{n-m}q_k-2(n-m-1)\ge b\ge 2(n-m)}.\notag\\
&[E_{a,b}^{(n)},g]_{1,\infty}\label{eq_general_estimation_divided_coupled}\\
&\le (1_{m\neq 0}(m-1)!+1_{m=0})(1_{m\neq
 n-1}(n-m-2)!+1_{m=n-1})\notag\\
&\quad\cdot
 2^{-2a-2b}D^{-n+1-\frac{1}{2}(a+b)}\|\tilde{\cC}\|_{1,\infty}^{n-2}\notag\\
&\quad\cdot \sum_{p_1,q_1=2}^{N}1_{p_1,q_1\in
 2\N}2^{3p_1+3q_1}D^{\frac{p_1+q_1}{2}}([G_{p_1,q_1},g]_{1,\infty}\|\tilde{\cC}\|_{1,\infty}+[G_{p_1,q_1},\tilde{\cC}]_{1,\infty}\|g\|_{1,\infty})\notag\\
&\quad\cdot\prod_{j=2}^{m+1}\Bigg(
\sum_{p_j=4}^N2^{3p_j}D^{\frac{p_j}{2}}\|G_{p_j}^{s_j}\|_{1,\infty}\Bigg)
\prod_{k=2}^{n-m}\Bigg(
\sum_{q_k=4}^N2^{3q_k}D^{\frac{q_k}{2}}\|G_{q_k}^{t_k}\|_{1,\infty}\Bigg)\notag\\
&\quad\cdot 1_{\sum_{j=1}^{m+1}p_j-2m\ge a\ge 2m+2}
1_{\sum_{k=1}^{n-m}q_k-2(n-m-1)\ge b\ge 2(n-m)}.\notag\\
&\|\widehat{E}_{a,b}^{(n)}\|_{1}\label{eq_general_estimation_divided_one}\\
&\le (1_{m\neq 0}(m-1)!+1_{m=0})(1_{m\neq
 n-1}(n-m-2)!+1_{m=n-1})\notag\\
&\quad\cdot2^{-2a-2b}D^{-n+1-\frac{1}{2}(a+b)}\|\tilde{\cC}\|_{1,\infty}^{n-1} \sum_{p_1,q_1=2}^{N}1_{p_1,q_1\in
 2\N}2^{3p_1+3q_1}D^{\frac{p_1+q_1}{2}}\|G_{p_1,q_1}\|_{1,\infty}\notag\\
&\quad\cdot\prod_{j=2}^{m+1}\Bigg(
\sum_{p_j=2}^N2^{3p_j}D^{\frac{p_j}{2}}
(1_{s_j\neq
 n}\|G_{p_j}^{s_j}\|_{1,\infty}+1_{s_j=n}\|F_{p_j}\|_1)\Bigg)\notag\\
&\quad\cdot\prod_{k=2}^{n-m}\Bigg(
\sum_{q_k=2}^N2^{3q_k}D^{\frac{q_k}{2}}
(1_{t_k\neq
 n}\|G_{q_k}^{t_k}\|_{1,\infty}+1_{t_k=n}\|F_{q_k}\|_1)\Bigg)\notag\\
&\quad\cdot 1_{\sum_{j=1}^{m+1}p_j-2m\ge a\ge 2m+2-21_{n\in \{s_j\}_{j=2}^{m+1}}}
1_{\sum_{k=1}^{n-m}q_k-2(n-m-1)\ge b\ge 2(n-m)-21_{n\in
 \{t_k\}_{k=2}^{n-m}}}.\notag\\
&[\widehat{E}_{a,b}^{(n)},g]_{1}\label{eq_general_estimation_divided_coupled_one}\\
&\le (1_{m\neq 0}(m-1)!+1_{m=0})(1_{m\neq
 n-1}(n-m-2)!+1_{m=n-1})\notag\\
&\quad\cdot
 2^{-2a-2b}D^{-n+1-\frac{1}{2}(a+b)}\|\tilde{\cC}\|_{1,\infty}^{n-2}\notag\\
&\quad\cdot \sum_{p_1,q_1=2}^{N}1_{p_1,q_1\in
 2\N}2^{3p_1+3q_1}D^{\frac{p_1+q_1}{2}}([G_{p_1,q_1},g]_{1,\infty}\|\tilde{\cC}\|_{1,\infty}+[G_{p_1,q_1},\tilde{\cC}]_{1,\infty}\|g\|_{1,\infty})\notag\\
&\quad\cdot\prod_{j=2}^{m+1}\Bigg(
\sum_{p_j=2}^N2^{3p_j}D^{\frac{p_j}{2}}
(1_{s_j\neq
 n}\|G_{p_j}^{s_j}\|_{1,\infty}+1_{s_j=n}\|F_{p_j}\|_1)\Bigg)\notag\\
&\quad\cdot\prod_{k=2}^{n-m}\Bigg(
\sum_{q_k=2}^N2^{3q_k}D^{\frac{q_k}{2}}
(1_{t_k\neq
 n}\|G_{q_k}^{t_k}\|_{1,\infty}+1_{t_k=n}\|F_{q_k}\|_1)\Bigg)\notag\\
&\quad\cdot 1_{\sum_{j=1}^{m+1}p_j-2m\ge a\ge 2m+2-21_{n\in \{s_j\}_{j=2}^{m+1}}}
1_{\sum_{k=1}^{n-m}q_k-2(n-m-1)\ge b\ge 2(n-m)-21_{n\in \{t_k\}_{k=2}^{n-m}}}.\notag
\end{align}
\end{lemma}

\begin{remark}
In fact the inequalities \eqref{eq_general_estimation_divided_single},
 \eqref{eq_general_estimation_divided_single_coupled} are same as 
``(4.16)'', ``(4.18)'' of \cite[\mbox{Lemma 4.4}]{K_BCS_II}
 respectively. However, we present them for convenience in the subsequent
 application.
\end{remark}

\begin{proof}
The existence of the bi-anti-symmetric kernels satisfying the claimed
 properties is essentially implied by \cite[\mbox{Lemma
 3.3}]{K_BCS_I}. In fact the kernels are explicitly given in
 \cite[\mbox{(4.15)}]{K_BCS_II} in a more general setting. 
To make clear, let us present the kernel
 $\widehat{E}_{a,b}^{(n)}:I^a\times I^b\to \C$ for $n\in \N$, $a,b\in
 \{2,4,\cdots,N\}$. For $\bX=(X_1,\cdots,X_a)\in I^a$, $\bY=(Y_1,\cdots,Y_b)\in I^b$
\begin{align*}
&\widehat{E}_{a,b}^{(n)}(\bX,\bY)\\
&=\sum_{p_1,q_1=2}^N1_{p_1,q_1\in 2\N}
\sum_{u_1=0}^{p_1}(1_{m=0}+1_{m\neq
 0}1_{u_1\le p_1-1})\left(\begin{array}{c} p_1\\ u_1\end{array}\right)\\
&\quad\cdot \sum_{v_1=0}^{q_1}(1_{m=n-1}+1_{m\neq n-1}1_{v_1\le
	      q_1-1})\left(\begin{array}{c} q_1\\ v_1\end{array}\right)\\
&\quad\cdot 
\left(\frac{1}{h}\right)^{p_1+q_1-u_1-v_1}
\sum_{\bW_1\in I^{p_1-u_1}}
\sum_{\bZ_1\in I^{q_1-v_1}}
G_{p_1,q_1}((\bW_1,\bX_1'),(\bZ_1,\bY_1'))\\
&\quad\cdot \prod_{j=2}^{m+1}\Bigg(1_{s_j\neq
 n}\sum_{p_j=4}^N
\sum_{u_j=0}^{p_j-1}
\left(\begin{array}{c} p_j\\ u_j\end{array}\right)
\left(\frac{1}{h}\right)^{p_j-u_j}
\sum_{\bW_j\in
 I^{p_j-u_j}}G_{p_j}^{s_j}(\bW_j,\bX_j')\\
&\qquad + 1_{s_j=
 n}\sum_{p_j=2}^N\sum_{u_j=0}^{p_j-1}\left(\begin{array}{c} p_j\\ u_j\end{array}\right)\left(\frac{1}{h}\right)^{p_j-u_j}
\sum_{\bW_j\in
 I^{p_j-u_j}}F_{p_j}(\bW_j,\bX_j')\Bigg)\\
&\quad\cdot \prod_{k=2}^{n-m}\Bigg(1_{t_k\neq
 n}\sum_{q_k=4}^N
\sum_{v_k=0}^{q_k-1}
\left(\begin{array}{c} q_k\\ v_k\end{array}\right)
\left(\frac{1}{h}\right)^{q_k-v_k}
\sum_{\bZ_k\in
 I^{q_k-v_k}}G_{q_k}^{t_k}(\bZ_k,\bY_k')\\
&\qquad + 1_{t_k=
 n}\sum_{q_k=2}^N\sum_{v_k=0}^{q_k-1}\left(\begin{array}{c} q_k\\ v_k\end{array}\right)\left(\frac{1}{h}\right)^{q_k-v_k}
\sum_{\bZ_k\in
 I^{q_k-v_k}}F_{q_k}(\bZ_k,\bY_k')\Bigg)\\
&\quad\cdot
 Tree(\{s_j\}_{j=1}^{m+1},\cC)\prod_{j=1}^{m+1}\psi_{\bW_j}^{s_j}
\Bigg|_{\psi^{s_j}=0\atop
 (\forall j\in \{1,\cdots,m+1\})}
Tree(\{t_k\}_{k=1}^{n-m},\cC)\prod_{k=1}^{n-m}\psi_{\bZ_k}^{t_k}
\Bigg|_{\psi^{t_k}=0\atop
 (\forall k\in \{1,\cdots,n-m\})}\\
&\quad \cdot (-1)^{\sum_{j=1}^mu_j\sum_{i=j+1}^{m+1}(p_i-u_i)+
\sum_{k=1}^{n-m-1}v_k\sum_{i=k+1}^{n-m}(q_i-v_i)}1_{\sum_{j=1}^{m+1}u_j=a}1_{\sum_{k=1}^{n-m}v_k=b}
\\
&\quad\cdot \frac{1}{a!b!}\sum_{\s\in \S_a\atop \tau\in \S_b}\sgn(\s)\sgn(\tau)
1_{(\bX_1',\cdots,\bX_{m+1}')=(X_{\s(1)},\cdots,X_{\s(a)})}
1_{(\bY_1',\cdots,\bY_{n-m}')=(Y_{\tau(1)},\cdots,Y_{\tau(b)})}.
\end{align*}
By considering \eqref{eq_covariance_time_independence},
 \eqref{eq_time_vanishing_property}
we can substitute the constraints
\begin{align*}
&u_1\ge 1,\quad u_j\ge 21_{s_j\neq n},\quad (\forall j\in
 \{2,\cdots,m+1\}),\\
&v_1\ge 1,\quad v_k\ge 21_{t_k\neq n},\quad (\forall k\in
 \{2,\cdots,n-m\}).
\end{align*}
Moreover, by using the fact that $a$, $b$ must be even
we have that 
\begin{align*}
\widehat{E}_{a,b}^{(n)}(\bX,\bY)=
1_{a\ge 2m+2-21_{n\in \{s_j\}_{j=2}^{m+1}}}
1_{b\ge  2(n-m)-21_{n\in \{t_k\}_{k=2}^{n-m}}}
\widehat{E}_{a,b}^{(n)}(\bX,\bY),
\end{align*}
and thus
\begin{align*}
&\|\widehat{E}_{a,b}^{(n)}\|_1=1_{a\ge 2m+2-21_{n\in \{s_j\}_{j=2}^{m+1}}}
1_{b\ge  2(n-m)-21_{n\in \{t_k\}_{k=2}^{n-m}}}
\|\widehat{E}_{a,b}^{(n)}\|_1,\\
&[\widehat{E}_{a,b}^{(n)},g]_1=1_{a\ge 2m+2-21_{n\in \{s_j\}_{j=2}^{m+1}}}
1_{b\ge  2(n-m)-21_{n\in \{t_k\}_{k=2}^{n-m}}}
[\widehat{E}_{a,b}^{(n)},g]_1
\end{align*}
for any anti-symmetric function $g:I^2\to \C$. Then we can apply
 ``(3.37)'' of \cite[\mbox{Lemma 3.3}]{K_BCS_I} (or ``(4.21)'' of
 \cite[\mbox{Lemma 4.4}]{K_BCS_II}), ``(4.23)'' of
 \cite[\mbox{Lemma 4.4}]{K_BCS_II}  to obtain
 \eqref{eq_general_estimation_divided_one},
 \eqref{eq_general_estimation_divided_coupled_one} respectively. 
By the same consideration based on
 \eqref{eq_covariance_time_independence},
 \eqref{eq_time_vanishing_property} and the parity of $a$, $b$ we
 see that 
\begin{align*}
&\|E_{a,b}^{(n)}\|_{1,\infty}=1_{a\ge 2m+2}
1_{b\ge  2(n-m)}
\|E_{a,b}^{(n)}\|_{1,\infty},\\
&[E_{a,b}^{(n)},g]_{1,\infty}=1_{a\ge 2m+2}
1_{b\ge  2(n-m)}
[E_{a,b}^{(n)},g]_{1,\infty}
\end{align*}
for any anti-symmetric function $g:I^2\to \C$. Then combination with
 ``(3.36)'' of \cite[\mbox{Lemma 3.3}]{K_BCS_I} 
(or ``(4.20)'' of \cite[\mbox{Lemma 4.4}]{K_BCS_II}), 
``(4.22)'' of \cite[\mbox{Lemma 4.4}]{K_BCS_II} leads to 
\eqref{eq_general_estimation_divided},
 \eqref{eq_general_estimation_divided_coupled} respectively.
\end{proof}
 
\subsection{Double-scale
  integration}\label{subsec_double_scale_integration}

In this subsection we construct a double-scale integration scheme based
on some general properties of a couple of covariances.  
With $c_0\in \R_{\ge 1}$, $A$, $B\in \R_{>0}$ the
covariances $\cC_0$, $\cC_1:I_0^2\to \C$ are assumed to satisfy the following
conditions.
\begin{itemize}
\item 
\begin{align}
\cC_0(\orho\rho\bx s,\oeta\eta \by t)=\cC_0(\orho\rho\bx 0, \oeta\eta
 \by 0),\quad (\forall (\orho,\rho,\bx,s),\ (\oeta,\eta,\by,t)\in I_0).
\label{eq_first_covariance_time_independence}
\end{align}
\item
\begin{align}
\cC_1(\cR_{\beta}(\bX+s))=\cC_1(\bX),\quad \left(\forall \bX\in I_0^2,\
 s\in \frac{1}{h}\Z\right).
\label{eq_covariance_time_translation}
\end{align}
\item 
\begin{align}
&|\det(\<\bu_i,\bw_j\>_{\C^m}\cC_a(X_i,Y_j))_{1\le i,j\le n}|\le
 (c_0(1_{a=0}A+1_{a=1}))^n,\label{eq_double_determinant_bound}\\
&(\forall m,n\in \N,\ \bu_i,\bw_i\in \C^m\text{ with
 }\|\bu_i\|_{\C^m},\|\bw_i\|_{\C^m}\le 1,\ X_i,Y_i\in I_0\
 (i=1,\cdots,n),\notag\\
&\quad a\in \{0,1\}).\notag
\end{align}
\item 
\begin{align}
\|\tilde{\cC}_a\|_{1,\infty}\le c_0B,\quad (\forall a\in
 \{0,1\}).\label{eq_double_norm_bound}
\end{align}
\item 
\begin{align}
\|\tilde{\cC}_a\|_{1,\infty}'\le c_0A,\quad (\forall a\in
 \{0,1\}).\label{eq_double_norm_bound_dash}
\end{align}
\end{itemize}

\noindent
We should think of them as generalizations of the covariances
$C_0$, $C_1$ introduced in Subsection \ref{subsec_covariances}. It
is efficient to define the covariances by abstracting 
the dependency on the
physical parameters at this stage. On the contrary, we explicitly 
define the input Grassmann polynomials to the double-scale
integration process as follows.
\begin{align*}
&V^{0-1,0}(u)(\psi):=\left(\frac{1}{h}\right)^2\sum_{\bX\in
 I^2}V^{0-1,0}_2(u)(\bX)\psi_{\bX},\\
&V^{0-2,0}(u)(\psi):=\left(\frac{1}{h}\right)^4\sum_{\bX,\bY\in
 I^2}V^{0-2,0}_{2,2}(u)(\bX,\bY)\psi_{\bX}\psi_{\bY},
\end{align*}
where the anti-symmetric kernel $V^{0-1,0}_2(u):I^2\to \C$ and the
bi-anti-symmetric kernel $V^{0-2,0}_{2,2}(u):I^2\times I^2\to \C$ are
defined by
\begin{align*}
&V_2^{0-1,0}(u)(\orho_1\rho_1\bx_1s_1\xi_1,\orho_2\rho_2\bx_2s_2\xi_2)\\
&:=-\frac{1}{2}uL^{-d}h1_{(\orho_1,\rho_1,\bx_1,s_1)=(\orho_2,\rho_2,\bx_2,s_2)}1_{\orho_1=1}(1_{(\xi_1,\xi_2)=(1,-1)}-
 1_{(\xi_1,\xi_2)=(-1,1)}),\\
&V_{2,2}^{0-2,0}(u)(\orho_1\rho_1\bx_1s_1\xi_1,\orho_2\rho_2\bx_2s_2\xi_2,
                    \oeta_1\eta_1\by_1t_1\zeta_1,\oeta_2\eta_2\by_2t_2\zeta_2)\\
&:=-\frac{1}{4}uL^{-d}h^2(h1_{s_1=t_1}-\beta^{-1})1_{(\rho_1,\bx_1,s_1,\eta_1,\by_1,t_1)=
(\rho_2,\bx_2,s_2,\eta_2,\by_2,t_2)}\\
&\qquad\cdot \sum_{\s,\tau\in
 \S_2}\sgn(\s)\sgn(\tau)1_{(\orho_{\s(1)},\orho_{\s(2)},\oeta_{\tau(1)},
\oeta_{\tau(2)})=(1,2,2,1)}
1_{(\xi_{\s(1)},\xi_{\s(2)},\zeta_{\tau(1)},
\zeta_{\tau(2)})=(1,-1,1,-1)}.
\end{align*}
Here $u$ is a complex parameter and 
should be considered as 
an extension of the coupling constant $U$. 
Though the definitions seem complicated, they can be
simply rewritten as follows. 
\begin{align}
V^{0-1,0}(u)(\psi)&=-\frac{u}{L^dh}\sum_{(\rho,\bx)\in \cB\times
 \G}\sum_{s\in [0,\beta)_h}\opsi_{1\rho\bx s}\psi_{1\rho\bx s},\notag\\
V^{0-2,0}(u)(\psi)&=-\frac{u}{L^dh}\sum_{(\rho,\bx),(\eta,\by)\in \cB\times
 \G}\sum_{s\in [0,\beta)_h}\opsi_{1\rho\bx s}\psi_{2\rho\bx s}
\opsi_{2\eta\by s}\psi_{1\eta\by s}\label{eq_0_1_2_reduction}\\
&\quad +
\frac{u}{\beta L^dh^2}\sum_{(\rho,\bx),(\eta,\by)\in \cB\times
 \G}\sum_{s,t\in [0,\beta)_h}\opsi_{1\rho\bx s}\psi_{2\rho\bx s}
\opsi_{2\eta\by t}\psi_{1\eta\by t}.\notag
\end{align}

\noindent
We adopt \cite[\mbox{Lemma 3.6}]{K_BCS_II} as the formulation of our
system. We can see from \eqref{eq_0_1_2_reduction} and \cite[\mbox{Lemma
3.6}]{K_BCS_II} that the Grassmann polynomial
$V^{0-1,0}(U)(\psi)+V^{0-2,0}(U)(\psi)$ appears in the Grassmann
integral formulation as the effective interaction. Our first goal in
this subsection is to construct an analytic continuation of the
$\bigwedge_{even}\cV$-valued function 
\begin{align*}
u\mapsto \log\Bigg(
\int e^{V^{0-1,0}(u)(\psi^0+\psi)+V^{0-2,0}(u)(\psi^0+\psi)}d\mu_{\cC_0}(\psi^0)
\Bigg) 
\end{align*}
in a neighborhood of the origin. Let us remark that we integrate
with the time-independent covariance $\cC_0$ as the first step, while
the integration with the time-independent covariance was performed in
the last step of the multi-scale integrations in \cite{K_BCS_I},
\cite{K_BCS_II}. The determinant bound on $\cC_0$ is the main
problematic contribution from the sliced covariances, while the 
$\|\cdot\|_{1,\infty}$-norm
bound on the time-independent covariance was so in \cite{K_BCS_I},
\cite{K_BCS_II}. We integrate with the covariance $\cC_0$ first in order
to remove the main burden on the possible magnitude of
$u$. The output of the integration with $\cC_0$ will be integrated with
$\cC_1$ in the second step. 

It will help us to organize our analysis if we prepare some sets of
$\bigwedge_{even}\cV$-valued functions in advance. For $r\in \R_{>0}$,
set $D(r):=\{z\in \C\ |\ |z|<r\}$. In the following $\alpha$ denotes a
parameter belonging to $\R_{\ge 1}$. Admitting the convention concerning
choice of a norm of $\bigwedge_{even}\cV$ explained in the beginning of 
\cite[\mbox{Subsection 4.4}]{K_BCS_II}, for any domain $D$ of $\C^n$ we
let $C(\overline{D},\bigwedge_{even}\cV)$,
$C^{\o}(D,\bigwedge_{even}\cV)$ denote the set of continuous
maps from $\overline{D}$ to $\bigwedge_{even}\cV$, 
the set of analytic
maps from $D$ to $\bigwedge_{even}\cV$
respectively. Let us also refer to the beginning of
\cite[\mbox{Subsection 4.4}]{K_BCS_II} for the definitions of the norm 
$\|\cdot\|_{1,\infty,r}$ of $C(\overline{D(r)},\C)$ and 
$C(\overline{D(r)},\Map(I^m,\C))$ and the measurement
$[\cdot,\cdot]_{1,\infty,r}$ for a coupling between a function
belonging to $C(\overline{D(r)},\Map(I^m,\C))$ and an anti-symmetric
function on $I^2$. For $r\in \R_{>0}$ we define the subsets $\cQ(r)$,
$\cR(r)$ of $\Map(\overline{D(r)},\bigwedge_{even}\cV)$ as follows.

$f\in \cQ(r)$ if and only if 
\begin{itemize}
\item 
$$
f\in C\Bigg(\overline{D(r)}, \bigwedge_{even}\cV\Bigg)\cap C^{\o}\Bigg(D(r),\bigwedge_{even}\cV\Bigg).
$$
\item 
For any $u\in \overline{D(r)}$ the anti-symmetric kernels $f(u)_m:I^m\to
      \C$ $(m=2,4,\cdots,N)$ satisfy
      \eqref{eq_time_translation_invariance} and 
\begin{align}
&\frac{h}{N}\|f_0\|_{1,\infty,r}\le \alpha^{-1}AB^{-1}L^{-d},\label{eq_grassmann_set_inverse_volume}\\
&\sum_{m=2}^Nc_0^{\frac{m}{2}}\alpha^m\|f_m\|_{1,\infty,r}\le
 (A+1)B^{-1}L^{-d}.\notag
\end{align}
\end{itemize}

$f\in \cR(r)$ if and only if
\begin{itemize}
\item 
$$
f\in C\Bigg(\overline{D(r)}, \bigwedge_{even}\cV\Bigg)\cap C^{\o}\Bigg(D(r),\bigwedge_{even}\cV\Bigg).
$$
\item There exist $f_{p,q}\in C(\overline{D(r)},\Map(I^p\times I^q,\C))$
      $(p,q\in \{2,4,\cdots,N\})$ such that for any $u\in
      \overline{D(r)}$, $p,q\in \{2,4,\cdots,N\}$ $f_{p,q}(u):I^p\times I^q\to
      \C$ is bi-anti-symmetric, satisfies
      \eqref{eq_time_translation_invariance},
      \eqref{eq_time_vanishing_property}, 
\begin{align*}
&f(u)(\psi)=\sum_{p,q=2}^N1_{p,q\in
 2\N}\left(\frac{1}{h}\right)^{p+q}\sum_{\bX\in I^p\atop \bY\in
 I^q}f_{p,q}(u)(\bX,\bY)\psi_{\bX}\psi_{\bY}
\end{align*}
and 
\begin{align}
&\sum_{p,q=2}^Nc_0^{\frac{p+q}{2}}\alpha^{p+q}\|f_{p,q}\|_{1,\infty,r}\le
 B^{-1},\label{eq_grassmann_set_divided}\\
&\sum_{p,q=2}^Nc_0^{\frac{p+q}{2}}\alpha^{p+q}[f_{p,q},g]_{1,\infty,r}\le
 B^{-1}(\|g\|_{1,\infty}'+AB^{-1}\|g\|_{1,\infty})L^{-d}
\label{eq_grassmann_set_divided_coupled}
\end{align}
for any anti-symmetric function $g:I^2\to \C$.
\end{itemize}

Next we arrange the Grassmann polynomials 
\begin{align}
\frac{1}{n!}\left(\frac{d}{dz}\right)^n\log\Bigg(
\int e^{zV^{0-1,0}(u)(\psi^0+\psi)+zV^{0-2,0}(u)(\psi^0+\psi)}d\mu_{\cC_0}(\psi^0)
\Bigg)\Bigg|_{z=0}\label{eq_taylor_one_term}
\end{align}
$(n\in \N)$ in the same way as in \cite[\mbox{Subsection 3.4}]{K_BCS_I}.
One apparent difference is that here we have the covariance $\cC_0$ rather than $\cC_1$. 
The difference in the index of the covariances results in the difference
in the second superscript of the Grassmann polynomials. Let us remark
that here the input polynomials have 0 and the output polynomials have 1
in the second superscript. In 
\cite[\mbox{Subsection 3.4}]{K_BCS_I} the Grassmann polynomials had the
opposite numbers in the second superscript. For $n\in \N$ we define 
$V^{0-1-1,1,(n)}$, $V^{0-1-2,1,(n)}$, $V^{0-2,1,(n)}\in
\Map(\C,\bigwedge_{even}\cV)$ as follows.
\begin{align*}
&V^{0-1-1,1,(n)}(u)(\psi)\\
&:=\frac{1}{n!}Tree(\{1,\cdots,n\},\cC_{0})
\prod_{j=1}^n\Bigg(\sum_{b_j\in\{1,2\}}V^{0-b_j,0}(u)(\psi^j+\psi)\Bigg)\Bigg|_{\psi^{j}=0\atop(\forall
 j\in\{1,\cdots,n\})}1_{\exists
 j(b_j=1)},\notag\\
&V^{0-1-2,1,(n)}(u)(\psi)\notag\\
&:=\frah^{4}\sum_{\bX,\bY\in
 I^2}V_{2,2}^{0-2,0}(u)(\bX,\bY)
 \frac{1}{n!}Tree(\{1,\cdots,n+1\},\cC_{0})\notag\\
&\quad\cdot (\psi^1+\psi)_{\bX}(\psi^2+\psi)_{\bY}
\prod_{j=3}^{n+1}V^{0-2,0}(u)(\psi^j+\psi)\Bigg|_{\psi^{j}=0\atop(\forall
 j\in\{1,\cdots,n+1\})},\notag\\
&V^{0-2,1,(n)}(u)(\psi)\notag\\
&:=\frac{1}{n!}\sum_{m=0}^{n-1}\sum_{(\{s_j\}_{j=1}^{m+1},
 \{t_k\}_{k=1}^{n-m})\in S(n,m)}
\frah^{4}\sum_{\bX,\bY\in
 I^2}V_{2,2}^{0-2,0}(u)(\bX,\bY)\notag\\
&\quad\cdot Tree(\{s_j\}_{j=1}^{m+1},\cC_{0})(\psi^{s_1}+\psi)_{\bX}\prod_{j=2}^{m+1}V^{0-2,0}(u)(\psi^{s_j}+\psi)
\Bigg|_{\psi^{s_j}=0\atop(\forall
 j\in\{1,\cdots,m+1\})}\notag\\
&\quad\cdot Tree(\{t_k\}_{k=1}^{n-m},\cC_{0})(\psi^{t_1}+\psi)_{\bY}\prod_{k=2}^{n-m}V^{0-2,0}(u)(\psi^{t_k}+\psi)
\Bigg|_{\psi^{t_k}=0\atop(\forall
 k\in\{1,\cdots,n-m\})},\notag
\end{align*}
where 
\begin{align*}
S(n,m):=\left\{(\{s_j\}_{j=1}^{m+1},
 \{t_k\}_{k=1}^{n-m})\ \Bigg|\
 \begin{array}{l}1=s_1<s_2<\cdots<s_{m+1}\le n,\\
                 1=t_1<t_2<\cdots<t_{n-m}\le n,\\
                \{s_j\}_{j=2}^{m+1}\cup
		 \{t_k\}_{k=2}^{n-m}=\{2,3,\cdots,n\},\\  
               \{s_j\}_{j=2}^{m+1}\cap \{t_k\}_{k=2}^{n-m}=\emptyset.
\end{array}
\right\}.
\end{align*}
The following equality is structurally same as
\cite[\mbox{(3.56)}]{K_BCS_I}, \cite[\mbox{(4.41)}]{K_BCS_II} and
originates from \cite[\mbox{(3.38)}]{M}, \cite[\mbox{(IV.15)}]{L}.
\begin{align}
&(\text{The Grassmann polynomial
 }\eqref{eq_taylor_one_term})\label{eq_decomposition_technique}\\
&= V^{0-1-1,1,(n)}(u)(\psi)+
 V^{0-1-2,1,(n)}(u)(\psi)+
V^{0-2,1,(n)}(u)(\psi).\notag
\end{align}
Moreover, we set 
\begin{align*}
&V^{0-1-j,1}(u)(\psi):=\sum_{n=1}^{\infty}V^{0-1-j,1,(n)}(u)(\psi),\quad
 (j=1,2),\\
&V^{0-1,1}(u)(\psi):=\sum_{j=1}^2V^{0-1-j,1}(u)(\psi),\quad
V^{0-2,1}(u)(\psi):=\sum_{n=1}^{\infty}V^{0-2,1,(n)}(u)(\psi),
\end{align*} 
if they converge in $\bigwedge_{even}\cV$. Bearing in mind that the
constant $A$ will
be $\beta$-dependent in practice, we want to prove the analyticity of
$u\mapsto V^{0-1,1}(u)$, $u\mapsto V^{0-2,1}(u)$ in an A-independent neighborhood of
the origin. The machinery which
essentially enables us to achieve this goal is the general estimations
summarized in Subsection \ref{subsec_general_estimation}.  
They are applicable in the proof below, mainly because
$V_2^{0-1,0}(u):I^2\to\C$, $V_{2,2}^{0-2,0}(u):I^2\times I^2\to \C$
satisfy \eqref{eq_time_translation_invariance},
$V_{2,2}^{0-2,0}(u)(\cdot)$ satisfies \eqref{eq_time_vanishing_property}
and the covariance $\cC_0$ satisfies
\eqref{eq_first_covariance_time_independence}. 

\begin{lemma}\label{lem_continuation_natural_parameter}
There exists $c\in \R_{>0}$ independent of any parameter such that if
 $\alpha \ge c$, 
\begin{align*}
V^{0-1,1}\in \cQ(c_0^{-2}\alpha^{-5}b^{-1}B^{-1}),\quad 
V^{0-2,1}\in \cR(c_0^{-2}\alpha^{-5}b^{-1}B^{-1}).
\end{align*}
\end{lemma}

\begin{proof}
We set $r:=c_0^{-2}\alpha^{-5}b^{-1}B^{-1}$. Let us
 begin by listing necessary bounds on the input. It follows from the
 definitions that 
\begin{align}
&\|V_2^{0-1,0}\|_{1,\infty,r}\le r L^{-d},\label{eq_norm_upper_two}\\
&\|V_{2,2}^{0-2,0}\|_{1,\infty,r}\le b
 r,\label{eq_norm_upper_double_two}\\
&\|V_4^{0-2,0}\|_{1,\infty,r}\le \|V_{2,2}^{0-2,0}\|_{1,\infty,r}\le
 br,\label{eq_norm_upper_double_four}\\
&[V_{2,2}^{0-2,0},g]_{1,\infty,r}\le
 rL^{-d}(\|g\|_{1,\infty}'+\beta^{-1}\|g\|_{1,\infty})\le 2r
 L^{-d}\|g\|_{1,\infty}'.
\label{eq_norm_upper_double_two_couple}
\end{align}
 
First let us consider $V^{0-1-1,1,(n)}$.
By ``(3.14)'' of \cite[\mbox{Lemma 3.1}]{K_BCS_I} or ``(4.6)'' of 
\cite[\mbox{Lemma 4.1}]{K_BCS_II}, \eqref{eq_double_determinant_bound}
 and \eqref{eq_norm_upper_two}, for $m\in \{0,2,\cdots,N\}$ 
\begin{align*}
\|V_m^{0-1-1,1,(1)}\|_{1,\infty,r}&\le
 \left(\frac{N}{h}\right)^{1_{m=0}}(c_0A)^{\frac{2-m}{2}}rL^{-d}1_{m\le 2}\\
&\le \left(\frac{N}{h}\right)^{1_{m=0}}c_0^{-\frac{m}{2}}
A^{1-\frac{m}{2}}\alpha^{-5}B^{-1}L^{-d}1_{m\le 2},
\end{align*}
where we also used that $c_0^{-1}\le 1$. Moreover, by
 \eqref{eq_general_estimation_direct},
 \eqref{eq_double_determinant_bound}, \eqref{eq_double_norm_bound}, 
\eqref{eq_norm_upper_two} and \eqref{eq_norm_upper_double_four} for any
$n\in \N_{\ge 2}$, $m\in \{0,2,\cdots,N\}$
\begin{align*}
\|V_m^{0-1-1,1,(n)}\|_{1,\infty,r}\le
 &\left(\frac{N}{h}\right)^{1_{m=0}}\sum_{l=1}^n\left(
\begin{array}{c} n \\
l\end{array}\right)(c_0A)^{-n+1-\frac{m}{2}}2^{-2m}(c_0B)^{n-1}\\
&\cdot(2^6c_0ArL^{-d})^l 
(2^{12}c_0^2A^2br)^{n-l}1_{2(n-l)+2\ge m\ge 2(n-l)}.
\end{align*}
Here we remark that when $m=0$, only the term with $l=n$ remains
 in the right-hand side of the above inequality. It follows that
\begin{align}
&\|V_0^{0-1-1,1,(1)}\|_{1,\infty,r}\le
 \frac{N}{h}AB^{-1}L^{-d}\alpha^{-5},\label{eq_0_1_1_1_0}\\
&\|V_0^{0-1-1,1,(n)}\|_{1,\infty,r}\le
 \frac{N}{h}A^{-n+1}B^{n-1}(2^6c_0ArL^{-d})^n\le
 \frac{N}{h}AB^{-1}L^{-d}(2^6\alpha^{-5})^n,\label{eq_0_1_1_n_0}\\
&\sum_{m=2}^Nc_0^{\frac{m}{2}}\alpha^m\|V_m^{0-1-1,1,(1)}\|_{1,\infty,r}\le
c_0\alpha^2r L^{-d}\le 
 \alpha^{-3}B^{-1}L^{-d},\label{eq_0_1_1_1_m}\\
&\sum_{m=2}^Nc_0^{\frac{m}{2}}\alpha^m\|V_m^{0-1-1,1,(n)}\|_{1,\infty,r}\label{eq_0_1_1_n_m}\\
&\le 
c\sum_{l=1}^n\left(\begin{array}{c} n \\ l
		   \end{array}\right)A^{-n+1}B^{n-1}(2^6c_0Ar L^{-d})^l
(2^{12}c_0^2A^2br)^{n-l}2^{-4(n-l)}A^{-(n-l)}\alpha^{2(n-l)}\notag\\
&\quad\cdot (1+A^{-1}\alpha^2)\notag\\
&\le c AB^{-1}(1+A^{-1}\alpha^2)
\sum_{l=1}^n\left(\begin{array}{c} n \\ l
		   \end{array}\right)
(2^6c_0Br L^{-d})^l
(2^8c_0^2B\alpha^2br)^{n-l}\notag\\
&\le c A B^{-1}(1+A^{-1}\alpha^2)\sum_{l=1}^n\left(\begin{array}{c} n \\ l
		   \end{array}\right)(2^6\alpha^{-5}L^{-d})^l(2^{8}\alpha^{-3})^{n-l}\notag\\
&\le c B^{-1}(A+\alpha^2)L^{-d}(2^9\alpha^{-3})^n.\notag
\end{align}

Next let us study $V^{0-1-2,1,(n)}$. We can apply
 \eqref{eq_general_estimation_double_one_field}, 
 \eqref{eq_double_determinant_bound}, \eqref{eq_double_norm_bound_dash}, 
\eqref{eq_norm_upper_double_two_couple} to derive that for $m\in
 \{0,2,\cdots,N\}$ 
\begin{align*}
&\|V_m^{0-1-2,1,(1)}\|_{1,\infty,r}\le c(c_0A)^{-2}(c_0A)^{2}rL^{-d}c_0A
 1_{m=2}\le c c_0^{-1}\alpha^{-5}AB^{-1}L^{-d}1_{m=2}.
\end{align*}
For $n\in \N_{\ge 2}$ we use \eqref{eq_general_estimation_double}
 instead of \eqref{eq_general_estimation_double_one_field} and
 \eqref{eq_norm_upper_double_four} together with
 \eqref{eq_norm_upper_double_two_couple} to derive that
\begin{align*}
&\|V_m^{0-1-2,1,(n)}\|_{1,\infty,r}\\
&\le c(c_0A)^{-n-\frac{m}{2}}2^{-2m}
(c_0B)^{n-1}(c_0A)^2r L^{-d}c_0A(2^{12}(c_0A)^2br)^{n-1}1_{m=2n}\\
&\le 
c c_0^{-\frac{m}{2}}AB^{-1}L^{-d}(2^8\alpha^{-5})^n1_{m=2n}.
\end{align*}
Thus
\begin{align}
&\sum_{m=2}^Nc_0^{\frac{m}{2}}\alpha^m\|V_m^{0-1-2,1,(1)}\|_{1,\infty,r}\le
 c\alpha^{-3}AB^{-1}L^{-d},\label{eq_0_1_2_1_m}\\
&\sum_{m=2}^Nc_0^{\frac{m}{2}}\alpha^m\|V_m^{0-1-2,1,(n)}\|_{1,\infty,r}\le
 cAB^{-1}L^{-d}(2^{8}\alpha^{-3})^n.\label{eq_0_1_2_n_m}
\end{align}

Assume that $\alpha^3\ge 2^{10}$. Then by \eqref{eq_0_1_1_1_0},
 \eqref{eq_0_1_1_n_0}, \eqref{eq_0_1_1_1_m}, \eqref{eq_0_1_1_n_m},
\eqref{eq_0_1_2_1_m} and \eqref{eq_0_1_2_n_m}
\begin{align*}
&\frac{h}{N}\sum_{n=1}^{\infty}\|V_0^{0-1-1,1,(n)}\|_{1,\infty,r}\le c
 \alpha^{-5}AB^{-1}L^{-d},\\
&\sum_{m=2}^Nc_0^{\frac{m}{2}}\alpha^m\sum_{n=1}^{\infty}(\|V_m^{0-1-1,1,(n)}\|_{1,\infty,r}+
 \|V_m^{0-1-2,1,(n)}\|_{1,\infty,r})\le c\alpha^{-3}(A+1)B^{-1}L^{-d}.
\end{align*}
These uniform convergent properties imply the well-definedness of
 $V^{0-1,1}$ and the claimed regularity with $u$. It follows from the
 statements of \cite[\mbox{Lemma 3.1}]{K_BCS_I} (or \cite[\mbox{Lemma
 4.1}]{K_BCS_II}), Lemma \ref{lem_general_estimation_direct}, 
Lemma \ref{lem_general_estimation_double} that the kernels of
 $V^{0-1,1}$ satisfy \eqref{eq_time_translation_invariance}. Moreover,
 the above inequalities ensure that if $\alpha \ge c$, $V^{0-1,1}$
 satisfies \eqref{eq_grassmann_set_inverse_volume}. Therefore,
 $V^{0-1,1}\in\cQ(r)$ on the assumption $\alpha\ge c$.

Let us treat $V^{0-2,1,(n)}$. By Lemma
 \ref{lem_general_estimation_divided}
(or more originally by \cite[\mbox{Lemma 3.3}]{K_BCS_I}, \cite[\mbox{Lemma
 4.4}]{K_BCS_II}) there are bi-anti-symmetric functions
 $V_{a,b}^{0-2,1,(n)}(u):I^a\times I^b\to \C$ $(n\in \N,$ $a,b\in
 \{2,4,\cdots,N\}$, $u\in \C)$ satisfying
 \eqref{eq_time_translation_invariance},
 \eqref{eq_time_vanishing_property} such that 
\begin{align*}
&V^{0-2,1,(n)}(u)(\psi)=\sum_{a,b=2}^N1_{a,b\in
 2\N}\left(\frac{1}{h}\right)^{a+b}\sum_{\bX\in I^a\atop \bY\in
 I^b}V_{a,b}^{0-2,1,(n)}(u)(\bX,\bY)\psi_{\bX}\psi_{\bY},\\
&(\forall
 n\in \N,\ u\in\C).
\end{align*}
By \eqref{eq_general_estimation_divided_single} and
 \eqref{eq_norm_upper_double_two}, for $a,b\in \{2,4,\cdots,N\}$ 
\begin{align*}
\|V_{a,b}^{0-2,1,(1)}\|_{1,\infty,r}\le
 \|V_{a,b}^{0-2,0}\|_{1,\infty,r}1_{a=b=2}\le
 c_0^{-2}\alpha^{-5}B^{-1}1_{a=b=2},
\end{align*}
and thus
\begin{align}
\sum_{a,b=2}^N1_{a,b\in
 2\N}c_0^{\frac{a+b}{2}}\alpha^{a+b}\|V_{a,b}^{0-2,1,(1)}\|_{1,\infty,r}\le
 \alpha^{-1}B^{-1}.\label{eq_0_2_1}
\end{align}
For $n\in \N_{\ge 2}$, $a,b\in \{2,4,\cdots,N\}$ the inequalities
 \eqref{eq_general_estimation_divided},
 \eqref{eq_double_determinant_bound}, \eqref{eq_double_norm_bound},
 \eqref{eq_norm_upper_double_two} and \eqref{eq_norm_upper_double_four}
 yield that 
\begin{align*}
&\|V_{a,b}^{0-2,1,(n)}\|_{1,\infty,r}\\
&\le
 \frac{1}{n!}\sum_{m=0}^{n-1}\sum_{(\{s_j\}_{j=1}^{m+1},\{t_k\}_{k=1}^{n-m})\in S(n,m)}\\
&\quad\cdot (1_{m\neq 0}(m-1)!+1_{m=0})(1_{m\neq
 n-1}(n-m-2)!+1_{m=n-1})\\
&\quad\cdot 2^{-2a-2b}(c_0A)^{-n+1-\frac{1}{2}(a+b)}(c_0B)^{n-1}(2^{12}c_0^2A^2br)^n1_{a=2m+2}1_{b=2(n-m)}\\
&=\frac{1}{n!}\sum_{m=0}^{n-1}\left(
\begin{array}{c}n-1 \\ m \end{array}
\right)(1_{m\neq 0}(m-1)!+1_{m=0})(1_{m\neq
 n-1}(n-m-2)!+1_{m=n-1})\\
&\quad\cdot 2^{8n-4}c_0^{-\frac{a+b}{2}}B^{-1}\alpha^{-5n}1_{a=2m+2}1_{b=2(n-m)}.
\end{align*}
Therefore, 
\begin{align}
\sum_{a,b=2}^N1_{a,b\in
 2\N}c_0^{\frac{a+b}{2}}\alpha^{a+b}\|V_{a,b}^{0-2,1,(n)}\|_{1,\infty,
 r}\le c \alpha^2B^{-1}(2^{8}\alpha^{-3})^n.\label{eq_0_2_n}
\end{align}
On the other hand, let us take an anti-symmetric function $g:I^2\to
 \C$. By \eqref{eq_general_estimation_divided_single_coupled} and
 \eqref{eq_norm_upper_double_two_couple}, for any $a,b\in
 \{2,4,\cdots,N\}$ 
\begin{align*}
&[V_{a,b}^{0-2,1,(1)},g]_{1,\infty,r}\le
 [V_{2,2}^{0-2,0},g]_{1,\infty,r}1_{a=b=2}\le 2 c_0^{-2}
 \alpha^{-5}B^{-1}L^{-d}\|g\|_{1,\infty}'1_{a=b=2}.
\end{align*}
Thus 
\begin{align}
\sum_{a,b=2}^N1_{a,b\in
 2\N}c_0^{\frac{a+b}{2}}\alpha^{a+b}[V_{a,b}^{0-2,1,(1)},g]_{1,\infty,r}\le
 2\alpha^{-1}B^{-1}L^{-d}\|g\|_{1,\infty}'.\label{eq_0_2_1_coupled}
\end{align}
For $n\in \N_{\ge 2}$, $a,b\in \{2,4,\cdots,N\}$ we can apply
 \eqref{eq_general_estimation_divided_coupled},
 \eqref{eq_double_determinant_bound}, \eqref{eq_double_norm_bound},
 \eqref{eq_double_norm_bound_dash}, \eqref{eq_norm_upper_double_four}
 and \eqref{eq_norm_upper_double_two_couple} to deduce that 
\begin{align*}
&[V_{a,b}^{0-2,1,(n)},g]_{1,\infty,r}\\
&\le
 \frac{c}{n!}\sum_{m=0}^{n-1}\sum_{(\{s_j\}_{j=1}^{m+1},\{t_k\}_{k=1}^{n-m})\in S(n,m)}\\
&\quad\cdot (1_{m\neq 0}(m-1)!+1_{m=0})(1_{m\neq
 n-1}(n-m-2)!+1_{m=n-1})\\
&\quad\cdot 2^{-2a-2b}(c_0A)^{-n+1-\frac{1}{2}(a+b)}(c_0B)^{n-2}
c_0^2A^2(rL^{-d}\|g\|_{1,\infty}'c_0B+ rL^{-d}c_0A\|g\|_{1,\infty})\\
&\quad\cdot (2^{12}c_0^2A^2br)^{n-1}1_{a=2m+2}1_{b=2(n-m)}\\
&\le \frac{c}{n!}\sum_{m=0}^{n-1}\left(
\begin{array}{c}n-1 \\ m \end{array}
\right)(1_{m\neq 0}(m-1)!+1_{m=0})(1_{m\neq
 n-1}(n-m-2)!+1_{m=n-1})\\
&\quad\cdot 2^{8n}c_0^{-\frac{a+b}{2}}B^{-1}\alpha^{-5n}L^{-d}
(\|g\|_{1,\infty}'+ AB^{-1}\|g\|_{1,\infty})
1_{a=2m+2}1_{b=2(n-m)},
\end{align*}
and thus
\begin{align}
&\sum_{a,b=2}^N1_{a,b\in 2\N}c_0^{\frac{a+b}{2}}\alpha^{a+b}
[V_{a,b}^{0-2,1,(n)},g]_{1,\infty,r}\label{eq_0_2_n_coupled}\\
&\le c\alpha^2B^{-1}L^{-d}(\|g\|_{1,\infty}'+
 AB^{-1}\|g\|_{1,\infty})(2^{8}\alpha^{-3})^n.\notag
\end{align}
Assume that $\alpha^3\ge 2^9$. By summing up \eqref{eq_0_2_1},
 \eqref{eq_0_2_n}, \eqref{eq_0_2_1_coupled}, \eqref{eq_0_2_n_coupled} we
 observe that 
\begin{align}
&\sum_{a,b=2}^N1_{a,b\in 2\N}c_0^{\frac{a+b}{2}}\alpha^{a+b}
\sum_{n=1}^{\infty}\|V_{a,b}^{0-2,1,(n)}\|_{1,\infty,r}\le (\alpha^{-1}+c\alpha^{-4})B^{-1},\label{eq_0_2_uniform_convergence}\\
&\sum_{a,b=2}^N1_{a,b\in 2\N}c_0^{\frac{a+b}{2}}\alpha^{a+b}
\sum_{n=1}^{\infty}[V_{a,b}^{0-2,1,(n)},g]_{1,\infty,r}\notag\\
&\le 
(2\alpha^{-1}+c\alpha^{-4})B^{-1}(\|g\|_{1,\infty}'+AB^{-1}\|g\|_{1,\infty})L^{-d}.\notag
\end{align}
The uniform convergence property \eqref{eq_0_2_uniform_convergence}
 ensures the well-definedness of $V^{0-2,1}$ and the claimed  regularity
 with $u$. On the assumption $\alpha\ge c$ we can conclude from
 the above inequalities that $V^{0-2,1}\in \cR(r)$.
\end{proof}

Lemma \ref{lem_continuation_natural_parameter} will support us in the
 derivation of the free energy density. In order to derive the thermal
expectations, on the other hand, we need to add an artificial
term to the input Grassmann polynomials and construct the
double-scale integration process by clarifying how the artificial term
 affects the output. 
Let us fix $(\hrho,\hbx)$, $(\heta,\hby)\in\cB\times \G_{\infty}$, which
are to represent the sites where the Cooper pair density is measured. 
The artificial Grassmann polynomial $V^{1,0}(\bla)(\psi)\in
\bigwedge_{even}\cV$ parameterized by the artificial parameter $\bla
=(\la_1,\la_2)\in \C^2$ is defined as follows.
\begin{align*}
V^{1,0}(\bla)(\psi):=\sum_{m\in
 \{2,4\}}\left(\frac{1}{h}\right)^m\sum_{\bX\in
 I^m}V_m^{1,0}(\bla)(\bX)\psi_{\bX}
\end{align*}
with the anti-symmetric kernels $V_m^{1,0}(\bla):I^m\to \C$ $(m=2,4)$
defined by 
\begin{align*}
&V_{2}^{1,0}(\bla)(\orho_1\rho_1\bx_1s_1\xi_1,\orho_2\rho_2\bx_2s_2\xi_2)\\
&:=-\frac{h}{2}1_{s_1=s_2}\sum_{\s\in \S_2}\sgn(\s)\Big(
\la_11_{((\orho_{\s(1)},\rho_{\s(1)}, \bx_{\s(1)}, \xi_{\s(1)}),
(\orho_{\s(2)},\rho_{\s(2)}, \bx_{\s(2)},
 \xi_{\s(2)}))\atop=((1,\hat{\rho},r_L(\hbx),1),(2,\hat{\rho},r_L(\hbx),-1))}\\
&\qquad\qquad\qquad\qquad+\la_21_{((\orho_{\s(1)},\rho_{\s(1)}, \bx_{\s(1)}, \xi_{\s(1)}),
(\orho_{\s(2)},\rho_{\s(2)}, \bx_{\s(2)},
 \xi_{\s(2)}))\atop=((1,\hat{\rho},r_L(\hbx),1),(1,\hat{\rho},r_L(\hbx),-1))}1_{(\hat{\rho},r_L(\hbx))=(\hat{\eta},r_L(\hby))}\Big),\\
&V_{4}^{1,0}(\bla)(\orho_1\rho_1\bx_1s_1\xi_1,\orho_2\rho_2\bx_2s_2\xi_2,
                  \orho_3\rho_3\bx_3s_3\xi_3,\orho_4\rho_4\bx_4s_4\xi_4)\\
&:=-\frac{h^3}{4!}\la_21_{s_1=s_2=s_3=s_4}\sum_{\s\in \S_4}\sgn(\s)\\
&\qquad\cdot 
1_{((\orho_{\s(1)},\rho_{\s(1)}, \bx_{\s(1)}, \xi_{\s(1)}),
(\orho_{\s(2)},\rho_{\s(2)}, \bx_{\s(2)}, \xi_{\s(2)}),
(\orho_{\s(3)},\rho_{\s(3)}, \bx_{\s(3)}, \xi_{\s(3)}),
(\orho_{\s(4)},\rho_{\s(4)}, \bx_{\s(4)}, \xi_{\s(4)}))\atop
=((1,\hat{\rho},r_L(\hbx),1),(2,\hat{\rho},r_L(\hbx),-1),
  (2,\hat{\eta},r_L(\hby),1),(1,\hat{\eta},r_L(\hby),-1))}.
\end{align*}
Remind us that the map $r_L:\G_{\infty}\to \G$ was defined just before
the statement of Theorem \ref{thm_infinite_volume_limit} in Subsection
\ref{subsec_main_results}. We can confirm that
\begin{align}
&V^{1,0}(\bla)(\psi)\label{eq_1_reduction}\\
&=-\frac{\la_1}{h}\sum_{s\in
 [0,\beta)_h}\opsi_{1\hat{\rho}r_L(\hbx)s}\psi_{2\hat{\rho}r_L(\hbx)s}
-1_{(\hat{\rho},r_L(\hbx))=(\hat{\eta},r_L(\hby))}
\frac{\la_2}{h}\sum_{s\in
 [0,\beta)_h}\opsi_{1\hat{\rho}r_L(\hbx)s}\psi_{1\hat{\rho}r_L(\hbx)s}\notag\\
&\quad-\frac{\la_2}{h}\sum_{s\in
 [0,\beta)_h}\opsi_{1\hat{\rho}r_L(\hbx)s}\psi_{2\hat{\rho}r_L(\hbx)s}
\opsi_{2\hat{\eta}r_L(\hby)s}\psi_{1\hat{\eta}r_L(\hby)s}.\notag
\end{align}
As the second goal of this subsection we construct an analytic
 continuation of the $\bigwedge_{even}\cV$-valued function
\begin{align*}
(u,\bla)\mapsto \log\Bigg(
\int e^{V^{0-1,0}(u)(\psi^0+\psi)+V^{0-2,0}(u)(\psi^0+\psi)+V^{1,0}(\bla)(\psi^0+\psi)}d\mu_{\cC_0}(\psi^0)
\Bigg) 
\end{align*}
in a neighborhood of the origin. The mission is seemingly close to that
in \cite[\mbox{Subsection 3.5}]{K_BCS_I}. However, the fact that the covariance
is independent of the time variables makes non-trivial
differences in analysis. Let us introduce sets of
$\bigwedge_{even}\cV$-valued functions in order to concisely
describe properties of the output of this single-scale integration. 
Let $r,r'\in \R_{>0}$. We use the norm $\|\cdot\|_{1,r,r'}$ on 
$C(\overline{D(r)}\times\overline{D(r')}^2,\C)$ and
$C(\overline{D(r)}\times\overline{D(r')}^2,\Map(I^m,\C))$ and 
the measurement $[\cdot,\cdot]_{1,r,r'}$ for a coupling between a
function belonging to
$C(\overline{D(r)}\times\overline{D(r')}^2,\Map(I^m,\C))$
and an anti-symmetric function on $I^2$. The definition of these notions
is found in \cite[\mbox{Subsection 4.5}]{K_BCS_II}. We define the
subset $\cQ'(r,r')$ of
$\Map(\overline{D(r)}\times\C^2,\bigwedge_{even}\cV)$ as follows.

$f\in \cQ'(r,r')$ if and only if 
\begin{itemize}
\item
\begin{align*}
f\in C\Bigg(\overline{D(r)}\times \C^2,\bigwedge_{even}\cV\Bigg)\cap 
    C^{\o}\Bigg(D(r)\times \C^2,\bigwedge_{even}\cV\Bigg).
\end{align*}
\item For any $u\in \overline{D(r)}$, $\bla\mapsto f(u,\bla)(\psi)$
     $:\C^2\to \bigwedge_{even}\cV$ is linear. 
\item For any $(u,\bla)\in \overline{D(r)}\times\C^2$ the anti-symmetric
     kernels $f(u,\bla)_m$ $:I^m\to \C$ $(m=2,4,\cdots,N)$ satisfy
     \eqref{eq_time_translation_invariance} and 
\begin{align}
\|f_0\|_{1,r,r'}\le
 \alpha^{-1}L^{-d},\quad
\sum_{m=2}^Nc_0^{\frac{m}{2}}\alpha^m\|f_m\|_{1,r,r'}\le L^{-d}.\label{eq_grassmann_set_inverse_volume_one}
\end{align}
\end{itemize}

We also need a set of $\bigwedge_{even}\cV$-valued functions with
bi-anti-symmetric kernels. Let us define the set $\cR'(r,r')$ as
follows.

$f\in \cR'(r,r')$ if and only if 
\begin{itemize}
\item 
\begin{align*}
f\in C\Bigg(\overline{D(r)}\times \C^2,\bigwedge_{even}\cV\Bigg)\cap 
    C^{\o}\Bigg(D(r)\times \C^2,\bigwedge_{even}\cV\Bigg).
\end{align*}
\item
For any $u\in \overline{D(r)}$, $\bla\mapsto f(u,\bla)(\psi)$ $:\C^2\to
      \bigwedge_{even}\cV$ is linear.
\item There exist $f_{p,q}\in C(\overline{D(r)}\times
      \C^2,\Map(I^p\times I^q,\C))$ $(p,q=2,4,\cdots,N)$ such that for
      any $(u,\bla)\in \overline{D(r)}\times \C^2$, $p,q\in
      \{2,4,\cdots,N\}$ $f_{p,q}(u,\bla):I^p\times I^q\to\C$ is
      bi-anti-symmetric, satisfies
      \eqref{eq_time_translation_invariance},
      \eqref{eq_time_vanishing_property}, 
\begin{align*}
f(u,\bla)(\psi)=\sum_{p,q=2}^N1_{p,q\in
 2\N}\left(\frac{1}{h}\right)^{p+q}\sum_{\bX\in I^p\atop \bY\in
 I^q}f_{p,q}(u,\bla)(\bX,\bY)\psi_{\bX}\psi_{\bY}
\end{align*}
and 
\begin{align}
&\sum_{p,q=2}^Nc_0^{\frac{p+q}{2}}\alpha^{p+q}\|f_{p,q}\|_{1,r,r'}\le
 1,\label{eq_grassmann_set_divided_one}\\
&\sum_{p,q=2}^Nc_0^{\frac{p+q}{2}}\alpha^{p+q}[f_{p,q},g]_{1,r,r'}\le 
(\|g\|_{1,\infty}'+AB^{-1}\|g\|_{1,\infty})L^{-d}\label{eq_grassmann_set_divided_one_coupled}
\end{align}
for any anti-symmetric function $g:I^2\to \C$.
\end{itemize}

We must prepare a set which can contain the direct descent from
$V^{1,0}$. The definition is as below. 

$f\in \cS(r,r')$ if and only if 
\begin{itemize}
\item 
\begin{align*}
f\in C\Bigg(\overline{D(r)}\times \C^2,\bigwedge_{even}\cV\Bigg)\cap 
    C^{\o}\Bigg(D(r)\times \C^2,\bigwedge_{even}\cV\Bigg).
\end{align*}
\item For any $u\in \overline{D(r)}$, $\bla\mapsto f(u,\bla)(\psi)$ $:\C^2\to
      \bigwedge_{even}\cV$ is linear.
\item For any $(u,\bla)\in \overline{D(r)}\times\C^2$ the anti-symmetric
     kernels $f(u,\bla)_m$ $:I^m\to \C$ $(m=2,4,\cdots,N)$ satisfy
     \eqref{eq_time_translation_invariance} and 
\begin{align}
\|f_0\|_{1,r,r'}\le
 \alpha^{-1},\quad
\sum_{m=2}^Nc_0^{\frac{m}{2}}\alpha^m\|f_m\|_{1,r,r'}\le 1.
\label{eq_grassmann_set_descendant}
\end{align}
\end{itemize}

Finally we define a set of $\bigwedge_{even}\cV$-valued functions
depending on $\bla$ at least quadratically. 

$f\in \cW(r,r')$ if and only if 
\begin{itemize}
\item 
\begin{align*}
f\in C\Bigg(\overline{D(r)}\times \overline{D(r')}^2,\bigwedge_{even}\cV\Bigg)\cap 
    C^{\o}\Bigg(D(r)\times D(r')^2,\bigwedge_{even}\cV\Bigg).
\end{align*}
\item For any $u\in D(r)$, $j\in \{1,2\}$
      $f(u,\b0)(\psi)=\frac{\partial}{\partial \la_j}f(u,\b0)(\psi)=0$.
\item For any $(u,\bla)\in \overline{D(r)}\times\overline{D(r')}^2$ the anti-symmetric
     kernels $f(u,\bla)_m$ $:I^m\to \C$ $(m=2,4,\cdots,N)$ satisfy
     \eqref{eq_time_translation_invariance} and 
\begin{align}
\|f_0\|_{1,r,r'}\le
 \alpha^{-1},\quad
\sum_{m=2}^Nc_0^{\frac{m}{2}}\alpha^m\|f_m\|_{1,r,r'}\le 1.
\label{eq_grassmann_set_quadratic}
\end{align}
\end{itemize}

Let us organize the Grassmann polynomials 
\begin{align}
\frac{1}{n!}\left(\frac{d}{dz}\right)^n
\log\Bigg(
\int e^{zV^{0-1,0}(u)(\psi^0+\psi)+zV^{0-2,0}(u)(\psi^0+\psi)+zV^{1,0}(\bla)(\psi^0+\psi)}d\mu_{\cC_0}(\psi^0)
\Bigg)\Bigg|_{z=0}\label{eq_taylor_one_term_artificial} 
\end{align}
in the same way as in \cite[\mbox{Subsection 3.5}]{K_BCS_I}.
The only difference from the previous work is that
here the second superscript of the input polynomials is 0 and that of
the output polynomials is 1. This is in accordance with
the index of the covariances. Define 
$V^{0,0}$, $V^{0,1,(n)}\in \Map(\C,\bigwedge_{even}\cV)$ $(n\in \N)$, 
$V^{1-3,1}\in \Map(\C\times \C^2,\bigwedge_{even}\cV)$ by
\begin{align*}
&V^{0,0}(u)(\psi):=V^{0-1,0}(u)(\psi)+V^{0-2,0}(u)(\psi),\\ 
&V^{0,1,(n)}(u)(\psi):=\frac{1}{n!}Tree(\{1,\cdots,n\},\cC_0)\prod_{j=1}^nV^{0,0}(u)(\psi^j+\psi)\Bigg|_{\psi^j=0\atop
 (\forall j\in \{1,\cdots,n\})},\\
&V^{1-3,1}(u,\bla)(\psi):=Tree(\{1\},\cC_0)V^{1,0}(\bla)(\psi^1+\psi)\Big|_{\psi^1=0}.
\end{align*}
Apparently $V^{1-3,1}$ is independent of $u$. However, by defining as
if it depends on $(u,\bla)$ we can estimate $V^{1-3,1}$ with the norm
$\|\cdot\|_{1,r,r'}$. This saves us introducing another norm. For $n\in \N_{\ge 2}$ we define $V^{1-1-1,1,(n)}$, $V^{1-1-2,1,(n)}$,
$V^{1-2,1,(n)}$, $V^{2,1,(n)}\in \Map(\C\times\C^2,\bigwedge_{even}\cV)$
as follows.  
\begin{align*}
&V^{1-1-1,1,(n)}(u,\bla)(\psi)\\
&:=\frac{1}{(n-1)!}Tree(\{1,\cdots,n\},\cC_{0})\\
&\quad\cdot \prod_{j=1}^{n-1}\Bigg(\sum_{b_j\in\{1,2\}}V^{0-b_j,0}(u)(\psi^j+\psi)\Bigg)
V^{1,0}(\bla)(\psi^n+\psi)\Bigg|_{\psi^{j}=0\atop(\forall
 j\in\{1,\cdots,n\})}
1_{\exists
 j(b_j=1)},\notag\\
&V^{1-1-2,1,(n)}(u,\bla)(\psi)\notag\\
&:=\frah^{4}\sum_{\bX,\bY\in
 I^2}V_{2,2}^{0-2,0}(u)(\bX,\bY)
 \frac{1}{(n-1)!}Tree(\{1,\cdots,n+1\},\cC_{0})\notag\\
&\quad\cdot (\psi^1+\psi)_{\bX}(\psi^2+\psi)_{\bY}
\prod_{j=3}^{n}V^{0-2,0}(u)(\psi^j+\psi)V^{1,0}(\bla)(\psi^{n+1}+\psi)
\Bigg|_{\psi^{j}=0\atop(\forall
 j\in\{1,\cdots,n+1\})},\notag\\
&V^{1-2,1,(n)}(u,\bla)(\psi)\notag\\
&:=\frac{1}{(n-1)!}\sum_{m=0}^{n-1}\sum_{(\{s_j\}_{j=1}^{m+1},
 \{t_k\}_{k=1}^{n-m})\in S(n,m)}
\frah^{4}\sum_{\bX,\bY\in
 I^2}V_{2,2}^{0-2,0}(u)(\bX,\bY)\notag\\
&\quad\cdot
 Tree(\{s_j\}_{j=1}^{m+1},\cC_{0})(\psi^{s_1}+\psi)_{\bX}\\
&\quad\cdot\prod_{j=2}^{m+1}
(1_{s_j\neq n}V^{0-2,0}(u)(\psi^{s_j}+\psi)+1_{s_j= n}V^{1,0}(\bla)(\psi^{s_j}+\psi))\Bigg|_{\psi^{s_j}=0\atop(\forall
 j\in\{1,\cdots,m+1\})}\notag\\
&\quad\cdot
 Tree(\{t_k\}_{k=1}^{n-m},\cC_{0})(\psi^{t_1}+\psi)_{\bY}\\
&\quad\cdot \prod_{k=2}^{n-m}
(1_{t_k\neq n}V^{0-2,0}(u)(\psi^{t_k}+\psi)+1_{t_k=n}V^{1,0}(\bla)(\psi^{t_k}+\psi))\Bigg|_{\psi^{t_k}=0\atop(\forall
 k\in\{1,\cdots,n-m\})},\notag\\
&V^{2,1,(n)}(u,\bla)(\psi)\\
&:=\frac{1}{n!}Tree(\{1,\cdots,n\},\cC_{0})
\prod_{j=1}^{n}\Bigg(\sum_{b_j\in\{0,1\}}V^{b_j,0}(\psi^j+\psi)\Bigg)
\Bigg|_{\psi^{j}=0\atop(\forall
 j\in\{1,\cdots,n\})}
1_{\sum_{j=1}^nb_j\ge 2}.
\end{align*}
Then, the following equality holds.
\begin{align*}
&(\text{The Grassmann polynomial }\eqref{eq_taylor_one_term_artificial})\\
&=V^{0,1,(n)}(u)(\psi)+1_{n=1}V^{1-3,1}(u,\bla)(\psi)\\
&\quad +1_{n\ge
 2}(V^{1-1-1,1,(n)}(u,\bla)(\psi)+V^{1-1-2,1,(n)}(u,\bla)(\psi)+V^{1-2,1,(n)}(u,\bla)(\psi)\\
&\qquad\qquad\ +V^{2,1,(n)}(u,\bla)(\psi)).
\end{align*}
We should remark that the above decomposition is essentially same as
that presented in \cite[\mbox{Subsection 3.5}]{K_BCS_I}.
Assuming their convergence, we set
\begin{align*}
&V^{0,1}(u)(\psi):=\sum_{n=1}^{\infty}V^{0,1,(n)}(u)(\psi),\\
&V^{1-1-j,1}(u,\bla)(\psi):=\sum_{n=2}^{\infty}V^{1-1-j,1,(n)}(u,\bla)(\psi),\quad(\forall
 j\in \{1,2\}),\\
&V^{1-1,1}(u,\bla)(\psi):=\sum_{j=1}^{2}V^{1-1-j,1}(u,\bla)(\psi),\\
&V^{1-2,1}(u,\bla)(\psi):=\sum_{n=2}^{\infty}V^{1-2,1,(n)}(u,\bla)(\psi),\quad
V^{2,1}(u,\bla)(\psi):=\sum_{n=2}^{\infty}V^{2,1,(n)}(u,\bla)(\psi).
\end{align*}
We want to prove that these $\bigwedge_{even}\cV$-valued functions are
analytic with $(u,\bla)$ in a neighborhood of the origin. In particular
the analyticity with $u$ must be ensured independently of $A$. We have
developed the general estimates \eqref{eq_general_estimation_direct_one}, 
\eqref{eq_general_estimation_double_one_norm},
\eqref{eq_general_estimation_divided_one},
\eqref{eq_general_estimation_divided_coupled_one} for this particular
purpose.

\begin{lemma}\label{lem_continuation_artificial_parameter}
There exists $c\in \R_{>0}$ independent of any parameter such that if
 $\alpha \ge c$, 
\begin{align*}
V^{1-1,1}\in \cQ'(r,r'),\quad V^{1-2,1}\in \cR'(r,r'),\quad V^{1-3,1}\in
 \cS(r,r'),\quad V^{2,1}\in \cW(r,r')
\end{align*}
with $r:=c_0^{-2}\alpha^{-5}b^{-1}B^{-1}$,
 $r':=(A+1)^{-2}(B+1)^{-1}(\beta+1)^{-1}c_0^{-2}\alpha^{-5}$.
\end{lemma}

\begin{proof}
We will repeatedly use the following inequalities, which can be directly
 derived from the definitions. 
\begin{align}
&\|V_2^{1,0}\|_{1,r,r'}\le 2\beta r',\label{eq_norm_upper_field}\\
&\sup_{\bla \in \overline{D(r')}^2}\|V_2^{1,0}(\bla)\|_{1,\infty}\le
 2r',\label{eq_norm_upper_field_infinity}\\
&\|V_4^{1,0}\|_{1,r,r'}\le \beta r',\label{eq_norm_upper_field_four}\\
&\sup_{\bla\in \overline{D(r')}^2}\|V_4^{1,0}(\bla)\|_{1,\infty}\le
 r'.\label{eq_norm_upper_field_four_infinity}
\end{align}

First let us summarize properties of $V^{1-3,1}$. By ``(4.7)'' of 
\cite[\mbox{Lemma 4.1}]{K_BCS_II} (or ``(3.15)'' of 
\cite[\mbox{Lemma 3.1}]{K_BCS_I}), \eqref{eq_double_determinant_bound}, 
\eqref{eq_norm_upper_field} and \eqref{eq_norm_upper_field_four}
\begin{align}
&\|V_0^{1-3,1}\|_{1,r,r'}\le
 c_0A\|V_2^{1,0}\|_{1,r,r'}+(c_0A)^2\|V_4^{1,0}\|_{1,r,r'}\le
 c\alpha^{-5},\label{eq_1_3_0}\\
&\|V_2^{1-3,1}\|_{1,r,r'}\le
 \|V_2^{1,0}\|_{1,r,r'}+c c_0A\|V_4^{1,0}\|_{1,r,r'}\le
 c(1+c_0A)\beta r'.\label{eq_1_3_2}
\end{align}
Since $V_4^{1-3,1}=V_4^{1,0}$, we can derive from
 \eqref{eq_norm_upper_field_four}, \eqref{eq_1_3_2} that 
\begin{align}
&\sum_{m=2}^Nc_0^{\frac{m}{2}}\alpha^m\|V_m^{1-3,1}\|_{1,r,r
'}\le c c_0 \alpha^2(1+c_0A)\beta r'+c_0^2\alpha^4\beta r'\le
 c\alpha^{-1}.\label{eq_1_3}
\end{align}
One part of the claims of \cite[\mbox{Lemma 4.1}]{K_BCS_II} or
 \cite[\mbox{Lemma 3.1}]{K_BCS_I} implies that the kernels of
 $V^{1-3,1}$ satisfy \eqref{eq_time_translation_invariance}. The
 linearity with $\bla\in \C^2$ is clear from the definition. Therefore
 we can conclude from \eqref{eq_1_3_0}, \eqref{eq_1_3} that if
 $\alpha\ge c$, $V^{1-3,1}\in \cS(r,r')$. 

Next let us consider $V^{1-1-1,1,(n)}$ $(n\in \N_{\ge 2})$. One can
 rewrite the defining equality as follows. 
\begin{align*}
&V^{1-1-1,1,(n)}(u,\bla)(\psi)\\
&=\frac{1}{(n-1)!}\sum_{l=2}^{n}\left(\begin{array}{c} n-1 \\
					l-1\end{array}\right)
Tree(\{1,\cdots,n\},\cC_{0})\\
&\quad\cdot V^{1,0}(\bla)(\psi^1+\psi)
\prod_{j=2}^{l}V^{0-1,0}(u)(\psi^j+\psi)
\prod_{k=l+1}^{n}V^{0-2,0}(u)(\psi^k+\psi)\Bigg|_{\psi^{j}=0\atop(\forall
 j\in\{1,\cdots,n\})}.
\end{align*}
Then we can apply \eqref{eq_general_estimation_direct_one},
 \eqref{eq_double_determinant_bound}, \eqref{eq_double_norm_bound},
 \eqref{eq_norm_upper_two}, \eqref{eq_norm_upper_double_four},
\eqref{eq_norm_upper_field} and \eqref{eq_norm_upper_field_four} to
 derive that for $m\in \{0,2,4,\cdots,N\}$
\begin{align*}
&\|V_m^{1-1-1,1,(n)}\|_{1,r,r'}\\
&\le c\sum_{l=2}^{n}\left(\begin{array}{c} n-1 \\
					l-1\end{array}\right)
(c_0A)^{-n+1-\frac{m}{2}}2^{-2m}(c_0B)^{n-1}\sum_{p\in
 \{2,4\}}2^{3p}(c_0A)^{\frac{p}{2}}\beta r'\\
&\quad\cdot(2^6c_0Ar L^{-d})^{l-1}(2^{12}c_0^2A^2br)^{n-l}1_{p+2(n-l)\ge m\ge
 2(n-l)}\\
&\le  c\sum_{l=2}^{n}\left(\begin{array}{c} n-1 \\
					l-1\end{array}\right)
c_0^{-\frac{m}{2}}A^{-\frac{m}{2}}2^{-2m}\sum_{p\in
 \{2,4\}}2^{3p}(c_0A)^{\frac{p}{2}}\beta r'\\
&\quad\cdot (2^6c_0Br L^{-d})^{l-1}(2^{12}c_0^2ABbr)^{n-l}1_{p+2(n-l)\ge m\ge
 2(n-l)}.
\end{align*}
Therefore,
\begin{align}
&\|V_0^{1-1-1,1,(n)}\|_{1,r,r'}\le c \sum_{p\in \{2,4\}}2^{3p}(c_0A)^{\frac{p}{2}}
\beta r'(2^6c_0Br
 L^{-d})^{n-1}\le cL^{-d}(2^6\alpha^{-5})^n,\label{eq_1_1_1_0}\\
&\sum_{m=2}^Nc_0^{\frac{m}{2}}\alpha^m\|V_m^{1-1-1,1,(n)}\|_{1,r,r'}\label{eq_1_1_1_m}\\
&\le c\sum_{l=2}^{n}\left(\begin{array}{c} n-1 \\
					l-1\end{array}\right)\notag\\
&\quad\cdot \sum_{p\in\{2,4\}}2^{3p}(c_0A)^{\frac{p}{2}}\beta r'
 (2^6c_0Br L^{-d})^{l-1}(2^{8}\alpha^2c_0^2Bbr)^{n-l}(1+A^{-1}\alpha^2+1_{p=4}A^{-2}\alpha^{4})\notag\\
&\le c L^{-d}c_0^2\alpha^4(A+1)^2\beta r'(2^6c_0B r+
 2^8\alpha^2c_0^2Bbr)^{n-1}\notag\\
&\le c \alpha^{-1}L^{-d}(2^9\alpha^{-3})^{n-1}.\notag
\end{align}

Next let us study $V^{1-1-2,l,(n)}$ $(n\in \N_{\ge 2})$. In this case
 the main tool is the inequality
 \eqref{eq_general_estimation_double_one_norm}. By combining
 \eqref{eq_general_estimation_double_one_norm} with 
\eqref{eq_double_determinant_bound}, \eqref{eq_double_norm_bound}, 
\eqref{eq_double_norm_bound_dash}, \eqref{eq_norm_upper_double_four}, 
\eqref{eq_norm_upper_double_two_couple}, \eqref{eq_norm_upper_field}, 
\eqref{eq_norm_upper_field_four} we observe that for any $m\in
 \{0,2,4,\cdots,N\}$ 
\begin{align*}
&\|V_m^{1-1-2,1,(n)}\|_{1,r,r'}\\
&\le
 c(c_0A)^{-n-\frac{m}{2}}2^{-2m}(c_0B)^{n-1}(c_0A)^2rL^{-d}c_0A(2^{12}c_0^2A^2br)^{n-2}\\
&\quad\cdot \sum_{p\in \{2,4\}}2^{3p}(c_0A)^{\frac{p}{2}}\beta
 r'1_{p+2n-4\ge m\ge 2n-2}\\
&\le c(c_0A)^{-\frac{m}{2}}2^{-2m}c_0^2ABr
 L^{-d}(2^{12}c_0^2ABbr)^{n-2}\sum_{p\in
 \{2,4\}}2^{3p}(c_0A)^{\frac{p}{2}}\beta r'1_{p+2n-4\ge m\ge 2n-2}.
\end{align*}
Since $n\ge 2$, this implies that 
\begin{align}
\|V_0^{1-1-2,1,(n)}\|_{1,r,r'}=0.\label{eq_1_1_2_0}
\end{align}
Moreover, 
\begin{align}
&\sum_{m=2}^{N}c_0^{\frac{m}{2}}\alpha^m\|V_m^{1-1-2,1,(n)}\|_{1,r,r
'}\label{eq_1_1_2_m}\\
&\le c c_0^2\alpha^2Br L^{-d}(2^8c_0^2\alpha^2Bbr)^{n-2}
\sum_{p\in \{2,4\}}2^{3p}(c_0A)^{\frac{p}{2}}\beta
 r'(1+1_{p=4}A^{-1}\alpha^2)\notag\\
&\le c L^{-d}(2^8\alpha^{-3})^n.\notag
\end{align}

By summing up \eqref{eq_1_1_1_0}, \eqref{eq_1_1_1_m},
 \eqref{eq_1_1_2_0}, \eqref{eq_1_1_2_m} and assuming $\alpha^3\ge
 2^{10}$ we obtain that 
\begin{align*}
&\sum_{n=2}^{\infty}(\|V_0^{1-1-1,1,(n)}\|_{1,r,r'}+\|V_0^{1-1-2,1,(n)}\|_{1,r,r'})\le
 c \alpha^{-10}L^{-d},\\
&\sum_{m=2}^Nc_0^{\frac{m}{2}}\alpha^m\sum_{n=2}^{\infty}(\|V_m^{1-1-1,1,(n)}\|_{1,r,r'}+\|V_m^{1-1-2,1,(n)}\|_{1,r,r'})\le
 c \alpha^{-4}L^{-d}.
\end{align*}
These inequalities imply that $V^{1-1,1}$ is well-defined and
 $$
V^{1-1,1}\in C\Bigg(\overline{D(r)}\times \overline{D(r')}^2,\bigwedge_{even}\cV\Bigg)\cap C^{\o}\Bigg(D(r)\times
 D(r')^2,\bigwedge_{even}\cV\Bigg).
$$
By the definition $V^{1-1,1}$ is linear with $\bla\in \C^2$. It is
 implied by Lemma \ref{lem_general_estimation_direct} and Lemma
 \ref{lem_general_estimation_double} that the kernels of $V^{1-1,1}$
 satisfy \eqref{eq_time_translation_invariance}. Thus by assuming
 $\alpha\ge c$ we can conclude from the above inequalities that
 $V^{1-1,1}\in \cQ'(r,r')$.

Next let us analyze $V^{1-2,1,(n)}$ $(n\in \N_{\ge 2})$. Lemma
 \ref{lem_general_estimation_divided} ensures the existence of
 bi-anti-symmetric functions $V_{a,b}^{1-2,1,(n)}(u,\bla):I^a\times
 I^b\to \C$ $(n\in \N_{\ge 2}$, $a,b\in \{2,4,\cdots,N\}$, $(u,\bla)\in
 \C\times \C^2)$ such that they satisfy
 \eqref{eq_time_translation_invariance},
 \eqref{eq_time_vanishing_property} and 
\begin{align*}
V^{1-2,1,(n)}(u,\bla)(\psi)=\sum_{a,b=2}^N1_{a,b\in
 2\N}\left(\frac{1}{h}\right)^{a+b}\sum_{\bX\in I^a\atop \bY\in
 I^b}V_{a,b}^{1-2,1,(n)}(u,\bla)(\bX,\bY)\psi_{\bX}\psi_{\bY}.
\end{align*}
It is clear from the definition that $\bla\mapsto
 V^{1-2,1,(n)}(u,\bla)(\psi)$ $:\C^2\to \bigwedge_{even}\cV$ is linear
 for any $n\in \N_{\ge 2}$, $u\in \C$. Once the uniform convergence of
 $\sum_{n=2}^{\infty}V^{1-2,1,(n)}(u,\bla)(\psi)$ with $(u,\bla)\in
 \overline{D(r)}\times \overline{D(r')}^2$ is proved, the properties
 \eqref{eq_time_translation_invariance},
 \eqref{eq_time_vanishing_property}, the linearity with $\bla$ and the
 claimed regularity with $(u,\bla)$ are automatically satisfied by
 $V^{1-2,1}$. Let us establish desirable norm bounds. The inequalities
 \eqref{eq_general_estimation_divided_one},
\eqref{eq_double_determinant_bound}, 
\eqref{eq_double_norm_bound},
\eqref{eq_norm_upper_double_two},
\eqref{eq_norm_upper_double_four},
\eqref{eq_norm_upper_field},
\eqref{eq_norm_upper_field_four} lead to that
for any $n\in \N_{\ge 2}$, $a,b\in \{2,4,\cdots,N\}$
\begin{align}
&\|V_{a,b}^{1-2,1,(n)}\|_{1,r,r'}\label{eq_1_2_n_pre}\\
&\le
 \frac{c}{(n-1)!}\sum_{m=0}^{n-1}\sum_{(\{s_j\}_{j=1}^{m+1},\{t_k\}_{k=1}^{n-m})\in S(n,m)}\notag\\
&\quad\cdot (1_{m\neq 0}(m-1)!+1_{m=0})(1_{m\neq n-1}(n-m-2)!+1_{m=n-1})\notag\\&\quad\cdot
 2^{-2a-2b}(c_0A)^{-n+1-\frac{1}{2}(a+b)}(c_0B)^{n-1}c_0^2A^2br
 \sum_{p\in \{2,4\}}2^{3p}(c_0A)^{\frac{p}{2}}\beta
 r'(2^{12}c_0^2A^2br)^{n-2}\notag\\
&\quad\cdot\big(1_{n\in \{s_j\}_{j=2}^{m+1}}1_{2m-2+p\ge a\ge
 2m}1_{b=2(n-m)}
+1_{n\in \{t_k\}_{k=2}^{n-m}}1_{a=2m+2}
1_{2(n-m)-4+p\ge b\ge
 2(n-m)-2}\big)\notag\\
&\le
 \frac{c}{(n-1)!}\sum_{m=0}^{n-1}\sum_{(\{s_j\}_{j=1}^{m+1},\{t_k\}_{k=1}^{n-m})\in S(n,m)}\notag\\
&\quad\cdot (1_{m\neq 0}(m-1)!+1_{m=0})(1_{m\neq n-1}(n-m-2)!+1_{m=n-1})\notag\\&\quad\cdot
 2^{-2a-2b}(c_0A)^{-\frac{1}{2}(a+b)}
\sum_{p\in \{2,4\}}2^{3p}(c_0A)^{\frac{p}{2}}\beta r'
(2^{12}c_0^2ABbr)^{n-1}\notag\\
&\quad\cdot\big(1_{n\in \{s_j\}_{j=2}^{m+1}}1_{2m-2+p\ge a\ge
 2m}1_{b=2(n-m)}
+1_{n\in \{t_k\}_{k=2}^{n-m}}1_{a=2m+2}
1_{2(n-m)-4+p\ge b\ge
 2(n-m)-2}\big).\notag
\end{align}
Thus
\begin{align}
&\sum_{a,b=2}^N1_{a,b\in
 2\N}c_0^{\frac{a+b}{2}}\alpha^{a+b}\|V_{a,b}^{1-2,1,(n)}\|_{1,r,r'}\label{eq_1_2_n}\\
&\le c\sum_{p\in \{2,4\}}2^{3p}(c_0A)^{\frac{p}{2}}\beta r'
(2^{12}c_0^2ABbr)^{n-1}2^{-4n}\alpha^{2n}A^{-n}(1+1_{p=4}\alpha^2
 A^{-1})\notag\\
&\le c c_0^2\alpha^4(1+A)\beta r'(2^{8}c_0^2\alpha^2Bbr)^{n-1}\notag\\
&\le c\alpha^{-1}(2^{8}\alpha^{-3})^{n-1}.\notag
\end{align}
On the other hand, by applying
 \eqref{eq_general_estimation_divided_coupled_one} instead of
 \eqref{eq_general_estimation_divided_one} and
 \eqref{eq_double_norm_bound_dash},
 \eqref{eq_norm_upper_double_two_couple} in addition we observe that for
 any $n\in\N_{\ge 2}$, $a,b\in \{2,4,\cdots,N\}$ and anti-symmetric
 function $g:I^2\to \C$, 
\begin{align*}
&[V_{a,b}^{1-2,1,(n)},g]_{1,r,r'}\\
&\le \frac{c}{(n-1)!}\sum_{m=0}^{n-1}\sum_{(\{s_j\}_{j=1}^{m+1},\{t_k\}_{k=1}^{n-m})\in S(n,m)}\notag\\
&\quad\cdot (1_{m\neq 0}(m-1)!+1_{m=0})(1_{m\neq n-1}(n-m-2)!+1_{m=n-1})\notag\\&\quad\cdot
 2^{-2a-2b}(c_0A)^{-n+1-\frac{1}{2}(a+b)}(c_0B)^{n-2}c_0^2A^2
(rL^{-d}\|g\|_{1,\infty}'c_0B+rL^{-d}c_0A\|g\|_{1,\infty})\\
&\quad\cdot \sum_{p\in \{2,4\}}2^{3p}(c_0A)^{\frac{p}{2}}\beta
 r'(2^{12}c_0^2A^2br)^{n-2}\notag\\
&\quad\cdot\big(1_{n\in \{s_j\}_{j=2}^{m+1}}1_{2m-2+p\ge a\ge
 2m}1_{b=2(n-m)}
+1_{n\in \{t_k\}_{k=2}^{n-m}}1_{a=2m+2}
1_{2(n-m)-4+p\ge b\ge
 2(n-m)-2}\big)\notag\\
&\le c L^{-d}(\|g\|_{1,\infty}'+AB^{-1}\|g\|_{1,\infty})\cdot(\text{R.H.S of
 }\eqref{eq_1_2_n_pre}).
\end{align*}
Therefore, by the same calculation as in \eqref{eq_1_2_n} we reach that
\begin{align}
&\sum_{a,b=2}^N1_{a,b\in
 2\N}c_0^{\frac{a+b}{2}}\alpha^{a+b}[V_{a,b}^{1-2,1,(n)},g]_{1,r,r'}
\le c L^{-d}(\|g\|_{1,\infty}'+AB^{-1}\|g\|_{1,\infty})
\alpha^{-1}(2^{8}\alpha^{-3})^{n-1}.\label{eq_1_2_n_coupled}
\end{align}
Assuming $\alpha^3\ge 2^9$, we deduce from \eqref{eq_1_2_n},
 \eqref{eq_1_2_n_coupled} that 
\begin{align*}
&\sum_{a,b=2}^N1_{a,b\in
 2\N}c_0^{\frac{a+b}{2}}\alpha^{a+b}\sum_{n=2}^{\infty}\|V_{a,b}^{1-2,1,(n)}\|_{1,r,r'}\le c\alpha^{-4},\\
&\sum_{a,b=2}^N1_{a,b\in
 2\N}c_0^{\frac{a+b}{2}}\alpha^{a+b}\sum_{n=2}^{\infty}[V_{a,b}^{1-2,1,(n)},g]_{1,r,r'}
\le c L^{-d}(\|g\|_{1,\infty}'+AB^{-1}\|g\|_{1,\infty})\alpha^{-4}.
\end{align*}
These inequalities enable us to conclude that if $\alpha\ge c$,
 $V^{1-2,1}\in \cR'(r,r')$.

Finally let us treat $V^{2,1,(n)}$ $(n\in \N_{\ge 2})$. Observe that for
 any $n\in \N_{\ge 2}$, $(u,\bla)\in \C\times \C^2$
\begin{align*}
&V^{2,1,(n)}(u,\bla)(\psi)\\
&=\frac{1}{n!}\sum_{l=2}^n\left(\begin{array}{c} n \\ l\end{array}
\right)\sum_{p=0}^{n-l}\left(\begin{array}{c} n-l \\ p\end{array}
\right)Tree(\{1,\cdots,n\},\cC_0)\prod_{j=1}^lV^{1,0}(\bla)(\psi^j+\psi)\\
&\quad\cdot \prod_{k=l+1}^{l+p}V^{0-1,0}(u)(\psi^k+\psi)
\prod_{i=l+p+1}^nV^{0-2,0}(u)(\psi^i+\psi)\Bigg|_{\psi^j=0\atop (\forall
 j\in \{1,\cdots,n\})}.
\end{align*}
We can see from this equality that 
\begin{align*}
&V^{2,1,(n)}\in C^{\o}\Bigg(\C\times \C^2,\bigwedge_{even}\cV\Bigg),\\
&V^{2,1,(n)}(u,\b0)(\psi)=\frac{\partial}{\partial
 \la_j}V^{2,1,(n)}(u,\b0)(\psi)=0,\quad(\forall j\in\{1,2\},\ u\in \C).
\end{align*}
Moreover, Lemma \ref{lem_general_estimation_direct} guarantees that for any
 $(u,\bla)\in \C\times \C^2$ the kernels of $V^{2,1,(n)}(u,\bla)(\psi)$
 satisfy \eqref{eq_time_translation_invariance}. If a uniform
 convergence of $\sum_{n=2}^{\infty}V^{2,1,(n)}(u,\bla)(\psi)$ with
 $(u,\bla)$ in a neighborhood of the origin is established, then
 $V^{2,1}(u,\bla)(\psi)$ will have the regularity with $(u,\bla)$ and
 the other properties described above in the domain. 
Thus it suffices to prove suitable norm bounds which imply the
 desired convergence of $\sum_{n=2}^{\infty}V^{2,1,(n)}$ together with
 the claimed inequalities. We can combine
 \eqref{eq_general_estimation_direct_one} with
 \eqref{eq_double_determinant_bound}, \eqref{eq_double_norm_bound},
 \eqref{eq_norm_upper_two}, \eqref{eq_norm_upper_double_four},
 \eqref{eq_norm_upper_field}, \eqref{eq_norm_upper_field_infinity},
 \eqref{eq_norm_upper_field_four},
 \eqref{eq_norm_upper_field_four_infinity}
to derive that for any $n\in \N_{\ge 2}$, $m\in \{0,2,\cdots,N\}$
\begin{align*}
&\|V^{2,1,(n)}_m\|_{1,r,r'}\\
&\le c\sum_{l=2}^n\left(\begin{array}{c} n \\ l\end{array}
\right)\sum_{p=0}^{n-l}\left(\begin{array}{c} n-l \\ p\end{array}
\right)(c_0A)^{-n+1-\frac{m}{2}}2^{-2m}(c_0B)^{n-1}\sum_{p_1\in \{2,4\}}
2^{3p_1}(c_0A)^{\frac{p_1}{2}}\beta r'\\
&\quad\cdot \prod_{j=2}^l\Bigg(\sum_{p_j\in
 \{2,4\}}2^{3p_j+1}(c_0A)^{\frac{p_j}{2}}r'
\Bigg)\prod_{k=l+1}^{l+p}(2^6c_0Ar
 L^{-d})\prod_{i=l+p+1}^n(2^{12}c_0^2A^2br)\\
&\quad\cdot 1_{\sum_{j=1}^lp_j+2n-4l-2p+2\ge m\ge 2(n-l-p)}\\
&\le c 2^{-2m}c_0^{-\frac{m}{2}}A^{-\frac{m}{2}}
\sum_{l=2}^n\left(\begin{array}{c} n \\ l\end{array}
\right)\sum_{p=0}^{n-l}\left(\begin{array}{c} n-l \\ p\end{array}
\right)\sum_{p_1\in \{2,4\}}2^{3p_1}(c_0A)^{\frac{p_1}{2}}\beta r'\\
&\quad\cdot \prod_{j=2}^l\Bigg(\sum_{p_j\in
 \{2,4\}}2^{3p_j+1}c_0^{\frac{p_j}{2}}A^{\frac{p_j}{2}-1}Br'\Bigg)
(2^6c_0Br L^{-d})^p(2^{12}c_0^2ABbr)^{n-l-p}\\
&\quad\cdot 1_{\sum_{j=1}^lp_j+2n-4l-2p+2\ge m\ge 2(n-l-p)}.
\end{align*}
It follows that 
\begin{align}
&\|V^{2,1,(n)}_0\|_{1,r,r'}\label{eq_2_1_0}\\
&\le c\sum_{l=2}^n\left(\begin{array}{c} n \\ l\end{array}
\right) c_0^2(A+1)^2\beta r'(2^{13}c_0^2(A+1)Br')^{l-1}
(2^6c_0Br L^{-d})^{n-l}\notag\\
&\le c\sum_{l=2}^n\left(\begin{array}{c} n \\ l\end{array}
\right) (2^{13}\alpha^{-5})^{l}
(2^6\alpha^{-5})^{n-l}\notag\\
&\le c(2^{14}\alpha^{-5})^n,\notag\\
&\sum_{m=2}^Nc_0^{\frac{m}{2}}\alpha^m\|V^{2,1,(n)}_m\|_{1,r,r'}\label{eq_2_1_n}\\
&\le c\sum_{l=2}^n\left(\begin{array}{c} n \\ l\end{array}
\right)\sum_{p=0}^{n-l}\left(\begin{array}{c} n-l \\ p\end{array}
\right)\sum_{p_1\in \{2,4\}}2^{3p_1}(c_0A)^{\frac{p_1}{2}}\beta r'\notag\\
&\quad\cdot \prod_{j=2}^l\Bigg(\sum_{p_j\in
 \{2,4\}}2^{3p_j+1}c_0^{\frac{p_j}{2}}A^{\frac{p_j}{2}-1}Br'\Bigg)
(2^6c_0Br L^{-d})^p(2^{8}c_0^2\alpha^2Bbr)^{n-l-p}\notag\\
&\quad\cdot (1+\alpha A^{-\frac{1}{2}})^{\sum_{j=1}^lp_j-2l+2}\notag\\
&\le c\sum_{l=2}^n\left(\begin{array}{c} n \\ l\end{array}
\right)\sum_{p=0}^{n-l}\left(\begin{array}{c} n-l \\ p\end{array}
\right)\sum_{p_1\in \{2,4\}}2^{3p_1}(c_0A)^{\frac{p_1}{2}}\beta r'
(1+\alpha A^{-\frac{1}{2}})^{p_1}
\notag\\
&\quad\cdot \Bigg(\sum_{m\in
 \{2,4\}}2^{3m+1}c_0^{\frac{m}{2}}A^{\frac{m}{2}-1}B r'
(1+\alpha A^{-\frac{1}{2}})^{m-2}
\Bigg)^{l-1}
(2^6c_0Br L^{-d})^p(2^{8}c_0^2\alpha^2Bbr)^{n-l-p}\notag\\
&\le c\sum_{l=2}^n\left(\begin{array}{c} n \\ l\end{array}
\right)c_0^2\alpha^4(A+1)^2\beta r'(2^{15}c_0^2\alpha^2(A+1)Br')^{l-1}
\notag\\
&\quad\cdot (2^6c_0Br L^{-d}+2^8c_0^2\alpha^2Bbr)^{n-l}\notag\\
&\le c
 \alpha^2(2^{15}\alpha^{-3}+2^6\alpha^{-5}+2^8\alpha^{-3})^n\notag\\
&\le c \alpha^2(2^{16}\alpha^{-3})^n.\notag
\end{align}
On the assumption $\alpha^3\ge 2^{17}$, the inequalities
 \eqref{eq_2_1_0}, \eqref{eq_2_1_n} yield that 
\begin{align*}
&\sum_{n=2}^{\infty}\|V_0^{2,1,(n)}\|_{1,r,r'}\le c\alpha^{-10},\quad \sum_{m=2}^Nc_0^{\frac{m}{2}}\alpha^m\sum_{n=2}^{\infty}\|V_m^{2,1,(n)}\|_{1,r,r'}\le
 c\alpha^{-4}.
\end{align*}
Assuming additionally that $\alpha\ge c$, we can conclude that $V^{2,1}\in \cW(r,r')$.
\end{proof}

Using the results obtained in Lemma
\ref{lem_continuation_natural_parameter} and Lemma
\ref{lem_continuation_artificial_parameter}, we can construct an
analytic continuation of the function 
\begin{align}
(u,\bla)\mapsto 
\log\Bigg(
\int e^{V^{0-1,0}(u)(\psi)+V^{0-2,0}(u)(\psi)+V^{1,0}(\bla)(\psi)}d\mu_{\cC_0+\cC_1}(\psi)\Bigg)\label{eq_target_full_integration} 
\end{align}
in a neighborhood of the origin. This can be achieved by integrating the
output of the first integration with the covariance $\cC_1$. 
We want to keep the analyticity with the variable $u$ in the same 
domain as in Lemma
\ref{lem_continuation_natural_parameter}, Lemma
\ref{lem_continuation_artificial_parameter}, while the domain of the
artificial variable $\bla$ can be taken smaller. We only need 
estimates previously proved in \cite[\mbox{Subsection 3.2}]{K_BCS_I}, 
\cite[\mbox{Subsection 4.2}]{K_BCS_II} for this purpose. We will not use
the estimates presented in Subsection
\ref{subsec_general_estimation} in the rest of this
paper. However, we need to argue differently from the previous final
integration steps \cite[\mbox{Lemma 3.8}]{K_BCS_I}, 
\cite[\mbox{Lemma 4.10}]{K_BCS_II}, since here the final covariance
$\cC_1$ depends on time variables. Let 
\begin{align*}
r=c_0^{-2}\alpha^{-5}b^{-1}B^{-1},\quad 
r'=(A+1)^{-2}(B+1)^{-1}(\beta+1)^{-1}c_0^{-2}\alpha^{-5}
\end{align*}
as we set in Lemma \ref{lem_continuation_artificial_parameter}. Then let us define
the functions $V^{end,(n)}$,
$V^{1-3,end}:\overline{D(r)}\times\overline{D(r')}^2\to\C$ $(n\in \N)$
by 
\begin{align*}
&V^{end,(n)}(u,\bla)\\
&:=\frac{1}{n!}Tree(\{1,\cdots,n\},\cC_1)
\\
&\quad\cdot \prod_{j=1}^n\Bigg(
\sum_{m=1}^2V^{0-m,1}(u)(\psi^j)+\sum_{k=1}^3V^{1-k,1}(u,\bla)(\psi^j)+
V^{2,1}(u,\bla)(\psi^j)
\Bigg)\Bigg|_{\psi^j=0\atop (\forall j\in \{1,\cdots,n\})},\\
&V^{1-3,end}(u,\bla):=Tree(\{1\},\cC_1)V^{1-3,1}(u,\bla)(\psi^1).
\end{align*}
Moreover, we set 
$$
V^{end}(u,\bla):=\sum_{n=1}^{\infty}V^{end,(n)}(u,\bla)
$$
if it converges. By the definition and the division formula of Grassmann
Gaussian integral (see e.g. \cite[\mbox{Proposition I.21}]{FKT}) one can
check that $V^{end}$ is an
analytic continuation of the function \eqref{eq_target_full_integration}
if it is proved to be analytic in a neighborhood of the origin. It is
obvious that $V^{end,1-3}$ is actually independent of the variable $u$
and linear with $\bla\in \C^2$. We write as if it depends on $u$ only
for notational consistency. The result is claimed as follows.

\begin{lemma}\label{lem_final_integration}
There exists $c\in \R_{>0}$ independent of any parameter such that if
 $\alpha\ge c$, $L^d\ge A+1$, the following statements hold.
\begin{itemize}
\item 
\begin{align}
V^{end}\in C\left(\overline{D(r)}\times \overline{D(\hat{r})}^2\right)\cap 
C^{\o}\left(D(r)\times D(\hat{r})^2\right).\label{eq_end_regularity}
\end{align}
\item
\begin{align}
\frac{h}{N}\sup_{u\in\overline{D(r)}}|V^{end}(u,\b0)|\le
 (A+1)B^{-1}L^{-d}.\label{eq_end_bound}
\end{align}
\item
\begin{align}
&\left|
\frac{\partial}{\partial \la_j}V^{end}(u,\b0)
-\frac{\partial}{\partial \la_j}V^{1-3,end}(u,\b0)
\right|\le
 (A+1)^3(B+1)(\beta+1)c_0^2\alpha^5L^{-d},\label{eq_end_difference}\\
&(\forall j\in
 \{1,2\},\ u\in D(r)).\notag
\end{align}
\end{itemize}
Here
\begin{align*}
r=c_0^{-2}\alpha^{-5}b^{-1}B^{-1},\quad
 \hat{r}:=2^{-1}(h+1)^{-1}(A+1)^{-2}(B+1)^{-2}(\beta+1)^{-1}c_0^{-2}\alpha^{-5}.
\end{align*}
\end{lemma}

\begin{proof}
The following inequalities will be often used. For $m\in
 \{0,2,\cdots,N\}$
\begin{align}
&\|V_m^{0-2,1}\|_{1,\infty,r}\le \sum_{p,q\in
 2\N}1_{p+q=m}\|V_{p,q}^{0-2,1}\|_{1,\infty,r},\label{eq_divided_connect_infinity}\\
&\|V_m^{1-2,1}\|_{1,r,r'}\le \sum_{p,q\in
 2\N}1_{p+q=m}\|V_{p,q}^{1-2,1}\|_{1,r,r'}.\label{eq_divided_connect_one}
\end{align}
The following inequality is essentially same as
 \cite[\mbox{(3.92)}]{K_BCS_I}, \cite[\mbox{Lemma
 4.9}]{K_BCS_II}.
\begin{align}
&\|V_m^{a,1}(u,\eps\bla)\|_{1,\infty}\le h\eps
 \|V_m^{a,1}\|_{1,r,r'},\label{eq_grassmann_scaling}\\
&(\forall u\in \overline{D(r)},\ \bla \in\overline{D(r')}^2,\
 \eps\in [0,1/2],\ a\in \{1-1,1-2,1-3,2\},\notag\\
&\quad m\in
 \{0,2,\cdots,N\}).\notag
\end{align}
Cauchy's integral formula can be used to prove it in
 the case $a=2$. 
Set $\eps:=2^{-1}(h+1)^{-1}(B+1)^{-1}$ so that $\eps \in(0,1/2]$. We can
 deduce from ``(3.16)'' of \cite[\mbox{Lemma
 3.1}]{K_BCS_I} (or ``(4.8)'' of \cite[\mbox{Lemma
 4.1}]{K_BCS_II}), \eqref{eq_double_determinant_bound},
 \eqref{eq_double_norm_bound}, 
\eqref{eq_grassmann_set_inverse_volume},
 \eqref{eq_grassmann_set_divided}, 
\eqref{eq_grassmann_set_inverse_volume_one},
\eqref{eq_grassmann_set_divided_one},
\eqref{eq_grassmann_set_descendant}, \eqref{eq_grassmann_set_quadratic},
\eqref{eq_divided_connect_infinity}, \eqref{eq_divided_connect_one},
\eqref{eq_grassmann_scaling} and the assumption $\alpha\ge 2^3$ that for $n\in \N_{\ge 2}$,
 $(u,\bla)\in \overline{D(r)}\times \overline{D(r')}^2$ 
\begin{align*}
&|V^{end,(n)}(u,\eps\bla)|\\
&\le \frac{N}{h}c_0^{-n+1}(c_0B)^{n-1}\Bigg(
\sum_{p=2}^N2^{3p}c_0^{\frac{p}{2}}\Bigg(\|V_p^{0-1,1}\|_{1,\infty,r}+1_{p\ge
 4}\|V_p^{0-2,1}\|_{1,\infty,r}\\
&\qquad\qquad\qquad\qquad\quad +\sum_{a\in
 \{1-1,1-3,2\}}\|V_p^{a,1}(u,\eps \bla)\|_{1,\infty}+1_{p\ge
 4}\|V_p^{1-2,1}(u,\eps \bla)\|_{1,\infty}\Bigg)\Bigg)^n\\
&\le
 \frac{N}{h}B^{n-1}\big(2^6\alpha^{-2}(A+1)B^{-1}L^{-d}+2^{12}\alpha^{-4}B^{-1}+ 
3\cdot 2^6\alpha^{-2}h\eps+2^{12}\alpha^{-4}h\eps\big)^n\\
&\le \frac{N}{h}B^{-1}(2^{14}\alpha^{-2})^n.
\end{align*}
In the last inequality we also used that $L^d\ge A+1$, $h\eps \le
 B^{-1}$. Thus, if $\alpha^2\ge 2^{15}$, 
\begin{align*}
\sum_{n=2}^{\infty}\sup_{(u,\bla)\in \overline{D(r)}\times
 \overline{D(\hat{r})}^2}|V^{end,(n)}(u,\bla)|<\infty,
\end{align*}
which implies \eqref{eq_end_regularity}.

To derive \eqref{eq_end_bound}, let us observe that for $n\in \N_{\ge
 1}$, $u\in \overline{D(r)}$
\begin{align}
&V^{end,(n)}(u,\b0)\label{eq_end_natural_expansion}\\
&=\frac{1}{n!}Tree(\{1,\cdots,n\},\cC_1)
\prod_{j=1}^n\Bigg(
\sum_{m=1}^2V^{0-m,1}(u)(\psi^j)\Bigg)\Bigg|_{\psi^j=0\atop (\forall
 j\in \{1,\cdots,n\})}\notag\\
&=\sum_{l=1}^n\left(\begin{array}{c} n \\ l\end{array}\right)
\frac{1}{n!}Tree(\{1,\cdots,n\},\cC_1)\notag\\
&\quad\cdot \prod_{j=1}^lV^{0-1,1}(u)(\psi^j)
\prod_{k=l+1}^nV^{0-2,1}(u)(\psi^k)\Bigg|_{\psi^j=0\atop (\forall
 j\in \{1,\cdots,n\})}\notag\\
&\quad +
 \frac{1}{n!}Tree(\{1,\cdots,n+1\},\cC_1)\sum_{p,q=2}^N\left(\frac{1}{h}\right)^{p+q}\sum_{\bX\in I^p\atop \bY\in I^q}V_{p,q}^{0-2,1}(u)(\bX,\bY)\psi_{\bX}^1\psi_{\bY}^2\notag\\
&\qquad\cdot \prod_{j=3}^{n+1}V^{0-2,1}(u)(\psi^j)\Bigg|_{\psi^j=0\atop (\forall
 j\in \{1,\cdots,n+1\})}\notag\\
&\quad
 +\frac{1}{n!}\sum_{m=0}^{n-1}\sum_{(\{s_j\}_{j=1}^{m+1},\{t_k\}_{k=1}^{n-m})\in S(n,m)}\sum_{p,q=2}^N\left(\frac{1}{h}\right)^{p+q}\sum_{\bX\in I^p\atop \bY\in I^q}V_{p,q}^{0-2,1}(u)(\bX,\bY)\notag\\
&\qquad\cdot Tree(\{s_j\}_{j=1}^{m+1},\cC_1)\psi_{\bX}^{s_1}\prod_{j=2}^{m+1}V^{0-2,1}(u)(\psi^{s_j})\Bigg|_{\psi^{s_j}=0\atop (\forall
 j\in \{1,\cdots,m+1\})}\notag\\
&\qquad\cdot Tree(\{t_k\}_{k=1}^{n-m},\cC_1)\psi_{\bY}^{t_1}\prod_{k=2}^{n-m}V^{0-2,1}(u)(\psi^{t_k})\Bigg|_{\psi^{t_k}=0\atop (\forall
 k\in \{1,\cdots,n-m\})}.\notag
\end{align}
The above transformation is based on the same idea as that behind
 \eqref{eq_decomposition_technique}. By the properties
 \eqref{eq_time_translation_invariance},
 \eqref{eq_time_vanishing_property} of the kernels of $V^{0-2,1}$ and
 \eqref{eq_covariance_time_translation} the third term in the right-hand
 side of \eqref{eq_end_natural_expansion} vanishes. Then, combination of 
``(3.14)'' of \cite[\mbox{Lemma 3.1}]{K_BCS_I}, 
``(3.24)'' of \cite[\mbox{Lemma 3.2}]{K_BCS_I} (or ``(4.6)'' of
 \cite[\mbox{Lemma 4.1}]{K_BCS_II}, 
``(4.11)'' of \cite[\mbox{Lemma 4.2}]{K_BCS_II}), 
\eqref{eq_double_determinant_bound}, \eqref{eq_double_norm_bound}, 
\eqref{eq_double_norm_bound_dash},
 \eqref{eq_grassmann_set_inverse_volume}, 
\eqref{eq_grassmann_set_divided_coupled} and the assumption $\alpha \ge
 2^2$ yields that
\begin{align*}
&\sup_{u\in\overline{D(r)}}|V^{end,(1)}(u,\b0)|\\
&\le \|V_0^{0-1,1}\|_{1,\infty,r}
+\frac{N}{h}\sum_{m=2}^Nc_0^{\frac{m}{2}}\|V_m^{0-1,1}\|_{1,\infty,r}+\frac{N}{h}c_0^{-1}\sum_{p,q=2}^N2^{2p+2q}c_0^{\frac{p+q}{2}}[V_{p,q}^{0-2,1},\tilde{\cC}_1]_{1,\infty,r}\\
&\le
 \frac{N}{h}\alpha^{-1}AB^{-1}L^{-d}+\frac{N}{h}\alpha^{-2}(A+1)B^{-1}L^{-d}
+c\frac{N}{h}\alpha^{-4} AB^{-1}L^{-d}\\
&\le c\frac{N}{h}\alpha^{-1}(A+1)B^{-1}L^{-d}.
\end{align*}
On the other hand, for $n\in \N_{\ge 2}$ we can use 
``(3.16)'' of \cite[\mbox{Lemma 3.1}]{K_BCS_I},
``(3.26)'' of \cite[\mbox{Lemma 3.2}]{K_BCS_I}
(or ``(4.8)'' of \cite[\mbox{Lemma 4.1}]{K_BCS_II},
``(4.13)'' of \cite[\mbox{Lemma 4.2}]{K_BCS_II}), 
\eqref{eq_double_determinant_bound}, \eqref{eq_double_norm_bound}, 
\eqref{eq_double_norm_bound_dash},
 \eqref{eq_grassmann_set_inverse_volume}, 
\eqref{eq_grassmann_set_divided},
 \eqref{eq_grassmann_set_divided_coupled},
 \eqref{eq_divided_connect_infinity} and the 
 assumptions $\alpha\ge 2^3$, $L^d\ge
 A+1$ to derive that 
\begin{align*}
&\sup_{u\in\overline{D(r)}}|V^{end,(n)}(u,\b0)|\\
&\le \frac{N}{h}\sum_{l=1}^n\left(\begin{array}{c} n \\ l
\end{array}\right) c_0^{-n+1}(c_0B)^{n-1}\\
&\quad\cdot \Bigg(
\sum_{m=2}^N2^{3m}c_0^{\frac{m}{2}}\|V_m^{0-1,1}\|_{1,\infty,r}
\Bigg)^l
\Bigg(
\sum_{p=4}^N2^{3p}c_0^{\frac{p}{2}}\|V_p^{0-2,1}\|_{1,\infty,r}
\Bigg)^{n-l}\\
&\quad + \frac{N}{h}c_0^{-n}(c_0B)^{n-1}\sum_{p,q=2}^N2^{3p+3q}c_0^{\frac{p+q}{2}}[V_{p,q}^{0-2,1},\tilde{\cC}_1]_{1,\infty,r}
\Bigg(
\sum_{m=4}^N2^{3m}c_0^{\frac{m}{2}}\|V_m^{0-2,1}\|_{1,\infty,r}
\Bigg)^{n-1}\\
&\le \frac{N}{h}\sum_{l=1}^n\left(\begin{array}{c} n \\ l
\end{array}\right)
B^{n-1}(2^6\alpha^{-2}(A+1)B^{-1}L^{-d})^l(2^{12}\alpha^{-4}B^{-1})^{n-l}
\\
&\quad +c\frac{N}{h}\alpha^{-4}AB^{-1}L^{-d}(2^{12}\alpha^{-4})^{n-1}\\
&\le c\frac{N}{h}(A+1)B^{-1}L^{-d}(2^{13}\alpha^{-2})^n.
\end{align*}
Therefore, on the assumption $\alpha^2\ge 2^{14}$ 
\begin{align*}
\frac{h}{N}\sum_{n=1}^{\infty}\sup_{u\in\overline{D(r)}}|V^{end,(n)}(u,\b0)|\le
 c\alpha^{-1}(A+1)B^{-1}L^{-d},
\end{align*}
which coupled with the further assumption $\alpha\ge c$ gives
 \eqref{eq_end_bound}. 

Finally let us prove \eqref{eq_end_difference}. For any $u\in D(r)$,
 $j\in \{1,2\}$
\begin{align*}
&\frac{\partial}{\partial \la_j}V^{end,(1)}(u,\b0)-
 \frac{\partial}{\partial \la_j}V^{1-3,end}(u,\b0)\\
&=\frac{1}{r'}Tree(\{1\},\cC_1)(V^{1-1,1}(u,r'\be_j)(\psi^1)+
V^{1-2,1}(u,r'\be_j)(\psi^1))\\
&=\frac{1}{r'}Tree(\{1\},\cC_1)V^{1-1,1}(u,r'\be_j)(\psi^1)\\
&\quad + \frac{1}{r'}Tree(\{1,2\},\cC_1)\sum_{p,q=2}^N\left(\frac{1}{h}\right)^{p+q}
\sum_{\bX\in I^p\atop \bY\in
 I^q}V^{1-2,1}_{p,q}(u,r'\be_j)(\bX,\bY)\psi_{\bX}^1\psi_{\bY}^2
\Bigg|_{\psi^j=0\atop (\forall j\in \{1,2\})},
\end{align*}
where $\be_1:=(1,0)$, $\be_2:=(0,1)\in \R^2$.
To derive the last equality, we transformed the integral of $V^{1-2,1}$
 in the same manner as in
 \eqref{eq_end_natural_expansion} and erased one part by taking into
 account the property
 \eqref{eq_time_vanishing_property} of the kernels of $V^{1-2,1}$ and
 \eqref{eq_covariance_time_translation}. Moreover, by 
``(3.15)'' of \cite[\mbox{Lemma 3.1}]{K_BCS_I}, 
``(3.25)'' of \cite[\mbox{Lemma 3.2}]{K_BCS_I}
(or ``(4.7)'' of \cite[\mbox{Lemma 4.1}]{K_BCS_II}, 
``(4.12)'' of \cite[\mbox{Lemma 4.2}]{K_BCS_II}), 
\eqref{eq_double_determinant_bound}, \eqref{eq_double_norm_bound}, 
\eqref{eq_double_norm_bound_dash},
 \eqref{eq_grassmann_set_inverse_volume_one}, 
\eqref{eq_grassmann_set_divided_one_coupled} and the assumption
 $\alpha\ge 2^{2}$
\begin{align}
&\left|\frac{\partial}{\partial \la_j}V^{end,(1)}(u,\b0)-
 \frac{\partial}{\partial
 \la_j}V^{1-3,end}(u,\b0)\right|\label{eq_end_difference_1}\\
&\le
 \frac{1}{r'}\|V_0^{1-1,1}\|_{1,r,r'}+\frac{1}{r'}\sum_{m=2}^{N}c_0^{\frac{m}{2}}\|V_m^{1-1,1}\|_{1,r,r'}
+\frac{1}{r'}c_0^{-1}\sum_{p,q=2}^N2^{2p+2q}c_0^{\frac{p+q}{2}}[V_{p,q}^{1-2,1},\tilde{\cC_1}]_{1,r,r'}\notag\\
&\le \frac{c}{r'}\alpha^{-1}(A+1)L^{-d}.\notag
\end{align}
Let $n\in \N_{\ge 2}$. Based on the properties
 \eqref{eq_time_translation_invariance}, \eqref{eq_time_vanishing_property}
 of the kernels of $V^{0-2,1}$,
the property \eqref{eq_time_translation_invariance} of the kernels of
 $V^{1-a,1}$ $(a=1,2,3)$
 and
 \eqref{eq_covariance_time_translation}, we can transform the defining
 equality in the same
 way as above and obtain that for $u\in D(r)$, $j\in \{1,2\}$
\begin{align*}
&\frac{\partial}{\partial \la_j}V^{end,(n)}(u,\b0)\\
&=\frac{1}{(n-1)!r'}Tree(\{1,\cdots,n\},\cC_1)\\
&\quad\cdot \sum_{a=1}^3V^{1-a,1}(u,r'\be_j)(\psi^1)
\prod_{k=2}^n\Bigg(
\sum_{p=1}^2V^{0-p,1}(u)(\psi^k)\Bigg)\Bigg|_{\psi^i=0\atop (\forall
 i\in \{1,\cdots,n\})}\notag\\
&=\frac{1}{(n-1)!r'}\sum_{l=2}^n
\left(\begin{array}{c} n-1 \\ l-1\end{array}\right)
Tree(\{1,\cdots,n\},\cC_1)\\
&\quad\cdot \sum_{a=1}^3V^{1-a,1}(u,r'\be_j)(\psi^1)
\prod_{k=2}^lV^{0-1,1}(u)(\psi^k)
\prod_{s=l+1}^nV^{0-2,1}(u)(\psi^s)\Bigg|_{\psi^i=0\atop (\forall
 i\in \{1,\cdots,n\})}\notag\\
&\quad + \frac{1}{(n-1)!r'} Tree(\{1,\cdots,n+1\},\cC_1)
\sum_{p,q=2}^N\left(\frac{1}{h}\right)^{p+q}\sum_{\bX\in I^p\atop \bY\in
 I^q}V_{p,q}^{0-2,1}(u)(\bX,\bY)\psi_{\bX}^1\psi_{\bY}^2\\
&\qquad\cdot \prod_{k=3}^nV^{0-2,1}(u)(\psi^k)\sum_{a=1}^3V^{1-a,1}(u,r'\be_j)(\psi^{n+1})
\Bigg|_{\psi^i=0\atop (\forall
 i\in \{1,\cdots,n+1\})}.
\end{align*}
In this situation we can apply 
``(3.17)'' of \cite[\mbox{Lemma 3.1}]{K_BCS_I}, 
``(3.27)'' of \cite[\mbox{Lemma 3.2}]{K_BCS_I}
(or ``(4.9)'' of \cite[\mbox{Lemma 4.1}]{K_BCS_II}, 
``(4.14)'' of \cite[\mbox{Lemma 4.2}]{K_BCS_II}), 
\eqref{eq_double_determinant_bound}, \eqref{eq_double_norm_bound}, 
\eqref{eq_double_norm_bound_dash},
 \eqref{eq_grassmann_set_inverse_volume},
 \eqref{eq_grassmann_set_divided}, \eqref{eq_grassmann_set_divided_coupled},
 \eqref{eq_grassmann_set_inverse_volume_one}, 
\eqref{eq_grassmann_set_divided_one},
\eqref{eq_grassmann_set_descendant},
 \eqref{eq_divided_connect_infinity}, \eqref{eq_divided_connect_one} and
 the inequalities $\alpha\ge 2^{3}$, 
 $L^d\ge A+1$ to
 deduce that 
\begin{align*}
&\left|\frac{\partial}{\partial \la_j}V^{end,(n)}(u,\b0)\right|\\
&\le \frac{1}{r'}
\sum_{l=2}^n
\left(\begin{array}{c} n-1 \\
      l-1\end{array}\right)c_0^{-n+1}(c_0B)^{n-1}\\
&\quad\cdot\sum_{p=2}^N2^{3p}c_0^{\frac{p}{2}}
(\|V_p^{1-1,1}\|_{1,r,r'}+1_{p\ge 4}\|V_p^{1-2,1}\|_{1,r,r'}
+\|V_p^{1-3,1}\|_{1,r,r'})\\
&\quad\cdot \Bigg(\sum_{q=2}^N2^{3q}c_0^{\frac{q}{2}}\|V_q^{0-1,1}\|_{1,\infty,r}
\Bigg)^{l-1}
\Bigg(\sum_{m=4}^N2^{3m}c_0^{\frac{m}{2}}\|V_m^{0-2,1}\|_{1,\infty,r}
\Bigg)^{n-l}\\
&\quad +
 \frac{1}{r'}c_0^{-n}(c_0B)^{n-1}\sum_{p,q=2}^N2^{3p+3q}c_0^{\frac{p+q}{2}}
[V_{p,q}^{0-2,1},\tilde{\cC}_1]_{1,\infty,r}
\Bigg(\sum_{m=4}^N2^{3m}c_0^{\frac{m}{2}}\|V_m^{0-2,1}\|_{1,\infty,r}
\Bigg)^{n-2}\\
&\qquad\cdot \Bigg(\sum_{s=2}^N2^{3s}c_0^{\frac{s}{2}}
(\|V_s^{1-1,1}\|_{1,r,r'}+1_{s\ge 4}\|V_s^{1-2,1}\|_{1,r,r'}
+\|V_s^{1-3,1}\|_{1,r,r'})\Bigg)\\
&\le \frac{c}{r'}B^{n-1}\alpha^{-2}
\sum_{l=2}^n
\left(\begin{array}{c} n-1 \\ l-1\end{array}\right)
(2^6\alpha^{-2}(A+1)B^{-1}L^{-d})^{l-1}(2^{12}\alpha^{-4}B^{-1})^{n-l}\\
&\quad +\frac{c}{r'}B^{n-2}\alpha^{-6}AL^{-d}(2^{12}\alpha^{-4}B^{-1})^{n-2}\\
&\le \frac{c}{r'}(A+1)L^{-d}(2^{13}\alpha^{-2})^n.
\end{align*}
Thus by assuming that $\alpha^2\ge 2^{14}$ we have that 
\begin{align}
\sum_{n=2}^{\infty}\left|\frac{\partial}{\partial
 \la_j}V^{end,(n)}(u,\b0)\right|\le \frac{c}{r'}\alpha^{-4}(A+1)L^{-d}.
\label{eq_end_difference_n}
\end{align}
By coupling \eqref{eq_end_difference_1} with \eqref{eq_end_difference_n}
 and assuming that $\alpha\ge c$ once more we reach
 \eqref{eq_end_difference}.
\end{proof}

\subsection{The infinite-volume limit}\label{subsec_infinite_volume}

Among all the lemmas prepared in this section so far, Lemma
\ref{lem_standard_covariance_bounds}, Lemma
\ref{lem_covariance_determinant_bound}, Lemma
\ref{lem_final_integration} are the main necessary tools to prove
Theorem \ref{thm_infinite_volume_limit}. With these lemmas we can
straightforwardly follow the arguments of \cite[\mbox{Subsection
5.2}]{K_BCS_II} to complete the proof of Theorem
\ref{thm_infinite_volume_limit}. Though we should not lengthen the paper
by repeating the same statements as before, let us state a few pivotal
lemmas for the sake of readability. These are close to lemmas
proved in \cite{K_BCS_I}, \cite{K_BCS_II} but are adjusted to the
present situation. Let us recall the definitions of $V(u)(\psi)$,
$W(u)(\psi)$ given in the beginning of \cite[\mbox{Subsection
4.4}]{K_BCS_II} and $A^1(\psi)$, $A^2(\psi)$, $A(\psi)$ given in 
\cite[\mbox{Section 3}]{K_BCS_II}. It is apparent from
\eqref{eq_0_1_2_reduction}, \eqref{eq_1_reduction} that 
\begin{align*}
V^{0-1,0}(u)(\psi)+V^{0-2,0}(u)(\psi)=-V(u)(\psi)+W(u)(\psi),\quad
 V^{1,0}(\bla)(\psi)=-A(\psi).
\end{align*}
A practical application of Lemma \ref{lem_standard_covariance_bounds}, 
Lemma \ref{lem_covariance_determinant_bound}, 
Lemma \ref{lem_final_integration} results in the following lemma.

\begin{lemma}\label{lem_practical_double_scale_result}
Set 
\begin{align*}
\hat{A}:=(e_{min}+\beta^{-1}+\beta^{-1}e_{min}^{-1}+1)\max\{e_{min}^{-1},e_{min}^{-d-1}\},\quad
 \hat{B}:=\max\{e_{min}^{-1},e_{min}^{-d-1}\}.
\end{align*}
Then there exist $c\in \R_{>0}$ independent of any parameter and
 $\hat{c}_0\in\R_{\ge 1}$ depending only on $d$, $b$,
 $(\hbv_j)_{j=1}^d$, $c_E$ such that the following statements hold for
 any $\alpha\in \R_{\ge 1}$, $h\in \frac{2}{\beta}\N$, $L\in \N$,
 $\phi\in \C$
 satisfying that 
\begin{align}
&\alpha\ge c,\quad L^d\ge
 \hat{A}+1,\notag\\
&h\ge\max\{\sqrt{e_{max}^2+|\phi|^2},1\}+\frac{1}{\beta}(3\pi+2).
\label{eq_final_h_assumption}
\end{align}
\begin{enumerate}[(i)]
\item\label{item_practical_double_scale_bound}
\begin{align*}
&e^{-4b\beta(\hat{A}+1)\hat{B}^{-1}}\le \left|
\int e^{-V(u)(\psi)+W(u)(\psi)}d\mu_{C(\phi)}(\psi)
\right|\le e^{4b\beta(\hat{A}+1)\hat{B}^{-1}},\\
&\left(\forall u\in \overline{D(\hat{c}_0^{-2}\alpha^{-5}b^{-1}\hat{B}^{-1})}
\right).
\end{align*}
\item\label{item_practical_double_scale_difference}
\begin{align*}
&\left|\frac{\int
 e^{-V(u)(\psi)+W(u)(\psi)}A^j(\psi)d\mu_{C(\phi)}(\psi)}{\int
 e^{-V(u)(\psi)+W(u)(\psi)}d\mu_{C(\phi)}(\psi)}-\int A^j(\psi)d\mu_{C(\phi)}(\psi)
\right|\\
&\le
 (\hat{A}+1)^3(\hat{B}+1)(\beta+1)\hat{c}_0^2\alpha^5L^{-d},\quad
 \left(\forall j\in \{1,2\},\ u\in\overline{D(\hat{c}_0^{-2}\alpha^{-5}b^{-1}\hat{B}^{-1})}
\right).
\end{align*}
\end{enumerate}
\end{lemma}

\begin{proof} 
We take the generalized covariances $\cC_0$, $\cC_1$ to be $C_0$,
 $C_1$, which were analyzed in Subsection \ref{subsec_covariances},
 respectively. We can see from Lemma
 \ref{lem_standard_covariance_bounds}, Lemma
 \ref{lem_covariance_determinant_bound} that on the assumption
 \eqref{eq_final_h_assumption} $c_0$, $A$, $B$ can be taken to be
 $\hat{c}_0$, $\hat{A}$, $\hat{B}$ respectively. Accordingly the claims
 of Lemma \ref{lem_final_integration} hold with $\hat{c}_0$, $\hat{A}$,
 $\hat{B}$ in place of $c_0$, $A$, $B$. By using the relation
 \eqref{eq_covariance_double_decomposition} and the gauge transform
 $\psi_{\orho\rho\bx s \xi}\mapsto
 e^{-i\xi\frac{\pi}{\beta}s}\psi_{\orho\rho\bx s \xi}$ we can prove that
if $|u|$, $\|\bla\|_{\C^2}$ are sufficiently small, 
\begin{align*}
&\Re \int e^{-V(u)(\psi)+W(u)(\psi)-A(\psi)}d\mu_{C(\phi)}(\psi)>0,\\
&V^{end}(u,\bla)=\log\Bigg(\int e^{-V(u)(\psi)+W(u)(\psi)-A(\psi)}d\mu_{C(\phi)}(\psi)
\Bigg).
\end{align*}
For the proof of the above properties let us refer to the proof of 
\cite[\mbox{Lemma 4.13}]{K_BCS_I} or \cite[\mbox{Proposition
 5.9}]{K_BCS_II} where a similar claim was proved. Then it follows from
 \eqref{eq_end_regularity}, the identity theorem and continuity that on
 the assumptions of this lemma 
\begin{align}
&e^{V^{end}(u,\bla)}=\int
 e^{-V(u)(\psi)+W(u)(\psi)-A(\psi)}d\mu_{C(\phi)}(\psi),\label{eq_identity_theorem_end}\\
&\Big(\forall (u,\bla)\in
 \overline{D(\hat{c}_0^{-2}\alpha^{-5}b^{-1}\hat{B}^{-1})}\notag\\
&\qquad\qquad\quad\times 
\overline{D(2^{-1}(h+1)^{-1}(\hat{A}+1)^{-2}(\hat{B}+1)^{-2}(\beta+1)^{-1}\hat{c}_0^{-2}\alpha^{-5})}^2\Big).\notag
\end{align}
On the other hand, by the definition and the same gauge transform as
 above
\begin{align}
V^{1-3,end}=-\int A(\psi)d\mu_{C(\phi)}(\psi).\label{eq_source_term_end}
\end{align}
By combining \eqref{eq_identity_theorem_end},
 \eqref{eq_source_term_end} with \eqref{eq_end_bound},
 \eqref{eq_end_difference} we can derive the claimed inequalities.
\end{proof}

The next lemma is essentially based on \cite[\mbox{Proposition
 4.16}]{K_BCS_I}. The proof of \cite[\mbox{Proposition
 5.10}]{K_BCS_II} can be read as a guide to deduce the lemma from 
\cite[\mbox{Proposition 4.16}]{K_BCS_I}. 

\begin{lemma}\label{lem_infinite_volume_limit_perturbative}
Let $\hat{A}$, $\hat{B}$, $c$, $\hat{c}_0$ be those introduced in Lemma
 \ref{lem_practical_double_scale_result}. Assume that $L^d\ge \hat{A}+1$
 and $\alpha\ge c$. Then for any non-empty compact set $Q$ of $\C$ 
\begin{align*}
&\lim_{h\to\infty\atop h\in \frac{2}{\beta}\N}\int 
e^{-V(u)(\psi)+W(u)(\psi)}d\mu_{C(\phi)}(\psi),\quad 
\lim_{L\to\infty\atop L\in \N}\lim_{h\to\infty\atop h\in \frac{2}{\beta}\N}\int 
e^{-V(u)(\psi)+W(u)(\psi)}d\mu_{C(\phi)}(\psi)
\end{align*}
converge in $C(Q\times
 \overline{D(2^{-1}\hat{c}_0^{-2}\alpha^{-5}b^{-1}\hat{B}^{-1})})$
as sequences of functions of the variable $(\phi,u)$. Here we consider $C(Q\times
 \overline{D(2^{-1}\hat{c}_0^{-2}\alpha^{-5}b^{-1}\hat{B}^{-1})})$ as
 the Banach space with the uniform norm.
\end{lemma}

Now we can describe how to derive the claims of Theorem
\ref{thm_infinite_volume_limit} by following the final part of the proof of
\cite[\mbox{Theorem 1.3}]{K_BCS_II} presented in \cite[\mbox{Subsection
5.2}]{K_BCS_II}.

\begin{proof}[Proof of Theorem \ref{thm_infinite_volume_limit}]

The proof of the claims ``(i), (ii), (iii), (iv), (v)'' of
 \cite[\mbox{Theorem 1.3}]{K_BCS_II} straightforwardly applies to prove
 \eqref{item_partition_positivity}, \eqref{item_free_energy},
\eqref{item_SSB}, \eqref{item_ODLRO}, \eqref{item_CD} of Theorem
 \ref{thm_infinite_volume_limit} respectively. In the proof of
 \cite[\mbox{Theorem 1.3}]{K_BCS_II} the basic lemmas ``Lemma 3.1'',
 ``Lemma 3.2'', ``Lemma 3.6'', ``Lemma 5.11'' of \cite{K_BCS_II} were
 frequently used. We should remark that here the same statements as these
 lemmas hold for any $\beta\in \R_{>0}$, $\theta\in \R$ including the
 case $\beta\theta/2\in \pi(2\Z+1)$. This is because in this
 paper the free partition function does not vanish for any $\theta\in
 \R$ thanks to the assumption \eqref{eq_one_particle_lowest_band}. Let
 us fix $\alpha\in \R_{\ge 1}$ satisfying the condition $\alpha\ge c$
 required in Lemma \ref{lem_practical_double_scale_result} and 
Lemma \ref{lem_infinite_volume_limit_perturbative}. Set
 $c':=4^{-1}\hat{c}_0^{-2}\alpha^{-5}$. We see that $c'\in (0,1]$, it
 depends only on $d$, $b$,
 $(\hbv_j)_{j=1}^d$, $c_E$ and 
\begin{align*}
\left(-\frac{2c'}{b}\min\{e_{min},e_{min}^{d+1}\},0\right)\subset 
 \overline{D(2^{-1}\hat{c}_0^{-2}\alpha^{-5}b^{-1}\hat{B}^{-1})}.
\end{align*}
This means that the inequalities and the convergence properties stated
 in Lemma \ref{lem_practical_double_scale_result}, Lemma
 \ref{lem_infinite_volume_limit_perturbative} are applicable to the
 Grassmann integral formulation with the coupling constant 
$$U\in \left(-\frac{2c'}{b}\min\{e_{min},e_{min}^{d+1}\},0\right).$$
Subsequently, for $U$ belonging to this open interval the claims of
 Theorem \ref{thm_infinite_volume_limit} can be proved.

Here we only summarize which lemmas are necessary to
 conclude the claims of Theorem \ref{thm_infinite_volume_limit} if we
 straightforwardly follow the proof of \cite[\mbox{Theorem 1.3}]{K_BCS_II}.
We avoid fully repeating the same arguments as before. 
The key point of translating the proof of \cite[\mbox{Theorem
 1.3}]{K_BCS_II} into the proof of Theorem
 \ref{thm_infinite_volume_limit} is to replace ``Proposition 5.9
 (i),(ii)'', ``Proposition 5.10'' of \cite{K_BCS_II} by Lemma
 \ref{lem_practical_double_scale_result}
 \eqref{item_practical_double_scale_bound},\eqref{item_practical_double_scale_difference}, 
Lemma \ref{lem_infinite_volume_limit_perturbative} respectively.
We can prove 
\eqref{item_partition_positivity}, \eqref{item_SSB}, \eqref{item_ODLRO},
 \eqref{item_CD}, \eqref{item_free_energy} in this order as in the proof
 of \cite[\mbox{Theorem 1.3}]{K_BCS_II}. 

\eqref{item_partition_positivity}: ``Lemma 3.1'', ``Lemma 3.2'', ``Lemma
 3.6 (i),(iii),(iv)'' of \cite{K_BCS_II}, Lemma
 \ref{lem_practical_double_scale_result}
 \eqref{item_practical_double_scale_bound} and Lemma
 \ref{lem_infinite_volume_limit_perturbative} of this paper.

\eqref{item_SSB}: ``Lemma 3.1'', ``Lemma 3.6 (i),(iii)'', ``Lemma
 5.11'', ``Lemma A.1'' of \cite{K_BCS_II}, 
Lemma \ref{lem_standard_covariance_bounds}
 \eqref{item_full_determinant}, 
Lemma \ref{lem_practical_double_scale_result}
 \eqref{item_practical_double_scale_bound},\eqref{item_practical_double_scale_difference}
 and Lemma \ref{lem_infinite_volume_limit_perturbative} of this paper.

\eqref{item_ODLRO}, \eqref{item_CD}: ``Lemma 3.1'', ``Lemma 3.6 (i),(iii)'', ``Lemma
 5.11'', ``Lemma A.2'', ``Lemma A.3'' of \cite{K_BCS_II}, 
Lemma \ref{lem_standard_covariance_bounds}
 \eqref{item_full_determinant}, Lemma \ref{lem_practical_double_scale_result}
 \eqref{item_practical_double_scale_bound},\eqref{item_practical_double_scale_difference}
 and Lemma \ref{lem_infinite_volume_limit_perturbative} of this paper.

\eqref{item_free_energy}:  ``Lemma 3.1'', ``Lemma 3.2'', ``Lemma 3.6 (iii)'', ``Lemma A.4'' of \cite{K_BCS_II}.
\end{proof}

\appendix

\section{A special matrix-valued function}\label{app_special_function}

Here we construct a matrix-valued function, which is used to prove that
the function $\tau(\cdot)$ can have more than one local minimum points
in Subsection \ref{subsec_shape}. 

\begin{lemma}\label{lem_special_function}
For any $d$, $b\in \N$, basis $(\hbv_j)_{j=1}^d$ of $\R^d$, $s,t\in
 \R_{>0}$ satisfying $0<s<t<1$, $e_{max}$, $e_{min}\in\R_{>0}$
 satisfying $0<e_{min}<e_{max}$ there exists $E\in \cE(e_{min},e_{max})$
 such that 
\begin{align*}
&D_d|\{\bk\in \G_{\infty}^{*}\ |\ \Tr|E(\bk)|=b e_{max}\}|=s,\\
&D_d|\{\bk\in \G_{\infty}^*\ |\ \Tr|E(\bk)|=b e_{min}\}|=1-t,
\end{align*}
where $|S|$ denotes the Lebesgue measure of a measurable set $S$ $(\subset
 \R^d)$. 
\end{lemma}

\begin{proof}
By a standard procedure one can construct a function $\phi$ $(\in
 C^{\infty}(\R))$ satisfying that 
\begin{align}
&\phi(x)=(e_{max}-e_{min})^{\frac{1}{d}}\text{ if }|x-\pi|\le \pi
 s^{\frac{1}{d}},\notag\\
&\phi(x)=0\text{ if }|x-\pi|\ge \pi
 t^{\frac{1}{d}},\notag\\
&\phi(x)\in (0,(e_{max}-e_{min})^{\frac{1}{d}})\text{ if }\pi
 s^{\frac{1}{d}}<|x-\pi|<\pi t^{\frac{1}{d}},\notag\\
&\phi(\pi+x)=\phi(\pi-x),\quad (\forall x\in
 \R).\label{eq_deep_time_reversal}
\end{align}
Let us define the function $\Phi$ $(\in C^{\infty}(\R^d))$ by 
$\Phi(x_1,\cdots,x_d):=\prod_{j=1}^d\phi(x_j)+e_{min}$.
Observe that
\begin{align*}
&\Phi(x_1,\cdots,x_d)=e_{max}\text{ if }|x_j-\pi|\le \pi
 s^{\frac{1}{d}}\ (\forall j\in \{1,\cdots,d\}),\notag\\
&\Phi(x_1,\cdots,x_d)=e_{min}\text{ if }\exists j\in
 \{1,\cdots,d\}\text{ s.t. }|x_j-\pi|\ge \pi
 t^{\frac{1}{d}},\notag\\
&\Phi(x_1,\cdots,x_d)\in (e_{min},e_{max})\text{ otherwise.}
\end{align*}
Then let us define the matrix-valued function $\hat{E}:\G_{\infty}^*\to
 \Mat(b,\C)$ by $\hat{E}(\bk):=\Phi((\hbv_1,\cdots,\hbv_d)^{-1}\bk)I_b$
 $(\bk\in \G_{\infty}^*)$. We can periodically extend $\hat{E}$ to be a
 map from $\R^d$ to $\Mat(b,\C)$. If $E$ denotes the extension, it
 follows that $E\in \cE(e_{min},e_{max})$. Let us confirm the property
 \eqref{eq_time_reversal_symmetry}. The other properties are obvious. 
Take $\bk\in \R^d$. There exist $\hat{k}_j\in [0,2\pi)$, $m_j\in \Z$
 $(j=1,\cdots,d)$ such that $\bk=\sum_{j=1}^d(\hat{k}_j+2\pi
 m_j)\hbv_j$. By the periodicity and \eqref{eq_deep_time_reversal}, 
\begin{align*}
\overline{E(-\bk)}&=E\left(\sum_{j=1}^d(2\pi-\hat{k}_j)\hbv_j\right)=\Phi(2\pi-\hat{k}_1,\cdots,2\pi-\hat{k}_d)I_b=\Phi(\hat{k}_1,\cdots,\hat{k}_d)I_b\\
&=E(\bk).
\end{align*}
Moreover, we can verify that
\begin{align*}
&D_d|\{\bk\in \G^*_{\infty}\ |\ \Tr |E(\bk)| = b e_{max}\}|\\
&=D_d|\{\bk\in
 \G_{\infty}^*\ |\ (\hbv_1,\cdots,\hbv_d)^{-1}\bk\in [\pi-\pi
 s^{\frac{1}{d}},\pi+\pi s^{\frac{1}{d}}]^d
\}|=s,\\
&D_d|\{\bk\in \G^*_{\infty}\ |\ \Tr |E(\bk)| = b e_{min}\}|\\
&=D_d|\{\bk\in
 \G_{\infty}^*\ |\ (\hbv_1,\cdots,\hbv_d)^{-1}\bk\in
 [0,2\pi]^d\backslash (\pi-\pi
 t^{\frac{1}{d}},\pi+\pi t^{\frac{1}{d}})^d
\}|=1-t.
\end{align*}
\end{proof}

\section{A definite integral formula}\label{app_integral}

Here we derive an explicit formula of a definite integral, which is used
in the proof of Proposition \ref{prop_one_band_result}. 

\begin{lemma}\label{lem_definite_integral}
For $x$, $t\in \R_{\ge 0}$
\begin{align}
&\frac{1}{2\pi}\int_0^{2\pi}dk \frac{1}{1+x (t(\cos k+1)+1)^2}\label{eq_definite_integral}\\
&=\frac{((2t+1)^2x+1)^{\frac{1}{2}}+(x+1)^{\frac{1}{2}}}{\sqrt{2}(x+1)^{\frac{1}{2}}((2t+1)^2x+1)^{\frac{1}{2}}\big((x+1)^{\frac{1}{2}}((2t+1)^2x+1)^{\frac{1}{2}}+(2t+1)x+1\big)^{\frac{1}{2}}
}.\notag
\end{align}
\end{lemma}

\begin{proof}
When $x=0$ or $t=0$, the equality obviously holds. Let us assume that
 $x>0$, $t>0$. One can prove by applying the residue theorem that
\begin{align}
\int_0^{\infty}ds \frac{1}{s^2+r
 e^{i\theta}}=\frac{\pi}{2}r^{-\frac{1}{2}}e^{-i\frac{\theta}{2}},\quad
 (\forall r\in \R_{>0},\ \theta \in
 (0,\pi)).\label{eq_residue_application}
\end{align}
By introducing a new variable $s$ by $s=\tan(k/2)$ we observe that
\begin{align*}
(\text{L.H.S of }\eqref{eq_definite_integral})&=\frac{2}{\pi x
 t^2}\int_0^{\infty}ds
 (1+s^2)^{-1}\Bigg(\left(\frac{2}{1+s^2}+\frac{1}{t}\right)^2+\frac{1}{xt^2}\Bigg)^{-1}\\
&=\frac{1}{i\pi}\Bigg((x^{\frac{1}{2}}-i)^{-1}\int_{0}^{\infty}ds 
\Bigg(s^2+\frac{(2t+1)x+1}{x+1}+\frac{2x^{\frac{1}{2}}t}{x+1}i\Bigg)^{-1}
\\
&\qquad\quad -(x^{\frac{1}{2}}+i)^{-1}\int_{0}^{\infty}ds 
\Bigg(s^2+\frac{(2t+1)x+1}{x+1}-\frac{2x^{\frac{1}{2}}t}{x+1}i\Bigg)^{-1}
\Bigg).
\end{align*}
Here we can apply \eqref{eq_residue_application} with 
\begin{align*}
r=\left(\frac{(2t+1)^2x+1}{x+1}\right)^{\frac{1}{2}},\quad
 \theta=\tan^{-1}\left(\frac{2t x^{\frac{1}{2}}}{(2t+1)x+1}\right)\
 \left(\in\left(0,\frac{\pi}{2}\right)\right)
\end{align*}
to derive that
\begin{align*}
(\text{L.H.S of }\eqref{eq_definite_integral})&=
\frac{1}{i\pi}\left((x^{\frac{1}{2}}-i)^{-1}\frac{\pi}{2}r^{-\frac{1}{2}}e^{-i\frac{\theta}{2}}-(x^{\frac{1}{2}}+i)^{-1}\frac{\pi}{2}r^{-\frac{1}{2}}e^{i\frac{\theta}{2}}\right)\\
&=\frac{\cos\left(\frac{\theta}{2}\right)\left(1-x^{\frac{1}{2}}\tan\left(\frac{\theta}{2}\right)\right)}{((2t+1)^2x+1)^{\frac{1}{4}}(x+1)^{\frac{3}{4}}}.
\end{align*}
Substitution of the equalities
\begin{align*}
&\tan\left(\frac{\theta}{2}\right)=\frac{1}{2tx^{\frac{1}{2}}}\left(
(x+1)^{\frac{1}{2}}((2t+1)^2x+1)^{\frac{1}{2}}-((2t+1)x+1)\right),\\
&\cos\left(\frac{\theta}{2}\right)=\frac{\big((x+1)^{\frac{1}{2}}((2t+1)^2x+1)^{\frac{1}{2}}+(2t+1)x+1\big)^{\frac{1}{2}}}{\sqrt{2}(x+1)^{\frac{1}{4}}((2t+1)^2x+1)^{\frac{1}{4}}}
\end{align*}
leads to the right-hand side of \eqref{eq_definite_integral}.
\end{proof}

\section*{Supplementary List of Notations}

\begin{center}
\begin{tabular}{lll}
Notation & Description & Reference \\
\hline
$e_{min}$ & minimum of magnitude of free dispersion relation & Subsection
 \ref{subsec_main_results}\\
$e_{max}$ & maximum of magnitude of free dispersion relation & Subsection
 \ref{subsec_main_results}\\
$\cE(e_{min},e_{max})$ & set of matrix-valued functions & Subsection
 \ref{subsec_main_results}\\
$g_E(\cdot)$ & real-valued function on $\R_{>0}\times \R\times \R$ & Subsection
 \ref{subsec_main_results}\\
$\beta_c$ & critical inverse temperature & Lemma \ref{lem_critical_temperature}\\
$c_E$ & positive constant depending only on $E(\cdot)$ & \eqref{eq_band_spectra_total_derivative}
\end{tabular}
\end{center}


\begin{thebibliography}{0}
\bibitem{AK} N. O. Abeling and S. Kehrein, Quantum quench dynamics in
	the transverse field Ising model at nonzero temperatures, 
        Phys. Rev. B {\bf 93} (2016), 104302.
\bibitem{BCS} J. Bardeen, L. N. Cooper and J. R. Schrieffer, 
Theory of superconductivity, Phys. Rev. {\bf 108} (1957), 1175--1204. 
\bibitem{BBD} U. Bhattacharya, S. Bandyopadhyay and A. Dutta, 
        Mixed state dynamical quantum phase transitions, 
        Phys. Rev. B {\bf 96} (2017), 180303(R).
\bibitem{BP} J.-B. Bru and W. de Siqueira Pedra, Effect of a locally
	repulsive interaction on s-wave superconductors, 
        Rev. Math. Phys. {\bf 22} (2010), 233--303.
\bibitem{FKT} J. Feldman, H. Kn\"orrer and E. Trubowitz, 
	     Fermionic functional integrals and the renormalization
	group,	     CRM monograph series No. 16, American Mathematical
	Society, Providence, R.I., 2002.
\bibitem{Fetal} N. Fl\"aschner, D. Vogel, M. Tarnowski, B. S. Rem,
	D.-S. L\"uhmann, M. Heyl, J. C. Budich, L. Mathey, K. Sengstock
	and C. Weitenberg, Observation of dynamical vortices after
	quenches in a system with topology, Nat. Phys. {\bf 14} (2018), 
        265--268.
\bibitem{H} M. Heyl, Dynamical quantum phase transitions: a review, 
        Rep. Prog. Phys. {\bf 81} (2018), 054001.
\bibitem{HB} M. Heyl and J. C. Budich, Dynamical topological quantum
        phase transitions for mixed states,  
        Phys. Rev. B {\bf 96} (2017), 180304(R).
\bibitem{HPK} M. Heyl, A. Polkovnikov and S. Kehrein, Dynamical quantum
	phase transitions in the transverse-field Ising model, 
        Phys. Rev. Lett. {\bf 110} (2013), 135704.
\bibitem{Jetal} P. Jurcevic, H. Shen, P. Hauke, C. Maier, T. Brydges,
	C. Hempel, B. P. Lanyon, M. Heyl, R. Blatt and C. F. Roos, 
        Direct observation of dynamical quantum phase transitions in an
	interacting many-body system, Phys. Rev. Lett. {\bf 119} (2017),
	080501.
\bibitem{K_2010} Y. Kashima, Exponential decay of correlation functions
	in many-electron systems, J. Math. Phys. {\bf 51} (2010),
	063521. 
\bibitem{K_RG} Y. Kashima, Renormalization group analysis of multi-band 
        many-electron systems at half-filling, ``the special issue for the
	20th anniversary'', J. Math. Sci. Univ. Tokyo. {\bf 23} (2016),
	1--288.
\bibitem{K_BCS_I} Y. Kashima, Superconducting phase in the BCS model with
	imaginary magnetic field, accepted for publication 
        in J. Math. Sci. Univ. Tokyo, arXiv:1609.06121.
\bibitem{K_BCS_II} Y. Kashima, Superconducting phase in the BCS model with
	imaginary magnetic field. II. Multi-scale infrared analysis, 
        accepted for publication in J. Math. Sci. Univ. Tokyo,
	arXiv:1709.06714.
\bibitem{KP} S. G. Krantz and H. R. Parks, A primer of real analytic
	functions, 2nd edition, Birkh\"auser, 2002.
\bibitem{L} A. Lesniewski, Effective action for the $\text{Yukawa}_2$
	quantum field theory, Commun. Math. Phys. {\bf 108} (1987),
	437--467.   
\bibitem{M} V. Mastropietro, Mass generation in a fermionic model with
	finite range time dependent interactions, Commun. Math. Phys. {\bf
	269} (2007), 401--424.
\bibitem{Metal} B. Mera, C. Vlachou, N. Paunkovi\'c, V. R. Vieira and
	O. Viyuela, Dynamical phase transitions at finite temperature
	from fidelity and interferometric Loschmidt echo induced metrics, 
        Phys. Rev. B {\bf 97} (2018), 094110. 
\bibitem{SFS} N. Sedlmayr, M. Fleischhauer and J. Sirker, Fate of
	dynamical phase transitions at finite temperatures and in open
	systems, Phys. Rev. B {\bf 97} (2018), 045147. 
\bibitem{PS} W. de Siqueira Pedra and M. Salmhofer, Determinant bounds and the
	Matsubara UV problem of many-fermion systems,
	Commun. Math. Phys. {\bf 282} (2008), 797--818.
\bibitem{ZPHKBKGGM} J. Zhang, G. Pagano, P. W. Hess, A. Kyprianidis,
	P. Becker, H. Kaplan, A. V. Gorshkov, Z.-X. Gong and C. Monroe, 
        Observation of a many-body dynamical phase transition with a
	53-qubit quantum simulator, Nature {\bf 551} (2017), 601--604.
\bibitem{Z} A. A. Zvyagin, Dynamical quantum phase transitions (review
	article), Low. Temp. Phys. {\bf 42} (2016), 971. 
\end{thebibliography}
\end{document}